\documentclass[A4,10pt]{article}
\usepackage[margin=1in,dvips]{geometry}

\pdfoutput=1 
\usepackage[dvipdfmx]{graphicx}

\usepackage{color, soul}
\usepackage{amsmath, amsthm, slashed}
\usepackage{amssymb}

\usepackage{marvosym}
\usepackage{rotating}
\usepackage{mathrsfs}
\usepackage{upgreek}

\usepackage{verbatim}
\usepackage[affil-it]{authblk}
\usepackage{rotating}
\usepackage{accents}
\usepackage{harmony, slashed}
\usepackage{tikz}
\usepackage{dsfont}
\usepackage{makeidx}
\usepackage{bbm}

\usepackage{tocloft}

\usepackage{thmtools}
\usepackage{thm-restate}

\usepackage{comment}

\usepackage{enumerate}

\usepackage[colorlinks=true,linkcolor=black,citecolor=blue]{hyperref}

\usepackage{cleveref}

\newtheorem{theorem}{Theorem}[section]

\newtheorem*{conjecture*}{Conjecture}
\newtheorem*{theorem*}{Theorem}

\newtheorem{proposition}{Proposition}[subsection]

\newtheorem{corollary}[proposition]{Corollary}
\newtheorem*{corollary*}{Corollary}
\newtheorem{lemma}[proposition]{Lemma}
\newtheorem{remark}[proposition]{Remark}

\newtheorem{assumption}[proposition]{Assumption}

\numberwithin{equation}{section}

\newcommand\lesslessk{\overset{\scalebox{.5}{\mbox{\small $\ll$}} } {\mbox{\tiny $k$}} }

\newcommand{\musicchi}{
 \accentset{\scalebox{.45}{\mbox{\tiny \text{\Acht}}}}\chi}

\newcommand{\musicrho}{
 \accentset{\scalebox{.45}{\mbox{\tiny \text{\Acht}}}}\rho}

\newcommand{\Dk}{\mathfrak{D}^{\mathbf{k}}}

\newcommand{\Jpk}{\accentset{\scalebox{.6}{\mbox{\tiny $(p)$}}}{\underaccent{\scalebox{.8}{\mbox{\tiny $k$}}}{J}}}

\newcommand{\Jp}{\accentset{\scalebox{.6}{\mbox{\tiny $(p)$}}}{J}}

\newcommand{\Jzero}{\accentset{\scalebox{.6}{\mbox{\tiny $(0)$}}}{J}}

\newcommand{\Kzero}{\accentset{\scalebox{.6}{\mbox{\tiny $(0)$}}}{K}}

\newcommand{\Kpk}{\accentset{\scalebox{.6}{\mbox{\tiny $(p)$}}}{\underaccent{\scalebox{.8}{\mbox{\tiny $k$}}}{K}}}

\newcommand{\Kp}{\accentset{\scalebox{.6}{\mbox{\tiny $(p)$}}}{K}}

\newcommand{\Hap}{\accentset{\scalebox{.6}{\mbox{\tiny $(p)$}}}{H}}

\newcommand{\Fzerok}{\accentset{\scalebox{.6}{\mbox{\tiny $(0)$}}}{\underaccent{\scalebox{.8}{\mbox{\tiny $k$}}}{\mathcal{F}}}}

\newcommand{\Fpk}{\accentset{\scalebox{.6}{\mbox{\tiny $(p)$}}}{\underaccent{\scalebox{.8}{\mbox{\tiny $k$}}}{\mathcal{F}}}}

\newcommand{\Fpprimek}{\accentset{\scalebox{.6}{\mbox{\tiny $(p')$}}}{\underaccent{\scalebox{.8}{\mbox{\tiny $k$}}}{\mathcal{F}}}}

\newcommand{\Epk}{\accentset{\scalebox{.6}{\mbox{\tiny $(p)$}}}{\underaccent{\scalebox{.8}{\mbox{\tiny $k$}}}{\mathcal{E}}}}

\newcommand{\Ezerokprime}{\accentset{\scalebox{.6}{\mbox{\tiny $(0)$}}}{\underaccent{\scalebox{.8}{\mbox{\tiny $k'$}}}{\mathcal{E}}}}

\newcommand{\Xpk}{\accentset{\scalebox{.6}{\mbox{\tiny $(p)$}}}{\underaccent{\scalebox{.8}{\mbox{\tiny $k$}}}{\mathcal{X}}}}

\newcommand{\Epprimekprime}{\accentset{\scalebox{.6}{\mbox{\tiny $(p')$}}}{\underaccent{\scalebox{.8}{\mbox{\tiny {$k'$}}}}{\mathcal{E}}}}

\newcommand{\Epprimek}{\accentset{\scalebox{.6}{\mbox{\tiny $(p')$}}}{\underaccent{\scalebox{.8}{\mbox{\tiny {$k$}}}}{\mathcal{E}}}}

\newcommand{\Xpprimekprime}{\accentset{\scalebox{.6}{\mbox{\tiny $(p')$}}}{\underaccent{\scalebox{.8}{\mbox{\tiny {$k'$}}}}{\mathcal{X}}}}

\newcommand{\Efancypk}{\accentset{\scalebox{.6}{\mbox{\tiny $(p)$}}}{\underaccent{\scalebox{.8}{\mbox{\tiny $k$}}}{\mathfrak{E}}}}

\newcommand{\Etwominusdelk}{\accentset{\scalebox{.6}{\mbox{\tiny $(2-\delta)$}}}{\underaccent{\scalebox{.8}{\mbox{\tiny $k$}}}{\mathcal{E}}}}

\newcommand{\Xtwominusdelk}{\accentset{\scalebox{.6}{\mbox{\tiny $(2-\delta)$}}}{\underaccent{\scalebox{.8}{\mbox{\tiny $k$}}}{\mathcal{X}}}}

\newcommand{\Xtwominusdelkminusone}{\accentset{\scalebox{.6}{\mbox{\tiny $(2-\delta)$}}}{\underaccent{\scalebox{.8}{\mbox{\tiny $k\mkern-6mu -\mkern-6mu 1$}}}{\mathcal{X}}}}

\newcommand{\Xtwominusdelkminustwo}{\accentset{\scalebox{.6}{\mbox{\tiny $(2-\delta)$}}}{\underaccent{\scalebox{.8}{\mbox{\tiny $k\mkern-6mu -\mkern-6mu 2$}}}{\mathcal{X}}}}

\newcommand{\Xtwominusdelkminusthree}{\accentset{\scalebox{.6}{\mbox{\tiny $(2-\delta)$}}}{\underaccent{\scalebox{.8}{\mbox{\tiny $k\mkern-6mu -\mkern-6mu 3$}}}{\mathcal{X}}}}

\newcommand{\Edeltaminusonek}{\accentset{\scalebox{.6}{\mbox{\tiny $(\delta-1)$}}}{\underaccent{\scalebox{.8}{\mbox{\tiny $k$}}}{\mathcal{E}}}}

\newcommand{\Eminusone}{\underaccent{\scalebox{.8}{\mbox{\tiny $ -\mkern-3mu 1$}}}{\mathcal{E}}}

\newcommand{\Epminusonek}{\accentset{\scalebox{.6}{\mbox{\tiny $(p\mkern-6mu-\mkern-6mu1)$}}}{\underaccent{\scalebox{.8}{\mbox{\tiny $k$}}}{\mathcal{E}}}}

\newcommand{\Xtwominusdeltalesslessk}{\accentset{\scalebox{.6}{\mbox{\tiny $(2\mkern-6mu-\mkern-6mu\delta)$}}}{\underaccent{\scalebox{.8}{\mbox{\tiny $\ll \mkern-6mu k$}}}{\mathcal{X}}}}

\newcommand{\Eoneminusdelk}{\accentset{\scalebox{.6}{\mbox{\tiny $(1-\delta)$}}}{\underaccent{\scalebox{.8}{\mbox{\tiny $k$}}}{\mathcal{E}}}}

\newcommand{\Eoneminusdelkminusfour}{\accentset{\scalebox{.6}{\mbox{\tiny $(1-\delta)$}}}{\underaccent{\scalebox{.8}{\mbox{\tiny $k\mkern-6mu -\mkern-6mu 4$}}}{\mathcal{E}}}}

\newcommand{\Eonek}{\accentset{\scalebox{.6}{\mbox{\tiny $(1)$}}}{\underaccent{\scalebox{.8}{\mbox{\tiny $k$}}}{\mathcal{E}}}}
\newcommand{\Eonekprimeprime}{\accentset{\scalebox{.6}{\mbox{\tiny $(1)$}}}{\underaccent{\scalebox{.8}{\mbox{\tiny $k''$}}}{\mathcal{E}}}}
\newcommand{\Eonelesslessk}{\accentset{\scalebox{.6}{\mbox{\tiny $(1)$}}}{\underaccent{\scalebox{.8}{\mbox{\tiny $\ll \mkern-6mu k$}}}{\mathcal{E}}}}

\newcommand{\Xonek}{\accentset{\scalebox{.6}{\mbox{\tiny $(1)$}}}{\underaccent{\scalebox{.8}{\mbox{\tiny $k$}}}{\mathcal{X}}}}

\newcommand{\Epminusonekminusone}{\accentset{\scalebox{.6}{\mbox{\tiny $(p\mkern-6mu-\mkern-6mu1)$}}}{\underaccent{\scalebox{.8}{\mbox{\tiny $k\mkern-6mu -\mkern-6mu 1$}}}{\mathcal{E}}}}

\newcommand{\Eonekminusone}{\accentset{\scalebox{.6}{\mbox{\tiny $(1)$}}}{\underaccent{\scalebox{.8}{\mbox{\tiny $k\mkern-6mu -\mkern-6mu 1$}}}{\mathcal{E}}}}

\newcommand{\Edeltaminusonekminusone}{\accentset{\scalebox{.6}{\mbox{\tiny $(\delta-1)$}}}{\underaccent{\scalebox{.8}{\mbox{\tiny $k\mkern-6mu -\mkern-6mu 1$}}}{\mathcal{E}}}}

\newcommand{\Eonekminusfour}{\accentset{\scalebox{.6}{\mbox{\tiny $(1)$}}}{\underaccent{\scalebox{.8}{\mbox{\tiny $k\mkern-6mu -\mkern-6mu 4$}}}{\mathcal{E}}}}

\newcommand{\Xk}{\underaccent{\scalebox{.8}{\mbox{\tiny $k$}}}{\mathcal{X}}}

\newcommand{\Xkminusone}{\underaccent{\scalebox{.8}{\mbox{\tiny $k\mkern-6mu -\mkern-6mu 1$}}}{\mathcal{X}}}

\newcommand{\Xonekminusone}{\accentset{\scalebox{.6}{\mbox{\tiny $(1)$}}}{\underaccent{\scalebox{.8}{\mbox{\tiny $k\mkern-6mu -\mkern-6mu 1$}}}{\mathcal{X}}}}

\newcommand{\Xonekminusfour}{\accentset{\scalebox{.6}{\mbox{\tiny $(1)$}}}{\underaccent{\scalebox{.8}{\mbox{\tiny $k\mkern-6mu -\mkern-6mu 4$}}}{\mathcal{X}}}}

\newcommand{\Xonekminusfive}{\accentset{\scalebox{.6}{\mbox{\tiny $(1)$}}}{\underaccent{\scalebox{.8}{\mbox{\tiny $k\mkern-6mu -\mkern-6mu 5$}}}{\mathcal{X}}}}

\newcommand{\Epminusonekminustwo}{\accentset{\scalebox{.6}{\mbox{\tiny $(p\mkern-6mu-\mkern-6mu1)$}}}{\underaccent{\scalebox{.8}{\mbox{\tiny $k\mkern-6mu -\mkern-6mu 2$}}}{\mathcal{E}}}}

\newcommand{\Epkminusone}{\accentset{\scalebox{.6}{\mbox{\tiny $(p)$}}}{\underaccent{\scalebox{.8}{\mbox{\tiny $k\mkern-6mu -\mkern-6mu 1$}}}{\mathcal{E}}}}

\newcommand{\Xpkminusone}{\accentset{\scalebox{.6}{\mbox{\tiny $(p)$}}}{\underaccent{\scalebox{.8}{\mbox{\tiny $k\mkern-6mu -\mkern-6mu 1$}}}{\mathcal{X}}}}

\newcommand{\Xpkminustwo}{\accentset{\scalebox{.6}{\mbox{\tiny $(p)$}}}{\underaccent{\scalebox{.8}{\mbox{\tiny $k\mkern-6mu -\mkern-6mu 2$}}}{\mathcal{X}}}}

\newcommand{\Etwominusdelkminusone}{\accentset{\scalebox{.6}{\mbox{\tiny $(2-\delta)$}}}{\underaccent{\scalebox{.8}{\mbox{\tiny $k\mkern-6mu -\mkern-6mu 1$}}}{\mathcal{E}}}}
\newcommand{\Etwominusdelkminustwo}{\accentset{\scalebox{.6}{\mbox{\tiny $(2-\delta)$}}}{\underaccent{\scalebox{.8}{\mbox{\tiny $k\mkern-6mu -\mkern-6mu 2$}}}{\mathcal{E}}}}

\newcommand{\Etwominusdelkminusfour}{\accentset{\scalebox{.6}{\mbox{\tiny $(2-\delta)$}}}{\underaccent{\scalebox{.8}{\mbox{\tiny $k\mkern-6mu -\mkern-6mu 4$}}}{\mathcal{E}}}}

\newcommand{\Epkplusone}{\accentset{\scalebox{.6}{\mbox{\tiny $(p)$}}}{\underaccent{\scalebox{.8}{\mbox{\tiny $k\mkern-6mu +\mkern-6mu 1$}}}{\mathcal{E}}}}

\newcommand{\Ezerok}{\accentset{\scalebox{.6}{\mbox{\tiny $(0)$}}}{\underaccent{\scalebox{.8}{\mbox{\tiny $k$}}}{\mathcal{E}}}}

\newcommand{\Xzerok}{\accentset{\scalebox{.6}{\mbox{\tiny $(0)$}}}{\underaccent{\scalebox{.8}{\mbox{\tiny $k$}}}{\mathcal{X}}}}

\newcommand{\Xzeroplusk}{\accentset{\scalebox{.6}{\mbox{\tiny $(0+)$}}}{\underaccent{\scalebox{.8}{\mbox{\tiny $k$}}}{\mathcal{X}}}}

\newcommand{\Ezerokminusone}{\accentset{\scalebox{.6}{\mbox{\tiny $(0)$}}}
{\underaccent{\scalebox{.8}{\mbox{\tiny $k\mkern-6mu -\mkern-6mu 1$}}}{\mathcal{E}}}}
\newcommand{\Ezerokminustwo}{\accentset{\scalebox{.6}{\mbox{\tiny $(0)$}}}
{\underaccent{\scalebox{.8}{\mbox{\tiny $k\mkern-6mu -\mkern-6mu 2$}}}{\mathcal{E}}}}

\newcommand{\Ezerokminussix}{\accentset{\scalebox{.6}{\mbox{\tiny $(0)$}}}
{\underaccent{\scalebox{.8}{\mbox{\tiny $k\mkern-6mu -\mkern-6mu 6$}}}{\mathcal{E}}}}

\newcommand{\Xzerokminusone}{\accentset{\scalebox{.6}{\mbox{\tiny $(0)$}}}
{\underaccent{\scalebox{.8}{\mbox{\tiny $k\mkern-6mu -\mkern-6mu 1$}}}{\mathcal{X}}}}

\newcommand{\Xzerokminustwo}{\accentset{\scalebox{.6}{\mbox{\tiny $(0)$}}}
{\underaccent{\scalebox{.8}{\mbox{\tiny $k\mkern-6mu -\mkern-6mu 2$}}}{\mathcal{X}}}}

\newcommand{\Xzerokminusthree}{\accentset{\scalebox{.6}{\mbox{\tiny $(0)$}}}
{\underaccent{\scalebox{.8}{\mbox{\tiny $k\mkern-6mu -\mkern-6mu 3$}}}{\mathcal{X}}}}

\newcommand{\Xzerokminussix}{\accentset{\scalebox{.6}{\mbox{\tiny $(0)$}}}
{\underaccent{\scalebox{.8}{\mbox{\tiny $k\mkern-6mu -\mkern-6mu 6$}}}{\mathcal{X}}}}

\newcommand{\Xzerokminusseven}{\accentset{\scalebox{.6}{\mbox{\tiny $(0)$}}}
{\underaccent{\scalebox{.8}{\mbox{\tiny $k\mkern-6mu -\mkern-6mu 7$}}}{\mathcal{X}}}}

\newcommand{\Ezerolesslessk}{\accentset{\scalebox{.6}{\mbox{\tiny $(0)$}}}{\underaccent{\scalebox{.8}{\mbox{\tiny 
$\ll\mkern-6mu k$}}}{\mathcal{E}}}}

\newcommand{\Xzerolesslessk}{\accentset{\scalebox{.6}{\mbox{\tiny $(0)$}}}{\underaccent{\scalebox{.8}{\mbox{\tiny 
$\ll\mkern-6mu k$}}}{\mathcal{X}}}}

\newcommand{\Xzeropluslesslessk}{\accentset{\scalebox{.6}{\mbox{\tiny $(0+)$}}}{\underaccent{\scalebox{.8}{\mbox{\tiny 
$\ll\mkern-6mu k$}}}{\mathcal{X}}}}

\newcommand{\Xonelesslessk}{\accentset{\scalebox{.6}{\mbox{\tiny $(1)$}}}{\underaccent{\scalebox{.8}{\mbox{\tiny 
$\ll\mkern-6mu k$}}}{\mathcal{X}}}}

\newcommand{\Xplesslessk}{\accentset{\scalebox{.6}{\mbox{\tiny $(p)$}}}{\underaccent{\scalebox{.8}{\mbox{\tiny 
$\ll\mkern-6mu k$}}}{\mathcal{X}}}}

\newcommand{\Ezerominusoneminusdeltak}{\accentset{\scalebox{.6}{\mbox{\tiny $(-1\mkern-6mu-\mkern-6mu\delta)$}}}{\underaccent{\scalebox{.8}{\mbox{\tiny $k$}}}{\mathcal{E}}}}

\newcommand{\Ezerominusthreeminusdeltakminusone}{\accentset{\scalebox{.6}{\mbox{\tiny $(-3\mkern-6mu-\mkern-6mu\delta)$}}}{\underaccent{\scalebox{.8}{\mbox{\tiny $k\mkern-6mu -\mkern-6mu 1$}}}{\mathcal{E}}}}

\newcommand{\Ezerominusoneminusdeltakminusone}{\accentset{\scalebox{.6}{\mbox{\tiny $(-1\mkern-6mu-\mkern-6mu\delta)$}}}{\underaccent{\scalebox{.8}{\mbox{\tiny $k\mkern-6mu -\mkern-6mu 1$}}}{\mathcal{E}}}}

\newcommand{\Ezerominusoneminusdeltakminustwo}{\accentset{\scalebox{.6}{\mbox{\tiny $(-1\mkern-6mu-\mkern-6mu\delta)$}}}{\underaccent{\scalebox{.8}{\mbox{\tiny $k\mkern-6mu -\mkern-6mu 2$}}}{\mathcal{E}}}}

\newcommand{\Ezerominusoneminusdeltalesslessk}{\accentset{\scalebox{.6}{\mbox{\tiny $(-1\mkern-6mu-\mkern-6mu\delta)$}}}{\underaccent{\scalebox{.8}{\mbox{\tiny $\ll\mkern-6mu k$}}}{\mathcal{E}}}}

\newcommand{\nablaslash}{\slashed{\nabla}}

\DeclareMathAlphabet\mathbfcal{OMS}{cmsy}{b}{n}

\makeatletter \@addtoreset{equation}{section}  \makeatother

\title{Quasilinear wave equations on  Kerr  black holes\\  in the full subextremal range $|a|<M$}

\author[$\ddag$ $\S$]{Mihalis Dafermos}
\author[$\dag$ $*$]{Gustav Holzegel}
\author[$\ddag$]{Igor Rodnianski}
\author[$\dag$]{Martin Taylor}

\affil[$\dag$]{\small Imperial College London,
Department of Mathematics,
South~Kensington~Campus,~London~SW7~2AZ,~United~Kingdom\vskip.2pc \ }

\affil[$\ddag$]{\small Princeton University, Department of Mathematics, Fine~Hall,~Washington~Road,~Princeton,~NJ~08544,~United~States~of~America\vskip.2pc \ }

\affil[$\S$]{\small University of Cambridge, Department of Pure Mathematics and Mathematical
Statistics, Wilberforce~Road,~Cambridge~CB3~0WA,~United~Kingdom\vskip.2pc \ }

\affil[$*$]{\small Universit\"at M\"unster,~Mathematisches~Institut, ~{\ \ \ \ \ \ \ \ \ }~Einsteinstrasse~62~48149~M\"unster,~Bundesrepublik~Deutschland~{\ \ \ }}

\date{4 October 2024}

\begin{document}

\maketitle

\begin{abstract}
 We prove global existence, 
 boundedness and decay for small data solutions $\psi$ to a general class of
quasilinear wave equations on 
  Kerr black hole backgrounds in the full sub-extremal range ${|a|<M}$. The method extends our previous~\cite{DHRT22}, which considered such equations
on a wide class of background spacetimes, 
including  Kerr, but restricted in that case to the very slowly rotating regime $|a|\ll M$ (which
may  be treated simply as a perturbation of Schwarzschild $a=0$).
 To handle the general $|a|<M$ case,
 our present proof is based on two ingredients:
 (i) the linear inhomogeneous
  estimates on Kerr backgrounds proven in~\cite{partiii}, further refined however 
 in order to gain a derivative in elliptic frequency regimes,
  and (ii) the existence of
   appropriate physical space currents satisfying degenerate coercivity properties,
  but which now must be tailored to a finite number of wave packets
  defined by suitable frequency projection.
  The above ingredients can be thought of as adaptations of the two basic 
  ingredients of~\cite{DHRT22}, exploiting however the peculiarities of the
  Kerr geometry.
 The novel frequency decomposition in (ii),
 inspired by the boundedness arguments 
of~\cite{DafRodKerr, partiii},
  is defined using only azimuthal and stationary frequencies, and serves both to separate the superradiant and non-superradiant parts of
the solution and to localise trapping to small regions of spacetime. 
 The strengthened (i), on the other hand, allows us to relax
  the  required coercivity properties of our wave-packet dependent  currents, so as in particular to accept top order
  errors provided that they are 
localised to the elliptic frequency regime. These error terms are 
  analysed with the help of the full Carter separation.
  \end{abstract}

  \tableofcontents

\section{Introduction}
Consider the $4$-dimensional Kerr manifold $(\mathcal{M},g_{a,M})$, for general
subextremal
parameters $|a|<M$~\cite{Kerr}.
In this paper we study quasilinear wave equations of the form
\begin{equation}
\label{theequationzero}
\Box_{g(\psi,x)}\psi = N(\partial\psi, \psi, x)
\end{equation}
 where $g(\psi,x)$ and $N(\partial\psi,\psi,x)$ are
appropriate nonlinear terms and 
where 
\begin{equation}
\label{linearisestokerr}
g(0,x)=g_{a,M}(x).
\end{equation}
We will consider the extended exterior region of Kerr (bounded by a future spacelike
hypersurface $\mathcal{S}$ slightly inside the event horizon $\mathcal{H}^+$) and define initial 
data $(\uppsi,\uppsi')$ for~\eqref{theequationzero}
on a suitable hypersurface $\Sigma_0$ which coincides with an outgoing null cone for $r\ge R$ asymptoting to null infinity $\mathcal{I}^+$.
With respect to the Kerr metric, $\Sigma_0$ is a Cauchy hypersurface for its future in this
extended region,
which is foliated by the stationary-time translates $\Sigma(\tau)$ of $\Sigma_0$, $\tau\in [0,\infty)$.
Refer to Figure~\ref{introdiagfigurelargeathe}.
\begin{figure}
\centering{
\def\svgwidth{15pc}
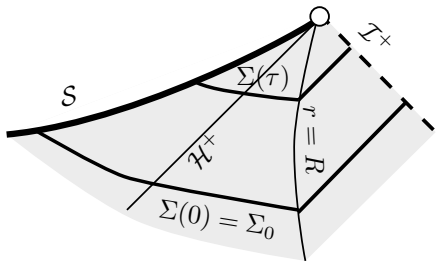}
\caption{The  extended Kerr spacetime $(\mathcal{M},g_{a,M})$ and the foliation $\Sigma(\tau)$}\label{introdiagfigurelargeathe}
\end{figure}

The main result of the present paper is a proof of global existence, together with orbital and asymptotic stability results,  for
 small-data solutions of~\eqref{theequationzero}:
\begin{theorem*}
For equations~\eqref{theequationzero}  on  
the extended exterior region of an arbitrary subextremal ($|a|<M$) Kerr spacetime
$(\mathcal{M}, g_{a,M}(x))$ satisfying~\eqref{linearisestokerr} and under appropriate   
assumptions
on the quasilinear term~$g(\psi,x)$ and semilinear term $N(\partial\psi,\psi,x)$,
we have 
\begin{itemize}
\item
{\bf Global existence of small data solutions.} Solutions $\psi$ arising from smooth
 initial data  $(\uppsi,\uppsi')$ 
on $\Sigma_0$, sufficiently small 
as measured in a suitable $r^p$-weighted Sobolev energy, exist globally in $\cup_{\tau\in[0,\infty)}\Sigma(\tau)$.
\item
{\bf Orbital stability.} Under the above assumptions,
the above-mentioned $r^p$-weighted energy flux through the time translates $\Sigma(\tau)$ 
is uniformly bounded by a constant
times its initial value on $\Sigma_0$.
\item
{\bf Asymptotic stability.} Under the above assumptions, a suitable  lower order unweighted
energy flux through $\Sigma(\tau)$ decays inverse polynomially in $\tau$
(implying also pointwise inverse polynomial decay for~$\psi$).
\end{itemize}
\end{theorem*}

The precise statement corresponding to the above 
is given as Theorem~\ref{largeamaintheorem} of Section~\ref{largeamaintheoremsection}.
We note already that the precise assumptions  on the nonlinear terms $g(\psi,x)$
and $N(\psi,x)$ in our main theorem above will be exactly as in~\cite{DHRT22} (see 
the discussion immediately below),
with $g(0,x)$ specialised however to the Kerr metric~\eqref{linearisestokerr}.

\paragraph{The general framework of~\cite{DHRT22} and the $|a|\ll M$ case.}
We recall that our previous~\cite{DHRT22} provided stability results analogous to the above
theorem
for equations~\eqref{theequationzero} under a wide class of assumptions on $g(\psi,x)$ and
$N(\psi,x)$.  In that work, $g_0:=g(0,x)$ is always assumed to be a background  asymptotically
flat stationary metric satisfying some basic geometric restrictions (allowing in particular for the
presence of event horizons, trapped null geodesics and ergoregions) and
the semilinear term $N(\partial\psi,\psi,x)$ is assumed to satisfy  $N(0,0,x)=0$ together with a generalisation of the null condition~\cite{KlNull},
while the quasilinear term is assumed for convenience to satisfy $g(\psi,x)=g_0$ for $x\ge R/2$ for sufficiently large~$R$. 

Our theorem above is not included in~\cite{DHRT22} for general $|a|<M$, however, because the results of~\cite{DHRT22} depended in addition on being able to show two ``black box''
assumptions on the stationary metric~$g_0$:
\begin{enumerate}
\item
a quantitative boundedness and (possibly degenerate) integrated local energy decay statement
for the inhomogeneous equation $\Box_{g_0}\psi =F$
\item
the existence of an auxiliary physical-space top order energy current for $\Box_{g_0}$, with bulk and boundary  coercivity properties 
up to specific allowed error terms.
\end{enumerate}
The significance of having a physical space energy current is that simply by re-expressing it in covariant form it
can then be applied
\emph{directly} to the quasilinear wave operator $\Box_{g(\psi,x)}$, the current's coercivity properties
immediately carrying over, up to allowed error terms. It is this fact, in conjunction with the decay provided by~1., which allows to prove top order energy boundedness estimates
for solutions of~\eqref{theequationzero} without loss of derivatives and indeed 
infer global existence
and stability. See already Section~\ref{reviewsection} for a review of the proof of~\cite{DHRT22},
in particular also the novel consecutive spacetime slab iteration method introduced there, which allows  the core of the analysis to be reduced to time-translation invariant estimates and identities above, rendering
the entire argument in the spirit of the $r^p$ method~\cite{DafRodnew}.

Turning to the particular Kerr case~\eqref{linearisestokerr} considered in the present paper, the black box assumption~1.~indeed holds~\cite{partiii} for the
full subextremal range $|a|<M$. The auxiliary identity 2., however, has only
been obtained for the very slowly rotating case $|a|\ll M$ (treating it  essentially as an arbitrary stationary perturbation of
Schwarzschild $a=0$), and thus, only in this regime do the
stability
results of~\cite{DHRT22} apply to the general quasilinear equation~\eqref{theequationzero} 
with~\eqref{linearisestokerr}. 
  (For the simpler  \emph{semilinear} case of~\eqref{theequationzero}, where 
$g(\psi,x)=g_{a,M}(x)$ \emph{identically}, the auxiliary 2.~is not in fact necessary and the results of~\cite{DHRT22}
already apply for the full subextremal range $|a|<M$.)

\paragraph{The general subextremal case $|a|<M$: wave packet decomposition distinguishing superradiance and
localising trapping.}
To handle the general quasilinear~\eqref{theequationzero} with~\eqref{linearisestokerr}  for
the full subextremal Kerr range $|a|<M$, the present paper 
circumvents the difficulty of showing the required 
general physical space identity 2.~by applying    a mild
frequency decomposition on~$\psi$, partitioning the solution 
into finitely many ``wave packets'' $\psi_n$, $n=0,\ldots, N$, for each of which
one can obtain a suitable physical space divergence identity
 with weak coercivity properties sufficient
for our argument. 
This ``frequency localised'' version of~2.~is inspired by the original energy boundedness arguments of~\cite{DafRodKerr, partiii} and by the fundamental property that superradiance and trapping are disjoint in frequency space, i.e.~that wave packets localised to the superradiant regime of frequencies
\begin{equation}
\label{superradianthere}
 0< \frac{\omega}{am} <\frac{1}{2Mr_+}
\end{equation}
are not subject to trapping.  (Here $r_+$ denotes the Boyer--Lindquist coordinate value
of the event horizon $\mathcal{H}^+$.) For the non-superradiant wave packets on the
other hand, by making the decomposition suitably fine we can ensure that trapping
for each such wave packet $\psi_n$ is localised to an $n$-dependent $r$-interval which 
is sufficiently small to admit a timelike Killing field $\partial_{t^*}+\alpha_n \partial_{\phi^*}$.
Our $n$-dependent currents realising the analogue of property 2.~exploit these properties
to satisfy coercivity up to allowed error terms.
Let us note already, however, that these error terms  still cannot be controlled by the integrated decay statement~1.~of~\cite{partiii} as stated there, but require a further  improved linear
homogeneous estimate
which follows by the methods
of~\cite{partiii} but which we must also obtain in the present paper.
We give more details about the nature of the decomposition, the coercive 
properties of the  $n$-dependent currents and the necessary improved linear estimate in Section~\ref{freqloccoercurintro} below.

We emphasise that our 
mild frequency decomposition \emph{only uses the frequencies $\omega$
and $m$ associated to the Fourier transform
with respect to the Kerr star coordinates $t^*$ and $\phi^*$}, and \underline{not}  Carter's full separation~\cite{carter1968hamilton} 
which was used 
in~\cite{partiii} to prove the integrated local energy statement~1.~itself.
Thus, nonlinear commutation errors introduced by applying this decomposition to $\psi$ 
are completely elementary to handle using the usual pseudodifferential
calculus. 
Carter's separation will be used however in the present paper 
both to prove the necessary improved version of the estimate
of~\cite{partiii} and to infer the weak coercivity  property
of our currents.

\paragraph{Background and outlook.}
We refer the reader to~\cite{DHRT22} for further discussion of the stability problem for 
equations of the form~\eqref{theequationzero} 
and the modern history of its study without symmetry, starting from the classical works~\cite{KlNull,christonull}. Let us, however, explicitly point
the reader to 
various previous approaches to quasilinear problems in the absence of symmetry in the Schwarzschild $a=0$ and very slowly rotating $|a|\ll M$ Kerr case~\cite{MR3082240, lindblad2016global, lindtoh, lindblad2022weak, Pasqualotto:2017rkh, giorgietal} (all based on previous linear work) discussed in~\cite{DHRT22} in more detail. We also refer to~\cite{DHRT22} for
comparison with the $\Lambda>0$ case, but specifically point the reader 
to~\cite{hintzvasyglobal} as well as the more recent~\cite{mavrogiannis, fang}.
(In brief, as should be clear already from the discussion in~\cite{mavrogiannis}, 
the requisite top order estimate here arising from 2.~can be shown in the $\Lambda>0$ case 
using Cauchy stability alone, provided that black box estimate 1.~is improved to a ``relatively'' non-degenerate estimate of the form proven in~\cite{holzegel2020note, mavrogiannis2023morawetz}. Thus, the difficulties addressed in 
the present paper are in fact entirely absent from the $\Lambda>0$ case.)
Finally, we note that in view of~\cite{civinthesis}, we expect
that the constructions of the present paper carry over essentially unchanged 
to the study of~\eqref{theequationzero} on 
Kerr--Newman backgrounds $g_{a, M, Q}$ (the charged analogue of the Kerr metric)
for general subextremal parameters  $a^2+Q^2<M^2$.
For more on Kerr--Newman, we refer the reader to~\cite{giorgi2021carter}.
For a very recent study of the linear scalar wave equation on  backgrounds settling
down to Kerr sufficiently fast, see~\cite{ma2024energymorawetzestimateswaveequation}.

As described in~\cite{DHRT22}, 
one of the motivations of the present work is to develop a setup allowing one 
to \emph{directly} promote linear stability results  on fixed Kerr black hole
spacetimes, independently of how exactly these are proven, 
to the statement of the nonlinear stability 
of the  spacetimes themselves as solutions to the Einstein vacuum equations. 
In that context, our approach to the linear theory, initiated in~\cite{holzstabofschw}, centres on the extremal components 
of the Newman--Penrose formalism, which decouple from the full system and
satisfy a generalisation of the linear wave equation called the Teukolsky equation
and completely parametrise the gauge invariant degrees of freedom. In our approach,
the Teukolsky equation is in turn estimated by first passing to a new
quantity satisfying the better behaved Regge--Wheeler 
equation, inspired by~\cite{Chandraschw}. 
In the nonlinear
setting (cf.~\cite{dhrtplus}), the analogous quantities
satisfy ``almost decoupled'' wave equations which can be handled by similar techniques
to~\eqref{theequationzero}. Indeed, the analogue of the black box estimate 
1.~for the Teukolsky equation on Kerr has been
shown recently for the full subextremal range $|a|<M$
in the breakthrough result~\cite{RitaShlap, RitaShlap2}.  Moreover, 
not only does~\cite{RitaShlap, RitaShlap2}  provide the black box assumption~1.,
but, just as the current constructions of~\cite{partiii} do in the present paper, the currents
introduced in~\cite{RitaShlap, RitaShlap2}   would appear to provide precisely the key to constructing
the analogue of the necessary auxiliary $n$-dependent currents giving our new version of 2.
It is  reasonable thus to  expect that the characteristic ``quasilinear'' difficulties of the stability problem for
the Einstein equations  can indeed  be entirely addressed by combining our method with the results of~\cite{RitaShlap, RitaShlap2}.
This will appear elsewhere.

To complete our understanding of quasilinear wave equations~\eqref{theequationzero} on Kerr black hole exterior backgrounds, it remains only to elucidate the extremal case $|a|=M$. Here, even the linear theory is not completely understood, and it is subject to the Aretakis instability~\cite{Aretakis2} already in axisymmetry, and to worse instabilities for higher azimuthal modes~\cite{gajicazimuthal}.
Mode stability, on the other hand, has indeed been shown in~\cite{ritamode}. 
In the simpler extremal  Reissner--Nordstr\"om case, semilinear equations have been handled 
in~\cite{nonlinearextreme}, under an additional structural assumption on $N$ at the horizon which appears to be necessary. For conjectures concerning the nonlinear stability of extremal black
holes under the Einstein equations, see~\cite{dhrtplus} (Conjecture IV.2) and 
Section~1.4 of~\cite{KU23}.

\paragraph{Outline of the introduction.}
In the rest of this introduction, we shall flesh out the above brief discussion of our proof with some more details.
We begin in Section~\ref{reviewsection} with a  review of~\cite{DHRT22}. In Section~\ref{freqloccoercurintro},
we then introduce the main new elements of our paper (the refinement of the integrated
local energy decay statement of~\cite{partiii} and the wave-packet localised currents which replace assumption 
2.~of~\cite{DHRT22} in $|a|<M$) and how to obtain from these the fundamental top-order estimate. Finally, in Section~\ref{outlinesection} we shall give an outline of the body of the paper.

\subsection{Review of~\cite{DHRT22}: the two black-box assumptions, the fundamental estimates on ${\rm L}$-slabs and the dyadic iteration scheme}
\label{reviewsection}

We review in this section the general scheme of~\cite{DHRT22} applied to the Kerr metric $g_{a,M}$. 
(The reader may wish to consult also the introduction of~\cite{DHRT22} for a more detailed overview.)

As discussed already above, the work~\cite{DHRT22} was based
on two ``black box'' assumptions for the linear operator $\Box_{g_0}$. These will be reviewed in Sections~\ref{reviewoffirstassumption} 
and~\ref{reviewofsecondassumption} below, including their extensions to higher order,
$r^p$-weighted estimates and coercive identities. Under the assumption of appropriate null structure for
the semilinear terms $N(\partial\psi, \psi,x)$, this  led to a hierarchy of time-translation invariant estimates,
reviewed in Section~\ref{estimatehierarchintroold}, 
for solutions now of the nonlinear equation~\eqref{theequationzero} on a spacetime slab of 
time-length ${\rm L}$, estimates which in particular
``close'' both with respect to their order  and with respect to their $r^p$ weights.
Finally, this estimate hierarchy leads immediately to  global existence and stability 
by a pigeonhole argument and suitable iteration in consecutive spacetime slabs, briefly discussed
in Section~\ref{dyadicintro}.

\subsubsection{Assumption~1.: integrated local energy decay}
\label{reviewoffirstassumption}
We give here a simplified discussion of the black box assumption~1.~of~\cite{DHRT22},
specialised to Kerr (in which case the assumption was indeed proven in~\cite{partiii} for the
entire subextremal regime $|a|<M$).

Let us first  introduce some notation for energies: Given a suitably regular function $\psi$, 
we define
\begin{equation}
\label{tocontroltheintegrand}
\mathcal{E}(\tau)[\psi]: = \int_{\Sigma(\tau)} |L\psi|^2+\iota_{r\le R} |\underline L\psi|^2 + |\nablaslash \psi|^2  + r^{-2}|\psi|^2.
\end{equation}
Here  $L$ and $\underline{L}$ are  null vector fields
in the region $r\ge R$ tangent to and transverse to $\Sigma(\tau)\cap \{r\ge R\}$, respectively,
while $|\nablaslash\psi|^2$ is comparable to the induced norm squared of the 
induced gradient on spheres
$\{r=c\}\cap \Sigma(\tau)$. The definition of $L$, $\underline L$ and $|\nablaslash \psi|$ is then extended to $r\le R$ in such a way so as
for  $|L\psi|^2+\iota_{r\le R} |\underline L\psi|^2 + |\nablaslash \psi|^2 $ to be 
a coercive expression
in first derivatives. (See already Section~\ref{commutationcommutation} for details.) 

The expression~\eqref{tocontroltheintegrand} 
should be thought of as the flux of energy through $\Sigma(\tau)$. 
(Note that because
$\Sigma(\tau)$ is outgoing null when restricted to $r\ge R$, the ingoing null derivative $\underline{L}\psi$
does not appear in the flux in this region, 
hence the presence of the indicator function $\iota_{r \le R}$.)

We also define the following energy quantities appearing naturally in bulk integrals:
We define the non-degenerate
\begin{equation}
\label{alsocomparetonondeg}
\mathcal{E}'(\tau)[\psi]: = \int_{\Sigma(\tau)}   r^{-1-\delta} ( |L\psi|^2+ |\underline L\psi|^2 + |\nablaslash \psi|^2)  
\end{equation}
and the 
 degenerate 
 \begin{equation}
\label{alsocompareto}
{}^\chi \mathcal{E}'(\tau)[\psi]: = \int_{\Sigma(\tau)} \chi   r^{-1-\delta} ( |L\psi|^2+ |\underline L\psi|^2 + |\nablaslash \psi|^2)  ,
\end{equation}
where $\chi(r)\ge 0$ is a nonnegative  ``degeneration function'' which in the Kerr case vanishes identically
in some spacetime region $\mathcal{D}$,   
while $\chi=1$ in $r_0 \le r\le r_2$ and $r\ge R/4$, for an $r_2>r_+$. The region $\mathcal{D}$ 
is such that all future trapped null geodesics eventually enter $\mathcal{D}$. 
In particular, $\chi$ vanishes at all $r$-values to which trapped null geodesics asymptote.
Finally, we define  the lower-order non-degenerate quantity:
\begin{equation}
\label{noteherelowerlargea}
\Eminusone'(\tau)[\psi] :=\int_{\Sigma(\tau)} r^{-3-\delta}\psi^2.
\end{equation}

The first ``black box'' assumption of~\cite{DHRT22} on the operator $\Box_{g_0}$, proved
in~\cite{partiii} for $g_0=g_{a,M}$, 
was  the validity of an integrated local energy estimate
\begin{equation}
\label{simplifiedestimate}
\mathcal{E}(\tau_1)[\psi] +
 \int_{\tau_0}^{\tau_1}( \,{}^\chi\mathcal{E}'(\tau') [\psi] + \Eminusone'(\tau')[\psi]  )d\tau' 
\lesssim \mathcal{E}(\tau_0)[\psi] + \int_{\mathcal{R}(\tau_0,\tau_1)}
| \partial \psi  \cdot \Box_{g_0}\psi| +| \Box_{g_0}\psi|^2,
\end{equation}
for arbitrary $\tau_0\le \tau_1$.
The final integral on the right hand side of~\eqref{simplifiedestimate} is a spacetime integral in the region $\mathcal{R}(\tau_0,\tau_1)$
bounded by $\Sigma(\tau_0)$ and $\Sigma(\tau_1)$. In the above, we have  left the 
inhomogeneous term $\partial \psi\cdot \Box_{g_0}\psi$ in schematic form (there are in fact $r$-weights associated to this, and moreover we may allow for a slightly weaker expression in the near region $r\le R$: see already the inhomogeneous terms in Theorem~\ref{blackboxoneforkerr} (specialised to the
case  $k=0$) for the precise form of the right hand side of the inhomogeneous estimate required).  

It was further shown in~\cite{DHRT22} that, given~\eqref{simplifiedestimate} and general assumptions on $g_0$ (which obtain in the Kerr case in the full subextremal range $|a|<M$),
one may extend~\eqref{simplifiedestimate} to an estimate for a $k+1$'th order $r^p$-weighted
energy~$\Epk(\tau)[\psi]$ 
(defined by applying an $r^p$-weighted version of~\eqref{tocontroltheintegrand} 
to $\mathfrak{D}^{\bf k}\psi$ in place of $\psi$, where 
 $\mathfrak{D}^{\bf k}\psi$  is a string of commutation operators of order $k=|{\bf k}|$; see already Section~\ref{commutationcommutation})
\begin{equation}
\label{simplifiedestimatepweighthigherorderlargea}
 \Epk(\tau_1)[\psi] + \int_{\tau_0}^{\tau_1} (\, {}^\chi\Epminusonek'(\tau') [\psi] + \Epminusonekminusone'(\tau') [\psi] )d\tau' \lesssim \Epk(\tau_0)[\psi]+\sum_{|{\bf k}|\le k} \int_{\mathcal{R}(\tau_0,\tau_1)}
|\accentset{\scalebox{.6}{\mbox{\tiny $(p)$}}} \partial \mathfrak{D}^{\bf k}\psi  \cdot \mathfrak{D}^{\bf k}  \Box_{g_0}\psi| +\ldots .
\end{equation}
The quantities $ {}^\chi\Epminusonek'(\tau') [\psi] $ and $\Epminusonekminusone'(\tau') [\psi] )$
represent
higher order $r^p$-weighted versions of~\eqref{alsocompareto} and~\eqref{alsocomparetonondeg}, respectively, where $k$ labels the number of commutations.
On the right hand side above, we have denoted only one of the resulting homogeneous terms---and this schematically,
with the notation $\accentset{\scalebox{.6}{\mbox{\tiny $(p)$}}} \partial $ invoking the presence
of appropriate $r^p$ weights!---to indicate the order
of what will turn out to be the most dangerous term later. See already Section~\ref{reviewedenergynotations} for the definition of all energies and 
 Theorem~\ref{blackboxoneforkerr} for the precise full estimate.

For a more detailed discussion of the $r^p$ weights appearing in the definitions, see
the introduction of~\cite{DHRT22} 
or see already Section~\ref{reviewedenergynotations} of the present paper. Here we note only
the fundamental
hierarchical relation of the $r^p$ method between bulk and boundary energies, namely
\begin{equation}
\label{downwithhierarchy}
  \Epminusonek' \gtrsim \Epminusonek  {\rm\ for\ } p\ge 1+\delta, \qquad   \Epminusonek' \gtrsim \Ezerok  {\rm\ for\ }p\ge 1 .
\end{equation}
It is this which allows one to extract decay in $\tau$ for the lower order unweighted
energy~\eqref{tocontroltheintegrand} in the case of the homogeneous
wave equation $\Box_{g_0}\psi=0$ via application of the pigeonhole principle and dyadic iteration 
(see the original~\cite{DafRodnew}).

We note already  that if the estimate~\eqref{simplifiedestimatepweighthigherorderlargea} is to be
applied directly to solutions $\psi$ of the quasilinear equation~\eqref{theequationzero},
we must express $\Box_{g_0}\psi= (\Box_{g_0}-\Box_g ) \psi +N(\partial\psi,\psi,x)$.
One sees then that the final displayed term  on the right hand side of~\eqref{simplifiedestimatepweighthigherorderlargea} 
is indeed nonlinear (thus ``small'') but of higher order than
the left hand side (order
$k+2$ in $\psi$ as opposed to order $k+1$). It follows that the estimate~\eqref{simplifiedestimatepweighthigherorderlargea} ``loses
derivatives'' (see already equation~\eqref{secondhigh}) 
and cannot be used on its own to control solutions of~\eqref{theequationzero}. (In contrast, in the semilinear case of~\eqref{theequationzero} (where $g(\psi,x):=g_{a,M}(x)$ identically), estimate~\eqref{simplifiedestimatepweighthigherorderlargea} suffices
for global stability results,
and this is why the results of~\cite{DHRT22} already apply in that case
to the full subextremal range
$|a|<M$.) 
It is overcoming this
 loss of derivative which motivates the second assumption, described immediately below.

\subsubsection{Assumption~2.: an auxiliary top order energy current}
\label{reviewofsecondassumption}

The key to an estimate which does not lose derivatives in the quasilinear setting of
equation~\eqref{theequationzero} is
to exploit covariant energy identities.
For a given metric $g$, the general covariant physical space energy identity takes the form
\begin{equation}
\label{genidentityintro}
\nabla^\mu {J_g}_\mu [\psi] = K_g [\psi] +{\rm Re} (H_g[\psi] \overline{\Box_{g}\psi}),
\end{equation}
where 
\begin{equation}
\label{labelthecurrents}
{J_g}_\mu[\psi],\, \,  K_g[\psi] ,  \, \, H_g[\psi]
\end{equation}
are energy 
currents depending covariantly on the metric $g$ and on the $1$-jet of $\psi$ (see already Section~\ref{Covariantenergyidentitiessec}).

The black box assumption 2.~of~\cite{DHRT22} (which we recall
was shown there  for $g_0=g_{a,M}$ in the very slowly rotating regime $|a|\ll M$ but is 
\underline{not} known 
to hold in the general $|a|<M$ case!)~was the existence of a set of time-translation invariant currents~\eqref{labelthecurrents} 
satisfying~\eqref{genidentityintro} with $g=g_0$ 
and enjoying the following coercivity assumptions:
\begin{itemize}
\item
The integrand ${J_{g_0}}[\psi] \cdot {\rm n}_{\mathcal{S}}$ arising in the boundary terms
on $\mathcal{S}$  from integration of~\eqref{genidentityintro} in a spacetime domain
\begin{figure}
\centering{
\def\svgwidth{15pc}
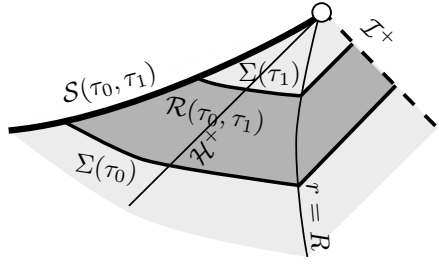}
\caption{The region $\mathcal{R}(\tau_0,\tau_1)$}\label{introdiagfigurelargea}
\end{figure}
(see Figure~\ref{introdiagfigurelargea})
is pointwise coercive.
The asymptotic integrand ${J_{g_0}}[\psi] \cdot {\rm n}_{\mathcal{I}^+}$ is likewise
nonnegative. 
\item 
The integrand ${J_{g_0}}[\psi] \cdot {\rm n}_{\Sigma(\tau)}$ arising in the boundary terms
on $\Sigma(\tau)$ is pointwise coercive
and controls the integrand of~\eqref{tocontroltheintegrand}.
Setting
\begin{equation}
\label{fromexactfluxes}
\mathfrak{E}(\tau)[\psi]:= \int_{\Sigma(\tau)} {J_{g_0}}_\mu[\psi]  \cdot {\rm n}_{\Sigma(\tau)},
\end{equation}
we have thus $\mathfrak{E}(\tau)[\psi] \sim \mathcal{E}(\tau)[\psi]$.
\item
The bulk integrand $K_{g_0}[\psi]$ is pointwise coercive near $\mathcal{S}$ and for large $r$.
It is not necessarily globally nonnegative, however, but satisfies
\begin{equation}
\label{partialcoercivit}
-K_{g_0}[\psi]  \lesssim \iota_{\{\chi= 1\}} |\psi|^2,
\end{equation}
where $\chi$ is the degeneration function appearing in assumption~1.,  and $\iota_S$ denotes
the indicator function of the set $S$. We thus may write
\begin{equation}
\label{wemaythuswrite}
\int_{\mathcal{R}(\tau_0,\tau_1)} K_{g_0}[\psi] \gtrsim \int_{\tau_0}^{\tau_1} {}^\rho \mathcal{E}' (\tau')[\psi]
d\tau' - C \int_{\mathcal{R}(\tau_0,\tau_1)}  \iota_{\{\chi= 1\}} |\psi|^2
\end{equation}
where $ {}^\rho \mathcal{E}' (\tau')$ is defined by the expression~\eqref{alsocompareto} 
with $\rho$ replacing $\chi$,
where $0\le \rho\le \chi$, where again  $\rho=1$ for $r\le r_2$, for some $r_+<r_1<r_2$ and say $r\ge  R/4$.
\end{itemize}

Let us note that the above three assumptions are in fact \emph{stable}
in the class of stationary metrics satisfying the basic geometric assumptions of~\cite{DHRT22}. Thus, they can be inferred to hold in the very slowly rotating Kerr case $|a|\ll M$ as a simple corollary of
the fact that they hold in the Schwarzschild case $a=0$.

As in Section~\ref{reviewoffirstassumption}, 
we may now extend the above generalised coercivity statements to appropriately apply
to $k+1$'th order, $r^p$ weighted currents 
\begin{equation}
\label{rpweightedcurrentsintro}
{\Jpk}_{g_0}[\psi],\,\, {\Kpk}_{g_0}[\psi], 
\end{equation}
which 
involve summing weighted versions of~\eqref{labelthecurrents} 
applied to $\mathfrak{D}^{\bf k} \psi$ over all $|{\bf k}|\le k$.
One obtains an energy estimate 
\begin{equation}
\label{topenergyidentityofthatpaper}
\Efancypk (\tau_1)[\psi] 
+ c \int_{\tau_0}^{\tau_1}{}^\rho\Epminusonek'(\tau')[\psi] d\tau' 
\leq \Efancypk (\tau_0)[\psi]
+C \int_{\tau_0}^{\tau_1}(\, {}^\chi\Epminusonekminusone'(\tau')[\psi] + \Epminusonekminustwo'(\tau')[\psi]) d\tau' 
+C \int  |\accentset{\scalebox{.6}{\mbox{\tiny $(p)$}}} \partial \mathfrak{D}^{\bf k} \psi  \cdot \mathfrak{D}^{\bf k}  \Box_{g_0}\psi|  +\ldots,
\end{equation}
where again
${}^\rho\Epminusonek'(\tau')[\psi]$ is a higher order $r^p$-weighted
version of ${}^\rho\mathcal{E}'(\tau')[\psi]$ appearing in~\eqref{wemaythuswrite}, and again
 we have only displayed one of the arising inhomogeneous terms schematically (and we have
 moreover
dropped some additional good flux terms from the left hand side).

As applied to solutions of $\Box_{g_0}\psi =F$, the
estimate~\eqref{topenergyidentityofthatpaper} appears to be strictly worse than~\eqref{simplifiedestimatepweighthigherorderlargea}. The point, however, is that
because it arises from a covariant energy identity~\eqref{genidentityintro}, 
the  identity can also be applied  to solutions $\psi$ of~\eqref{theequationzero} \underline{replacing
$g_0$ with $g(\psi,x)$ in the definitions of~\eqref{labelthecurrents}}. 
This will generate various additional error terms 
in the arising estimate analogous to~\eqref{topenergyidentityofthatpaper}, but
for the final  displayed term  of~\eqref{topenergyidentityofthatpaper} we may now replace $\Box_g\psi$ with $N(\partial\psi, \psi,x)$,
which is now only $k+1$'th order in $\psi$ (and not $k+2$'th order), 
and there is thus no loss of derivatives.
We describe this in more detail in Section~\ref{estimatehierarchintroold} immediately below.

Before proceeding to the quasilinear~\eqref{theequationzero}, 
let us remark that there is another sense in which
the estimate of~\eqref{topenergyidentityofthatpaper} is better than~\eqref{simplifiedestimatepweighthigherorderlargea}, namely one has the exact constant $1$
in front of the term  $\Efancypk (\tau_0)[\psi]$ on the right hand side. (Note 
that~\eqref{topenergyidentityofthatpaper} has $\leq$ and not $\lesssim$.)
This is because  these energies  are defined as exact energy fluxes  (cf.~\eqref{fromexactfluxes})
arising from identities. This fact was exploited in the scheme 
of~\cite{DHRT22} to obtain the most precise results, at the expense, however,  of having to always carry forward the distinction
between the quantities $\Efancypk$ and $\Epk$.  Since we shall not be able to make use of
this in the present paper, we will describe a simplified scheme where we do not carry this 
information further. In what follows then, we may thus apply the relation
\begin{equation}
\label{referbutbecareful}
\Efancypk (\tau)[\psi] \sim \Epk(\tau)[\psi],
\end{equation}
which follows from appropriate elliptic estimates, to already substitute 
$\Epk(\tau)$ in~\eqref{topenergyidentityofthatpaper}, replacing $\leq$ by $\lesssim$.

\subsubsection{The estimate hierarchy on ${\rm L}$-slabs for the quasilinear equation~\eqref{theequationzero}}
\label{estimatehierarchintroold}

To describe the resulting estimate when~\eqref{genidentityintro} is applied with $g=g(\psi,x)$ to solutions
$\psi$ of the quasilinear equation~\eqref{theequationzero}, 
it is useful to introduce the additional ``master'' energies:
\begin{eqnarray}
\label{defforintronondeg}
\Xpk(\tau_0,\tau_1)					&:=&		 \sup_{\tau_0\le\tau\le \tau_1} \Epk (\tau')
			+ \int_{\tau_0}^{\tau_1}
			 \,\, \Epminusonek'(\tau') 
			 d\tau'  + \ldots, \\
\label{defforintrozero}
{}^\chi \Xpk(\tau_0,\tau_1)				&:=&  \sup_{\tau_0\le\tau\le \tau_1}\Epk (\tau)
+\int_{\tau_0}^{\tau_1}\left( \, \,   {}^{\chi} \Epminusonek' (\tau') + \Epminusonekminusone' (\tau') \right)d\tau' +\cdots, \\
\label{defforintroone}
{}^\rho \Xpk(\tau_0,\tau_1)&:= &
\sup_{\tau_0\le\tau\le \tau_1}\Epk(\tau)+\int_{\tau_0}^{\tau_1}
 {}^\rho \Epminusonek'(\tau) d\tau
+\cdots , \,\, 
\end{eqnarray}
for all
$\delta\le p\le 2-\delta$, where we have omitted various additional terms (including null fluxes). The
master energies~\eqref{defforintrozero} and~\eqref{defforintroone} 
combine
all the terms controlled on the left hand side of estimates~\eqref{simplifiedestimatepweighthigherorderlargea} and~\eqref{topenergyidentityofthatpaper}, respectively, while the master
energy~\eqref{defforintronondeg} satisfies
\begin{equation}
\label{fundrelhere}
\Xpkminusone(\tau_0,\tau_1)		\leq {}^\chi \Xpk(\tau_0,\tau_1)	.
\end{equation}
The case $p=0$ is slightly anomalous and we shall not give the formulae here. (See the introduction
of~\cite{DHRT22} or already Section~\ref{masterenergiessec} for the full definitions
of the master energies.)

We consider then a solution $\psi$ 
of~\eqref{theequationzero} defined on a 
 spacetime slab $\mathcal{R}(\tau_0,\tau_1:=\tau_0+{\rm L})$ of time-length ${\rm L}\ge 1$ satisfying an additional 
 lower order smallness estimate
\begin{equation}
\label{contingent}
\Xzerolesslessk(\tau_0,\tau_1) \lesssim \varepsilon
\end{equation}
which should be viewed as a bootstrap assumption.  One obtains from assumptions 1.~and 
2.~and our basic assumptions on the 
quasilinearity $g(\psi,x)$ and the semilinearity $N(\partial \psi, \psi,x)$
the following hierarchy of estimates
for $p=2-\delta$ and $p=1$, and for $\delta>0$ sufficiently small:
\begin{align}
{}^\rho \Xpk(\tau_0,\tau_1)
&\lesssim
\Epk(\tau_0)
+{}^\chi \Xpkminusone(\tau_0,\tau_1)	
&+\sqrt{{}^\rho \Xpk (\tau_0,\tau_1)}\sqrt{{}^\rho\Xzerok(\tau_0,\tau_1)}\sqrt{\Xplesslessk(\tau_0,\tau_1)}  +\cdots\,\,
\label{highes}
&+{}^\rho \Xpk(\tau_0,\tau_1)
\sqrt{\Xzerolesslessk(\tau_0,\tau_1)} \sqrt{{\rm L}},\\
\label{secondhigh}
{}^\chi \Xpkminusone (\tau_0,\tau_1) &\lesssim \Epkminusone(\tau_0) 
&+  
\sqrt{{}^\rho \Xpkminusone (\tau_0,\tau_1)}\sqrt{{}^\rho\Xzerokminusone(\tau_0,\tau_1)}\sqrt{\Xplesslessk(\tau_0,\tau_1)} +\cdots\,\,
&+{}^\rho \Xpk(\tau_0,\tau_1)
\sqrt{\Xzerolesslessk(\tau_0,\tau_1)} \sqrt{{\rm L}}.
\end{align}
The hierarchy extends to $p=0$ in a suitable form, which is anomalous and we shall omit here.

We explain briefly the origin and significance of the various terms in~\eqref{highes}--\eqref{secondhigh}.

Note that the left hand side of~\eqref{highes} is an energy with $k$ commutations whereas the left
hand side of~\eqref{secondhigh} is with $k-1$, i.e.~is one order lower.  
Equation~\eqref{secondhigh} arises from 
applying the linear inhomogeneous estimate~\eqref{simplifiedestimatepweighthigherorderlargea} to $\psi$, viewing $\psi$ as a solution to the inhomogeneous equation $\Box_{g_0}\psi =F:= (\Box_g-\Box_{g_0})\psi
+N(\partial\psi,\psi, x)$, whereas equation~\eqref{highes}  arises from integrating the 
divergence identity corresponding to the currents~\eqref{rpweightedcurrentsintro} 
where we replace $g_0$
with $g(\psi,x)$.

The term
${}^\chi \Xpkminusone(\tau_0,\tau_1)$ on the right hand side of~\eqref{highes}
arises precisely from the linear
term on the right hand side of~\eqref{topenergyidentityofthatpaper} in view of the 
definition~\eqref{defforintrozero} and the relation~\eqref{referbutbecareful} (which we note
also relies on~\eqref{contingent} in the nonlinear setting). The key to the scheme
is that this term corresponds precisely to the term on the left hand side of~\eqref{secondhigh}.
Thus, examining the terms on the right hand side of~\eqref{secondhigh}, 
one may already replace this term by $\Epkminusone(\tau_0)$ which in turn may be absorbed
into the term $\Epk(\tau_0)$ already present.

The nonlinear  term $\sqrt{{}^\rho \Xpk (\tau_0,\tau_1)}\sqrt{{}^\rho\Xzerok(\tau_0,\tau_1)}\sqrt{\Xplesslessk(\tau_0,\tau_1)}$ on the
right hand side of~\eqref{highes} (together with similary omitted terms) arises from the
 semilinearity $N(\partial\psi,\psi,x)$
in the region $r\ge R$.  The fundamental point is that the appropriately measured total $r^p$ weight in these terms 
is no higher than
the $r^p$ weight  of the energy on the left hand side. 
The role of the null structure
assumption (see already Assumption~\ref{largeanullcondassumption}) 
on $N(\partial\psi,\psi,x)$ is precisely to ensure this property. 

Identical considerations apply to the term
$\sqrt{{}^\rho \Xpkminusone (\tau_0,\tau_1)}\sqrt{{}^\rho\Xzerokminusone(\tau_0,\tau_1)}\sqrt{\Xplesslessk(\tau_0,\tau_1)} $ on the right
hand side of~\eqref{secondhigh}, which is one order lower.

Finally, the nonlinear term ${}^\rho \Xpk(\tau_0,\tau_1)
\sqrt{\Xzerolesslessk(\tau_0,\tau_1)} \sqrt{{\rm L}}$ on the right hand side 
of~\eqref{highes} and~\eqref{secondhigh} arises in both cases from both quasilinear error
terms and semilinear terms in $r\le R$ (for the semilinear terms, however,
as remarked earlier,
for~\eqref{secondhigh} we could here replace $k$ by $k-1$ in the last term on the right
hand side).
Note that in the case of~\eqref{highes}, this term is of the same order as the left hand side, 
while it is of one order more than the left hand side
in the case of~\eqref{secondhigh}, hence the ``loss of derivatives'' in~\eqref{secondhigh}.
Note that in this estimate we bound a spacetime integral 
using the supremum term in the definition~\eqref{defforintroone}
and Sobolev estimates generating $\sqrt{\Xzerolesslessk(\tau_0,\tau_1)}$, at the expense of
an estimate that depends explicitly on the length ${\rm L}$ of the slab.  Exploiting the presence
of the spacetime integrated terms in~\eqref{defforintronondeg}, 
however, the ${\rm L}$-dependence is only $\sqrt{{\rm L}}$.

\subsubsection{The pigeonhole argument and iteration in consecutive spacetime slabs}
\label{dyadicintro}
By our above comments, we may add the two estimates~\eqref{highes} 
and~\eqref{secondhigh} and obtain (under our bootstrap
assumption~\eqref{contingent})
\begin{equation}
\label{addedhereletssee}
{}^\rho \Xpk(\tau_0,\tau_1)
+{}^\chi \Xpkminusone(\tau_0,\tau_1)
\lesssim
\Epk(\tau_0)
+
\sqrt{{}^\rho \Xpk (\tau_0,\tau_1)}\sqrt{{}^\rho\Xzerok(\tau_0,\tau_1)}\sqrt{\Xplesslessk(\tau_0,\tau_1)}
+{}^\rho \Xpk(\tau_0,\tau_1)
\sqrt{\Xzerolesslessk(\tau_0,\tau_1)} \sqrt{{\rm L}},
\end{equation}
for $p=2-\delta, 1$. (Again, we omit  the precise form of the inequality analogous to~\eqref{addedhereletssee} for the $p=0$ case.)
In view of~\eqref{fundrelhere},  this allows
us to obtain global existence and estimates on an ${\rm L}$-slab under the assumption
\begin{equation}
\label{justthis}
\Eonekprimeprime(\tau_0) \lesssim \varepsilon, \qquad
\Ezerokprime(\tau_0)\lesssim \varepsilon {\rm L}^{-1} 
\end{equation}
for any $\lesslessk \le k'-1 < k' \le k'' \le k$. (Note we indeed need to assume estimates for some $p>0$ as in the first inequality of~\eqref{justthis} because
the anomalous $p=0$ case requires an estimate on a higher $p$ weighted norm in order for the estimates to close globally
in the slab.)

The above assumption~\eqref{justthis} on its own is insufficient to obtain 
existence beyond the fixed ${\rm L}$-slab $\mathcal{R}(\tau_0,\tau_1)$. If we embed~\eqref{justthis}, however, in a suitable hierarchy of ${\rm L}$-dependent estimates,
then not only may we propagate estimates in an ${\rm L}$-slab using~\eqref{addedhereletssee} and its $p=0$
analogue, but, we may derive---with an appeal to the pigeonhole principle in view of
the fundamental relation~\eqref{downwithhierarchy}---new 
estimates at the ``top'' $\Sigma(\tau_1)$ of the slab which allow one to iterate the estimate in suitable consecutive
slabs, to infer \emph{global} existence
and stability.

An example of such an ``iterating'' hierarchy is 
\begin{equation}
\label{iteratinghierarchy}
\Eonek(\tau_i) \lesssim \varepsilon  {\rm L}_i^{\beta}, \qquad 
\Etwominusdelkminustwo(\tau_i) \lesssim \varepsilon {\rm L}_i^{\beta}, \qquad 
\Ezerokminustwo(\tau_i)\lesssim\varepsilon {\rm L}_i^{-1+\beta}, \qquad
 \Eonekminusfour(\tau_i) \lesssim\varepsilon {\rm L}_i^{-1+\delta+\beta},  \qquad
 \Ezerokminussix(\tau_i) \lesssim\varepsilon  {\rm L}_i^{-2+\delta+ \beta},
\end{equation} 
where $\beta>0$ is suitably small and $\tau_i:= \alpha^i$ where $\alpha\gg 1$ is suitably large,  ${\rm L}_i=\tau_{i+1}-\tau_i$
and the constants implicit
in $\lesssim$ are chosen accordingly.

This yields the global existence statement of the main theorem
and the asymptotic stability statement, for instance the decay
\begin{equation}
\label{andthedecayandthedecay}
 \Ezerokminussix(\tau) \lesssim\varepsilon  \tau^{-2+\delta+ \beta}
\end{equation}
for the lower order unweighted energy.

 Let us note that the above hierarchy~\eqref{iteratinghierarchy}  upon iteration appears to yield
 $\tau^\beta$ growth for the highest order norm $\Eonek(\tau)$. Revisiting 
 however~\eqref{addedhereletssee} in the global (non-dyadic!)~slab $\mathcal{R}(1,\tau)$, one may estimate the nonlinear 
 terms on the right hand side by summing the estimates on the dyadic subintervals already obtained
 and thus obtain the uniform boundedness statement
 \begin{equation}
 \label{uniformboundednessintro}
 \Eonek(\tau)\lesssim \varepsilon.
 \end{equation} 
 Thus we obtain also the top order orbital stability statement of the main theorem.

Note that the iterating hierarchy~\eqref{iteratinghierarchy} presented in this section 
was not the one actually 
used in Section~6.3.4 of~\cite{DHRT22} 
but is related to that
of Section 6.2.5 of that paper (which was there used however only in the context of the restricted case (ii) considered in that section).
In contrast, the original hierarchy of Section~6.3.4 of~\cite{DHRT22}
exploited the distinction between
the two energies~\eqref{referbutbecareful}   to obtain slightly more precise results, in particular
to obtain the uniform boundedness~\eqref{uniformboundednessintro} 
without having to revisit the estimates in the global slab $\mathcal{R}(1,\tau)$. 
As we shall not be able to do this
in the present paper, we have described in this section
exactly the simplified scheme which we shall actually
use for the general $|a|<M$ case. (See already Section~\ref{proofofmainthelargea}.)

\subsection{The wave-packet localised coercive currents and the top order energy estimate
for $|a|<M$}
\label{freqloccoercurintro}

We describe in this section in more detail the main new ideas of the present paper.

We first review in Section~\ref{reviewoffreq}  the basic frequency analysis
associated to $\Box_{g_{a,M}}$ on Kerr as used in~\cite{partiii}, distinguishing
the $(t^*,\phi^*)$ analysis (for which we will introduce here a rudimentary pseudo-differential 
calculus) from Carter's full separation.
 We then proceed in Section~\ref{reviewofILEDandimprove} with a discussion of 
an improvement to the local integrated decay estimate
(assumption~1.~of~\cite{DHRT22} reviewed
in Section~\ref{reviewoffirstassumption}) which
can easily be shown from~\cite{partiii} for the full subextremal range $|a|<M$ of Kerr and will play
an important role in our argument. With this, we shall describe our wave-packet localised constructions in 
Section~\ref{wavepacketlocintrosec}, replacing
the role of 2.~above. We will then describe in Section~\ref{toporderenerintro} how in conjunction with the improved integrated decay estimate strengthening 1., these result in a suitable top order estimate. Finally, we will briefly remark how to complete the proof in Section~\ref{finishit}.

\subsubsection{Frequency analysis of $\Box_{g_{a,M}}$: The $(t^*,\phi^*)$ frequency analysis and pseudo-differential calculus and Carter's full separation}
\label{reviewoffreq}

The Killing symmetries $\partial_{t^*}$ and $\partial_{\phi^*}$ of the Kerr metric (here
$(t^*,\phi^*, r,\theta)$ denote Kerr star coordinates---see already Section~\ref{coordsheresec}) 
allow
us to define frequency projections via the Fourier transform:
\begin{equation}
\label{Fourierintro}
\mathfrak{F}[\Psi] (r,\theta, \omega,m) := \int_{-\infty}^{\infty} \int_0^{2\pi} e^{i\omega t^*} e^{-im \phi^*} 
\Psi (r,\theta,t^*, \phi^*)dt^* d\phi^*,
\end{equation}
where $\omega\in \mathbb R$, $m\in \mathbb Z$.
It follows then that given a suitable smooth function  $0\le \musicchi(\omega, m)\le 1$, the
``projection'' operator 
\begin{equation}
\label{definitionofPintro}
P = \mathfrak{F}^{-1}( \musicchi \, \mathfrak{F}),
\end{equation}
defined on an appropriate functional domain,
commutes with $\Box_{g_{a,M}}$. Under the usual decay assumptions on $\omega$-derivatives
and $m$-differences of~$\musicchi(\omega, m)$, operators of the form~\eqref{definitionofPintro} will be the simplest examples
of what we will term \emph{$(t^*,\phi^*)$ pseudodifferential operators}.  Versions of
the standard pseudodifferential calculus apply for operators of this type and operators that
arise by commuting~\eqref{definitionofPintro} with suitable spacetime differential operators. 
See already Section~\ref{elementarycalculus}. In practice, this will allow us
to commute $\Box_{g(\psi,x)}$ with appropriate zeroth order pseudodifferential projection operators like~\eqref{definitionofPintro} 
above and estimate without loss of derivatives, in an entirely analogous
way as commutation with usual differential operators.

The Fourier transform~\eqref{Fourierintro} already allows us to define the superradiant frequencies
$(\omega, m)$ to be those 
satisfying~\eqref{superradianthere}.
(Note that such frequencies exist only when $a\ne 0$. For convenience, let us assume $a\ne 0$ in what follows.)
 All other frequencies $(\omega, m)$ will be known as \emph{nonsuperradiant}. Let us note that if $\musicchi$ in~\eqref{definitionofPnintro} is supported
in the nonsuperradiant frequencies, then
\begin{equation}
\label{trytolocalise}
\int_{\mathcal{H}^+}  J^T[ P\psi] \cdot{\rm n}_{\mathcal{H}^+}\ge 0
\end{equation}
where $J^T_{\mu}$ denotes the standard energy current corresponding to 
the stationary Killing vector field $T=\partial_{t^*}$. 
(The separation of the superradiant and nonsuperradiant part of the solution
via projections as above to ensure in particular~\eqref{trytolocalise} was already exploited
in~\cite{DafRodKerr}.)

Let us note already that another source of error terms which will arise
in this paper from applying pseudodifferential operators concern when we try to localise
integrated coercivity and nonnegativity statements like~\eqref{trytolocalise} to smaller domains. 
These error terms will  be ``linear'' but will typically---but not always!---be lower order. 

On top of the Fourier transform~\eqref{Fourierintro}, we shall also use 
 Carter's full separation~\cite{carter1968hamilton} 
 to analyse functions $\Psi$, i.e.~we will apply to $U=\mathfrak{F}[\Psi]$
 the further transformation
\begin{equation}
\label{carterfortheintro}
\mathfrak{C}[U] (r,\ell, \omega,m ) =    \int_{0}^{\pi} U(r,\theta,\omega,m)  \, \overline{S_{m\ell}^{(a\omega)}}(\theta)  \, d\theta       .
\end{equation}
Here $S_{m\ell}^{(a\omega)}$ 
label the so-called oblate spheroidal harmonics. A rescaled and adjusted version of
$\mathfrak{C}[U] $, denoted $u(r, \ell, \omega, m)$,
 satisfies an ordinary differential equation
\begin{equation}
\label{odesgoodpdesbad}
u''+\omega^2 u - V(\omega, m, \ell)u  = G
\end{equation}
where the inhomogeneity $G$ arises from $\Box_{g_{a,M}}\Psi$ and prime-notated
differentiation $'$ is to be understood with respect
to a rescaled version of $r$ denoted $r^*$, which satisfies $r^*\to -\infty$ as $r\to r_+$. We shall only apply~\eqref{carterfortheintro} 
in some  region $r_+<r\le R_{\rm freq} <R$
outside the Kerr event horizon~$\mathcal{H}^+$.

The additional frequency parameter $\ell$ allows one to localise the
trapping phenomenon to the frequency range $\omega^2\sim \Lambda(\omega, m, \ell)\gg 1$, where $\Lambda$ is related to the eigenvalues associated to the oblate spheroidal harmonics $S_{m\ell}^{(a\omega)}$.
The proof of~\eqref{simplifiedestimatepweighthigherorderlargea} 
in~\cite{partiii} was based on fixed $(\omega, m, \ell)$-frequency 
energy currents  which in particular capture   the unstable property of trapping. This in turn
is related to the fact that the potential $V$ in~\eqref{odesgoodpdesbad} 
has a unique non-degenerate maximum $\hat{V}_{\rm max}$ which
in the true trapping frequency range satisfies $\omega^2\approx \hat{V}_{\rm max}$.
See already Section~\ref{translationsection} for a review of these currents in the 
formalism of~\cite{partiii}
and their relation with  physical space currents~\eqref{labelthecurrents}.

As noted at the beginning of this paper, we emphasise that we shall not attempt to define pseudodifferential operators using projections defined by~\eqref{carterfortheintro} or commute such operators with 
$\Box_{g(\psi,x)}$.   The role of Carter's separation~\eqref{carterfortheintro} will be to revisit~\cite{partiii} in order to prove a stronger
version of~\eqref{simplifiedestimatepweighthigherorderlargea} (see Section~\ref{reviewofILEDandimprove} immediately below) and to infer coercivity of our currents modulo error terms precisely
related to this stronger estimate (see already Section~\ref{wavepacketlocintrosec}).

\subsubsection{An improved integrated local decay estimate}
\label{reviewofILEDandimprove}
Before discussing the replacement of assumption 2.~of~\cite{DHRT22}, let us
describe how a stronger version of the integrated local decay assumption~\eqref{simplifiedestimate}
may be obtained in the Kerr case in the full subextremal regime $|a|<M$. For details, see already Section~\ref{refinedlinearhere}.

We first note that in~\cite{partiii}, 
a stronger estimate than~\eqref{simplifiedestimate} is in fact already shown, 
which can only however be expressed in terms of the frequency
analysis with respect to Carter's full separation reviewed in Section~\ref{reviewoffreq}. 
In particular, restricted to superradiant
frequencies~\eqref{superradianthere}, the bulk integral does not in fact degenerate, 
whereas, for the non-superradiant frequencies, the degeneration may be localised
to an $(\omega, m, \ell)$-frequency dependent fixed $r$-value related to the maximum 
$\hat{V}_{\rm max}$
of the potential $V$ in~\eqref{odesgoodpdesbad} when this maximum  moreover satisfies
$\hat{V}_{\rm max}\approx \omega^2$. 
See already
Section~\ref{capturingprecise}.

What is more however, by exploiting appropriate fixed-$(\omega, m, \ell)$ currents for~\eqref{odesgoodpdesbad}, it turns out that the frequency localised estimate of~\cite{partiii} can further be improved
in the frequency range
\begin{equation}
\label{ellipticrangeintheintro}
V-\omega^2 \gtrsim b_{\rm elliptic} \omega^2
\end{equation}
 in order to gain suitable powers of $\omega$ and
$m$, rendering the estimate effectively elliptic.
Thus, after commutation, we in fact control a stronger quantity than~${}^\chi \Xpkminusone (\tau_0,\tau_1)$ in~\eqref{secondhigh}, a quantity we shall here
denote as
\begin{equation}
\label{controlcontrol}
{}^{\chi\natural}_{\scalebox{.6}{\mbox{\tiny{\boxed{+1}}}}}\,\Xpkminusone(\tau_0,\tau_1), 
\end{equation}
which in fact contains one order  more control of $\psi$ restricted to certain 
$(\omega, m, \ell)$-frequencies and
regions of spacetime (hence, the ${\scalebox{.6}{\mbox{\tiny{\boxed{+1}}}}}$). See already Sections~\ref{ellipticimprovement}--\ref{finalimprovedsec}. 

Controlling~\eqref{controlcontrol} will allow us
to in fact accept certain \emph{top order} errors in our coercivity, provided that these are localised in
both physical and Carter frequency space, to this ``elliptic''
region.
 We will see this in practice in Section~\ref{wavepacketlocintrosec} below (see 
 already~\eqref{globalstatementfirstintroline}).

\subsubsection{The frequency projections $P_n$, wave-packet localised currents and their coercivity}
\label{wavepacketlocintrosec}

With the above preliminaries, we are ready to describe in more detail the novel replacement of 
assumption~2.~of~\cite{DHRT22}
and how it leads to a weaker version of~\eqref{topenergyidentityofthatpaper}---still sufficient for
our purposes however!

Fix a spacetime slab $\mathcal{R}(\tau_0,\tau_1)$  and let 
 $\psi$ be a smooth function defined on that slab.

We first introduce the frequency projections which will allow
us to partition $\psi$ into wave packets. 

We
define $\musicchi_n(\omega,m)$ 
to be appropriate smooth functions, $n=0,\ldots, N$,
which will make the projection
\begin{equation}
\label{definitionofPnintro}
P_n = \mathfrak{F}^{-1}( \musicchi_n\, \mathfrak{F})
\end{equation}
into a zeroth order $(t^*,\phi^*)$-pseudodifferential operator,
and decompose $\psi$ into wavepackets 
\begin{equation}
\label{decompintopack}
\psi_n := P_n (\chi^2_{\tau_0,\tau_1}\psi)
\end{equation}
(defined with the help of a spacetime cutoff function $\chi_{\tau_0,\tau_1}$ which localises
$\psi$ to a spacetime slab $\mathcal{R}(\tau_0,\tau_1)$).

As mentioned already, 
our decomposition into wave packets  $\psi_n$ given by~\eqref{decompintopack}
will in particular separate out the non-superradiant
part of $\psi$ from the superradiant part supported on frequencies~\eqref{superradianthere}. It  will further
localise the original $\psi$, however, in finer frequency ranges based on the value of $\omega/am$.
Thus we will have in fact a set $\psi_0, \ldots, \psi_{N_s}$ of  superradiant wave
packets and a remaining set $\psi_{N_s+1}, \ldots, \psi_N$ of nonsuperradiant wave packets.
(See already Sections~\ref{projectiondef} and~\ref{fixingthecovering}.)
The fundamental  goal of the decomposition into~\eqref{decompintopack}  is the following:
\begin{enumerate}
\item 
For non-superradiant packets $\psi_n$, by requiring that the decomposition
be sufficiently fine,
one can ensure that  trapping is localised in physical space 
to an (arbitrary) small $r$-range $\mathcal{D}_n$, depending only on $n$, 
on which there is a fixed timelike Killing field 
\begin{equation}
\label{fixedtimelikevf}
T+\alpha_n \Omega_1
\end{equation}
in the span of $T=\partial_t$ and $\Omega_1=\partial_\phi$. We note that such a vector field
and a related decomposition already played a role in showing the boundedness 
part of statement~\eqref{simplifiedestimate}
in Section 13 of~\cite{partiii}. 
\item
For superradiant packets $\psi_n$, on the other hand, which by the above  remarks are not  subject to trapping, again
by requiring that the decomposition be sufficiently fine, one can ensure that   there is an
$r$-range depending only on $n$ where the wave operator $\Box_{g_{a,M}}$ 
is effectively elliptic and the derivative  improvement
of Section~\ref{reviewofILEDandimprove} applies.
\end{enumerate}

Property~1.~above is to be compared with the discussion of 
Section~\ref{reviewofILEDandimprove}, where we note that for fixed $(\omega, m, \ell)$, the degeneration associated to trapping is associated to a unique $r(\omega,m,\ell)$ related
to the maximum $V_{\rm max}(\omega, m, \ell)$. As remarked in Section~\ref{reviewofILEDandimprove} above, however,
true trapping requires also this maximum
to satisfy  $V_{\rm max}(\omega, m, \ell)\approx \omega^2$ and this only happens at an $r$-value which is uniquely
determined by (and depends continuously on) $\omega/am$. See already 
Section~\ref{propsnonsuper}.

Property~2.~above is in fact related to the original proof of $(\omega, m, \ell)$-localised
energy estimates for superradiant frequencies in~\cite{partiii}, where for fixed $\ell$,
the existence of an $r$-interval satisfying~\eqref{ellipticrangeintheintro} was exploited.
The fact that one may choose such an interval depending only on $\omega/am$  follows
from the monotonicity properties of $V$ in $\ell$. The existence then of a suitable interval
\underline{depending now only on $n$} follows by continuity for a sufficiently fine decomposition. See already 
Section~\ref{propssuper}.

The above properties of the decomposition allow for the construction of $n$-dependent currents 
\begin{equation}
\label{reffertothecurrent}
{J^{{\rm main},n}_{g_0}}, 
\qquad K^{{\rm main}, n}_{g_0}
\end{equation}
related by~\eqref{genidentityintro}, which in turn  have the following coercivity properties:
\begin{itemize}
\item
In the region $r\ge R_{\rm freq}$ for large enough $R_{\rm freq}\le R/4$, the currents are $n$-independent and have
the good boundary and bulk coercivity properties identical to those of their analogues in~\cite{DHRT22} described in Section~\ref{reviewofsecondassumption} above.
\item
For all $n$, 
 the boundary terms   $\int {J_{g_0}^{{\rm main}, n}}[\psi_n]\cdot{\rm n}_{\mathcal{S}}$, integrated
 on general subsets $\mathcal{S}(\bar\tau_1,\bar\tau_2)\subset \mathcal{S}$, are nonnegative
 up to allowed errors arising from pseudodifferential localisation. See already 
 Proposition~\ref{asintheproofofthisprop}.
For superradiant $n$, we in fact have the pointwise coercivity 
\begin{equation}
\label{pointwisesuperyes}
{J_{g_0}^{{\rm main}, n}}[\psi]   \cdot {\rm n}_{\mathcal{S}} \ge 
c( |L\psi|^2 +|\underline{L}\psi|^2
+|\nablaslash\psi|^2)  - C |\psi|^2 \qquad \text{\ on\ }\mathcal{S}, 
\end{equation}
(see already Proposition~\ref{bcgsc}).
\item 
For non-superradiant frequency ranges, $\int K^{{\rm main},n}_{g_0}[\psi_n]$
is  coercive   after integration in spacetime \emph{in the region
where} ${J_{g_0}^{{\rm main}, n}}\cdot {\rm n}_{\Sigma(\tau)}$ \emph{fails to be coercive},
up to lower order terms, but also  top order terms in the elliptic frequency range discussed in Section~\ref{reviewofILEDandimprove}, and errors
arising from pseudodifferential localisation.
More precisely, we may define for each  $n$ a sufficiently regular
degeneration function $\chi_n(r)$, vanishing identically in the region $\mathcal{D}_n$ 
discussed above,
and with 
\begin{equation}
\label{doesntdegenerate}
\chi_n=1\text{\ on\ }\mathcal{M}\setminus \widetilde{\mathcal{D}}_n,
\end{equation} 
where
$\widetilde{\mathcal{D}}_n\supset \mathcal{D}_n$ is a slightly larger region,
such that
\begin{eqnarray}
\label{wemaythuswritenewhere}
\int  K^{{\rm main}, n}_{g_0} 
 [\mathring{\mathfrak{D}}^{\bf k} \psi_{n}] 
&\gtrsim&  
\int  \chi_ n r^{-2}  \left( | L\mathring{\mathfrak{D}}^{\bf k} \psi_{n} |^2 + | \underline{L}\mathring{\mathfrak{D}}^{\bf k} \psi_{n} | ^2 
+|\slashed\nabla \mathring{\mathfrak{D}}^{\bf k} \psi_{n} |^2   \right)   \\
\label{globalstatementfirstintroline}
&&\qquad   -C\, \, {}^{\chi\natural}_{\scalebox{.6}{\mbox{\tiny{\boxed{+1}}}}} \, \Xpkminusone(\tau_0,\tau_1)\\
\nonumber
&&\qquad - \ldots 
\end{eqnarray}
where integration is over
 $\mathcal{R}(\bar\tau_1,\bar\tau_2)\cap \{r_+ \le r \le R_{\rm fixed}\} $,
where  $\mathcal{R}(\bar\tau_1,\bar\tau_2)\subset \mathcal{R}(\tau_1,\tau_2)$ is
a general subregion, and where $\mathring{\mathfrak{D}}^{\bf k}$ is a string of commutation
operators involving only $T$ and $\Omega_1$ (and thus commutes with $P_n$).
Inequality~\eqref{wemaythuswritenewhere} 
should be compared with~\eqref{wemaythuswrite} (after commutation with
$\mathring{\mathfrak{D}}^{\bf k}$). Note the presence of the
stronger quantity~\eqref{controlcontrol} on line~\eqref{globalstatementfirstintroline} which
now includes also top order terms. 
The omitted error terms arise
from pseudodifferential localisation. 

On the region $\widetilde{\mathcal{D}}_n$ on the other
hand, we have pointwise  coercivity of the boundary current, i.e.
\begin{equation}
\label{coercofboundintrohere}
{J_{g_0}^{{\rm main},n}} [\psi] \cdot {\rm n}_{\Sigma(\tau)} \gtrsim |L\psi|^2 +|\underline{L}\psi|^2
+|\nablaslash\psi|^2 +|\psi|^2 \text{\rm\ on\ } \widetilde{\mathcal{D}}_n \cap \Sigma(\tau).
\end{equation}
\item
For superradiant frequency ranges, $\int K_{g_0}^{{\rm main}, n}[\psi_n]$ is  coercive
 after integration in spacetime \emph{without degeneration}, up to lower order terms
and top order terms in the elliptic frequency range, and errors arising from
pseudodifferential localisation. That is to say, we have~\eqref{wemaythuswritenewhere} 
without the presence
of the degeneration function $\chi_n$.  For consistency, it will be convenient
to set $\widetilde{\mathcal{D}}_n:=\emptyset$ and define $\chi_n:=1$ identically. 
Then~\eqref{doesntdegenerate} and~\eqref{wemaythuswritenewhere} hold as stated for  both superradiant and nonnsuperradiant
frequency ranges $n$.
\end{itemize}

To prove our bulk and boundary coercivity 
properties~\eqref{wemaythuswritenewhere},~\eqref{coercofboundintrohere}  and~\eqref{pointwisesuperyes}
we build our current from
various components (see already Section~\ref{introducingthecurrents}).

The most delicate component of the current ensuring the bulk coercivity~\eqref{wemaythuswritenewhere}
will be a current of the form $\tilde{J}^{y_n}$ associated to a vector field $2y_n(r)\partial_{r^*}$
(see already~\eqref{introducingthecurrentyinappendix}), where $y_n(r)$ is a non-decreasing function 
of $r$ which \underline{vanishes identically} in $\mathcal{D}_n$.
The associated bulk current  $\tilde{K}^{y_n}$ then also vanishes identically in
$\mathcal{D}_n$. 
For superradiant ranges $n$, in which case $\mathcal{D}_n=\emptyset$,
  we choose $y_n$ to simply change sign at some value
of $r$ in the elliptic region provided by 2.

It is in order to show 
nonnegativity of  $\tilde{K}^{y_n}$ modulo our ``allowed'' error terms
that
 we must use 
Carter's separation and the current formalism of~\cite{partiii}.
This is because the allowed error term~\eqref{globalstatementfirstintroline} itself is characterised
by the elliptic improvement discussed in Section~\ref{reviewofILEDandimprove}, 
which in turn depends on
Carter's separation.
To show then the requisite nonnegativity, we 
frequency decompose $\tilde{K}^{y_n}[\mathring{\mathfrak{D}}^{\bf k} \psi]$
and examine its action on fixed frequencies $(\omega, m, \ell)$.
The required coercivity can then be reduced to a fixed $(\omega, m,\ell)$-frequency coercivity statement with error terms which are allowed to be $(\omega, m, \ell)$-dependent but
must be controlled by~\eqref{globalstatementfirstintroline}. 
In practice, this means that at the fixed $(\omega, m, \ell)$ level,
one can even accept top order failure of positivity  in the
region~\eqref{ellipticrangeintheintro},
as this corresponds precisely
to  the elliptic region discussed in Section~\ref{reviewofILEDandimprove}, 
and thus suitable control is indeed
provided by~\eqref{globalstatementfirstintroline}.
It thus suffices to  design $y_n$ so that the fixed $(\omega, m, \ell)$-frequency
bulk has nonnegativity properties outside of~\eqref{ellipticrangeintheintro}. 
It is precisely the properties of the potential $V$ discussed above
which ensure that this is indeed possible.
See already Propositions~\ref{bulkfornonsuperduper} and~\ref{superrangepropfreq}.

The
 boundary coercivity property~\eqref{coercofboundintrohere} 
  on~$\Sigma(\tau)\cap\widetilde{\mathcal{D}}_n$ 
 for nonsuperradiant frequency regimes
  is ensured by adding a large multiple times the current
  \begin{equation}
  \label{addingthisofcourse}
  J^{T+\alpha_{n}\Omega_1}
  \end{equation}
  associated to~\eqref{fixedtimelikevf}, which we recall is timelike in $\mathcal{D}_n$
(and thus also in some $\widetilde{\mathcal{D}}_n\supset \mathcal{D}_n$), suitably cut off for large $r$.
  The boundary coercivity property~\eqref{pointwisesuperyes} on the other hand for superradiant
  frequency regimes is ensured by
  adding a large multiple of~\eqref{addingthisofcourse} with the choice
  $\alpha_n:= \frac{a}{2Mr_+}$ in which case $T+\alpha_{n}\Omega_1$ coincides with the Hawking
  vector field~\eqref{Hawkingdef} which is future directed, null on the event horizon $\mathcal{H}^+$. Requiring
  now $r_0$ to be sufficiently near $r_+$, and adding a small multiple times the
  redshift current of~\cite{mihalisnotes}, one indeed ensures~\eqref{pointwisesuperyes}.
Adding the redshift  current also ensures non-degeneracy of bulk terms near $\mathcal{H}^+$.

Finally, the integrated boundary coercivity on $\mathcal{S}$ 
property valid for non superradiant frequencies essentially
follows from the analogue of~\eqref{trytolocalise} applied to the current~\eqref{addingthisofcourse}.
See already Section~\ref{fundcoercivityboundapp}.

The remaining coercivity properties referred to above and the $n$-independence of the currents
in the region $r\ge R_{\rm freq}$ are ensured
by adding additional well-chosen current components. As mentioned above, one must cut off
the component~\eqref{addingthisofcourse}
so as to coincide with an $n$-independent multiple of $J^T$ for large $r$. This is somewhat delicate
in the case of non-superradiant $n$ (see already~\eqref{seehere}). 
For superradiant $n$ on the other hand, the cutoff can be chosen to be supported in the $r$-range (depending only on $n$) for which the $\Box_{g_{a,M}}$ is elliptic (recall property~2.~above),
and thus all error terms generated by the cutoff can be easily absorbed by~\eqref{globalstatementfirstintroline}.

\subsubsection{The new top order energy estimate on ${\rm L}$-slabs}
\label{toporderenerintro}
We finally turn to the quasilinear equation~\eqref{theequationzero}. As before,
we fix a slab $\mathcal{R}(\tau_0,\tau_1)$ and suppose
we have a smooth solution $\psi$ of~\eqref{theequationzero} defined
on the slab. We define the wave packets $\psi_n$ by~\eqref{decompintopack} as before.

We now define energies   ${\Efancypk}_{\rm bdry}$  and ${\Efancypk}_{\rm bulk}$ 
corresponding to quantities which are indeed controlled coercively
 by~\eqref{coercofboundintrohere} and~\eqref{wemaythuswritenewhere} in 
the boundary
and bulk respectively,
summed over wave packets:
\begin{align*}
\Efancypk{}_{{\rm bdry}}(\tau) &:=\sum_n \sum_{|{\bf k}|=k}
 \int_{\widetilde{\mathcal{D}}_n \cap \Sigma(\tau)} |L\mathring{\mathfrak{D}}^{\bf k} \psi_{n}|^2
+|\underline{L}\mathring{\mathfrak{D}}^{\bf k}\psi_{n}|^2  +|\nablaslash \mathring{\mathfrak{D}}^{\bf k}\psi_{n}|^2  \\
&\qquad \qquad + r^p\text{-weighted terms in\ }\{r\ge R_{\rm freq}\} \\
&\qquad\qquad + \text{improved quantities near $\mathcal{S}$},\\
\Efancypk{}_{{\rm bulk}}(\tau) &: =\sum_n \sum_{|{\bf k}|=k} \int_{(\{r_0\le r\le R_{\rm freq}\} \setminus \widetilde{\mathcal{D}}_n )\cap  \Sigma(\tau)}  |L\mathring{\mathfrak{D}}^{\bf k} \psi_{n}|^2
+|\underline{L}\mathring{\mathfrak{D}}^{\bf k}\psi_{n}|^2  +|\nablaslash \mathring{\mathfrak{D}}^{\bf k}\psi_{n}|^2 \\
&\qquad \qquad + r^p\text{-weighted terms in\ }\{r\ge R_{\rm freq}\} \\
&\qquad\qquad + \text{improved quantities near $\mathcal{S}$},
\end{align*}
where we recall the convention $\widetilde{\mathcal{D}}_n:=\emptyset$ for superradiant ranges $n$.
The additional terms added on the right hand side are $n$-independent.
See already Section~\ref{toporderergiescomp} for precise definitions.

Defining
\[
\Efancypk(\tau):= \Efancypk{}_{{\rm bdry}}(\tau) + \Efancypk{}_{{\rm bulk}}(\tau)
\]
then using appropriate elliptic estimates  (see already Section~\ref{ellipticestimatessection}) 
and in view also of the extra terms added in the  second
and third lines of the above definitions, we have 
\begin{equation}
\label{morerelationshere}
\Epk(\tau)
 \lesssim 
\Efancypk (\tau) +\ldots, \qquad
{}^\chi \Epminusonek'(\tau) \lesssim \Efancypk{}_{{\rm bulk}}(\tau) +\ldots,
\end{equation}
while, using also the properties of the pseudodifferential calculus and local existence, we have
\begin{equation}
\label{morerelationshereintegrated}
\int_{\tau}^{\tau+1}
\Efancypk (\tau') d\tau'  \lesssim 
\Epk(\tau) +\ldots ,
\end{equation}
where the omitted terms are all error terms which are ``allowed''. (See already Proposition~\ref{reversecomparison}.)

Integrating in a general subregion
 $\mathcal{R}(\bar\tau_0,\bar\tau_1)
\subset \mathcal{R}(\tau_0,\tau_1)$
 the divergence identities of the currents~\eqref{reffertothecurrent} applied to $\psi_n$,
where we covariantly now replace $g_0$ with $g(\psi,x)$ in their definitions,
and summing over $n$ and $|{\bf k}|=k$ we obtain
\begin{align}
\label{fundidentity}
{\Efancypk}_{\rm bdry}(\bar\tau_1) + \int_{\bar\tau_0}^{\bar\tau_1} 
{\Efancypk}_{\rm bulk}(\tau) d\tau  \lesssim {\Efancypk}_{\rm bdry}(\bar\tau_0)
+  {\Efancypk}_{\rm bulk}(\bar\tau_0) + {\Efancypk}_{\rm bulk}(\bar\tau_1)\\
\label{refertothelinehere}
+ {}^{\chi\natural}_{\scalebox{.6}{\mbox{\tiny{\boxed{+1}}}}}\,\Xpkminusone(\tau_0,\tau_1)\\
\label{nonlinearagain}
+\sqrt{{}^\rho \Xpk (\tau_0,\tau_1)}\sqrt{{}^\rho\Xzerok(\tau_0,\tau_1)}\sqrt{\Xplesslessk(\tau_0,\tau_1)}+\cdots
+{}^\rho \Xpk(\tau_0,\tau_1)
\sqrt{\Xzerolesslessk(\tau_0,\tau_1)} \sqrt{{\rm L}}\\
\nonumber
+\cdots.
\end{align}
(In obtaining the above, we have also separately added in an $n$-independent divergence identity in the far away region
and another such identity  close to $\mathcal{S}$.)
We note the term~${}^{\chi\natural}_{\scalebox{.6}{\mbox{\tiny{\boxed{+1}}}}}\,\Xpkminusone(\tau_0,\tau_1)$ on line~\eqref{refertothelinehere}
which arises from~\eqref{globalstatementfirstintroline}, 
and the nonlinear terms~\eqref{nonlinearagain} which arise
as before. There are additional error terms, linear and nonlinear, arising
from pseudodifferential commutation, but these can finally be absorbed in the other
errors terms or finally in the main terms estimated. We shall not discuss the omitted
terms further here. See already~\eqref{mainestimate}--\eqref{cancellingfluxes} for the full 
estimate.

At first glance, the presence of the future boundary  term 
\begin{equation}
\label{firstglance}
 {\Efancypk}_{\rm bulk}(\bar\tau_1)
 \end{equation}
on the right hand
side may
seem an obstacle to infer boundedness estimates at top order from~\eqref{fundidentity}.
It turns out that given a solution $\psi$ of~\eqref{theequationzero} existing globally in a slab $\mathcal{R}(\tau_0,\tau_1)$, the inequality~\eqref{fundidentity} would be sufficient to obtain quantitative top order boundedness 
if only we had some a priori nonquantitative control of~\eqref{firstglance} for all times $\bar\tau_1$
and we could apply the inequality  in regions 
$\mathcal{R}(\bar\tau_0,\bar\tau_1)$ where $\bar\tau_1$ is possibly far in the future of
$\tau_1$. While  at this point of the argument
we do not even know that $\psi$ exists in the future of
$\Sigma(\tau_1)$, let alone that it be
well behaved to the future, we may, by cutting off the nonlinearities near $\Sigma(\tau_1)$ (instead of cutting off $\psi$) and then re-solving the equation, extend $\psi$ to a $\widehat\psi$ which for later $\tau\ge \tau_1$  is a solution to the linear equation:
\begin{equation}
\label{satisfies}
\Box_{g_{a,M}}\widehat\psi =0, \qquad \tau\ge \tau_1.
\end{equation}
See already Section~\ref{auxsec}.
All the above estimates and relations 
in fact then apply with $\widehat\psi$ in place of $\psi$. 
In view of the uniform boundedness estimate for solutions
of~\eqref{satisfies} (which follows from~\eqref{simplifiedestimatepweighthigherorderlargea}), 
we have the nonquantitative statement:
\begin{equation}
\label{nonquantitativeintheintro}
\lim_{\bar\tau_1\to \infty}  {\Efancypk}_{\rm bulk}(\bar\tau_1)[\widehat\psi]<\infty.
\end{equation}
It follows that applying~\eqref{fundidentity} with a  large enough $\bar\tau_1$, we may obtain a quantitative  estimate on 
 ${\Efancypk}_{\rm bulk}(\tau_{\rm late})[\widehat\psi] $ for some $\tau_{\rm late}\ge \tau_1$ by
applying the pigeonhole principle to the bulk integral on the left hand side. 
We may now plug this estimate back into~\eqref{fundidentity} 
to remove any dependence on a future boundary
term~\eqref{firstglance}.   See already Lemma~\ref{lemmaformainestimate}.

In view of the relations~\eqref{morerelationshere}--\eqref{morerelationshereintegrated} 
and the above comments, we finally obtain from~\eqref{fundidentity} 
the estimate
\begin{equation}
{}^\rho \Xpk(\tau_0,\tau_1)
\lesssim
\Epk(\tau_0)
+ {}^{\chi\natural}_{\scalebox{.6}{\mbox{\tiny{\boxed{+1}}}}}\,\Xpkminusone(\tau_0,\tau_1)
+\sqrt{{}^\rho \Xpk (\tau_0,\tau_1)}\sqrt{{}^\rho\Xzerok(\tau_0,\tau_1)}\sqrt{\Xplesslessk(\tau_0,\tau_1)} +\cdots
\label{obtainpreciselynow}
+{}^\rho \Xpk(\tau_0,\tau_1)
\sqrt{\Xzerolesslessk(\tau_0,\tau_1)} \sqrt{{\rm L}},
\end{equation}
which we recognise to be~\eqref{highes} 
with $ {}^{\chi\natural}_{\scalebox{.6}{\mbox{\tiny{\boxed{+1}}}}}\, \Xpkminusone(\tau_0,\tau_1)$ in
place of ${}^{\chi}\Xpkminusone(\tau_0,\tau_1)$.

\subsubsection{Finishing the proof}
\label{finishit}

In view of the fact that by our comments in Section~\ref{reviewofILEDandimprove}, 
we have that a stronger version of~\eqref{secondhigh} holds with 
 $ {}^{\chi\natural}_{\scalebox{.6}{\mbox{\tiny{\boxed{+1}}}}}\,\Xpkminusone(\tau_0,\tau_1)$ replacing
 ${}^{\chi}\Xpkminusone(\tau_0,\tau_1)$ on the left hand side, then this together 
 with~\eqref{obtainpreciselynow} may again be combined and used essentially
 unchanged in the scheme outlined
 previously in Section~\ref{dyadicintro} involving the pigeonhole principle and iteration
 in consecutive spacetime slabs. 
 We obtain in particular global existence, the  top order boundedness 
 statement~\eqref{uniformboundednessintro} 
 and the lower order decay~\eqref{andthedecayandthedecay}.
 See already Section~\ref{proofofmainthelargea}
  for the details. 
 This will then complete the proof.

\subsection{Outline of the paper}
\label{outlinesection}

The remainder of this paper is structured as follows:
In Section~\ref{notationsetc}, we shall review the Kerr manifold and
basic notations from~\cite{DHRT22}. We shall then state in Section~\ref{replacement} 
the black box assumption 1.~as proven already in~\cite{partiii} (and deduce the refinement discussed
already in Section~\ref{reviewofILEDandimprove} above), and introduce the new set of weakly coercive currents applied to projections discussed above in Section~\ref{wavepacketlocintrosec}, which will replace assumption 2. We will turn to the study of the quasilinear 
equation~\eqref{theequationzero} itself
in Section~\ref{quasilinearsectionlargea} and deduce in particular the main top order estimate
discussed above in Section~\ref{toporderenerintro}.
We give a precise statement of the main theorem in Section~\ref{largeamaintheoremsection} 
and
carry out its proof in Section~\ref{proofofmainthelargea}, using the main estimate above
and the scheme outlined in Section~\ref{dyadicintro}.  All arguments making use of
Carter's separation are deferred to Appendix~\ref{carterestimatesappend}, which
reviews in detail the fixed frequency current formalism of~\cite{partiii}.

The reader will notice that we have followed closely the structure of~\cite{DHRT22}  wherever
possible so as to facilitate comparison.

\paragraph{Acknowledgements.}
MD acknowledges support through NSF grant DMS-2005464. GH acknowledges support by the Alexander von Humboldt Foundation in the framework of the Alexander von Humboldt Professorship endowed by the Federal Ministry of Education and Research and ERC Consolidator Grant 772249. 
IR acknowledges support through NSF grants DMS-2005464 and a Simons Investigator Award.
MT acknowledges support through Royal Society Tata University Research Fellowship URF\textbackslash R1\textbackslash 191409.

\section{The subextremal Kerr manifold and basic notations}
\label{notationsetc}
In this section we review 
the Kerr manifold and the framework of~\cite{DHRT22}, recalling in particular
all notations to be used in the paper.

We begin in Section~\ref{subextremalkerrsec} reviewing how Kerr is put in the setting of our general 
assumptions of~\cite{DHRT22}. 
We shall review the formalism of covariant energy identities for wave operators in 
Section~\ref{Covariantenergyidentitiessec}
and finally, review our basic energy notations from~\cite{DHRT22} in Section~\ref{reviewedenergynotations}.

\subsection{The subextremal Kerr manifold in the framework of~\cite{DHRT22}}
\label{subextremalkerrsec}

We have already described in~\cite{DHRT22} (see Section~2.7.3 and Appendices A and C.2 
therein) how
the Kerr metric  fits into the framework of that paper for the full subextremal range $|a|<M$. 
We review briefly here.

\subsubsection{Boyer--Lindquist and Kerr star coordinates, differentiable structure and the function $r$}
\label{coordsheresec}

Fixing subextremal parameters $|a|<M$, the Kerr metric is given in   Boyer--Lindquist coordinates $(t,r_{BL},\theta, \phi)$ by the expression:
\begin{eqnarray}
\nonumber
g_{a,M} &=& -\frac{\Delta}{r_{BL}^2+a^2\cos^2\theta} (dt-a\sin^2\theta d\phi)^2 +
\Delta^{-1}(r_{BL}^2+a^2\cos^2\theta)dr_{BL}^2\\
\label{Kerrmetric}
&&\qquad
 +(r_{BL}^2+a^2\cos^2\theta) d\theta^2
+\sin^2\theta (r^2_{BL}+a^2\cos^2\theta)^{-1}(adt - (r_{BL}^2+a^2)d\phi)^2,
\end{eqnarray}
where 
\[
\Delta:  = r_{BL}^2 -2Mar+a^2 =(r_{BL}-r_-)(r_{BL}-r_+), \qquad r_\pm: =M\pm \sqrt{M^2-a^2}.
\]
The metric expression~\eqref{Kerrmetric} 
extends to be defined on a larger region which can be put in the form described in Section~2.7.3 of~\cite{DHRT22}.

The extension can be defined with the help of so-called Kerr star coordinates
$(r_{BL}, t^*, \phi^*, \theta)$ which in $r_+>r_{BL}$ are related to Boyer--Lindquist
coordinates by the transformations
\[
t^*=  t+ \bar{t}(r_{BL}), \qquad \phi^* = \phi+\bar\phi(r_{BL}) \mod 2\pi,
\]
where for the explicit form of the functions
$\bar{t}$, $\bar{\phi}$, see~\cite{partiii}. (These functions  may be chosen
so as to be supported in $r_{BL}\le 9M/4$ for instance.)
Given $r_0$ with $r_-<r_0<r_+$,  the expression would define a smooth metric on the ambient 
manifold-with-boundary
 $\mathbb R^4_{(x^0,x^1,x^2,x^3)} \setminus \{ r_{BL} < r_0\} $ where $r_{BL}$ is identified with the
 usual radial function of the spatial Cartesian coordinates, 
 and $t^*$ is identified with $x^0$. We do a slightly different 
 identification, that is,  given $r_0$ as above and $R$ sufficiently large, we define  the metric on
  $\mathcal{M}:= \mathbb R^4_{(x^0,x^1,x^2,x^3)} \setminus \{ r < r_0\} $ 
where $r$ is the usual radial function, $\theta$ and $\phi^*$ are the standard spherical coordinates
defined by the relations~$x^1=r\sin\theta \cos\phi^*$, $x^2= r\sin\theta \sin \phi^*$, $x^3=r\cos\theta$, and we relate $r= r(r_{BL}, \theta, \phi^*)$ 
where this function in particular has the property that  for
$r_{BL}\le R/2$ we have $r=r_{BL}$. For a suitable choice
of $r$ (see~\cite{DafLuk1}) 
the fixed-$r$ spheres for $r\ge R$ are then the spheres of a double null foliation.
In particular, there are two smooth functions $u(r,t^*)$ and $v(r,t^*)$ whose level sets are null
hypersurfaces in $r\ge R$ and such that $u(R,0)=0$.

We have defined thus our Lorentzian manifold $(\mathcal{M},g_{a, M})$ with boundary.
The manifold can be seen to be time-orientable and we choose the unique time orientation
making future pointing
the Kerr star coordinate vector $\partial_{t^*}$ at some point of sufficiently large $r$, say
$\partial_{t^*}(r=4M,t=0,\theta=0,\phi=0)$ (for which the vector is indeed timelike).

Note that $\mathcal{S}:=\{r=r_0\} = \partial\mathcal{M}$ is a spacelike
hypersurface with respect to $g_{a,M}$, and under our choice of time orientation it constitutes
a future boundary of $\mathcal{M}$. The null hypersurface $r=r_+$ will be denoted
by $\mathcal{H}^+$ and represents the event horizon.

We define finally the  part spacelike, part null hypersurface 
\[
\Sigma_0=( \{ t^*=0\} \cap \{r \le R\} )  \cup ( \{ u=0 \}\cap \{r\ge R\}).
\]

\subsubsection{The Killing fields $T$ and $\Omega_1$ and the ergoregion}
The stationary and axisymmetric Killing fields will be denoted $T$ and $\Omega_1$ and correspond to 
$\partial_t$ and $\partial_\phi$ of Boyer--Lindquist coordinates (wherever these are defined) and
to $\partial_{t^*}$ and $\partial_{\phi^*}$ of Kerr star coordinates.
We recall that while the span of $T$ and $\Omega_1$ is timelike in $r>r_+$, there is no single
Killing
vector field in the span of $T$ and $\Omega_1$ which is globally timelike in $r>r_+$ 
(unless $a=0$ in which case $T$ is such a vector field).

We define the \emph{(extended) ergoregion} to be the subset of $\mathcal{M}$ defined by
\begin{equation}
\label{ergoregion}
\{x\in \mathcal{M}:  g(T , T)\ge  0 \}  =\{ r\le M+\sqrt{M^2 -a^2 \cos^2\theta} \}  .
\end{equation}
In particular, if $r>2M+\epsilon$, then $-g(T,T) \ge b(\epsilon)>0$. 
Note that by our choice of time orientation, $T$ is future directed outside the ergoregion.
(Note that in the $a=0$ case, this ``extended ergoregion''~\eqref{ergoregion} 
coincides exactly with the
part of $\mathcal{M}$ in the black hole interior $r_0\le r\le 2M$, including the horizon $\mathcal{H}^+$.)

For points outside the horizon but lying in  the ergoregion~\eqref{ergoregion}, 
it is clear that we may make the following statement:
Fixing  any
$r_{\rm fixed}>r _+$ and restricted to $\{ r_{\rm fixed}-\epsilon\le r\le r_{\rm fixed}+\epsilon\}$ for small enough $\epsilon>0$,  
we may  indeed find a timelike vector field in the span of 
$T$ and $\Omega_1$. This fact will play an important role in Section~\ref{generalisedcoerc}.

Finally, we note that on the event horizon $\mathcal{H}^+$ (where $r=r_+$) the so-called Hawking vector field
\begin{equation}
\label{Hawkingdef}
Z:= T+\frac{a}{2Mr_+}\Omega_1
\end{equation}
lies in the direction of the future null generator. 

\subsubsection{Other vector fields and the commutation operators $\mathfrak{D}_i$,
$\mathring{\mathfrak{D}}_i$
and $\widetilde{\mathfrak{D}}_i$}
\label{commutationcommutation}

All vector fields to be described below are translation invariant with respect to both $T$ and $\Omega_1$.

Given $r_0<r_+<r_1<r_2$ all sufficiently close to $r_+$, recall from~\cite{DHRT22} that we may define an additional vector field $Y$,  supported in $r\le r_1+(r_2-r_1)/2$ and
capturing the redshift properties in a neighbourhood
of $\mathcal{H}^+$.  The vectors $Y$, $Z$ (defined by~\eqref{Hawkingdef}) 
and a standard basis of rotation vector fields $\Omega_i$, $i=1,\ldots, 3$
(defined
by the usual formulas with respect to $(\theta, \phi^*)$ coordinates), span the tangent space
in $r_0\le r\le r_1+ (r_2-r_1)/4$.

We recall  also from~\cite{DHRT22}
the globally defined spanning set of vector fields $L$, $\underline{L}$,
$\Omega_1$, $\Omega_2$, $\Omega_3$,  
the notation $|\nablaslash\psi|^2 := \sum_i r^{-2} |\Omega_i\psi|^2$,
and the fact that  the expression 
\[
|L\psi|^2+|\underline L\psi|^2 +|\nablaslash \psi|^2
\] 
is a reference coercive expression on first derivatives of $\psi$.

We also recall our two sets of commutation operators from~\cite{DHRT22}, 
the set  $\mathfrak{D}_i$:
\[
\mathfrak{D}_1= T, \qquad \mathfrak{D}_2=\Omega_1,\qquad  \mathfrak{D}_3=Y, \qquad
\mathfrak{D}_4 =\zeta L, \qquad \mathfrak{D}_5 =\zeta \underline L, \qquad
\mathfrak{D}_{5+i} =(\zeta+\hat\zeta)\Omega_i, i=1,\ldots, 3
\]
and the set
$\widetilde{\mathfrak{D}}_i$:
\[
\widetilde{\mathfrak{D}}_1= L, \qquad \widetilde{\mathfrak{D}}_2 = \underline L,
\qquad \widetilde{\mathfrak{D}}_{2+i}=\Omega_i,    \qquad i=1,\ldots, 3.
\]
Here $\zeta$ is a smooth cutoff such that $\zeta=1$ for $r\ge 3R/4$ and $\zeta=0$ for $r\le R/2$
while $\hat\zeta$ is a smooth cutoff such that $\hat\zeta=1$ for $r_0\le r\le  r_1+(r_2-r_1)/4$
and $\hat\zeta=0$ for $r\ge r_1+(r_2-r_1)/2$.
We define the multi-index notation
\[
\mathfrak{D}^{\bf k} = \mathfrak{D}_1^{k_1} \cdots \mathfrak{D}_8^{k_8},
\qquad \widetilde{\mathfrak{D}}^{\bf k} =
\widetilde{\mathfrak{D}}_1^{k_1} \cdots \widetilde{\mathfrak{D}}_5^{k_5}
\]
and denote $|{\bf k}|=\sum k_i$.

We will here introduce also the notation
\begin{equation}
\label{recallnewcommutation}
\mathring{\mathfrak{D}}_1=T, \qquad \mathring{\mathfrak{D}}_2=\Omega_1,
\qquad  \mathring{\mathfrak{D}}^{\bf k} =\mathring{\mathfrak{D}}_1^{k_1}\mathring{\mathfrak{D}}_2^{k_2}
\end{equation}
for a set of operators involving commutation only with the Killing directions.

The significance of the untilded $\mathring{\mathfrak{D}}_i$ and 
 $\mathfrak{D}_i$  is that they have better commutation 
properties with $\Box_{g_{a,M}}$ (the former in fact commute exactly), while
the tilded $\widetilde{\mathfrak{D}}_i$ represent 
a complete spanning set of the tangent
space at every point.
Thus, it is typically the 
tilded $\widetilde{\mathfrak{D}}_i$ which will appear in our basic energies (see below), but
in the course of our proofs we will commute equation~\eqref{theequationzero} only
with  the untilded $\mathring{\mathfrak{D}}_i$  or  $\mathfrak{D}_i$,
relying on additional elliptic estimates to recover estimates for derivatives in the remaining
directions in the span of~$\widetilde{\mathfrak{D}}_i$.

\subsubsection{Hypersurfaces, spacetime regions and volume forms}
\label{hypersurfacesetalrecalled}

We may foliate our manifold $\mathcal{M}$  by the hypersurfaces $\Sigma(\tau)$
\[
\Sigma(\tau):= ( \{ t^*=\tau \} \cap \{r \le R\})  \cup ( \{ u=\tau  \}\cap \{r\ge R\})
\]
which
 intersect
the boundary $\mathcal{S}$ transversally.
Note that    $\Sigma(\tau)\cap \{r\le R\}$ is spacelike while $\Sigma(\tau)\cap \{r\ge R\}$ is null
and these are translates of $\Sigma_0$ by the one-parameter family of diffeomorphisms
generated by $T$.

In addition, the region  $r\ge R$ is foliated by translation invariant ``ingoing'' null cones
$\underline{C}_v$ 
parametrised by the smooth function $v$ defined on $r\ge R$ referred
to already in Section~\ref{coordsheresec}
and the vector fields $\underline L$, $\Omega_i$ mentioned above 
are tangent to $\underline{C}_v$
and $\underline{L}$ is in the direction of the null generator  and satisfies 
$g(\underline{L},T)\sim-1$, $g(L,\underline{L})=-1$ in $r\ge R$.
We define $\tau(v)$ by the relation $\underline{C}_v\cap \{r=R\}= \Sigma(\tau(v))\cap \{r=R\}$.

We introduce the notations
\[
\mathcal{R}(\tau_0,\tau_1):= \cup_{\tau_0\le\tau\le\tau_1} \Sigma(\tau),
\]
\[
\mathcal{S}(\tau_0,\tau_0) := \mathcal{S}\cap \mathcal{R}(\tau_0,\tau_1),
\]
and,  for
$\tau_0\le \tau_1 \le \tau(v)$,  
\[
\mathcal{R}(\tau_0,\tau_1,v) :=\mathcal{R}(\tau_0,\tau_1)\setminus\bigcup_{\tilde{v}> v} \underline{C}_{\tilde{v}} ,
\]
\[
\Sigma(\tau,v):=\Sigma(\tau)\setminus\bigcup_{\tilde{v}> v} \underline{C}_{\tilde{v}} .
\]

The spacetime region $\mathcal{R}(\tau_0,\tau_1,v)$ is a compact subset of spacetime
with past boundary  $\Sigma(\tau_0,v)$ 
and future boundary $\mathcal{S}(\tau_0,\tau_1)\cup \Sigma(\tau_1,v) \cup 
 \underline{C}_{v}$.

The volume form of $(\mathcal{M},g_{a,M})$ may be expressed as
\[
dV_{\mathcal{M}} =\upsilon(r, \theta,\phi^*) d\tau\,  r^2 dr  \sin\theta\, d\theta \, d\phi^*
\]
and, for $\Sigma(\tau)\cap \{r \ge R \}$, with the choice $L$ for the null normal, as
\[
dV_{\Sigma(\tau)\cap \{r\ge R\} } := \tilde\upsilon(r,\theta,\phi^*) r^2 dr \sin\theta\, d\theta  \, d\phi^* 
\] 
and for  $\underline{C}_v$, with the choice $\underline{L}$ for the null normal, as
\[
dV_{ \underline{C}_v } := \tilde{\tilde\upsilon}(r,\theta,\phi^*) r^2 dr \sin\theta\, d\theta \, d\phi^*  ,
\]
where  
\[
\upsilon \sim \tilde\upsilon \sim \tilde{\tilde\upsilon} \sim 1.
\]
The volume
form of $(\mathcal{M},g_{a.M})$ is related to the volume form of $\Sigma(\tau)$
\[
dV_{\mathcal{M}} \sim d\tau \, dV_{\Sigma(\tau)},
\]
where $\sim$ is interpreted for $4$-forms in the obvious sense.

Let us note that in the region $r_0\le r\le  R/2$, where in particular $r=r_{BL}$, the spacetime volume form may be reexpressed 
in Kerr-star coordinates as
\begin{equation}
\label{volumeform}
dV_{\mathcal{M}} = (r^2+a^2\cos^2\theta)  dt^* \, dr \,  \sin\theta\,  d\theta\,  d\phi^*,
\end{equation}
or alternatively, in the region $r_+< r \le R/2$, 
in Boyer--Lindquist coordinates as
\[
dV_{\mathcal{M}} = (r^2+a^2\cos^2\theta)  dt \, dr \,  \sin\theta\,  d\theta\,  d\phi.
\]

Note that when volume forms are omitted from integrals, the above induced volume forms from the metric $g_{a,M}$ will be 
understood, unless otherwise noted.

\subsubsection{Note on parameters and constants}
\label{noteonconstants}

The above constructions already referred to some $r$-parameters 
\[
r_0<r_+<r_1<r_2 <R,
\]
where $r_0$, $r_1$, $r_2$ may be taken arbitrarily close to $r_+$
and $R$ may be taken arbitrarily large.
We will in fact only fix these parameters later in the paper as they will be constrained
by certain additional parameters to be introduced in the course of our arguments.
We give the full list of $r$-parameters  we will need (in increasing order) in Table~\ref{firsttable},
including comments about their significance. 
The role of most parameters has been discussed in~\cite{DHRT22}. 
We note the important new parameters $r_{\rm pot} < R_{\rm pot}< R_{\rm freq}$ which will
be constrained in the process of our frequency analysis in Appendix~\ref{carterestimatesappend}. 
We will then require $r_2<r_{\rm pot}$ and $R_{\rm freq}\le R/4$, 
so that in particular~\eqref{volumeform} holds
in $r\le R_{\rm freq}$.

\begin{table}
\caption{Table of $r$-parameters adapted from~\cite{DHRT22}}
\begin{center}
\label{firsttable}
\begin{tabular}{ |c|c| } 
\hline
$r_0$ & $\mathcal{S}=\{r=r_0\}$ \\
\hline
$r_+$ & $\mathcal{H}^+=\{r=r_+\}$ is the Kerr Killing horizon with positive surface gravity;\\
&  span of $T$ and $\Omega_1$ is timelike for $r> r_+$ \\
\hline
$r_1$  & parameter related to the vector field $Y$  \\
\hline
$r_1+(r_2-r_1)/4$ &  $Z$, $Y$, $\Omega_i$ span the tangent space for $r_0\le r_1+(r_2-r_1)/4$\\
\hline
$r_1+(r_2-r_1)/2$ &  commutation vector fields $\mathfrak{D}$  all Killing for $r_1+(r_2-r_1)/2\le r\le R/2$\\
\hline
$r_2$ &
			see above\\
\hline
$r_{\rm pot}$ &  
			$r$-value connected to properties of $V_0$; $\chi=1$ for $r_0\le r\le r_{\rm pot}$\\			
\hline
$2.01M$ &   $T$ timelike for $r>2M$, strictly for $r\ge 2.01M$ \\
\hline
$5M$ & high frequency Carter potential  satisfies $V'_0<0$ for $r\ge 5M$ \\
\hline
$R_{\rm pot}$ & $r$-value connected to properties of $V_0$; $\chi= 1$, $\tilde\rho=1$ for $r\ge R_{\rm pot}$  \\
\hline
$R_{\rm freq}$ & delimits region $r_0\le r\le R_{\rm freq}$ where frequency analysis is applied \\
\hline
$R/2$ &  commutation vector fields $\mathfrak{D}$  all Killing for $r_1+(r_2-r_1)/2\le r\le R/2$\, ; \\
& $g=g_{a,M}$ for $r\ge R/2$ \, ; $r=r_{BL}$ for $r\le R/2$ \\
\hline
$8R/9$ &   the generalised null structure assumption concerns $r\ge 8R/9$  \\
\hline
$R$ &   $\Sigma(\tau)\cap \{r\ge R\}$ is null,  $\underline{C}_v\subset \{r\ge R\}$ \\
\hline
\end{tabular}
\end{center}
\end{table}

As in~\cite{DHRT22}, let us also fix a parameter
\begin{equation}
\label{largeafixeddelta}
0<\delta<\frac1{10}.
\end{equation}
This will be related to the degeneration of the $r^p$ method~\cite{DafRodnew} 
as one approaches $p=0$
and $p=2$, i.e.~we will constrain $p$ to be $p=0$ or $\delta\le p \le 2-\delta$.

We will use $k$ to denote integers which will parametrise number of derivatives.

We will follow the conventions for constants from~\cite{DHRT22}: 
Namely, in inequalities, we denote by $C$ and $c$ generic positive constants, which may be chosen so as to \emph{eventually} (see remarks below!)~depend only on (a) the Kerr parameters $a$, $M$ 
(b) the choice~\eqref{largeafixeddelta} of $\delta$ (when $\delta$-dependent
quantities or energies are invoked), and (c)
if there is $k$-dependence
in the relevant statement, also on~$k$.
(We use $C$ for large constants and $c$ for small constants.)
The notation $[\ldots ]\lesssim [\ldots] $ will be used to denote $[\ldots ]\le C[\ldots]$,
for a $C$ as above.

In addition to the above $r$-parameters and the parameter $\delta$, 
in the course of our proofs we will need to introduce a large number of additional fixed parameters,
which again, can eventually be chosen to depend only on $a$ and $M$, but will often
only be fixed later in the paper, for instance:
\begin{equation}
\label{forinstancepar}
b_{\rm trap}, b_{\rm elliptic}, \gamma_{\rm elliptic}
\end{equation}
which are used in characterising some frequency ranges or
\begin{equation}
\label{currentparameters}
E, e_{\rm red}, B_{\rm indep}
\end{equation} 
which appear as parameters in the construction of currents. 

 We shall explicitly describe the order in which parameters such
 as~\eqref{forinstancepar}  or~\eqref{currentparameters}
must be chosen in the text. \emph{We warn the reader that before~\eqref{forinstancepar} are fixed, some generic constants $C$ may depend
also on these parameters.}
Moreover, when we write that a generic constant $C$ depends on other quantities,  $C=C(x,y, \ldots)$, we mean that this dependence
on $x$, $y$, etc.~is \emph{in addition} to $a$ and $M$ (and $k$, $\delta$ as appropriate) and potentially any parameters as above not yet chosen.

Let us make an explicit remark already about two particular additional classes of such parameters:
Parameters which are small \underline{and which often multiply other constants} on the right hand
side of estimates will typically 
be denoted 
by variants of $\epsilon$ and 
$\varepsilon$. The parameters $\epsilon$ are related to showing coercivity and other properties
of our current constructions while the parameters $\varepsilon$ appear in smallness
assumptions for our norms.

Regarding $\epsilon$-parameters (for instance $\epsilon_{\rm cutoff}$ from~\eqref{announcingepsiloncutoff}), until these are fixed, then according to our conventions above,
generic constants $c$, $C$ may depend
(unfavourably)
also on these as $\epsilon\to 0$. We will have to track this dependence carefully.
To assist the reader, we will always  adhere to the following additional convention for $\epsilon$-parameters: (1) Small constants $c$ may always be chosen independent of  $\epsilon$ and
(2) large constants $C$ (explicitly displayed or implicit in a $\lesssim$) \emph{multiplying} such an $\epsilon$ can always be chosen
independent of $\epsilon$.  (When appealing to this convention, we will in particular explain
why the constants of our inequalities can indeed be chosen as claimed.)
This will allow us to absorb such $C\epsilon$ terms (for small enough $\epsilon$), thus allowing
us to eventually fix $\epsilon$ depending only on $a$, $M$ (and $k$, $\delta$ as appropriate), eliminating thus
the $\epsilon$-dependence  of other generic constants $C$. 

For $\varepsilon$-parameters, on the other hand, there will never be unfavourable dependence
of generic constants $C$ in the limit $\varepsilon\to0$.

\subsection{Covariant energy identities}
\label{Covariantenergyidentitiessec}
As in~\cite{DHRT22}, covariant energy identities will play a fundamental role in the present
paper as these will be applied at highest order. We review briefly the formalism to set the notation.
For a much more general discussion of such identities, see~\cite{christodoulou2016action}.

Our physical space identities are associated
to a quadruple $(V,w,q,\varpi)$. Here, $V^\mu$ is a vector field on $\mathcal{M}$, $w$ is a scalar function, and $q_\mu$, $\varpi_{\mu\nu}$ are $1$-forms and $2$-forms, respectively.
Given  $(V,w,q,\varpi)$, a general Lorentzian metric $g$ on $\mathcal{M}$ and a suitably regular 
(complex-valued) 
function $\psi$,
define
\begin{align}
\label{generalJdef}
J^{V,w,q,\varpi}_\mu [g, \psi] &:= T_{\mu\nu} [g,\psi] V^\nu + w {\rm Re} ( \psi \overline{\partial_\mu \psi} )+ q_\mu |\psi|^2  
+ * d \big( |\psi|^2 \varpi \big)_{\mu}\, ,
\\
\label{generalKdef}
K^{V,w,q}[g, \psi]  &:=  \pi^{V}_{\mu\nu} [g]  T^{\mu\nu}[g, \psi]  +  \nabla^\mu w {\rm Re}(\psi \overline{\partial_\mu \psi}) +w\nabla^\mu\psi \overline{\partial_\mu \psi} +  \nabla^\mu q_\mu |\psi|^2 + 2{\rm Re}(\psi q_\mu g^{\mu\nu}
\partial_\nu \psi)\, , \\
\label{generalHdef}
H^{V,w}[\psi] &:= V^\mu\partial_\mu \psi + w\psi \, , 
\end{align}
where
\[
T_{\mu\nu} [g,\psi] := {\rm Re} (\partial_\mu\psi\overline{\partial_\nu\psi})-\frac12 g_{\mu\nu} \nabla^{\alpha}\psi\overline{\partial_\alpha\psi},
\qquad \pi^X_{\mu\nu} [g] : = \frac12 (\mathcal{L}_X g)_{\mu\nu} = \frac12(\nabla_\mu X_\nu+\nabla_\nu X_\mu),
\]
and where $* \colon \Lambda^3 \mathcal{M} \to \Lambda^1 \mathcal{M}$ is 
the Hodge star operator corresponding to $g$.
(Note we here allow complex valued $\psi$, as opposed to~\cite{DHRT22} where everything
was assumed real valued. This is primarily for the convenience of Fourier analysis.)
We will often denote the $g$ dependence with a subscript, i.e.~$K_g^{V,w,q}[\psi]$, etc.

For a general function $\psi$, the current $J^{V,w,q,\varpi}_\mu$ satisfies 
the following fundamental divergence identity
\begin{equation}
\label{dividentitynotintegrated}
\nabla_{g}^\mu J^{V,w,q, \varpi}_{\mu} [g,\psi] = K^{V,w,q}[g,\psi ] + {\rm Re} \left( H^{V,w} [\psi] \overline{\Box_{g} \psi} \right) ,
\end{equation}
which upon integration in the region  $\mathcal{R}(\tau_0,\tau_1,v)$ yields
\begin{align}
\nonumber
\int_{\Sigma(\tau_1,v)} J[\psi]\cdot {\rm n} +\int_{\mathcal{S}(\tau_0,\tau_1)}J[\psi]\cdot {\rm n} +
\int_{\underline{C}_v(\tau_0,\tau_1)}J[\psi]\cdot {\rm n} +
\int_{\mathcal{R}(\tau_0,\tau_1,v)} K[\psi] \\
\label{energyidentity}
=
\int_{\Sigma(\tau_0,v)} J[\psi]\cdot {\rm n}-\int_{\mathcal{R}(\tau_0,\tau_1,v)} {\rm Re} \left( H[\psi] \overline{\Box_{g} \psi} \right) ,
\end{align}
where the normals, contractions and volume forms are with respect to the metric $g$.
We will always apply~\eqref{energyidentity} either with the exact
Kerr metric $g_{a,M}$ or with metrics $g$ sufficiently 
close to Kerr
satisfying moreover $g=g_{a,M}$ in $r\ge R/2$.  Thus, 
$\mathcal{S}(\tau_0,\tau_1)$ and $\Sigma(\tau)\cap \{r\le R\}$ will be spacelike
with respect to $g$ and
$\underline{C}_v$ and $\Sigma(\tau)\cap
\{ r\ge R\}$
will be null, 
and where
moreover we may use the null generators  $L$ and $\underline{L}$ to fix the normal (and volume form) on the null parts
of $\Sigma(\tau)$ and on $\underline{C}_v$, respectively. See Section~\ref{hypersurfacesetalrecalled}.

The currents we shall define in this paper will essentially all arise more naturally
as suitable combinations of ``twisted currents'', which are expressed in terms of a ``twisted'' energy
momentum tensor and ``twisted derivatives''. 
These are reviewed in Appendix~\ref{translationsection}.
The total currents can always be re-expressed however in the form given above.

\subsection{Energies and their properties}
\label{reviewedenergynotations}

We recall now the energies defined in~\cite{DHRT22} specialised to the
subextremal Kerr manifold $(\mathcal{M},g_{a,M})$.

\subsubsection{Non-degenerate  boundary and bulk energies}
\label{boundaryenergies}
In this section, we will recall the definitions of the  non-degenerate boundary and bulk energies. 

Let us be given $k\ge 0$,  $\tau_0\le \tau\le \tau_1$, $v$,
and a (complex-valued) function $\psi$ defined on $\mathcal{R}(\tau_0,\tau_1)$.

We recall first the ``unweighted'' energies:
\begin{align}
\label{firstdefk}
\Ezerok(\tau)[\psi]	&:=  \sum_{|{\bf k}|\leq k} 
   \int_{\Sigma(\tau)} |L\widetilde{\mathfrak{D}}^{\bf k}\psi|^2+|\slashed\nabla \widetilde{\mathfrak{D}}^{\bf k}\psi|^2 + \iota_{r\le R} |\underline{L}\widetilde{\mathfrak{D}}^{\bf k}\psi|^2+r^{-2}|\widetilde{\mathfrak{D}}^{\bf k}\psi|^2,\\
\label{extradefk}
\Ezerok_{\mathcal{S}}(\tau_0,\tau_1)[\psi]	&:= \sum_{|{\bf k}|\leq k} 
  \int_{\mathcal{S}(\tau_0,\tau_1)} |L\widetilde{\mathfrak{D}}^{\bf k}\psi|^2+|\slashed\nabla\psi|^2 +
   |\underline{L}\widetilde{\mathfrak{D}}^{\bf k}\psi|^2+|\widetilde{\mathfrak{D}}^{\bf k}\psi|^2,\\
\label{seconddefk}
\Fzerok(v, \tau_0,\tau_1)[\psi]	&: = \sum_{|{\bf k}|\leq k}
   \int_{\underline{C}_v\cap \mathcal{R}(\tau_0,\tau_1)} 
   |\underline{L}\widetilde{\mathfrak{D}}^{\bf k}\psi|^2
   +|\slashed\nabla\widetilde{\mathfrak{D}}^{\bf k}\psi|^2
   +r^{-2}|\widetilde{\mathfrak{D}}^{\bf k}\psi|^2,\\
   \label{lastinthislist}
\Ezerominusoneminusdeltak'(\tau)[\psi]	&:= \sum_{|{\bf k}|\leq k}
\int_{\Sigma(\tau)} 
r^{-1-\delta} \left(
|L\widetilde{\mathfrak{D}}^{\bf k}\psi|^2+
|\underline{L}\widetilde{\mathfrak{D}}^{\bf k}\psi|^2
+|\slashed\nabla\widetilde{\mathfrak{D}}^{\bf k}\psi|^2\right)
+r^{-3-\delta}|\widetilde{\mathfrak{D}}^{\bf k}\psi|^2.
\end{align}
The $'$ on the 	quantity defined in~\eqref{lastinthislist} denotes that this quantity typically appears in \emph{bulk} integrals, i.e.~\emph{integrated} in~$\tau$.  Note that $k$ refers to the number of commutations, so energies
with a $k$ subscript are of order $k+1$ in $\psi$. Note moreover that we use the tilded
$\widetilde{\mathfrak{D}}^i$ as commutators. All integrals are with respect
to the volume forms induced by the Kerr metric $g_{a,M}$ according to the conventions
recalled in Section~\ref{hypersurfacesetalrecalled}. 
The term ``unweighted'' refers to the fact that the above expressions are naturally related to boundary and bulk terms of energy currents which near infinity are controlled by the boundary terms of the current $J^T$ associated with the stationary vector field $T$, without additional weights.

We will often omit the $[\psi]$ from~\eqref{firstdefk}--\eqref{lastinthislist} 
and the definitions which follow,
when it is understood which $\psi$ we are
applying this to. 

Let us also introduce the notations
\begin{align}
\label{energylocalisednearbound}
{\Ezerok}{}_{r\le r_1}(\tau)[\psi]	&:=  \sum_{|{\bf k}|\leq k} 
   \int_{\Sigma(\tau)\cap \{r\le r_1\}} |L\widetilde{\mathfrak{D}}^{\bf k}\psi|^2+|\slashed\nabla \widetilde{\mathfrak{D}}^{\bf k}\psi|^2 +  |\underline{L}\widetilde{\mathfrak{D}}^{\bf k}\psi|^2+
  |\widetilde{\mathfrak{D}}^{\bf k}\psi|^2,\\
  \label{energyannulus}
{\Ezerok}{}_{r_1\le r\le r_2}(\tau)[\psi] 	&:=  \sum_{|{\bf k}|\leq k} 
   \int_{\Sigma(\tau)\cap \{r_1\le r\le r_2\}} |L\widetilde{\mathfrak{D}}^{\bf k}\psi|^2+|\slashed\nabla \widetilde{\mathfrak{D}}^{\bf k}\psi|^2 +  |\underline{L}\widetilde{\mathfrak{D}}^{\bf k}\psi|^2+
  |\widetilde{\mathfrak{D}}^{\bf k}\psi|^2,\\
\label{othercommutatorsenergy}
\mathring{\Ezerok}(\tau)[\psi]	&:=  \sum_{|{\bf k}|\leq k} 
   \int_{\Sigma(\tau)} |L{\mathring{\mathfrak{D}}}^{\bf k}\psi|^2+|\slashed\nabla{\mathring{\mathfrak{D}}}^{\bf k}\psi|^2 + \iota_{r\le R} |\underline{L}{\mathring{\mathfrak{D}}}^{\bf k}\psi|^2+r^{-2}|{\mathring{\mathfrak{D}}}^{\bf k}\psi|^2,\\
\label{othercommutatorshorizonflux}
\mathring{\Ezerok}_{\mathcal{S}}(\tau_0,\tau_1)[\psi] &	:=  \sum_{|{\bf k}|\leq k} 
   \int_{\mathcal{S}(\tau_0,\tau_1)} |L{\mathring{\mathfrak{D}}}^{\bf k}\psi|^2+|\slashed\nabla{\mathring{\mathfrak{D}}}^{\bf k}\psi|^2
   + |\underline{L}{\mathring{\mathfrak{D}}}^{\bf k}\psi|^2+r^{-2}|{\mathring{\mathfrak{D}}}^{\bf k}\psi|^2.
\end{align}

For global existence and stability it is essential to control quantities with stronger weights
at infinity than those corresponding to the current $J^T$. These will correspond
to the weights of currents associated to the so-called $r^p$ method~\cite{DafRodnew}, 
which is
a hierarchy of estimates for various values of $p$.
We recall thus the $p$-weighted analogues of the above for $\delta\le p\le 2-\delta$,
which  are defined as below:
  \begin{align}
  \label{coeffshereandkorder}
  \Epk(\tau)[\psi]	 &: = \Ezerok(\tau)[\psi]	+  \sum_{|{\bf k}|\leq k} \int_{\Sigma(\tau)\cap \{r\ge R\} } r^p |r^{-1}L(r\widetilde{\mathfrak{D}}^{\bf k}\psi)|^2 + r^{\frac{p}2}|L\widetilde{\mathfrak{D}}^{\bf k}\psi|^2
  +r^{\frac{p}2-2}|\widetilde{\mathfrak{D}}^{\bf k}\psi|^2, \\
     \label{weightedpfluxkorder}
 \Fpk(v,\tau_0,\tau)[\psi]	 &:= \Fzerok(v, \tau_0,\tau)[\psi]	 +     \sum_{|{\bf k}|\leq k}  \int_{\underline{C}_v\cap \mathcal{R}(\tau_0,\tau)} r^p|\slashed\nabla\widetilde{\mathfrak{D}}^{\bf k}\psi|^2 + r^{p-2}|\widetilde{\mathfrak{D}}^{\bf k}\psi|^2,\\
      \label{coeffsheretootoosmallerkorder}
   \Epminusonek'(\tau) [\psi]	&: = \Ezerominusoneminusdeltak'(\tau)[\psi]	+
    \sum_{|{\bf k}|\leq k}  \int_{\Sigma(\tau) } r^{p-1} \left( |r^{-1}L(r\widetilde{\mathfrak{D}}^{\bf k}\psi)|^2
     + |L\widetilde{\mathfrak{D}}^{\bf k}\psi|^2 
     +      |\slashed\nabla\widetilde{\mathfrak{D}}^{\bf k} \psi|^2\right)
  +r^{p-3}|\widetilde{\mathfrak{D}}^{\bf k}\psi|^2 .
  \end{align}

We recall the fundamental hierarchical relations at the heart of the $r^p$ method:
\begin{equation}
\label{largeahigherorderfluxbulkrelation}
\Epk \gtrsim \Epprimek, \, \, \Fpk \gtrsim \Fpprimek {\rm\ \ for\ } p\ge p'\ge \delta {\rm\ or\ }p'=0, \qquad  
  \Epminusonek' \gtrsim \Epminusonek  {\rm\ for\ } p\ge 1+\delta, \qquad   \Epminusonek' \gtrsim \Ezerok  {\rm\ for\ }p\ge 1 .
\end{equation}

\subsubsection{Degenerate bulk energies and the degeneration functions $\chi$, $\tilde\rho$}
\label{degbulkenergiesdefsec}

We recall from Section~\ref{estimatehierarchintroold} 
that the scheme of~\cite{DHRT22}  involved a hierarchy of degenerate
energies, where
 the
degeneration was captured by a set of functions $\chi$, $\rho$, $\tilde\rho$ and $\xi$,
appearing in the blackbox assumptions~1.~and~2.,
whose support properties are moreover non-trivially related.

In the present paper we shall not make use of $\xi$, $\rho$,
 while we may define $\chi$
to be a function 
\begin{equation}
\label{chidefimportant}
\chi=0 \text{\ in\ } (r_{\rm pot} , R_{\rm pot}), \qquad \chi=1\text{\ in\ }[r_0, r_{\rm pot}] \cup [R_{\rm pot},\infty)
\end{equation}
 and
 \begin{equation}
 \label{rhodefimportant}
\tilde\rho=0\text{\ in\ } [r_0, R_{\rm pot}), \qquad \tilde\rho=1 \text{\ in\ }[R_{\rm pot},\infty)
\end{equation}
where $r_+<r_{\rm pot}<2.01M<5M<R_{\rm pot}$ are parameters chosen at the beginning of Section~\ref{fixingthecovering}.

We recall from~\cite{DHRT22} the definitions:
\begin{align}
{}^\chi\Ezerominusoneminusdeltak'(\tau)[\psi]&:= \sum_{|{\bf k}|\leq k} \int_{\Sigma(\tau)}
  r^{-1-\delta}\chi(r) \left( |L\mathfrak{D}^{\bf k}\psi|^2+|\underline{L}\mathfrak{D}^{\bf k}\psi|^2+|\slashed\nabla\mathfrak{D}^{\bf k}\psi|^2\right),\\
{}^{\tilde\rho} \Ezerominusthreeminusdeltakminusone'(\tau) [\psi]&:=  \sum_{|{\bf k}|\leq k} \int_{\Sigma(\tau)} 
 \tilde\rho (r)r^{-3-\delta} |\mathfrak{D}^{\bf k}\psi|^2 ,\\
   \label{coeffsheretootookorderchi}
 {}^\chi  \Epminusonek'(\tau) [\psi] &: = {}^\chi\Ezerominusoneminusdeltak'(\tau)[\psi]+  \sum_{|{\bf k}|\leq k} \int_{\Sigma(\tau)\cap \{r\ge R\} } r^{p-1} \left( |r^{-1}L(r\widetilde{\mathfrak{D}}^{\bf k}\psi)|^2
     + |L\widetilde{\mathfrak{D}}^{\bf k}\psi|^2 
     +      |\slashed\nabla\widetilde{\mathfrak{D}}^{\bf k} \psi|^2\right)
  +r^{p-3}|\widetilde{\mathfrak{D}}^{\bf k}\psi|^2 
  .
 \end{align}
 We remind the reader that unlike the boundary and bulk energies of 
 Section~\ref{boundaryenergies}, the energies above, \emph{which are to
 appear in bulk integrals only}, have been defined with respect to the untilded $\mathfrak{D}^i$
 operators (except for the terms only supported in $r\ge R$, for which the two definitions 
 would have been equivalent and it is more convenient to retain $\widetilde{\mathfrak{D}}^i$).

\subsubsection{Master energies}
\label{masterenergiessec}

It was convenient in~\cite{DHRT22} to bundle some of the energies together as ``master energies''.
We define thus the following for $\delta\le p \le 2-\delta$:
\begin{eqnarray*}
\Xpk(\tau_0,\tau_1)	[\psi]			&:=&		\sup_{\tau'\in[\tau_0,\tau_1]} \Epk (\tau')+\sup_{v:\tau_1\le v(\tau)} \Fpk(v,\tau_0,\tau_1)
			+ \int_{\tau_0}^{\tau_1}
			 \,\, \Epminusonek'(\tau') 
			 d\tau' ,\\
{}^\rho \Xpk(\tau_0,\tau_1)	[\psi]			&:=& \sup_{\tau'\in[\tau_0,\tau_1]} \Epk(\tau')+\sup_{v:\tau_1\le \tau(v)}  \Fpk(v,\tau_0,\tau_1)
			 +\int_{\tau_0}^{\tau_1}\left( \,\,    {}^{\chi} \Epminusonek'(\tau')+ {}^{\tilde\rho}\,\Ezerominusthreeminusdeltakminusone'(\tau')+ \Epminusonekminustwo'
			 										(\tau')	\right)d\tau'	,	\\
{}^\chi \Xpk(\tau_0,\tau_1)		[\psi]			&:=&  \sup_{\tau'\in[\tau_0,\tau_1]}\Epk (\tau')+\sup_{v:\tau_1\le \tau(v)} \Fpk(v,\tau_0,\tau_1) 
			+\int_{\tau_0}^{\tau_1}\left( \, \,   {}^{\chi} \Epminusonek' (\tau') + \Epminusonekminusone' (\tau') \right)d\tau' .
\end{eqnarray*}
Note that we have retained the $\rho$ notation on ${}^\rho \Xpk$ for ease of comparison 
with~\cite{DHRT22}, even though there is no longer a function $\rho$. Thus, $\chi$ replaces
$\rho$ on the right hand side in the definition. Note that ${}^\rho \Xpk(\tau_0,\tau_1)$
and ${}^\chi \Xpk(\tau_0,\tau_1)$ now differ only with respect to lower order terms.
(Note however  that in contrast  to~\cite{DHRT22}, we included an extra $\Epminusonekminustwo'(\tau')$ term for convenience in the definition of ${}^\rho \Xpk(\tau_0,\tau_1)$.)

For $p=0$, in the analogous definitions, $p-1$ is replaced by $-1-\delta$:
\begin{eqnarray*}
\Xzerok(\tau_0,\tau_1) 	[\psi]				&:=&   \sup_{\tau'\in[\tau_0,\tau_1]}\Ezerok (\tau')+\sup_{v:\tau_1\le \tau(v)} \Fzerok(v,\tau_0,\tau_1)
			 +\int_{\tau_0}^{\tau_1}  \, \Ezerominusoneminusdeltak'(\tau') d\tau' ,\\
 {}^\rho \Xzerok(\tau_0,\tau_1) 	[\psi]			&:=&  \sup_{\tau'\in[\tau_0,\tau_1]} \Ezerok (\tau')+\sup_{v:\tau_1\le \tau(v)}  \Fzerok(v,\tau_0,\tau_1)
 			+\int_{\tau_0}^{\tau_1}  \left(\,\, {}^{\chi} \Ezerominusoneminusdeltak' (\tau') + 
								{}^{\tilde\rho}\,\Ezerominusthreeminusdeltakminusone'	(\tau')+\Ezerominusoneminusdeltakminustwo' (\tau') \right)d\tau' , \\
{}^\chi \Xzerok(\tau_0,\tau_1)		[\psi]		&:=&	  \sup_{\tau'\in[\tau_0,\tau_1]}\Ezerok (\tau')+\sup_{v:\tau_1\le \tau(v)} \Fzerok(v,\tau_0,\tau_1)
			+\int_{\tau_0}^{\tau_1} \left(\,\,  {}^{\chi} \Ezerominusoneminusdeltak'(\tau') 
								+ \Ezerominusoneminusdeltakminusone'(\tau')\right) d\tau'.
\end{eqnarray*}

As discussed in~\cite{DHRT22}, the case $p=0$ is anomalous, and thus we will need
to also consider the following stronger energies which will appear \emph{on the right hand side} of $p=0$ estimates:
\begin{eqnarray*}
\Xzeroplusk(\tau_0,\tau_1) [\psi]				&:=&  \sup_{\tau'\in[\tau_0,\tau_1]}\Ezerok (\tau')+\sup_{v:\tau_1\le \tau(v)} \Fzerok(v,\tau_0,\tau_1) 
			+\int_{\tau_0}^{\tau_1}    \, \Edeltaminusonek'(\tau') d\tau' , 			\\
 {}^\rho \Xzeroplusk(\tau_0,\tau_1) 	[\psi]		&:=&   \sup_{\tau'\in[\tau_0,\tau_1]}\Ezerok (\tau')+\sup_{v:\tau_1\le \tau(v)} \Fzerok(v,\tau_0,\tau_1)
 			 +\int_{\tau_0}^{\tau_1}  \left(  {}^{\chi} \Edeltaminusonek' (\tau')+ {}^{\tilde\rho}\,\Ezerominusthreeminusdeltakminusone'(\tau')
			 						\right) d\tau' , \\
{}^\chi \Xzeroplusk(\tau_0,\tau_1)[\psi]			&:=&  \sup_{\tau'\in[\tau_0,\tau_1]}\Ezerok (\tau')+\sup_{v:\tau_1\le \tau(v)} \Fzerok(v,\tau_0,\tau_1) 
			+\int_{\tau_0}^{\tau_1}   \left( {}^{\chi} \Edeltaminusonek'(\tau') + \Edeltaminusonekminusone'(\tau')\right) d\tau'.
\end{eqnarray*}
We recall also the general properties that $p'\ge p$, $k'\ge k$ 
implies $\Xpprimekprime\gtrsim \Xpk$, ${}^\chi\Xpprimekprime\gtrsim {}^\chi \Xpk$,
${}^\rho \Xpprimekprime\gtrsim {}^\rho \Xpk$, while
\begin{equation}
\label{therelationweknow}
\Xpkminusone \lesssim {}^\chi \Xpk.
\end{equation}

For the bounding of errors arising from pseudodifferential commutation,
it will be useful to define finally 
 \begin{eqnarray*}
  \Xzerok{}^*(\tau_0,\tau_1) 	[\psi]			&:=&  \sup_{\tau'\in[\tau_0,\tau_1]} \Ezerok (\tau')
 +\int_{\tau_0}^{\tau_1}   \Ezerominusoneminusdeltak' _{r\le R_{\rm freq}}(\tau') d\tau', 
 \end{eqnarray*}
 where the $r\le R_{\rm freq}$ subscript here denotes that the integrand in the definition
 of the energy is to be restricted to that set.
Note that for $\tau_1\ge \tau_0+1$, this satisfies
\begin{equation}
\label{couldrefertothisexpl}
    \Xzerok{}^*(\tau_0,\tau_1) 		\lesssim (\tau_1-\tau_0)
  \sup_{\tau'\in[\tau_0,\tau_1]} \Ezerok (\tau').
\end{equation}

We will use the notation $\ll \mkern-6mu k$ 
or $\lesslessk$
to denote some particular positive integer, depending on $k$,
which may be different on different instances of our use of the notation,
such that $\lesslessk \le k$, and in fact, $\lesslessk$ is ``much less than $k$'', 
 provided $k$ is sufficiently
large. In particular, for all positive integers $n$, 
we assume there exists a $k(n)$ for which $k\ge k(n)$ implies $\lesslessk \le k-n$.

\subsubsection{Sobolev inequalities and interpolation of $p$-weighted energies}
We recall the elementary Sobolev and interpolation estimates
used in~\cite{DHRT22}.

\begin{proposition}[Proposition~3.6.10 of~\cite{DHRT22}]
\label{largeasobolevforfunctions}
Let $\psi$ be a smooth function defined on some neighbourhood of $\Sigma(\tau)$.
Then  we have
\begin{equation}
\label{largeasobolev}
\sum_{|{\bf k}|\le k-3}
\sup_{x\in \Sigma(\tau) \cap \{r \le R\} }|\widetilde{\mathfrak{D}}^{\bf k} \psi(x)|^2 \lesssim
\min\left\{ \,
\Ezerominusoneminusdeltak'(\tau)[\psi]	,\,  \Ezerok(\tau)[\psi]	 \right\}.
\end{equation}
\end{proposition}

\begin{proposition}[Proposition~3.6.11 of~\cite{DHRT22}]
\label{largeainterpolationprop}
For $\delta$ as fixed in~\eqref{largeafixeddelta}, one has 
the following interpolation inequalities:
\begin{align}
\label{largeainterpolationstatementnewback}
\Eonek (\tau)[\psi]	
&\lesssim 
  \left( \, \, \Eoneminusdelk  (\tau)[\psi]	\right)^{1-\delta} \left(\, \, \Etwominusdelk (\tau)[\psi]	\right)^{\delta},
  \\
  \label{largeainterpolationstatementbulk}
 \Edeltaminusonek' (\tau)[\psi]	 &\lesssim  \left( \, \, \Ezerominusoneminusdeltak'(\tau) [\psi]	\right)^{\frac{1-\delta}{1+\delta}} \left( \, \, \Ezerok'(\tau) [\psi]	\right)^{\frac{2\delta}{1+\delta}}.
\end{align}
\end{proposition}

\subsubsection{Elliptic estimates}
\label{ellipticestimatessection}

Finally, we recall the following elliptic estimates of~\cite{DHRT22}, specialised to Kerr:
\begin{proposition}[Proposition 3.6.3 of~\cite{DHRT22}]
\label{ellproplargea}
Fix $|a|<M$ and 
let $(\mathcal{M},g_{a,M})$ denote the Kerr manifold.
Let $\psi$ be a smooth function on $\mathcal{R}(\tau_0,\tau_1)$
and let $\tau_0\le \tau'\le \tau_1$.  
Then for all $k\ge 1$ and for all $r_+< r'_-<r'<r''<r''_+\le R$,
\begin{equation}
\label{ellestimnoYforlargea}
\int_{\Sigma(\tau')\cap \{ r' \le r\le r'' \}  } \sum_{|{\bf  k}|\le k+1}
|\widetilde{\mathfrak{D}}^{\bf {{k}}}\psi|^2\\
 \lesssim 
\int_{\Sigma(\tau')\cap \{r'_- \le r\le  r''_+\} } \sum_{0\le  |{\bf k}|\le k}\sum_{|{\bf k}'|=1 }|\mathring{\mathfrak{D}}^{\bf k}{\mathfrak{D}}^{{\bf k}'}   \psi|^2 
+\sum_{|{\bf  k}|\le 1}
|\widetilde{\mathfrak{D}}^{\bf{k}}\psi|^2
+\sum_{|{\bf  k}|\le k-1} |\widetilde{\mathfrak{D}}^{\bf  k}\Box_{g_{a,M}}\psi|^2,
\end{equation}
where here $\lesssim$ depends on the choice of $r_+<r'_-<r'<r''<r_+''$. 
We also have the estimate, for all $r_0\le r' \le r_1$,
\begin{equation}
\label{ellestimnorestrictionforlargea}
\int_{\Sigma(\tau')\cap \{r\ge r' \}  } \sum_{|{\bf  k}|\le k+1}
|\widetilde{\mathfrak{D}}^{\bf {{k}}}\psi|^2\\
 \lesssim 
\int_{\Sigma(\tau')\cap \{r\ge r' \} } \sum_{|{\bf k}|\le  k+1} |\mathfrak{D}^{\bf k}\psi|^2 
+\sum_{|{\bf  k}|\le k-1} |\widetilde{\mathfrak{D}}^{\bf  k}\Box_{g_{a,M}}\psi|^2.
\end{equation}

Note that the analogous statements to~\eqref{ellestimnoYforlargea}, \eqref{ellestimnorestrictionforlargea} 
with integration on $\mathcal{R}(\tau_0,\tau_1) \cap \{r'\le r \le r''\}$, etc.,
follow immediately in view of the coarea formula.

In fact, even without assuming $r''_+\le R$ (i.e.~allowing in particular $r''>R$), we still have the spacetime elliptic estimate:
\begin{equation}
\label{ellipticestimnoYspacetime}
\int_{\mathcal{R}(\tau_0,\tau_1)\cap \{ r' \le r\le r'' \}  } \sum_{|{\bf  k}|\le k+1}
|\widetilde{\mathfrak{D}}^{\bf {{k}}}\psi|^2\\
 \lesssim 
\int_{\mathcal{R}(\tau_0,\tau_1)\cap \{r'_- \le r\le  r''_+\} } \sum_{0\le  |{\bf k}|\le k}\sum_{|{\bf k}'|=1 }|\mathring{\mathfrak{D}}^{\bf k}{\mathfrak{D}}^{{\bf k}'}   \psi|^2 
+\sum_{|{\bf  k}|\le 1}
|\widetilde{\mathfrak{D}}^{\bf{k}}\psi|^2
+\sum_{|{\bf  k}|\le k-1} |\widetilde{\mathfrak{D}}^{\bf  k}\Box_{g_{a,M}}\psi|^2,
\end{equation}
where again $\lesssim$ depends on the choice of $r_+<r'_-<r'<r''<r_+''$. 
\end{proposition}

\begin{remark}
We remind the reader that~\eqref{ellipticestimnoYspacetime}  in the case $r''>R$ 
must be proven directly as a spacetime estimate and not as a corollary of an estimate
of the type~\eqref{ellestimnoYforlargea}
on fixed slices $\Sigma(\tau)$, precisely 
because the $\Sigma(\tau)$ slices are null in the region $r\ge R$.
(See the proof of Proposition 3.6.3 of~\cite{DHRT22} for more details.)
\end{remark}

\section{Integrated local energy decay estimates and physical space currents for $\Box_{g_{a,M}}\psi =F$ on subextremal Kerr $|a|<M$}
\label{replacement}

In this section we will introduce the  estimates and identities 
concerning the linear wave operator $\Box_{g_{a,M}}$ on subextremal Kerr $(\mathcal{M}, g_{a,M})$, $|a|<M$
which as 
discussed already in the introduction constitute the basis of our method.

The first estimate (Theorem~\ref{blackboxoneforkerr})
is the black box degenerate  integrated local energy decay 
estimate
for solutions $\psi$ of $\Box_{g_{a,M}}\psi = F$
 corresponding to~black box assumption 1.~of~\cite{DHRT22}. 
We will state  this in Section~\ref{recallingablackbox}. This estimate
was proven in~\cite{partiii}.

For the purposes of the remainder of the section, we will need to review
the geometric Fourier analysis with respect to the Killing fields $T$ and $\Omega_1$ and
the full 
Carter separation~\cite{carter1968hamilton} of $\Box_{g_{a,M}}$.
This will be done in Section~\ref{frequencyanalysisreviewsec}.
We will in particular introduce the notion of a $(t^*,\phi^*)$-pseudodifferential operator 
satisfying a rudimentary pseudodifferential calculus.

As discussed already in Section~\ref{reviewofILEDandimprove}, it turns out that we will need in fact a slight refinement of Theorem~\ref{blackboxoneforkerr},
 including now elliptic estimates on certain non-trapped
frequency ranges, which can be defined with respect to Carter's separation. 
We will state that refinement in Section~\ref{refinedlinearhere}
as Theorem~\ref{refinedblackboxoneforkerr}.

The main new element of the present paper will be presented
in Section~\ref{generalisedcoerc}. There, we shall introduce the set of
coercive currents  $J^{{\rm main}, n}$, $K^{{\rm main},n}$, discussed already in 
Section~\ref{wavepacketlocintrosec},
which will replace the role of assumption 2.~of~\cite{DHRT22}.
The coercivity properties of these currents hold only
when  applied to a corresponding set of wave packets 
defined by a finite set of frequency projection operators $P_n$, introduced here.
These operators may be understood as zeroth order $(t^*,\phi^*)$-pseudodifferential 
operators in the formalism of Section~\ref{frequencyanalysisreviewsec}.
Moreover, the allowed errors to coercivity  of $K^{{\rm main},n}$ 
include highest order terms which can only be controlled by
the elliptic improvement of Theorem~\ref{refinedblackboxoneforkerr}. In order to capture this,
however, we must further frequency analyse our expressions $K^{{\rm main},n}$ 
with respect to Carter's separation.  Since the very design of the currents is intimately
connected to showing this property, the precise definitions of the functions 
constituting the currents   $J^{{\rm main}, n}$, $K^{{\rm main},n}$
 is in fact deferred to Appendix~\ref{carterestimatesappend}, 
 which contains all arguments making use
of Carter's separation. Thus, the proof of most of the results of Section~\ref{generalisedcoerc}
will depend on constructions and proofs 
postponed to Appendix~\ref{carterestimatesappend}.

 Finally, in Section~\ref{absorbingetcsec} we shall introduce additional currents
 to obtain
  improved coercivity properties near the
 boundary $\mathcal{S}$ and in the far-away region, further  exploiting the red-shift property
 and the $r^p$ estimate hierarchy.
 
As defined in the sections above, the above currents all still depend on some free parameters which
must be taken large or small accordingly.
 In Section~\ref{fixersection}, we shall finalise our choice of all parameters appearing
 in our currents, in particular fixing thus the choice of the $r$-parameters of 
 Section~\ref{noteonconstants}
 so as to depend only on $a$, $M$.
 
\subsection{The black box integrated local energy decay assumption for 
subextremal Kerr~\cite{partiii}}
\label{recallingablackbox}
 
 We review here  the black box local energy decay assumption (statement 1.)~required by~\cite{DHRT22} and proven in~\cite{partiii} for the Kerr case in the general
subextremal range $|a|<M$.
 The statement will refer to the degeneration function $\chi$ defined in Section~\ref{degbulkenergiesdefsec}.
 For the right hand side, it will also refer to a vector field $V_p$  and function  $w_p$
which are defined in Appendix B.1 of~\cite{DHRT22} but may also be taken to be the
functions (with respect to the parametrisation of currents in Section~\ref{Covariantenergyidentitiessec})
 corresponding to the current $\Jp^{\rm total}$ of Section~\ref{absorbingetcsec}.
 \begin{theorem}
 \label{blackboxoneforkerr}
Fix $|a|<M$ and 
let $(\mathcal{M},g_{a,M})$ denote the Kerr manifold and let $\psi$ satisfy
$\Box_{g_{a,M}} \psi =F$ on $\mathcal{R}(\tau_0,\tau_1)$. 

For all $k\ge 0$,  for all $0<\delta\le p \le 2-\delta$ and for all
$\tau_0\le \tau\le \tau_1$ we have the following statement:
 \begin{align*}
\sup_{v:\tau\le \tau(v)}\Fpk(v,\tau_0,\tau)[\psi]&+ \Epk(\tau)[\psi] +\Ezerok_{\mathcal{S}}(\tau_0,\tau)[\psi]+ \int_{\tau_0}^{\tau}{}^\chi\Epminusonek'(\tau')[\psi]
d\tau' + \int_{\tau_0}^{\tau}  \Epminusonekminusone'(\tau') [\psi]d\tau' \\
\nonumber
&\lesssim  \Epk(\tau_0)[\psi] +\int_{\mathcal{R}(\tau_0,\tau)\cap \{r\ge R\}}
\sum_{|{\bf k}|\le k}\left(|V^\mu_p\partial_\mu ({\mathfrak{D}}^{\bf k}\psi)|+|w_p{\mathfrak{D}}^{\bf k}\psi| \right)|{\mathfrak{D}}^{\bf k}F|+\int_{\mathcal{R}(\tau_0,\tau)}\sum_{|{\bf k}| \le k}|{\mathfrak{D}}^{\bf k} {F}|^2\\
&\qquad +
\sum_{|{\bf k}|\le k}
\sqrt{\int_{\mathcal{R}(\tau_0,\tau)\cap \{r\le R\}}
\left( |L\mathfrak{D}^{\bf k} \psi|
+|\underline{L} \mathfrak{D}^{\bf k} \psi | +| \slashed\nabla \mathfrak{D}^{\bf k} \psi| 
+|\mathfrak{D}^{\bf k}\psi| \right)^2
}
\sqrt{
\int_{\mathcal{R}(\tau_0,\tau)\cap \{r\le R\}}
|{\mathfrak{D}}^{\bf k}F|^2
}\\
&\qquad+  \int_{\mathcal{R}(\tau_0,\tau)\cap \{r\le R\}}\sum_{|{\bf k}| \le k-1}|{\widetilde{\mathfrak{D}}}^{\bf k} F |^2
+\int_{\Sigma(\tau)\cap\{r\le R\}}\sum_{|{\bf k}|\le k-1}|\widetilde{\mathfrak{D}}^{\bf k}F|^2.
\end{align*}
In the case $p=0$, an identical statement holds where we replace $p-1$ by $-1-\delta$.
\end{theorem}

\begin{proof} The $p=0$, $k=0$, $F=0$ statement is contained in the proof of~\cite{partiii}, 
given our definition of $\chi$ in~\eqref{chidefimportant} 
and the relation of the parameters $r_{\rm pot}$,
$R_{\rm pot}$  (as these are finally chosen in 
Section~\ref{fixingthecovering}) to the properties of the potential $V_0$.
For the statement for general $F$, see Theorem D.1 of~\cite{DHRT22} (where
note we only here apply the weaker version whose proof is contained in Section D).
Finally, to promote this statement to higher $k$ and $0<\delta \le p \le 2-\delta$,  see
Section 3.6 of~\cite{DHRT22}, in particular Remark 3.6.9.  
\end{proof}

\subsection{The $(t^*,\phi^*)$ frequency analysis, superradiance
and Carter's complete separation}
\label{frequencyanalysisreviewsec}

The remainder of this section will need more precise statements which require
frequency analysis.
We  review here the natural  geometric frequency analysis 
associated to the operator $\Box_{g_{a,M}}$. 
We will distinguish between Fourier analysis with respect to Kerr star coordinates 
$t^*$ and $\phi^*$, on the one hand, which just depends on the fact that $T$ and $\Omega_1$
are Killing
fields, and Carter's full separation, which arises from the presence of an additional ``hidden
symmetry''.

We begin in Section~\ref{Schwartzandcutoffs}  with the definition of some basic spaces, cutoffs
and norms which will appear naturally when applying our frequency analysis.
The $(t^*,\phi^*)$ frequency analysis and pseudodifferential calculus is then introduced
in Section~\ref{elementarycalculus}. This will in particular allow us to discuss superradiant
and non-superradiant frequencies in Section~\ref{supernonsuper}.
We shall introduce an alternative Fourier transform with respect to Boyer--Lindquist coordinates
in Section~\ref{altfourier}, which is used in our discussion of Carter's complete separation,
given in Section~\ref{Carterscompleteseparation}.  Finally, we will collect some important Plancherel identities 
in Section~\ref{Plancherelsection}.

\subsubsection{Schwartz space, temporal cutoff functions $\chi_{\tau_0,\tau_1}$ and global mixed 
Sobolev
norms}
\label{Schwartzandcutoffs}

We will only apply Fourier analysis in the the region $r_0\le r\le R_{\rm freq}$, for
a parameter $R_{\rm freq}\le R/4$ to be determined later.
Let us recall that we have an ambient Cartesian coordinate system $(x^0, x^1, x^2, x^3)$ and
the Kerr star coordinates $(t^*, \phi^*, \theta, r)$, where  $x^0=t^*$. 
Let us use $\beta= (\beta_0, \beta_1,\beta_2,\beta_3)$ to denote multi-indices related to 
the Cartesian coordinates and $\alpha=(\alpha_0,\alpha_1)$ to denote multi-indices
related to the $(t^*,\phi^*)$ pair.

It will be natural to consider our Fourier based operators as defined in the Schwartz class with respect to $t^*=x^0$
in the region $r_0\le r\le R_{\rm freq}$, i.e.~to define
\begin{equation}
\label{Schwartzspace}
\mathscr{S}(\mathcal{M}\cap \{r_0\le r\le R_{\rm freq}\}):=\{
\Psi\in C^{\infty}(\mathcal{M}\cap \{r_0\le r\le R_{\rm freq}\} : \forall\beta_0, \tilde\beta\, \exists C_{\beta_0,\tilde\beta}: \sup  | t_*^{\beta_0} \partial^{\tilde\beta} \Psi| \le C_{\beta_0,\tilde\beta}  \} \, .
\end{equation}

To put smooth spacetime functions into the space~\eqref{Schwartzspace}, we will typically
cut off in time. 
Given fixed  $\tau_0+2\le \tau_1$, we
 introduce a cutoff $\chi_{\tau_0,\tau_1}$ which is $0$ for $\tau\le \tau_0$ and
$\tau\ge \tau_1$ 
and $1$ for $\tau_0+\frac12 \le\tau \le \tau_1-\frac12$, and where
the translates  $\chi(\tau-\tau_0)|_{[0,1/2]}$
and $\chi(\tau- \tau_1)_{[-1/2,0]}$ are independent of $\tau_0$, $\tau_1$ respectively.
Clearly, if $\psi\in C^\infty(\mathcal{R}(\tau_0,\tau_1))$, then $\chi_{\tau_0,\tau_1}\psi\in
\mathscr{S}(\mathcal{M}\cap \{r_0\le r\le R_{\rm freq}\})$, $\chi^2_{\tau_0,\tau_1}\psi\in
\mathscr{S}(\mathcal{M}\cap \{r_0\le r\le R_{\rm freq}\})$, etc.

We may also define the global $L^2$ norm on~\eqref{Schwartzspace}
by the expression:
\[
\| \Psi\|^2_
{L^2(\mathcal{M}\cap \{r_0\le r\le R_{\rm freq}) )}
:= \int_{r_0}^{R_{\rm freq}} \int_{-\infty}^{\infty} \int_0^{2\pi}\int_0^{\pi} 
|\Psi|^2 (r^2+a^2\cos^2\theta)  dr\,dt^*\,\sin\theta\,  d\phi^* \, d\theta
\]
and the global mixed Sobolev norm:
\begin{equation}
\label{Sobolevnorm}
\|\Psi \|_{H^{k,j}(\mathcal{M}\cap \{r_0\le r\le R_{\rm freq}) )} := \sum_{|\alpha| \le k,\,  |\beta|\le j}    \|\partial_{t^*}^{\alpha_0}
\partial_{\phi^*}^{\alpha_1} \partial_{x^0}^{\beta_0} \partial_{x^1}^{\beta_1}
\partial_{x^2}^{\beta_2}\partial_{x^3}^{\beta_3} \Psi\|_{L^2(\mathcal{M}\cap \{r_0\le r\le R_{\rm freq}\} )}.
\end{equation}
We will also consider a similar  norm on $\mathcal{S}$:
\begin{equation}
\label{Sobolevnormhorizon}
\|\Psi \|_{H^{k,j}(\mathcal{S})} := \sum_{|\alpha| \le k,\,  |\beta|\le j}    \|\partial_{t^*}^{\alpha_0}
\partial_{\phi^*}^{\alpha_1} \partial_{x^0}^{\beta_0} \partial_{x^1}^{\beta_1}
\partial_{x^2}^{\beta_2}\partial_{x^3}^{\beta_3}\Psi\|_{L^2(\mathcal{S})},
\end{equation}
as well as the restriction of~\eqref{Sobolevnorm} to other subsets of $\mathcal{M}$.
Let us note that by our definition of the ambient Cartesian coordinates in Section~\ref{coordsheresec} in terms of Kerr star coordinates,
we have the commutation relations 
\begin{equation}
\label{commutationsuccess}
[\partial_{t^*}, \partial_{x^\mu}]=0,\qquad [\partial_{\phi^*}, \partial_{x^3}]= 0,
\end{equation}
\begin{equation}
\label{commutationfailure}
[\partial_{\phi^*}, \partial_{x^1}]=\partial_{x^2}, \qquad [\partial_{\phi^*}, \partial_{x^2}]=-\partial_{x^1},
\end{equation}
between the two sets of coordinate vector fields.

We may extend the above definitions to negative integers $k\le 0$ in the standard way, using the
equivalent chacterisation
\begin{align}
\label{Sobolevnormequiv}
\|\Psi \|_{H^{k,j}(\mathcal{M}\cap \{r_0\le r\le R_{\rm freq}\} )} \\
\nonumber
\sim 
\sum_{|\alpha| = |k|,\,  |\beta|\le j}    \| (1+ |\omega|)^{{\rm sgn} (k)\alpha_0}
(1+|m|)^{{\rm sgn}(k) \alpha_1}  \mathfrak{F}  [ \partial_{x^0}^{\beta_0} \partial_{x^1}^{\beta_1}
\partial_{x^2}^{\beta_2}\partial_{x^3}^{\beta_3} \Psi]  (\omega, m, r,\theta) \|_{L^2(\mathbb R)l^2(\mathbb Z)L^2( r^2 dr \sin \theta d\theta)},
\end{align}
etc.

\subsubsection{The $(t^*,\phi^*)$ frequency analysis and elementary pseudodifferential calculus}
\label{elementarycalculus}

We define now $\mathfrak{F}$ to be the Fourier transform operator 
with respect to $t^*$ and $\phi^*$,
\[
\mathfrak{F}[\Psi] (r,\theta, \omega,m) = \int_{-\infty}^{\infty} \int_0^{2\pi} e^{i\omega t^*} e^{-im \phi^*} 
\Psi (t^*, \phi^*, r, \theta )dt^* d\phi^*.
\]

Denoting $\widetilde{\mathscr{S}}:= \mathfrak{F} (\mathscr{S})$ we have that 
$\mathfrak{F}^{-1}:\widetilde{\mathscr{S}} \to \mathscr{S}$ is given by the inverse Fourier transform:
\[
\mathfrak{F}^{-1}[U] ( t^*,\phi^*, r,\theta) = \int_{-\infty}^{\infty} \sum_{m\in \mathbb Z} e^{-i\omega t^*} e^{im \phi^*}  U (r,\theta,\omega, m) d \omega.  
\]
We moreover have the Plancherel formula
\begin{equation}
\label{firstPlanch}
\int_0^{2\pi} \int_{-\infty}^\infty |\Psi|^2 (t^*,\phi^*, r,\theta) d\phi^*  dt^*
= \int_{-\infty}^\infty \sum_{m}  | \mathfrak{F}[\Psi](r, \theta, \omega, m)|^2 d\omega.
\end{equation}

We have the relations (note the signs!):
\[
\partial_{t^*}\Psi (t^*,\phi^*, r,\theta) = \frac{-i}{\sqrt {2\pi}} \int_{-\infty}^\infty 
\sum_{m} \omega \, \mathfrak{F}[\Psi](r,\theta,\omega,m)  e^{-i\omega t^*} e^{im\phi^*} d\omega,
\]
\[
\partial_{\phi^*} \Psi (t^*, \phi^*,r,\theta) = \frac{i}{\sqrt {2\pi}} \int_{-\infty}^\infty 
\sum_{m} m \,  \mathfrak{F}[\Psi](r,\theta,\omega,m) e^{-i\omega t^*} e^{im\phi^*} d\omega
\]
and the additional Plancherel formulae:
\begin{equation}
\label{babyPlanchstar}
\int_0^{2\pi}  \int_{-\infty}^\infty  |\partial_{t^*} \Psi|^2 (t^*, \phi^*, r,\theta) \, d\phi^* \,  dt^* = 
\int_{-\infty}^\infty \sum_{m} \omega^2 \, | \mathfrak{F}[\Psi](r, \theta, \omega, m)|^2 d\omega,
\end{equation}
\begin{equation}
\label{babyPlanchtwostar}
\int_0^{2\pi} \int_{-\infty}^\infty  |\partial_{\phi^*} \Psi|^2 (t^*,\phi^*,r,\theta) \,  d\phi^* \, dt^* = 
\int_{-\infty}^\infty \sum_{m} m^2 \, | \mathfrak{F} [\Psi](r, \theta, \omega, m)|^2 d\omega.
\end{equation}

We will consider \emph{$(t^*,\phi^*)$-pseudodifferential operators} of the form
\begin{equation}
\label{operatorsoftheform}
Q[\Psi] (t^*,\phi^*, r,\theta) = \int_{-\infty}^{\infty} \sum_{m\in \mathbb Z} q_{\rm symb}(\omega, m,x^0, x^1,x^2,x^3) e^{-i\omega t^*} e^{im \phi^*} 
\mathfrak{F}[\Psi ](r,\theta,\omega , m )d\omega 
\end{equation}
where
\begin{equation}
\label{orderdef}
| \partial_{x^0}^{\beta_0} \partial_{x^1}^{\beta_1}
\partial_{x^2}^{\beta_2}\partial_{x^3}^{\beta_3}  \partial_\omega^{\alpha_0}  D_m^{\alpha_1}q_{\rm symb}(\omega,m,x^0,x^1, x^2 ,x^3) |  \le
C_\beta (1+|\omega| +|m|)^{s-\alpha_0-\alpha_1} .
\end{equation}
In the above
\[
D_m^{0}U:= U , \qquad D_m^{1} U(\omega, m, x^0, x^1, x^2, x^3)  := U(\omega, m, x^0, x^1, x^2, x^3)-
U(\omega, m-1, x^0, x^1, x^2, x^3)
\]
and for $\alpha_1\ge 2$,
\[
D_m^{\alpha_1} := D_m (D_m^{\alpha_1-1} U).
\]
We call $s$ the \emph{order} of the operator. For convenience we will only consider
integer order $s\in \mathbb Z$. We will say that such an operator is 
$(t^*,\phi^*)$-\emph{smoothing}
if~\eqref{orderdef} holds for all $s<0$  and $C_\beta(s)$.

We will need only the most rudimentary facts about the pseudo-differential calculus
for operators of the form~\eqref{operatorsoftheform}.
We first note (cf.~\eqref{commutationsuccess}, \eqref{commutationfailure}) the following
  computation:
\begin{lemma}
Let $Q$ be a $(t^*,\phi^*)$-pseudodifferential  operator of order $s$ with symbol $q_{\rm symb}(\omega,m,x^0,x^1,x^2,x^3)$ and let $\Psi\in \mathscr{S}$. Then we have the commutation formulae:
\begin{equation}
\label{theimmediateones}
[\partial_{x^0},Q]\Psi = \tilde{Q}_0\Psi , \qquad [\partial_{x^3}, Q]\Psi = \tilde{Q}_3\Psi,
\end{equation}
\begin{equation}
\label{similarformulato}
[\partial_{x^1} , Q ]\Psi = \tilde{Q}_1\Psi - \hat{Q}\partial_{x^2}\Psi -\hat{Q}'\partial_{x^1}\Psi, \qquad
[\partial_{x^2} , Q ]\Psi = \tilde{Q}_2\Psi + \hat{Q}\partial_{x^1}\Psi + \hat{Q}'\partial_{x^2}\Psi
\end{equation}
where
$\tilde{Q}_j$ is a $(t^*,\phi^*)$-pseudodifferential operator of order $s$ with symbol $\partial_{x^j}q_{\rm symb}$, for $j=0,\ldots, 3$, while 
$\hat{Q}$, $\hat{Q}'$  are $(t^*,\phi^*)$-pseudodifferential operators of order $s-1$, with symbols
defined respectively by  
\[
\frac1{2i} (D^1_{m+1} +D^1_{m})q_{\rm symb}, \qquad  \frac1{2} (D^1_{m+1} -D^1_{m})q_{\rm symb}.
\]
\end{lemma}
\begin{proof}
The commutation formulae~\eqref{theimmediateones} are immediate. 
For~\eqref{similarformulato}, noting the relations
between coordinate vector fields of the two coordinate systems,
we then compute explicitly the commutators using also~\eqref{theimmediateones}. 
\end{proof}

We may now easily show the following propositions,
whose proofs, in view of the above lemma, are elementary adaptations
of standard results of the pseudodifferential calculus (see e.g.~Chapter 6 of~\cite{stein1993harmonic}).

\begin{proposition}
\label{mixedsobolevboundednessprop}
Let $Q$ be a $(t^*,\phi^*)$-pseudodifferential operator of order $\hat{s}$. Then 
for the global  mixed Sobolev norms defined 
by~\eqref{Sobolevnorm} or~\eqref{Sobolevnormhorizon}, we have 
\[ 
\|  Q[\Psi] \|_{H^{k-\hat{s}, j}} \lesssim C(Q, k, j) \| \Psi \|_{H^{k, j}}.
\]
Thus, if $Q$ is  $(t^*,\phi^*)$-smoothing, then
\begin{equation}
\label{quantestimhere}
\|  Q[\Psi] \|_{H^{k+s, j}} \lesssim C(Q, k, j, s) \| \Psi \|_{H^{k, j}}
\end{equation}
for all $s>0$.
\end{proposition}

\begin{proposition}
\label{smoothingprop}
Let $Q$ be a $(t^*,\phi^*)$-pseudodifferential operator  of order $\hat{s}$, and let $\mu_1$, $\mu_2$
be smooth compactly supported functions on  $\mathcal{M}\cap \{r_0\le r\le R_{\rm freq}\}$
supported in $\tau\le \tau_1$ and $\tau\ge \tau_2$, respectively, for some $\tau_2>\tau_1$.
Then $\mu_1Q\mu_2$ is $(t^*,\phi^*)$-smoothing
and the constant $C(\mu_1Q\mu_2,k, j, s)$ appearing in~\eqref{quantestimhere} can be bound by $C(Q,k, j, s, \tau_2-\tau_1, \|\mu_1\|_{H^{0,\tilde{s}}}, \|\mu_2\|_{H^{0,\tilde{s}}})$ for  an  $\tilde{s}$ depending only on
 $s$, $k$ and $j$, where the $\tau_2-\tau_1$ dependence of $C$ is such that
$C |\tau_2-\tau_1+1|^\ell <\infty$ for all $\ell\ge 0$.
\end{proposition}

\begin{proposition}
\label{commutatorestimate}
Let $Q$  be a $(t^*,\phi^*)$-pseudodifferential operator  of order $s$ and let $\mu$ be a smooth function of compact support on $\mathcal{M}\cap \{r_0\le r\le R_{\rm freq}\}$.
Then $[Q,\mu]$ is of order $s-1$.
\end{proposition}

\begin{proposition}
\label{fornonlinearpseudo}
Let $Q$  be a $(t^*,\phi^*)$-pseudodifferential operator   of order $0$ whose symbol
$q_{\rm symb}=q_{\rm symb}(\omega, m)$ is a function of frequency only. Let $D=\sum_{|\beta|\le \tilde{k}} 
a_{\beta} \partial_x^\beta $ where $\beta$ is a multiindex and
$\partial_{x^\beta}$ denotes the differential operator
$\partial^{\beta_0}_{x^0}\cdots \partial^{\beta_3}_{x^3}$  with respect to the ambient
Cartesian structure. Suppose the coefficients $\{a_{\beta}\}$ are smooth and compactly supported
in $\mathcal{M}\cap \{r_0\le r\le R_{\rm freq}$\}. Then for Schwartz functions
$\Psi$ on $\mathcal{M}\cap \{r_0\le r\le R_{\rm freq}\}$, we have
\begin{align}
\nonumber
\| [ Q, D] \Psi & \|_{H^{k+1, j-\tilde{k}}(\mathcal{M}\cap \{r_0\le r\le R_{\rm freq}\}) }\\
\nonumber
&\leq C(Q, k,\tilde{k}, j)  
 \|\{ a_{\beta}\} \|_{H^{0, s}( (\mathcal{M}\cap \{r_0\le r\le R_{\rm freq}\}) )}\sup_{\tau\in \mathbb R} \|\chi_{\tau,\tau+2}\Psi\|_{H^{k, j} (\mathcal{M}\cap \{r_0\le r\le R_{\rm freq}\}) }\\
 &\qquad + C(Q, k,\tilde{k}, j)   \|\{ a_{\beta}\} \|_{H^{0, s}( (\mathcal{M}\cap \{r_0\le r\le R_{\rm freq}\}) )}\| \Psi\|_{H^{k-2, j} (\mathcal{M}\cap \{r_0\le r\le R_{\rm freq}\}) }
\end{align}
for  some $s$ depending on $\tilde{k}$,  $j$ and $k$.
\end{proposition}
\begin{proof}
Note that 
the first term on the right hand side appears as a consequence of the standard Poisson bracket type commutator calculation, the bound of 
Proposition~\ref{mixedsobolevboundednessprop} and the  additional localisation to time intervals of size 2. The second term accounts for the error terms from the semi-classical commutation and the smoothing terms resulting from the localisation of the first term via Proposition~\ref{smoothingprop}.
\end{proof}

For related pseudodifferential setups in Kerr and Kerr--de Sitter, see~\cite{holzegel2023wave, 
mavrogianniskerrrelative}.

\subsubsection{Superradiant and non-superradiant frequencies}
\label{supernonsuper}

As discussed already in Section~\ref{reviewoffreq}, the $(t^*,\phi^*)$ frequency analysis of Section~\ref{elementarycalculus} 
already allows us to 
distinguish between so-called \emph{superradiant} and \emph{non-superradiant} 
frequencies.

We say a frequency pair $(\omega, m)$ is \emph{superradiant} if
\begin{equation}
\label{superradiantcondition}
m\ne0 \qquad {\text{and}} \qquad
0 < \frac{\omega}{am}  < \frac{1}{2Mr_+}.
\end{equation}
Otherwise, i.e.~if
\begin{equation}
\label{nonsuperradiantcondition}
a=0, \qquad {\text{or}} \qquad
m=0,\qquad {\text{or}} \qquad
\frac{\omega}{am} \le 0, \qquad {\text{or}}\qquad \frac{\omega}{am} \ge \frac{1}{2Mr_+},
\end{equation}
we say that $(\omega, m)$ is \emph{non-superradiant}.

We recall that the quantity $\frac{a}{2Mr_+}$ represents the \emph{angular velocity} 
of the horizon $\mathcal{H}^+$,
and  the null generator of $\mathcal{H}^+$ can be expressed as~\eqref{Hawkingdef}.
The significance of~\eqref{superradiantcondition} can be understood by considering the $J^T_{g_{a,M}}$ flux
through $\mathcal{H}^+$, where $J^T_{g_{a,M}}$ is the current~\eqref{generalJdef} corresponding
to $V=T$ and all other objects $0$.

\begin{proposition}
Let $\Psi$ be a function in the space~\eqref{Schwartzspace}
such that $\mathfrak{F}[\Psi](r_+, \theta, \omega, m)$ is supported
in the non-superradiant frequencies~\eqref{nonsuperradiantcondition}.
Then
\[
\int_{\mathcal{H}^+} J^T_{g_{a,M}} [\Psi] \cdot {\rm n}_{\mathcal{H}^+}  \ge 0.
\]
\end{proposition}
\begin{proof} 
This follows from the Plancherel identities~\eqref{babyPlanchstar} and~\eqref{babyPlanchtwostar} and the fact that~\eqref{Hawkingdef} is future causal.
\end{proof}

In the language of the pseudodifferential calculus of Section~\ref{elementarycalculus}, 
we have the following
immediate corollary:
\begin{corollary}
Let $Q$ be a $(t^*,\phi^*)$-pseudodifferential operator in the sense of Section~\ref{elementarycalculus} whose symbol is supported in~\eqref{nonsuperradiantcondition},  
let $\psi\in \mathcal{R}(\tau_0,\tau_1)$ and $\chi_{\tau_0,\tau_1}$
denote the cutoff of Section~\ref{Schwartzandcutoffs}. Then
\[
\int_{\mathcal{H}^+} J^T_{g_{a,M}} [Q\chi_{\tau_0,\tau_1}\psi] \cdot {\rm n}_{\mathcal{H}^+}  \ge 0.
\]
\end{corollary}

\subsubsection{The Fourier transform $\mathfrak{F}_{BL}$ with respect
to Boyer--Lindquist coordinates}
\label{altfourier}

Let us note that, restricted to $r>r_+$, we could also define the Fourier transform replacing
$t^*$ and $\phi^*$ with the Boyer--Lindquist $t$ and $\phi$. We will denote
this by $\mathfrak{F}_{BL}$, i.e.~we have
\begin{equation}
\label{unsuit}
\mathfrak{F}_{BL} [\Psi] (r,\theta, \omega,m) = \int_{-\infty}^{\infty} \int_0^{2\pi} e^{i\omega t} e^{-im \phi} 
\Psi (t,\phi, r,\theta)dt d\phi.
\end{equation}

 Note that $\mathfrak{F}[\Psi] (r,\theta, \omega,m) $ and
$\mathfrak{F}_{BL}[\Psi] (r,\theta, \omega,m) $ differ by a phase depending on $(r,t^*)$ which
is however poorly behaved at $\mathcal{H}^+$.

Note that with respect to $\mathfrak{F}_{BL}$, we have the relations (note again the signs!)
\[
\partial_t\Psi (t,\phi, r,\theta) = \frac{-i}{\sqrt {2\pi}} \int_{-\infty}^\infty 
\sum_{m} \omega \, \mathfrak{F}_{BL}[\Psi](r,\theta,\omega,m)  e^{-i\omega t} e^{im\phi} d\omega,
\]
\[
\partial_\phi \Psi (t,\phi, r,\theta) = \frac{i}{\sqrt {2\pi}} \int_{-\infty}^\infty 
\sum_{m} m \,  \mathfrak{F}_{BL}[\Psi](r,\theta,\omega,m) e^{-i\omega t} e^{im\phi} d\omega
\]
and the Plancherel relations (cf.~\eqref{babyPlanchstar},~\eqref{babyPlanchtwostar}):
\begin{equation}
\label{babyPlanch}
\int_0^{2\pi}  \int_{-\infty}^\infty  |\partial_t\Psi|^2 (t,\phi, r,\theta ) \, d\phi \,  dt = 
\int_{-\infty}^\infty \sum_{m} \omega^2 \, | \mathfrak{F}_{BL}[\Psi](r, \theta, \omega, m)|^2 d\omega,
\end{equation}
\begin{equation}
\label{babyPlanchtwo}
\int_0^{2\pi} \int_{-\infty}^\infty  |\partial_\phi \Psi|^2 (t,\phi, r,\theta) \,  d\phi \, dt = 
\int_{-\infty}^\infty \sum_{m} m^2 \, | \mathfrak{F}_{BL}[\Psi](r, \theta, \omega, m)|^2 d\omega.
\end{equation}

As we shall see,
it will be convenient to use~\eqref{unsuit} when invoking Carter's separation in the section
immediately following.

\subsubsection{Carter's complete separation and the radial ODE}
\label{Carterscompleteseparation}

It is remarkable that the wave operator $\Box_{g_{a,M}}$ fully separates,
reducing its study to the analysis of ordinary differential equations with respect
to $r$.
We refer the reader to the original work of Carter~\cite{carter1968hamilton}. 
The geometric origin of this fact is the presence of a hidden symmetry.

We review briefly this separation here:
We will denote by $\mathfrak{C}$ the transform corresponding to Carter's separation
\[
\mathfrak{C}[U] (r,\ell, \omega,m ) =    \int_{0}^{\pi} U(r,\theta,\omega,m)  \, \overline{S_{m\ell}^{(a\omega)}}(\cos \theta)  \, d\theta       
\]
and $\mathfrak{C}^{-1}$ its inverse,
\[
\mathfrak{C}^{-1} [u] (r,\theta,\omega, m) = \sum_{\ell \ge |m|} U(r,\theta,\omega,m) \, S_{m\ell}^{(a\omega)}(\cos \theta)  \, d\theta,
\]
where $S_{m\ell}^{(a\omega)}(\cos \theta)$ 
denote the eigenfunctions of the second order self-adjoint operator
\begin{equation}
\label{operatorhere}
P(a\omega, m) U = -(\sin \theta)^{-1}\frac{d}{d\theta} \left(\sin \theta \frac{d}{d\theta} U\right)+
(m^2 (\sin \theta)^{-2}   - (a\omega)^2 \cos^2\theta ) U,
\end{equation}
whose eigenvalues we denote $\lambda_{m\ell}^{(a\omega)}\in \mathbb R$, where these are labeled
by the set $\{\ell  \ge  |m|\}$.
We will informally refer to the set 
\[
\{ (\omega,m,\ell)\in \mathbb R\times \mathbb Z\times
\mathbb Z :\ell\ge |m| \}
\]
as \emph{Carter frequency space} to distinguish it from the frequency space of
Section~\ref{elementarycalculus}.

We define 
\begin{equation}
\label{definitionofLambda}
\Lambda_{\ell m}^{(a\omega)} := \lambda^{(a\omega)}_{m\ell} + a^2\omega^2.
\end{equation} 
It was shown in~\cite{partiii} that
\begin{align}
\label{proptyofLambda}
\Lambda &\ge |m|(|m|+1),\\
\label{proptyofLambdatwo}
\Lambda &\ge |am\omega|.
\end{align}

Following~\cite{partiii}, we will refer to triples $(\omega, m , \Lambda)$ with $\omega\in \mathbb R$, 
$m \in \mathbb Z$ and $\Lambda\in \mathbb R$ satisfying~\eqref{proptyofLambda}--\eqref{proptyofLambdatwo}
 as \emph{admissible frequency triples}, irrespectively of whether $\Lambda$ is related to $\omega$
 and $m$ by~\eqref{definitionofLambda}.

We have that  if $\psi$ is a smooth function on $\mathcal{R}(\tau_0,\tau_1)$, then 
\[
u_{\rm star}: =(r^2+a^2)^{1/2} \mathfrak{C} \circ \mathfrak{F} [\chi^2_{\tau_0,\tau_1} \psi ]
\]
is a function of $(r, \omega, m, \ell)$,  smooth in $r\in [r_0,\infty)$ and $\omega\in \mathbb R$.
If $\Box_{g_{a,M}} \psi =F$,  then $u_{\rm star}(r,\omega,m,\ell)$ satisfies for fixed $(\omega, m, \ell)$ 
an ordinary differential equation with respect to $r$.

In this paper, we could work directly with the equation satisfied by $u_{\rm star}$. 
To connect, however, with previous literature, and since we will only apply
Carter's separation in the region $r> r_+$, we will work instead with:
\[
u=u_{BL} : = (r^2+a^2)^{1/2} \mathfrak{C} \circ \mathfrak{F}_{BL} [\chi^2_{\tau_0,\tau_1} \psi ]
\]
and we will moreover consider this as a function of $r^*$,
defined in turn from $r$ by the relations 
\[
\frac{dr^*}{dr} = \frac{r^2+a^2}{\Delta} , \qquad r^*(3M)= 0.
\]
Note that $r^*\to -\infty$ as $r\to r_+$ while $r^*\to \infty$ as $r\to \infty$.
The function $u$ is smooth with respect to $r^*\in (-\infty,\infty)$ and satisfies an ODE 
 which may be written
\begin{equation}
\label{radialodewithpot}
u'' +(\omega^2- V) u = G
\end{equation}
where
\begin{equation}
\label{Vdefbody}
V:= \frac{4Mram\omega-a^2m^2+\Delta\Lambda}{(r^2+a^2)^2}+
 \frac{\Delta}{(r^2+a^2)^4}\left(a^2\Delta +2Mr(r^2-a^2)\right)
\end{equation}
and 
\begin{equation}
\label{Gdefbody}
G:= \frac{\Delta  \mathfrak{C} \circ \mathfrak{F}_{BL} ( (r^2+a^2 cos^2\theta)( \chi^2_{\tau_1,\tau_2}F + 2\nabla_{g_0}\chi^2_{\tau_1,\tau_2}\cdot \nabla_{g_0}\psi +(\Box_{g_0}\chi^2_{\tau_1,\tau_2} )\psi ) }{(r^2+a^2)^{3/2}}
\end{equation}
and
where $'$ in fact denotes the derivative with respect to $r^*$.
Despite thinking of $u$ as a function of $r^*$, we will often refer
to its values at various $r$ by $u(r)$, whereby we mean $u(r^*(r))$. 

We note that all information on $u_{\rm star}(r_+)$ 
is encoded in the limiting behaviour  $u(r^*)$  as $r^*\to -\infty$.

Let us note that we may view $V$ defined by~\eqref{Vdefbody} as a function of $r$
and of Carter frequencies $(\omega, m, \ell)$, i.e.~as a function
$V(\omega, m, \ell, r)$, where we substitute~\eqref{definitionofLambda} for $\Lambda$,
but alternatively, we may consider~\eqref{Vdefbody} 
to define $V(\omega, m, \Lambda, r)$
on the set of all admissible frequencies $(\omega, m, \Lambda)$, and
consider solutions of~\eqref{radialodewithpot} for general such admissible frequencies.

Note finally the relations:
\[
\partial_t\Psi (t,r,\theta,\phi) = \frac{-i}{\sqrt {2\pi}} \int_{-\infty}^\infty 
\sum_{m, \ell} \omega \,  \mathfrak{C} \circ \mathfrak{F}_{BL}[\Psi](r, \omega, m, \ell) S_{m\ell}(a\omega,\cos\theta)
e^{-i\omega t} e^{im\phi}   d\omega,
\]
\[
\partial_\phi \Psi (t,r,\theta,\phi) = \frac{i}{\sqrt {2\pi}} \int_{-\infty}^\infty 
\sum_{m, \ell} m \,  \mathfrak{C} \circ \mathfrak{F}_{BL}[\Psi](r, \omega, m, \ell)  S_{m\ell}(a\omega,\cos\theta)  e^{-i\omega t} e^{im\phi} d\omega.
\]

\subsubsection{The Plancherel formulae for Carter's separation}
\label{Plancherelsection}

We collect here the Plancherel formulae corresponding to Carter's separation
with respect to  Boyer--Lindquist coordinates  $\mathfrak{C} \circ \mathfrak{F}_{BL}$:

\[
\int_0^{2\pi} \int_0^\pi \int_{-\infty}^\infty |\Psi|^2 (t,r,\theta,\phi) \sin \theta d\phi d\theta dt
= \int_{-\infty}^\infty \sum_{m,\ell}  | \mathfrak{C} \circ \mathfrak{F}_{BL}[\Psi](r, \omega, m, \ell)|^2 d\omega,
\]
\[
\int_0^{2\pi} \int_0^\pi \int_{-\infty}^\infty  |\partial_r\Psi|^2 (t,r,\theta,\phi)\sin \theta d\phi d\theta dt = 
\int_{-\infty}^\infty \sum_{m,\ell} |\frac{d}{dr}  \mathfrak{C} \circ \mathfrak{F}_{BL}[\Psi](r, \omega, m, \ell)|^2 d\omega,
\]
\[
\int_0^{2\pi} \int_0^\pi \int_{-\infty}^\infty  |\partial_t\Psi|^2 (t,r,\theta,\phi)\sin\theta d\phi d\theta dt = 
\int_{-\infty}^\infty \sum_{m,\ell} \omega^2 | \mathfrak{C} \circ \mathfrak{F}_{BL}[\Psi](r, \omega, m, \ell)|^2 d\omega,
\]
\[
\int_0^{2\pi} \int_0^\pi \int_{-\infty}^\infty  |\partial_\phi \Psi|^2 (t,r,\theta,\phi)\sin\theta d\phi d\theta dt = 
\int_{-\infty}^\infty \sum_{m,\ell} m^2 | \mathfrak{C} \circ \mathfrak{F}_{BL}[\Psi](r, \omega, m, \ell)|^2 d\omega,
\]
\begin{align*}
&\int_0^{2\pi} \int_0^\pi \int_{-\infty}^\infty  \left(|\partial_\theta \Psi|^2 +\sin^{-2}\theta
|\partial_\phi\Psi|^2 \right)(t,r,\theta,\phi)\sin \theta d\phi d\theta dt 
+a^2 \int_0^{2\pi}\int_0^\pi \int_{-\infty}^\infty \cos^2\theta | \partial_t\Psi|^2 \sin \theta
d \phi d\theta dt\\
&\qquad\qquad= 
\int_{-\infty}^\infty \sum \lambda_{m,\ell}^{(a\omega)}| \mathfrak{C} \circ \mathfrak{F}_{BL}[\Psi](r, \omega, m, \ell)|^2d\omega.
\end{align*}

\subsection{A refined inhomogeneous estimate}
\label{refinedlinearhere}

The estimate provided by Theorem~\ref{blackboxoneforkerr} 
has the advantage that it can be expressed in terms of a
physical space norm. This is at the expense, however, of losing some more precise control
contained in the results of~\cite{partiii}, which can only be expressed in terms
of the frequency decomposition of Section~\ref{frequencyanalysisreviewsec}.

While this less precise information was sufficient for~\cite{DHRT22}, in the present paper,
it will be essential to use the full results of~\cite{partiii}.
As discussed already in Section~\ref{reviewofILEDandimprove}, 
we will in fact require a more precise estimate
than was proven in~\cite{partiii}. 
Specifically, in regions  of Carter
frequency space where the operator $\Box_{g_{a,M}}$ is essentially elliptic for some
suitable interval of $r$, we
will require a corresponding gain of one order of differentiability in the estimate.
As we shall see below, this can easily be obtained from the methods of~\cite{partiii}.

In Section~\ref{capturingprecise}, we shall first formulate a proposition capturing the precise frequency dependent statement of~\cite{partiii}. We shall then formulate our elliptic
improvement
in Section~\ref{ellipticimprovement}.  Finally, in Section~\ref{finalimprovedsec},
we shall combine the above two statements
with that of Theorem~\ref{blackboxoneforkerr}, 
giving Theorem~\ref{refinedblackboxoneforkerr}. 
This latter theorem will then be the statement that
will be used in the remainder of the paper.

\subsubsection{Capturing the precise result of~\cite{partiii}}
\label{capturingprecise}

Let us recall the ``trapped'' frequency range $\mathcal{G}_\natural$ of~\cite{partiii}
\[
 \mathcal{G}_\natural:= \left\{(\omega, m, \Lambda)\text{\ admissible\ } :
|\omega|\ge \omega_{\rm high} ,\, \epsilon_{\rm width}\Lambda \le \omega^2 \le \epsilon_{\rm width}^{-1} \Lambda,\, m\omega\notin (0,\frac{am^2}{2Mr_+} +\alpha \Lambda]\right\}
\]
where $\omega_{\rm high}$, $\epsilon_{\rm width}$ and $\alpha$ were positive parameters. 
We recall from~\cite{partiii} (see Proposition~\ref{genpropsofV0} of 
Section~\ref{basicdecompo} of the present paper),  that 
for $(\omega, m, \ell)\in  \mathcal{G}_\natural$, the high frequency potential $V_0$
(see~\eqref{defofVzero}) has a unique interior maximum point $r^0_{\rm max}$. We denote
the corresponding value of $V_0$ by $V_{\rm max}$. (Note, $V_{\rm max}$ thus denotes
the maximum of $V_0$ and not the maximum of $V$.)

The frequency range  $\mathcal{G}_\natural$ is actually larger than the true trapped
frequency range: In order for true trapping to occur, then, in addition to $V_0$ having a maximum, the maximum value $V_{\rm max}$ must be approximately equal to 
$\omega^2$. To capture thus genuine trapping, 
we introduce a parameter $b_{\rm trap}>0$, to be determined later,
and define the smaller range:
\begin{equation}
\label{betterdefhere}
 \mathcal{G}_\natural (b_{\rm trap}) := \mathcal{G}_\natural \cap \{ (\omega, m, \Lambda)\text{\ admissible\ }  : 
 |V_{\rm max} - \omega^2| \le b_{\rm trap}\omega^2 \}.
\end{equation}

We may now define a nonnegative function
\[
\chi_{\natural} (r, \omega, m, \Lambda)
\]
with the following properties:
\begin{enumerate}
\item
For $(\omega, m ,\Lambda )\not\in \mathcal{G}_\natural(b_{\rm trap})$, 
$\chi_{\natural} (r, \omega, m, \Lambda) =1$  identically.
 \item
 For  $(\omega, m ,\Lambda)\in \mathcal{G}_\natural(b_{\rm trap})$, the function
$\chi_{\natural}(r, \omega, m,\Lambda )$ is smooth in $r$ and
vanishes at the unique
maximum $r^0_{\rm max}$ of the high frequency potential $V_0$ (see~\eqref{defofVzero} of
Section~\ref{basicdecompo}).
\item
$\chi''_{\natural}\gtrsim 1$ wherever  $\chi_{\natural} \le 1/2$ (where by our conventions
from Section~\ref{noteonconstants}
the constant implicit in $\gtrsim$ is in particular independent of the frequency triple $(\omega, m, \Lambda)$).
\end{enumerate}

Let $u_{\bf k}:=  \mathfrak{C}\circ
\mathfrak{F}_{BL}  [  \mathring{\mathfrak{D}}^{\bf k} \chi^2_{\tau_0,\tau_1} \psi  ]$.
To capture the more precise estimate proven in~\cite{partiii},
we define the following quantity:
\begin{align*}
{}^{\natural}\Xk(\tau_0,\tau_1)[\psi]:= \sum_{|{\bf k} |\le k}
\int_{-\infty}^{R_{\rm freq}^*} \int_{-\infty}^\infty \sum_{m, \ell\ge |m|} &
\Delta r^{-2} \Big( |u'_{\bf k}(r, \omega, m ,\ell)|^2\\
&\quad+ \chi_{\natural} (r, \omega, m, \Lambda)(\omega^2
+ r^{-1} m^2+r^{-1} \Lambda^2)  |u_{\bf k} (r, \omega, m ,\ell)|^2 \Big) dr^* d\omega,
\end{align*}
where for $\Lambda$ we plug in~\eqref{definitionofLambda}.

\begin{remark}
\label{whythesefactors}
Concerning the $r^{-1}$ and $r^{-2}$ factors above:
Recall from Section~\ref{noteonconstants} that the $C$ implicit in $\lesssim$ will eventually
depend only $a$ and $M$ (and $k$), but at this point in the paper may depend in addition
on parameters not yet fixed.  Nonetheless, where possible, we try to state inequalities
to have as little as possible spurious dependence on other parameters, even if that dependence would be innocuous. For instance,
the $r^{-2}$  and $r^{-1}$ factors in the above definition has been included so that the $C$
implicit  in the inequality of Proposition~\ref{moreprecise} does not depend
on the choice of $R_{\rm freq}$. Our choice to include other (often non-sharp) powers of $r$ in estimates to follow can be understood in a similar spirit.
\end{remark}

We have the following:
\begin{proposition}
\label{moreprecise}
Fix $|a|<M$, 
let $(\mathcal{M},g_{a,M})$ denote the Kerr manifold and let $\psi$ satisfy
$\Box_{g_{a,M}} \psi =F$ on $\mathcal{R}(\tau_0,\tau_1)$. 
Then 
\begin{equation}
\label{anotherinequality}
 {}^{\natural}\Xk(\tau_0,\tau_1)[\psi]
\lesssim   \Ezerok(\tau_0)[\psi]
+{}^\chi \Xzerok(\tau_0,\tau_1)[\psi]	+ \int_{\mathcal{R}(\tau_0,\tau_1)\cap \{r_1\le r\le R\}}  
|\mathfrak{D}^{\bf k}F |^2.
\end{equation}
 \end{proposition}
\begin{proof}
Were $\chi_{\natural} (r, \omega, m, \Lambda)$ defined as above with $ \mathcal{G}_\natural$
replacing $\mathcal{G}_\natural (b_{\rm trap})$
then the
 statement would in fact be contained in~\cite{partiii}.  The improvement
 with the definition of $\chi_{\natural}$ as given (i.e.~with
 $\mathcal{G}_\natural(b_{\rm trap})$) is straightforward. 
(We will in fact show explicitly how to obtain  it in
 Remark~\ref{inviewofelliptic} of Section~\ref{refinedapp}, 
 immediately after 
 the proof of the more important improvement to be discussed in 
 Section~\ref{ellipticimprovement} below.)
\end{proof}

\begin{remark}
\label{asinthisremark}
We
note that, according to our convention from Section~\ref{noteonconstants}, the
constant implicit in~\eqref{anotherinequality} indeed depends also on $b_{\rm trap}$, until this is fixed
(in Section~\ref{fixparamsheremust}).
\end{remark}

\subsubsection{Gaining a derivative in the elliptic frequency ranges}
\label{ellipticimprovement}

The above estimate of Proposition~\ref{moreprecise}, together with 
Theorem~\ref{blackboxoneforkerr}, is still not
quite sufficient to handle the type of error terms which will occur in our highest order
estimates. For this we need an additional improvement, which will  be
stated below as Proposition~\ref{refinedproposition} and proven in Section~\ref{refinedapp},
using the methods of~\cite{partiii}.
This improvement concerns estimates which gain one derivative when
localised to regions in Carter frequency space and in physical space
where the wave
operator is in fact elliptic.

Given parameters $\gamma_{\rm elliptic}>0$,
$b_{\rm elliptic}>0$, to be determined later, 
 and  given a frequency pair $(\omega, m) \in \mathbb R\times \mathbb Z$
  let us define 
 \begin{equation}
 \label{rellipticdef}
 r_{\rm elliptic} (\omega, m) := 
     \begin{cases}
       r_+, &\text{\ if\ }
 a=0\text{\ or\ } m=0\text{\ or\ } 
  |\omega/am| \ge \gamma_{\rm elliptic}\\
        2.03 M &\text{\ if\ }a\ne0, m\ne 0 \text{\ and\ }|\omega/am |< \gamma_{\rm elliptic}.\\
\end{cases}
  \end{equation}
Given admissible frequencies $(\omega, m, \Lambda)$, we  consider the $r$-range 
 \begin{equation}
 \label{elllipticdefinitionreg}
 {\bf R}_{\rm elliptic}(\omega, m, \Lambda) : = \{  r: V_0 \ge (1+b_{\rm elliptic}) \omega^2 \}\cap 
 \{ r_{\rm elliptic}(\omega, m)  <r \le 1.2 R_{\rm  pot}\},
 \end{equation}
 where $V_0$ is the high frequency potential (see~\eqref{defofVzero}). Let 
 \[
 \iota_{\rm elliptic} (r, \omega, m ,\Lambda)
 \]
 denote the indicator function of the set~\eqref{elllipticdefinitionreg}.

Let $u_{\bf k}:=  \mathfrak{C}\circ
\mathfrak{F}_{BL}  [ \mathring{\mathfrak{D}}^{\bf k} \chi^2_{\tau_0,\tau_1} \psi  ]$.
We 
may now define 
the following
 integrated quantity (localised in both physical and Carter frequency  space)
\begin{equation}
\label{theellipticenergydefinition}
{}^{\rm elp}_{\scalebox{.6}{\mbox{\tiny{\boxed{+1}}}}}
\Xk(\tau_0,\tau_1)[\psi]:= \sum_{|{\bf k} |\le k}
\int_{-\infty}^{(1.2 R_{\rm pot})^*} \int_{-\infty}^\infty \sum_{m, \ell\ge |m|} 
\Delta r^{-2} \iota_{\rm elliptic}(r, \omega, m ,\Lambda) (V_0 -\omega^2)^2 |u_{\bf k} (r, \omega, m ,\ell)|^2 dr^* d\omega,
\end{equation}
where for $\Lambda$ we again plug in~\eqref{definitionofLambda}. 
Let us note that the quantity~\eqref{theellipticenergydefinition}  is formed
by applying $k$ commutations  (hence the $k$ subscript which always labels
number of commutation operators applied), but is of ``order'' $k+2$ in
certain frequency ranges,~i.e.~one order more than the quantity ${}^{\chi}\Xk$ say, which 
we recall is of total order $k+1$. Hence, the extra ${\scalebox{.6}{\mbox{\tiny{\boxed{+1}}}}}$
in the notation.

We may now state an improved estimate for this quantity:
\begin{proposition}
\label{refinedproposition}
Fix $|a|<M$, 
let $(\mathcal{M},g_{a,M})$ denote the Kerr manifold $g_{a,M}$ and let $\psi$ satisfy
$\Box_{g_{a,M}} \psi =F$ on $\mathcal{R}(\tau_0,\tau_1)$. 
Then we have the estimate:
\begin{align}
\label{thellipticestimate}
{}^{\rm elp}_{\scalebox{.6}{\mbox{\tiny{\boxed{+1}}}}}\Xk(\tau_0,\tau_1)[\psi]
&\lesssim  
   {}^\chi \Xzerok(\tau_0,\tau_1)[\psi]
+  {}^{\natural}\Xk(\tau_0,\tau_1)[\psi]+
\sum_{|{\bf k}|\le k} \int_{\mathcal{R}(\tau_0,\tau_1)\cap \{r_+\le r\le 1.3 R_{\rm pot}\}}  |\mathfrak{D}^{\bf k}F |^2.
\end{align}
 \end{proposition}
\begin{proof}
See Section~\ref{refinedapp}.
\end{proof}

\begin{remark}
As with Remark~\ref{asinthisremark}, we 
note that, according to our convention from Section~\ref{noteonconstants}, the
constant implicit in~\eqref{thellipticestimate} 
indeed depends also on $b_{\rm elliptic}$ and $\gamma_{\rm elliptic}$, until these are fixed
(in Section~\ref{fixparamsheremust}).
\end{remark}

\subsubsection{The final improved estimate}
\label{finalimprovedsec}

Let us introduce a new notation
\begin{equation}
\label{supontheright}
{}^{\chi\natural}_{\scalebox{.6}{\mbox{\tiny{\boxed{+1}}}}}\Xpk(\tau_0,\tau_1)[\psi] := {}^\chi\Xpk(\tau_0,\tau_1)  [\psi] 
+\Ezerok_{\mathcal{S}}(\tau_0,\tau_1)[\psi] +
 {}^{\rm elp}_{\scalebox{.6}{\mbox{\tiny{\boxed{+1}}}}}\Xk(\tau_0,\tau_1)[\psi] + {}^{\natural}\Xk(\tau_0,\tau_1)[\psi] .
\end{equation}
The symbol ${}^{\chi\natural}_{\scalebox{.6}{\mbox{\tiny{\boxed{+1}}}}}\Xpk$
is meant to combine all symbolic notations on the right hand side.

We emphasise that we do \underline{not} necessarily have
\begin{equation}
\label{donthavethiswithexact}
 {}^{\rm elp}_{\scalebox{.6}{\mbox{\tiny{\boxed{+1}}}}}\Xk(\tau^*_0,\tau^*_1)[\psi]\le  {}^{\rm elp}_{\scalebox{.6}{\mbox{\tiny{\boxed{+1}}}}}\Xk(\tau_0,\tau_1)[\psi],
 \qquad {}^{\natural}\Xk(\tau^*_0,\tau^*_1)[\psi] \le {}^{\natural}\Xk(\tau_0,\tau_1) [\psi]
\end{equation}
due to the nonlocal definition of these terms.
Nonetheless, we see easily that for $\tau_0+1\le \tau^*_0\le\tilde\tau^*_0+1\le
\tau_1^*\le \tau_1-1$, we have
\[
{}^{\chi\natural}_{\scalebox{.6}{\mbox{\tiny{\boxed{+1}}}}}\Xpk(\tau^*_0,\tau^*_1)[\psi]
\lesssim {}^{\chi\natural}_{\scalebox{.6}{\mbox{\tiny{\boxed{+1}}}}}\Xpk(\tau_0,\tau_1)[\psi],
\]
and more generally
\begin{equation}
\label{weseeeasilythatthis} 
{}^{\chi\natural}_{\scalebox{.6}{\mbox{\tiny{\boxed{+1}}}}}\Xpk(\tau_0,\tau_1)[\hat
\chi(\tau)\psi]
\lesssim {}^{\chi\natural}_{\scalebox{.6}{\mbox{\tiny{\boxed{+1}}}}}\Xpk(\tau_0,\tau_1)[\psi],
\end{equation}
where $\hat\chi(\tau)$ is a general cutoff function, and the constant implicit in $\lesssim$ depends
on derivatives of $\hat\chi$.

 \begin{theorem}
 \label{refinedblackboxoneforkerr}
Fix $|a|<M$, 
let $(\mathcal{M},g_{a,M})$ denote the Kerr manifold and let $\psi$ satisfy
$\Box_{g_{a,M}} \psi =F$ on $\mathcal{R}(\tau_0,\tau_1)$. 

For all $k\ge 1$, and for $p=0$ or  $0<\delta\le p \le 2-\delta$ and for all
$\tau_0\le \tau\le \tau_1$, 
we have the following statement:
 \begin{align*}
{}^{\chi\natural}_{\scalebox{.6}{\mbox{\tiny{\boxed{+1}}}}}\Xpk(\tau_0,\tau_1)[\psi]	
&\lesssim 
 \Epk(\tau_0)[\psi]	 +
 \int_{\mathcal{R}(\tau_0,\tau_1)\cap \{r\ge R\}}
\sum_{|{\bf k}|\le k}\left(|V^\mu_p\partial_\mu({\mathfrak{D}}^{\bf k}\psi)|+|w_p{\mathfrak{D}}^{\bf k}\psi| \right)|{\mathfrak{D}}^{\bf k}F|\\
&\qquad +
 \int_{\mathcal{R}(\tau_0,\tau_1)}
\sum_{|{\bf k}| \le k}| {\mathfrak{D}}^{\bf k} {F} |^2\\
&\qquad
+\sum_{|{\bf k}|\le k}
\sqrt{\int_{\mathcal{R}(\tau_0,\tau)\cap \{r\le R\}}
\left( |L\mathfrak{D}^{\bf k} \psi|
+|\underline{L} \mathfrak{D}^{\bf k} \psi | +| \slashed\nabla \mathfrak{D}^{\bf k} \psi| 
+|\mathfrak{D}^{\bf k}\psi| \right)^2
}
\sqrt{
\int_{\mathcal{R}(\tau_0,\tau)\cap \{r\le R\}}
|{\mathfrak{D}}^{\bf k}F|^2
}\\
&\qquad
+
 \int_{\mathcal{R}(\tau_0,\tau_1)\cap \{r\le R\}}\sum_{|{\bf k}| \le k-1}|{\widetilde{\mathfrak{D}}}^{\bf k} F |^2
+ 
\sup_{\tau_0\le \tau\le \tau_1}\int_{\Sigma(\tau)\cap\{r\le R\}}\sum_{|{\bf k}|\le k-1}| \widetilde{\mathfrak{D}}^{\bf k}F|^2.
\end{align*}
\end{theorem}

\begin{proof}
This follows immediately putting together Theorem~\ref{blackboxoneforkerr} and 
Propositions~\ref{moreprecise} and~\ref{refinedproposition}.
\end{proof}

\subsection{The projections $P_n$, the currents $J_{g_{a,M}}^{{\rm main}, n}$, 
$K_{g_{a,M}}^{{\rm main}, n}$ and  generalised coercivity statements}
\label{generalisedcoerc}

The second ingredient of this paper, discussed already in Section~\ref{wavepacketlocintrosec}, will concern a new coercivity statement for 
a finite set of natural currents 
\begin{equation}
\label{currentsthemselves}
J_{g_{a,M}}^{{\rm main},n}, \qquad K_{g_{a,M}}^{{\rm main},n}, \qquad n=0,\ldots, N,
\end{equation}
associated
with $\Box_{g_{a,M}}$.
The currents~\eqref{currentsthemselves} 
are themselves defined purely in physical space, i.e.~they are currents
precisely in the sense of Section~\ref{Covariantenergyidentitiessec}, 
but for each $n$, their most general  \emph{coercivity} properties will 
only be valid \emph{restricted to the image of a projection operator  $P_n$},
defined using the Fourier transform $\mathfrak{F}$ in $t^*$ and $\phi^*$ (but without using 
Carter's full separation operator $\mathfrak{C}$).
(Thus, in the language of  Section~\ref{elementarycalculus},
$P_n$ will be a zeroth order  $(t^*,\phi^*)$-pseudodifferential operators and will satisfy a standard
good commutation calculus.) 
 Moreover, coercivity of bulk terms 
will only be valid modulo a highest order ``error'' term which in practice will be absorbed
via the refined statement Theorem~\ref{refinedblackboxoneforkerr} of Section~\ref{refinedlinearhere} above. Thus, in general, in order to actually show this coercivity we will need to analyse our solution 
using Carter's frequency separation, even though the definition of the projection operators and
the currents are both
 independent of  the Carter frequency $\ell$.

We first define the projections $P_n$ in Section~\ref{projectiondef} below.
We define the currents~\eqref{currentsthemselves}
in Section~\ref{introducingthecurrents} and address boundary coercivity, 
while we treat bulk coercivity in Sections~\ref{nowforthebulk} and~\ref{bulkcoerfar}.
The definition of the currents~\eqref{currentsthemselves} depends 
on choices of functions which, though  defined independently of
Carter's separation, are best understood
in the context of showing coercivity using Carter's separation.
All constructions and proofs requiring Carter's separation, however, will be deferred to Appendix~\ref{carterestimatesappend}. Thus, the main results of this section will depend in particular
on a number of propositions whose proof has been relegated there

For convenience, in the remainder of Section~\ref{replacement}, we will assume
\[
a\ne 0,
\]
as our frequency decomposition will involve inverse powers of $a$.
We shall return to the general case in Section~\ref{generalclassconsideredrecalled}, as well as the issue of 
uniformity of our estimates in $a$.  (See already Remark~\ref{aquestionofuniformity}.)

\subsubsection{The frequency space covering $\mathcal{F}_n$ and the 
projections $P_n$, $n=0,\ldots, N$}
\label{projectiondef}
We define here the projection operators $P_n$.
The operators separate the superradiant and nonsuperradiant
parts of the solution and further decompose each into wave packets.
As discussed already in Section~\ref{wavepacketlocintrosec},
the point of projection is two-fold:
\begin{enumerate}
\item[(i)]
For the non-superradiant part, 
to localise the position of trapping in each packet to a fixed sufficiently small $r$-interval
on which  there exists a   timelike Killing field.
\item[(ii)] 
For the superradiant part,
to ensure that the elliptic frequency region corresponding to 
all frequencies in the wave packet contains a common 
$r$-interval.
\end{enumerate}
 It is also useful technically
to separate low $(\omega, m)$ frequencies, but these in practice can usually be estimated
together with the superradiant part.

To define $P_n$, we  first require a well chosen 
open covering of  $\mathbb R$ by finitely many open intervals $I_n$, $n=1, \ldots, N$:
\begin{equation}
\label{thisisanopencover}
\mathbb R = \cup_{n=1}^N I_n ,
\end{equation}
determined by  Proposition~\ref{determiningthecoveringprop} of Section~\ref{fixingthecovering}.  
This will allow us to define an associated 
 open cover $\mathcal{F}_0,\ldots, \mathcal{F}_N$ of 
$(\omega, m)$ frequency space:
\begin{equation}
\label{thisisalsoanopencover}
\mathbb R\times\mathbb Z = \cup_{n=0}^N \mathcal{F}_n.
\end{equation}
The operators $P_n$ will then be defined so as to project to frequencies supported in $\mathcal{F}_n$,
with the help of a suitable partition of unity subordinate to the cover.

We proceed with the details.

\paragraph{Definition of the $(\omega, m$)-frequency space cover $\mathcal{F}_n$.}

We introduce an open cover $\mathcal{F}_0,\ldots, \mathcal{F}_N$ 
of $(\omega, m)$~frequency space~\eqref{thisisalsoanopencover} defined as follows.

We define first
\[
\mathcal{F}_0:= \{ \omega^2 + m^2 <  2B_{\rm low} \},
\]
for some fixed $B_{\rm low}>0$ satisfying
$B_{\rm low}\le \frac12 \omega_{\rm high}$ (where $\omega_{\rm high}$ is the parameter
appearing
in the definition of~\eqref{betterdefhere}). 
We may in fact take $B_{\rm low}=1$ provided we assume (without loss of generality) that
 $\omega_{\rm high}\ge 2$). Note that for $(\omega, m, \Lambda)$ admissible with
 $(\omega, m)\in \mathcal{F}_0$ we thus  have in particular
 that
\begin{equation}
 \label{isidenticallyonehere}
 \chi_{\natural} (r, \omega, m, \Lambda)=1.
 \end{equation}

Consider now the specific  
open covering~\eqref{thisisanopencover}  of $\mathbb{R}$
determined by  Proposition~\ref{determiningthecoveringprop} of Section~\ref{fixingthecovering},
consisting of intervals
  $I_n$, $n=1, \ldots, N$.

We now may define
\begin{eqnarray*}
\mathcal{F}_n &:=& \{   \omega^2 + m^2 > B_{\rm low}, \quad m\ne 0, \,\,  \omega/(am)  \in I_n   \} ,\qquad 
n=1,\ldots, N-2,\\
\mathcal{F}_{N-1} &:= &\{   \omega^2 + m^2 > B_{\rm low}, \quad m\ne 0, \,\,  \omega/(am)  \in I_{N-1} \text{\rm\ or\ } m=0,\, \omega/a <0  \} ,\\
\mathcal{F}_{N} &:= &\{   \omega^2 + m^2 > B_{\rm low}, \quad m\ne 0, \,\,  \omega/(am)  \in I_N\text{\rm\ or\ } m=0, \, \omega/a  >0  \}.
\end{eqnarray*}

\paragraph{Properties of $\mathcal{F}_n$.}
In view of the definition above, 
the properties of $I_n$ established in 
Proposition~\ref{determiningthecoveringprop} of Section~\ref{fixingthecovering} translate
immediately into the following properties of the 
 ranges $\mathcal{F}_n$:
\begin{itemize}
\item
The collection $\mathcal{F}_n$, $n=0,\ldots, N$ defined above indeed cover
frequency space $\mathbb R\times\mathbb Z$, i.e.~\eqref{thisisalsoanopencover} holds.  Moreover, replacing $m\in \mathbb Z$ with a real valued variable $m\in \mathbb R$
in their definition, $\mathcal{F}_n$ can be viewed as an open covering of $\mathbb R\times
\mathbb R$.
\item
For $n= N_s+1,\ldots , N$, all $(\omega, m)\in \mathcal{F}_n$ 
are non-superradiant, i.e.~satisfy~\eqref{nonsuperradiantcondition},
and in fact $\omega/am$ is uniformly bounded away from $0$
and $\frac{1}{2Mr_+}$ or else $m=0$ (in which case we note that necessarily $n=N-1$ or $N$). 
We will refer to the ranges  
 $\mathcal{F}_{N_s+1}, \ldots \mathcal{F}_{N}$ (or sometimes simply to the indices $n=N_s+1, \ldots, N$ themselves)
  as the \emph{non-superradiant ranges}.
\item
For $n=3,\ldots, N_s$, all $(\omega, m)\in \mathcal{F}_n$ 
are superradiant, i.e.~satisfy~\eqref{superradiantcondition}, 
and in fact $\omega/am$ is uniformly bounded away from $0$
and $\frac{1}{2Mr_+}$. We will refer to the ranges $\mathcal{F}_3,\ldots, \mathcal{F}_{N_s}$ (or sometimes simply to the indices $n=3, \ldots, N_s$ themselves)  as \emph{the
strictly superradiant ranges}.
\item 
Concerning $n=1$, since $0\in I_1$ and $I_1$ is open, $\mathcal{F}_1$ consists of
both superradiant and nonsuperradiant
$(\omega, m)$  (such that  $\omega/am \in I_1$), while
concerning $n=2$, since  $\frac{1}{2Mr_+}\in I_2$ and $I_2$ is open, it follows that 
$\mathcal{F}_2$  again contains
both superradiant and nonsuperradiant $(\omega, m)$  (such that  $\omega/am \in I_2$).
Similarly, concerning $n=0$, $\mathcal{F}_0$ contains both superradiant and nonsuperradiant
frequencies.  Nonetheless, in our arguments, we shall usually handle the low frequencies
$\mathcal{F}_0$ and the overlap cases $\mathcal{F}_1$, $\mathcal{F}_2$ 
in a similar way as with the strictly superradiant frequencies. Thus we shall refer
collectively 
to the ranges
 $\mathcal{F}_0,\ldots, \mathcal{F}_{N_s}$
(or sometimes simply to the indices $n=0,\ldots, N_s$ themselves) as the \emph{generalised superradiant frequency ranges}.
\item
To each non-zero generalised
superradiant  range $n=1,\ldots, N_s$, then,
 for all
admissible $(\omega, m, \Lambda)$ with 
 $(\omega, m)\in
\mathcal{F}_n$, we have
\begin{equation}
\label{centrelowerboundherefirst}
\chi_{\natural}(r, \omega, m, \Lambda) =1
\end{equation}
identically. Concerning the low frequency generalised superradiant range $n=0$, note that we have already established~\eqref{centrelowerboundherefirst} 
for all admissible $(\omega, m,\Lambda)$ with $(\omega, m)\in \mathcal{F}_0$; see~\eqref{isidenticallyonehere}.  
\item
Moreover, associated to each non-zero generalised superradiant range $n=1,\ldots, N_s$, 
there exists a corresponding non-empty interval $[r_{n,1}, r_{n,2}]\subset (r_{\rm pot}, R_{\rm pot})$
such that $[r_{n,1}, r_{n,2}]\subset {\bf R}_{\rm elliptic} (\omega, m, \Lambda)$ for
all  admissible $(\omega, m, \Lambda)$ with 
$(\omega, m)\in \mathcal{F}_n$.
(We recall here that $r_{\rm pot}$, $R_{\rm pot}$ are parameters fixed in Section~\ref{fixingthecovering}, which themselves
will satisfy $r_1<r_{\rm pot} <2.01M<5M<R_{\rm pot}<2R_{\rm pot}<R_{\rm freq}<R/2$.)
For the low frequency generalised superradiant range $n=0$, we may define for instance 
\[
r_{0,1}:=4M, \, r_{0,2}:=5M.
\]
For all generalised superradiant ranges $n=0,\ldots, N_s$, 
we  now denote the analogous spacetime region by
\[
\mathcal{L}_n := \{ r_{n,1} \le r \le r_{n,2} \}.
\]
\item
For the nonsuperradiant ranges $n=N_{s}+1,\ldots, N$, there again exists a corresponding
nonempty interval $[r_{n, 1}, r_{n, 2}]\subset (r_{\rm pot}, R_{\rm pot})$,
but now with the property that
for all
admissible $(\omega, m, \Lambda)$ with 
 $(\omega, m)\in \mathcal{F}_n$
\begin{equation}
\label{centrelowerboundhere}
\chi_{\natural}(r, \omega, m , \Lambda) \gtrsim 1\qquad\text{for all\ }r\not\in (r_{n,1},r_{n,2}).
\end{equation}
 Let us denote the corresponding spacetime region by
 \[
 \mathcal{D}_n:=\{ r_{n, 1}\le  r \le r_{n, 2}  \} 
 \]
 and set $\mathcal{L}_n:=\emptyset$.
Moreover, associated to each $n=N_s+1,\ldots, N$, 
there is an $\alpha_n\in\mathbb R$ 
satisfying
\begin{equation}
\label{alphanbound}
0\le  \alpha_n \frac{2Mr_+}a \le 1
\end{equation}
such that the vector field $T+\alpha_n \Omega_1$ 
is future directed timelike on $\mathcal{D}_n$, in fact, such that 
\begin{equation}
\label{thisbound}
-g_{a,M} (T+\alpha_n \Omega_1, T+\alpha_n\Omega_1)\gtrsim 1.
\end{equation}
(See~\eqref{timelikeherereflink}.)
This implies that $T+\alpha_n\Omega$  is in fact timelike in (and satisfies a similar bound in)
an enlarged region 
\begin{equation}
\label{defofenlarged}
\widetilde{\mathcal{D}}_n:=\{r'_{n,1} \le r \le r'_{n,2}\}
\end{equation}
for some 
$r_{\rm pot}<r'_{n,1}<r_{n,1}$, $R_{\rm pot}>r'_{n,2} > r_{n,2}$.
which we now fix.
\item
Returning to  the generalised superradiant
ranges $n=0,\ldots, N_s$, let us set $\widetilde{\mathcal{D}}_n:=\emptyset$, 
$\widetilde{\mathcal{D}}_n:=\emptyset$
and set $\alpha_n: = \frac{a}{2Mr_+}$. Thus, in this case $T+\alpha_n\Omega_1$
coincides with $Z$ defined by~\eqref{Hawkingdef}, which we recall is a null generator  of the Kerr horizon~$\mathcal{H}^+$.
\end{itemize}

\paragraph{Definition of $P_n$.}

Now let $\musicrho_n:\mathbb R\to \mathbb R$, $n=1,\ldots, N$ be a smooth partition of unity
of  $\mathbb R$ subordinate to the open cover $I_n$, and let $\eta:\mathbb R\to \mathbb R$ be
a smooth function  such that $\eta(z)=1$ for $|z|\le 1.1B_{\rm low}$, $\eta(z)=0$ for $|z|\ge 1.9B_{\rm low}$. With this, define 
\begin{equation}
\label{partofunityexplicit}
\musicchi_0(\omega, m) = \eta(\omega^2+m^2) , \qquad
\musicchi_n (\omega, m) =(1- \eta (\omega^2+m^2)) \musicrho_n(\omega/am), \,\, n=1,\ldots, N
\end{equation}
where if $m=0$,  $\omega/a <0$,   we define
$\musicchi_{N-1} (\omega, m)=1-\eta(\omega^2+m^2)$
and  $\musicchi_n(\omega, m)=0$ for all $n \ne 0, N-1$, while 
while if $m=0$,  $\omega/a \ge 0$, we define 
$\musicchi_{N} (\omega, m)=1-\eta(\omega^2+m^2)$, while 
 $\musicchi_n(\omega, m)=0$ for all $n\ne 0, N$.
We will view these as functions  $\musicchi_n:\mathbb R\times \mathbb Z\to \mathbb R$ but they
are of course the restrictions of functions  $\musicchi_n:\mathbb R\times \mathbb R
\to \mathbb R$. 
 
In this sense, we see easily that~\eqref{partofunityexplicit} defines 
a smooth  partition of unity  $\musicchi_n$, $n=0,\ldots, N$, of $(\omega, m)$ frequency
space (interpreted for the moment as $\mathbb R\times\mathbb R$)
subordinate to the open cover $\mathcal{F}_n$.

Define finally
\begin{equation}
\label{definitionofPn}
P_n = \mathfrak{F}^{-1}( \musicchi_n\, \mathfrak{F}).  
\end{equation}
Note that 
\begin{equation}
\label{theysumtoid}
\sum_{n=0}^N P_n = Id.
\end{equation}

Let us note that in the region $r>r_+$, the $P_n$ may also be expressed
as $P_n =  \mathfrak{F}_{BL}^{-1}( \musicchi_n\, \mathfrak{F}_{BL})$.

Let us set
\begin{equation}
\label{thenpackets}
\psi_{n,{\bf k}} = P_n ( \mathring{\mathfrak{D}}^{\bf k} \chi^2_{\tau_0,\tau_1} \psi)
 \end{equation}
 where  $\chi_{\tau_0,\tau_1}$  is the cutoff of Section~\ref{Schwartzandcutoffs}
 and $\mathring{\mathfrak{D}}^{\bf k}$ are the entirely Killing commutation operators
 of~\eqref{recallnewcommutation}.
 (Let us note that $[P_n,  \mathring{\mathfrak{D}}^{\bf k}]=0$, and thus we may rewrite
 this
 as 
 \[
 \psi_{n,{\bf k}} =  \mathring{\mathfrak{D}}^{\bf k} P_n (\chi^2_{\tau_0,\tau_1} \psi)
 =  \mathring{\mathfrak{D}}^{\bf k} \psi_n
 \]
 in the notation~\eqref{decompintopack} of the introduction.)

\paragraph{Properties of $P_n$.}
Because the operators $P_n$ above
 are defined with smooth cutoffs in both physical and frequency space with
nice homogeneity properties, we have the following:

 \begin{proposition}
 The operators $P_n$, $P_n\chi_{\tau_0,\tau_1}$ are zeroth order $(t^*,\phi^*)$-pseudodifferential operators
 in the sense of Section~\ref{elementarycalculus}.
 \end{proposition}
\begin{proof}
This is elementary given the way $P_n$ is defined and the homogeneity properties
of~\eqref{partofunityexplicit}.
\end{proof}

Let us already state an improved physical space bound that applies
to functions projected by $P_n$. 

\begin{proposition}
\label{forpseudoerrors}
Let $\psi$ be a smooth function on $\mathcal{R}(\tau_0,\tau_1)$.

Define $\psi_{n, {\bf k}}$ by~\eqref{thenpackets}, for ${\bf k}$ with \underline{${|\bf k}|=k-1$}.
Then we have
\[
\int_{\mathcal{R}(\tau_0,\tau_1) \cap \{r_0\le r \le R_{\rm freq} \} \setminus\mathcal{D}_n} 
r^{-2} \left( | L\psi_{n,{\bf k}}    |^2 + |\underline L\psi_{n,{\bf k}} |^2
+|\nablaslash \psi_{n,{\bf k}}   |^2  \right)
\lesssim  {}^{\natural}\Xzerokminusone[\psi] + {}^\chi\Xzerokminusone[\psi] .
\]

More generally, 
consider a $(t^*,\phi^*)$ pseudodifferential operator of  order $k-1$
whose symbol $q_{\rm symb}$ is supported in $\mathcal{F}_n$  and in $\mathcal{R}(\tau_0,\tau_1)$.
Then
\[
\int_{\mathcal{R}(\tau_0,\tau_1)\cap \{r_0\le r \le R_{\rm freq} \} \setminus \mathcal{D}_n} 
r^{-2} \left (| LQ \chi_{\tau_0,\tau_1}    \psi |^2 +| \underline LQ  \chi_{\tau_0,\tau_1}  \psi|^2
+|\nablaslash Q   \chi_{\tau_0,\tau_1}  \psi|^2 \right) 
\lesssim  {}^{\natural}\Xzerokminusone[\psi] + {}^\chi\Xzerokminusone[\psi] .
\]
\end{proposition}

\begin{proof}
This follows from Plancherel and the lower bounds~\eqref{centrelowerboundherefirst}
and~\eqref{centrelowerboundhere}.
\end{proof}

\begin{remark}
Concerning the $r^{-2}$ factor in the integrands above, see Remark~\ref{whythesefactors}.
\end{remark}

We will need one final proposition which may be proven using the
pseudodifferential calculus:
\begin{proposition}
\label{finalpseudostatement}
Let $P_n$ be the projection~\eqref{definitionofPn} and let $\mu$ be a smooth
function on $\mathcal{M}\cap \{r_0\le r\le R_{\rm freq}\}$. 
Then 
\[
[P_n,\mu] = Q_1+Q_2
\]
where $Q_1$ is of order $-1$ and the support of its symbol 
in frequency space is contained in $\mathcal{F}_n$, while
$Q_2$ is of order $-2$. 
\end{proposition}

\subsubsection{The currents $J_{g_{a,M}}^{{\rm main},n}$ and boundary coercivity}
\label{introducingthecurrents}

We now introduce the fundamental new ingredient of the present 
paper replacing 2.~of~\cite{DHRT22}: a set
of translation invariant currents $J^{{\rm main},n}_{g_{a,M}}$,
  $n=0,\ldots, N$.

For each $n$, the corresponding current 
may be expressed in the form $J^{V_n, w_n, q_n, \varpi_n}_{g_{a,M}}$ from~\eqref{generalJdef}
for suitable choices of these functions.
It will be more natural, however,
to realise this current as a sum of ``twisted currents'' with
twisting function $(r^2+a^2)^{\frac12}$.
This is reviewed in Appendix~\ref{translationsection}, 
and, in the notation introduced there, the
current has the following form in the region $r\ge r_+$ for non-superradiant frequencies $n$:
\begin{eqnarray}
\nonumber
J^{{\rm main}, n}_{g_{a,M}}[\Psi] &:=& \tilde{J}^{y_n}[\Psi] +\tilde{J}^{f_n}[\Psi] + \tilde{J}^{f_{\rm fixed}} [\Psi] +\tilde{J}^{y_{\rm fixed}}[\Psi] +\tilde{J}^{\hat z}[\Psi]+e_{\rm red} \tilde{J}^z_{\rm red}[\Psi]\\
\label{willtakethefollowingformnonsuperbefore}
&&\qquad +E
\left (   \tilde{J}^T[\Psi]  +
\chi_{{\rm Killing},n }\frac{(2Mr_+ \alpha_n/a) }{(1 -2Mr_+ \alpha_n/a)}\tilde{J}^{Z}[\Psi] ) \right)
\end{eqnarray}
while for generalised superradiant frequencies $n$ 
it has the form:
\begin{eqnarray}
\nonumber
J^{{\rm main}, n}_{g_{a,M}}[\Psi] &:=& \tilde{J}^{y_n}[\Psi] +\tilde{J}^{f_n}[\Psi] + \tilde{J}^{f_{\rm fixed}} [\Psi] +\tilde{J}^{y_{\rm fixed}}[\Psi] +\tilde{J}^{\hat z}[\Psi]+e_{\rm red} \tilde{J}^z_{\rm red}[\Psi]\\
\label{willtakethefollowingformsuperbefore}
&&\qquad 
+E( \tilde{J}^T[\Psi] +\chi_{{\rm Killing},n} \alpha_n \tilde{J}^{\Omega_1}[\Psi]  ).
\end{eqnarray}
(We have dropped $g_{a,M}$ subscripts from the currents on the right hand side.)
Here $E>0$ and $e_{\rm red}>0$ are large and small positive parameters, respectively,
which must be chosen appropriately, and $\chi_{{\rm Killing}, n}$ is a cutoff which depends
on the parameter $R_{\rm freq}$, which also must be chosen appropriately.
We will define the corresponding bulk currents as $\tilde{K}^{y_n}$, etc.,
and the total bulk current by $K^{{\rm main}, n}$.

For the region $r_0\le r\le r_+$, the exact form of the currents~\eqref{willtakethefollowingformnonsuperbefore}--\eqref{willtakethefollowingformsuperbefore} is in fact
not important because
we shall be able to deduce the relevant properties
via continuity from the behaviour at $r=r_+$, after restriction of $r_0$  (see for instance already Proposition~\ref{bcgsc}). 
Thus, we may define $J^{{\rm main}, n}_{g_{a,M}}[\Psi]$ in the region  $r_0\le r\le r_+$ to be in fact arbitrary smooth
extensions of the above. (Since some of the component currents of~\eqref{willtakethefollowingformnonsuperbefore}--\eqref{willtakethefollowingformsuperbefore} (like the Killing currents) have natural extensions to $r\le r_+$, one may think of this as equivalent to choosing smooth 
extensions of  the $y_n$ and $f_n$ currents.)

We record only those properties which will be already necessary for the discussion of the 
statements
of the  propositions of the present section as well as those of Section~\ref{nowforthebulk} 
immediately below. All properties referred to below follow easily from 
the basic support properties of the definitions of the functions $y_n$, $f_n$, $f_{\rm fixed}$, etc.
as described in
 Section~\ref{fundcoercivityapp}, 
to which (along with Section~\ref{fundcoercivityboundapp}) the proofs of some of the propositions stated in this and the next section
are  also deferred:
 
\begin{itemize}
\item
The current components $\tilde{J}^{f_{\rm fixed}}$, $\tilde{J}^{\hat z}$ are $n$-independent and supported only in
$r\ge R_{\rm pot}$
(thus similarly for  $\tilde{K}^{f_{\rm fixed}}$, $\tilde{K}^{\hat z}$)
 while $\tilde{J}^{z}_{\rm red}$ is again $n$-independent and supported only in $r\le r_2$
 (thus similarly for  $\tilde{K}^{z}_{\rm red}$). (See Section~\ref{firstchoiceoffetc}.)
 \item
$ \tilde{J}^T_{a,M}$, $\tilde{J}_{a,M}^{\Omega_1}$ 
are the twisted energy currents corresponding to vector fields 
$T$ and $\Omega_1$ defined by expressions~\eqref{twisteddefT},~\eqref{twisteddefT}, respectively,
and $Z=T+\frac{a}{2Mr_+}\Omega_1$,  $\tilde{J}_{a,M}^Z= \tilde{J}^T_{a,M} +\frac{a}{2Mr_+} \tilde{J}_{a,M}^{\Omega_1}$.
\item
$\tilde{J}^z_{\rm red}$ is in fact the twisted energy current $\tilde{J}_{a,M}^X$ 
of a vector field $X$ such that, given arbitrary $e_{\rm red}>0$,
$e_{\rm red}X+ EZ$ is timelike on $r=r_0$ for $E>0$ sufficiently large
and $r_0$ sufficiently close to $r_+$. See~\eqref{reddefinitionsecretlyX}.
\item
In the generalised superradiant case, 
$\chi'_{{\rm Killing},n}$ will be supported
in $\mathcal{L}_n$,  $\chi_{{\rm Killing},n}=1$ near $r=r_0$ and $\chi_{{\rm Killing},n}=0$ for large $r$ (see~\eqref{definekilling}) 
and $\tilde{J}^{y_n}$ will vanish at some value of $r$ in $\mathcal{L}_n$ (see~\eqref{defineytosathere}).
Recall that in this case, $\alpha_n=\frac{a}{2Mr_+}$.
\item
In the non-superradiant case,  $\chi'_{{\rm Killing}, n}$ 
is supported in $[R_{\rm pot}, R_{\rm freq}]$,  $\chi_{{\rm Killing},n}=1$ near $r=r_0$ and $\chi_{{\rm Killing},n}=0$ for large $r$ (see~\eqref{chikillingconstraint})
and ${\tilde{J}}^{y_n}$,
${\tilde{J}}^{f_n}$ 
 (and thus $\tilde{K}^{y_n}$, $\tilde{K}^{f_n}$) will vanish identically in $\mathcal{D}_n$
 (see~\eqref{ypropertiesone}, \eqref{fndefnonsuper}). 
\item
The currents 
$ \tilde{J}^{y_n}[\Psi]$ are $n$-independent for $r\ge 1.2R_{\rm pot}$
and the currents $\tilde{J}^{f_n}[\Psi]$ vanish for $r\ge R_{\rm pot}$.
\end{itemize}
It follows in particular from the above that the total current~\eqref{willtakethefollowingformnonsuperbefore},~\eqref{willtakethefollowingformsuperbefore} is in fact 
$n$-independent for $r\ge R_{\rm freq}$,
where we may write
\[
J_{g_{a,M}}^{{\rm main}}[\Psi]: =  J_{g_{a,M}}^{{\rm main}, n}[\Psi] ,\qquad 
K_{g_{a,M}}^{{\rm main}}[\Psi]: =  K_{g_{a,M}}^{{\rm main}, n}[\Psi] .
\]

Using the properties listed above alone, we can deduce the following boundary coercivity statement in $\widetilde{\mathcal{D}}_n$ for the
currents corresponding to non-superradiant ranges $n$:
\begin{proposition}[Boundary coercivity in $\widetilde{\mathcal{D}}_n$ for non-superradiant
ranges]
Let $E>0$ be sufficiently large in the definitions~\eqref{willtakethefollowingformnonsuperbefore}
and~\eqref{willtakethefollowingformsuperbefore}.
For a general function $\Psi$ and all $0\le n\le N$, we have that in $\widetilde{\mathcal{D}}_n$
\begin{equation}
\label{stillhavethis}
{J_{g_{a,M}}^{{\rm main}, n}} [\Psi] \cdot {\rm n} \gtrsim  | L\Psi |^2 + | \underline{L}\Psi |^2 
+|\slashed\nabla \Psi|^2 + |\Psi|^2,
\end{equation}
where ${\rm n}={\rm n}_{\Sigma(\tau)}$.
\end{proposition}
\begin{proof}
(Recall that $\widetilde{\mathcal{D}}_n=\emptyset$ in the generalised superradiant case, so the proposition is only nontrivial in the case of non-superradiant $n$, where
$\widetilde{\mathcal{D}}_n$ is the enlarged region defined by~\eqref{defofenlarged}.)

Restricted to $\mathcal{D}_n$, the statement follows because by the above discussion,
the current reduces to
\[
E(1-2Mr_+\alpha_n/a)^{-1} ( \tilde{J}^T[\Psi] + \alpha_n \tilde{J}^{\Omega_1}[\Psi]  )=
E\tilde{J}^{T+\alpha_n\Omega_1}[\Psi].
\]
Thus the boundary term integrand on $\mathcal{D}_n\cap \Sigma(\tau)$ may be expressed 
\[
E(1-2Mr_+\alpha_n/a)^{-1} \tilde{T}( T+\alpha_n\Omega_1, n_{\Sigma(\tau)})
\]
where $\tilde{T}$ is the twisted energy momentum tensor~\eqref{twistedenergymom},
and this quantity is coercive since $E>0$ and $T+\alpha_n\Omega_1$ and $n_{\Sigma(\tau)}$
are timelike. Thus, in $\mathcal{D}_n$, the bound~\eqref{stillhavethis} would
follow for any $E>0$.

More generally, given $\widetilde{\mathcal{D}}_n$ as defined,
where $-g(T+\alpha_n\Omega_1,T+\alpha_n\Omega_1)\ge b$, then
if $E$ is sufficiently large we may absorb all other terms and again obtain~\eqref{stillhavethis}.
\end{proof}

For currents corresponding to generalised superradiant ranges, on the other hand, we have the following coercivity statement at $\mathcal{S}$:
\begin{proposition}[Boundary coercivity at $\mathcal{S}$ for generalised superradiant currents]
\label{bcgsc}
Let $e_{\rm red}>0$ be arbitrary. Let $E>0$ be sufficiently large in the definition~\eqref{willtakethefollowingformsuperbefore} and $r_0<r_+$ be sufficiently close to $r_+$, both depending on $e_{\rm red}$.
For a general function~$\Psi$, and for generalised  superradiant $n$, we have on $\mathcal{S}$
the pointwise relation:
\begin{equation}
\label{pointwisegood}
{J_{g_{a,M}}^{{\rm main}, n}} [\Psi] \cdot {\rm n}_{\mathcal{S}} \ge c(r_0) ( | L\Psi |^2 + | \underline{L}\Psi |^2 
+|\slashed\nabla \Psi|^2 ) - C |\Psi|^2, 
\end{equation}
\end{proposition}
\begin{proof}
The statement follows immediately from the definition $\alpha_n=\frac{a}{2Mr_+}$ in the
generalised  superradiant 
case, in view of the presence also of $e_{\rm red}\tilde{J}^z_{\rm red}$ and the
fact recorded earlier that  $e_{\rm red}X+EZ$ is timelike on $r=r_0$ for $r_0$ sufficiently
close to $r_+$ and $E>0$ sufficiently large, and thus $\tilde{J}^{e_{\rm red}X+EZ}$ is coercive
and in fact can absorb all other boundary terms. (That it can absorb the remaining boundary terms can be seen easily from their structure at $r=r_+$ and continuity; cf.~also the proof of Proposition~\ref{globalboundarySpositivity} below (in Section~\ref{fundcoercivityboundapp}) 
where this is done in detail in the more complicated
non-superradiant case. Thus, in particular, this property is independent of the precise form of
the smooth extension to $r_0\le r<r_+$.) 
We note that until $e_{\rm red}$ is fixed, the constant $c$ above also depends on $e_{\rm red}$.
\end{proof}
\begin{remark}
We recall that according to the conventions of Section~\ref{noteonconstants}, all
generic constants, like $c$ in~\eqref{pointwisegood}, may depend in principle 
on all parameters not yet fixed. We have chosen to  explicitly record the $r_0$-dependence
in~\eqref{pointwisegood}
for emphasis, as $c(r_0)$ indeed degenerates as $r_0\to r_+$.
\end{remark}

In the case of non-superradiant $n$,  the current $J^{{\rm main},n}$ will not 
be  \emph{pointwise} coercive at $\mathcal{S}$. It will however enjoy \emph{integrated} coercivity
properties \emph{when specialised to the image of $P_n$}, stemming precisely from
the non-superradiant property of $P_n\psi$.
We  first state the following global integrated coercivity property: 
\begin{proposition}[Global boundary coercivity on the image of $P_n$ for non-superradiant 
frequency ranges]
\label{globalboundarySpositivity}
Let  $e_{\rm red}>0$ be arbitrary. Let $E>0$ be sufficiently large
in the definition~\eqref{willtakethefollowingformnonsuperbefore}
 and $r_0<r_+$ be sufficiently close to $r_+$, both depending on $e_{\rm red}$.
 
 Let $\psi$ be a smooth function on $\mathcal{R}(\tau_0,\tau_1)$
and define $\psi_{n,{\bf k}}$ by~\eqref{thenpackets}.  
For non-superradiant $n$, then denoting $k=|{\bf k}|$
we have
\begin{equation}
\label{globalboundpos}
\int_{\mathcal{S}} {J_{g_{a,M}}^{{\rm main}, n}}[\psi_{n,{\bf k}}] \cdot {\rm n}_{\mathcal{S}}\ge
c(r_0) \int_{\mathcal{S}} ( | L\psi_{n,{\bf k}} |^2 +
 | \underline{L}\psi_{n,{\bf k}} |^2 
+|\slashed\nabla \psi_{n,{\bf k}}|^2 )  - C\, \Ezerokminusone_{\mathcal{S}}(\tau_0,\tau_1)[\psi] . 
\end{equation}
\end{proposition}

\begin{proof}
See Section~\ref{fundcoercivityboundapp}.
\end{proof}

We may now localise the above using the pseudodifferential calculus.
For this, it will be important to introduce an additional parameter
\begin{equation}
\label{announcingepsiloncutoff}
0<\epsilon_{\rm cutoff}<\frac14
\end{equation}
which will be fixed later (see already Section~\ref{toporderidentity}),
and, given $\tau_0+1 \le \tilde\tau_0 \le \tilde\tau_1-1 \le\tilde\tau\le  \tau_1-1$, an associated cutoff
$\tilde\chi_{\tilde\tau_0,\tilde\tau_1}$ such that $\tilde\chi_{\tilde\tau_0,\tilde\tau_1}=1$
in $\tilde\tau_0\le \tau\le \tau_1$ and $\tilde\chi_{\tilde\tau_0,\tilde\tau_1}=0$ in 
$\tau\le \tilde\tau_0-\epsilon_{\rm cutoff}$ and $\tau\ge \tilde\tau_1+\epsilon_{\rm cutoff}$.
Note  already the unfavourable $\epsilon_{\rm cutoff}$-dependence $|\partial_{t^*}\tilde\chi_{\tilde\tau_0,\tilde\tau_1}| \lesssim \epsilon_{\rm cutoff}^{-1}$, etc.
 Let us also recall the $\mathring{\Ezerok}_{\mathcal{S}}$ 
notation from~\eqref{othercommutatorshorizonflux}.
We have the following localised version of Proposition~\ref{globalboundarySpositivity}:
\begin{proposition}[Localised boundary coercivity on the image of $P_n$ for non-superradiant 
frequency ranges]
\label{asintheproofofthisprop}
Under the assumptions of Proposition~\ref{globalboundarySpositivity}, 
the following is moreover true:

For non-superradiant $n$ and for the corresponding $\psi_{n,{\bf k}}$, then denoting
$k=|{\bf k}|$ we have
for all  $\tau_0+1 \le \tilde\tau_0 \le \tilde\tau_1-1 \le\tilde\tau\le  \tau_1-1$ 
the inequality
\begin{align}
\label{nonsuperradiantboundflux}
\int_{\mathcal{S}(\tilde\tau_0,\tilde\tau_1)} {J_{g_{a,M}}^{{\rm main}, n}}[\psi_{n,{\bf k}}] \cdot {\rm n}_{\mathcal{S}} \ge&
 c(r_0)\int_{\mathcal{S}(\tilde\tau_0,\tilde\tau_1)} ( | L\psi_{n,{\bf k}} |^2 +
 | \underline{L}\psi_{n,{\bf k}} |^2 
+|\slashed\nabla \psi_{n, {\bf k}}|^2 ) \\
\label{termsonthislinehavealabel}
&\, -C\mathring{\Ezerok}_{\mathcal{S}} (\tilde\tau_0-\epsilon_{\rm cutoff},\tilde\tau_0+\epsilon_{\rm cutoff}) [\psi]-
C\mathring{\Ezerok}_{\mathcal{S}} (\tilde\tau_1-\epsilon_{\rm cutoff},\tilde\tau_1 +\epsilon_{\rm cutoff})[\psi]  \\
\label{secondlineoferrortermshere}
&\, -C\, \Ezerokminusone_{\mathcal{S}}(\tau_0,\tau_1)[\psi] ,
\end{align}
where the constants in both lines~\eqref{nonsuperradiantboundflux} and~\eqref{termsonthislinehavealabel} (but not~\eqref{secondlineoferrortermshere}!),~may be chosen independently
of $\epsilon_{\rm cutoff}$.
\end{proposition}
\begin{remark}
We may of course already replace $\Ezerokminusone_{\mathcal{S}}(\tau_0,\tau_1)$ above
by $ {}^{\chi\natural}_{\scalebox{.6}{\mbox{\tiny{\boxed{+1}}}}}\, \Xpkminusone(\tau_0,\tau_1) $
in view of the definition~\eqref{supontheright}. We note, however, that the
terms on line~\eqref{termsonthislinehavealabel} will be top order error terms
in our later argument, and thus it will be essential that we can exploit the smallness
of $\epsilon_{\rm cutoff}$.  (See already~\eqref{providedweadd}. Note how it is important
that the constants of~\eqref{nonsuperradiantboundflux} and~\eqref{termsonthislinehavealabel}
may be chosen independently of $\epsilon_{\rm cutoff}$.)
\end{remark}

\begin{proof}
We may apply~\eqref{globalboundpos} to 
\begin{equation}
\label{psiwithatildehere}
\widetilde\psi_{n,{\bf k}}:= P_n
 \mathring{\mathfrak{D}}^{\bf k} \tilde\chi^2_{\tilde\tau_0, \tilde\tau_1}\psi,
 \end{equation}
 where $\tilde\chi_{\tilde\tau_0,\tilde\tau_1}$ is the $\epsilon_{\rm cutoff}$-dependent
 cutoff defined above immediately following~\eqref{announcingepsiloncutoff}, 
 to obtain
\[
\int_{\mathcal{S}} {J_{g_{a,M}}^{{\rm main}, n}}[\widetilde\psi_{n,{\bf k}}] {\rm n}_{\mathcal{S}}\ge  
\int_{\mathcal{S}} c(r_0) ( | L \widetilde\psi_{n,{\bf k}} |^2 +
 | \underline{L} \widetilde\psi_{n,{\bf k}} |^2 
+|\slashed\nabla \widetilde \psi_{n, {\bf k}}|^2 )  - C\, \Ezerokminusone_{\mathcal{S}}(\tau_0,\tau_1) ,
\]
where we note that $c(r_0)$ may be chosen independently of $\epsilon_{\rm cutoff}$ (but not $C$!).
 Thus, to obtain~\eqref{nonsuperradiantboundflux}, it suffices to bound
 \begin{equation}
 \label{sufficesfirst}
\left| \int_{\mathcal{S}} {J_{g_{a,M}}^{{\rm main}, n}}[\widetilde\psi_{n,{\bf k}}] \cdot {\rm n}_{\mathcal{S}}- 
 \int_{\mathcal{S}(\tilde\tau_0,\tilde\tau_1)}{J_{g_{a,M}}^{{\rm main}, n}}[\psi_{n,{\bf k}}] \cdot {\rm n}_{\mathcal{S}}\right|
 \end{equation}
 and
 \begin{equation}
 \label{sufficessecond}
\left| \int_{\mathcal{S}}  ( | L \widetilde\psi_{n,{\bf k}} |^2 +
 | \underline{L} \widetilde\psi_{n,{\bf k}} |^2 
+|\slashed\nabla \widetilde \psi_{n, {\bf k}}|^2  )
- \int_{\mathcal{S}(\tilde\tau_0,\tilde\tau_1)}  (
 | L\psi_{n,{\bf k}} |^2 |+  | \underline{L}\psi_{n,{\bf k}} |^2 
+ |\slashed\nabla \psi_{n, {\bf k}}|^2 ) \right|
 \end{equation}
 in terms 
 of a small multiple times 
 \begin{equation}
 \label{saytwo}
 \int_{\mathcal{S}}  ( | L \widetilde\psi_{n,{\bf k}} |^2 +
 | \underline{L} \widetilde\psi_{n,{\bf k}} |^2 
+|\slashed\nabla \widetilde \psi_{n, {\bf k}}|^2  )
+
  \int_{\mathcal{S}(\tilde\tau_0,\tilde\tau_1)}  (
 | L\psi_{n,{\bf k}} |^2 |+  | \underline{L}\psi_{n,{\bf k}} |^2 
+ |\slashed\nabla \psi_{n, {\bf k}}|^2 ) ,
 \end{equation}
 and the error terms on lines~\eqref{termsonthislinehavealabel} 
 and~\eqref{secondlineoferrortermshere}. 
 
 The term~\eqref{sufficessecond} can be estimated by the sum of
  \begin{equation}
 \label{rewrittensecond}
\int_{\mathcal{S}(\tilde\tau_0,\tilde\tau_1)}  ( | L \widetilde\psi_{n,{\bf k}}- L\psi_{n,{\bf k}} | 
|L \widetilde\psi_{n,{\bf k}}+ L\psi_{n,{\bf k}}  | +
 | \underline L \widetilde\psi_{n,{\bf k}}- \underline L\psi_{n,{\bf k}} | 
|\underline L \widetilde\psi_{n,{\bf k}}+ \underline L\psi_{n,{\bf k}}  | 
+|\slashed\nabla \widetilde \psi_{n, {\bf k}}
-\slashed\nabla \psi_{n, {\bf k}}| | \slashed\nabla \widetilde \psi_{n, {\bf k}}
+\slashed\nabla \psi_{n, {\bf k}}|   ) 
 \end{equation}
 and
 \begin{equation}
 \label{othertermoverhere}
  \int_{\mathcal{S}\setminus\mathcal{S}(\tilde\tau_1,\tilde\tau_2)}   | L \widetilde\psi_{n,{\bf k}} |^2 +
 | \underline{L} \widetilde\psi_{n,{\bf k}} |^2 
+|\slashed\nabla \widetilde \psi_{n, {\bf k}}|^2  .
 \end{equation}
 
To estimate~\eqref{rewrittensecond}, by Cauchy--Schwarz, we may indeed estimate the one factor
 by a small multiple times~\eqref{saytwo}. It thus suffices to bound 
   \begin{equation}
 \label{rewrittensecondsuffices}
\int_{\mathcal{S}(\tilde\tau_0,\tilde\tau_1)}  ( | L \widetilde\psi_{n,{\bf k}}- L\psi_{n,{\bf k}} |^2 +
 | \underline L \widetilde\psi_{n,{\bf k}}- \underline L\psi_{n,{\bf k}} |^2
+|\slashed\nabla \widetilde \psi_{n, {\bf k}}
-\slashed\nabla \psi_{n, {\bf k}}|^2  ) .
 \end{equation}
 
This may be bounded by
\[
\| Q \chi_{\tau_0,\tau_1} \psi \|^2_{H^{0,1}(\mathcal{S})}
\]
where $Q$ is the pseudodifferential operator 
$\tilde\chi_{\tilde\tau_0,\tilde\tau_1} P_n\mathring{\mathfrak{D}}^{\bf k}(\tilde\chi^2_{\tilde\tau_0,\tilde\tau_1}-\chi_{\tau_0,\tau_1})$ 
and we may rewrite 
\begin{eqnarray*}
Q	&=&	Q_1+Q_2\\
	&=& P_n \mathring{\mathfrak{D}}^{\bf k} \tilde\chi_{\tilde\tau_0,\tilde\tau_1}(\tilde\chi^2_{\tilde\tau_0,\tilde\tau_1}-\chi^2_{\tau_0,\tau_1}) +[P_n
	\mathring{\mathfrak{D}}^{\bf k}, \tilde\chi_{\tilde\tau_0,\tilde\tau_1}] 
(\tilde\chi^2_{\tilde\tau_0,\tilde\tau_1}-\chi^2_{\tau_0,\tau_1}).
\end{eqnarray*}

Now we have from Proposition~\ref{mixedsobolevboundednessprop} that 
\begin{equation}
\label{newboundherefortheterms}
\| Q_1\psi \|^2_{H^{0,1}(\mathcal{S})} \lesssim 
\mathring{\Ezerok}_{\mathcal{S}} (\tilde\tau_0-\epsilon_{\rm cutoff},\tilde\tau_0+\epsilon_{\rm cutoff}) 
+\mathring{\Ezerok}_{\mathcal{S}} (\tilde\tau_1-\epsilon_{\rm cutoff},\tilde\tau_1+\epsilon_{\rm cutoff})
+  \Ezerokminusone_{\mathcal{S}}(\tau_0,\tau_1) ,
\end{equation}
where the implicit constant on the first two terms on the right hand side (but not the third!)~may be chosen
independently of $\epsilon_{\rm cutoff}$
while
\[
\|Q_2\psi\|^2_{H^{0,1}(\mathcal{S})} \lesssim  \Ezerokminusone_{\mathcal{S}}(\tau_0,\tau_1) ,
\]
where we have used Proposition~\ref{commutatorestimate} and we recall
the energy of~\eqref{othercommutatorshorizonflux}.

To estimate~\eqref{othertermoverhere} on the other hand, we note that this
may be estimated by
\[
\|Q\tilde\chi_{\tilde\tau_0,\tilde\tau_1} \psi\|^2_{H^{0,1}(\mathcal{S})}
\]
where $Q= (1-\tilde\chi_{\tilde\tau_0+\epsilon_{\rm cutoff},\tilde\tau_1-\epsilon_{\rm cutoff}} ) P_n \mathring{\mathfrak{D}}^{\bf k} \chi_{\tilde\tau_0,\tilde\tau_1}$
and thus this term may again be estimated by the right hand side of~\eqref{newboundherefortheterms}.

We have estimated~\eqref{sufficessecond} as desired. To estimate~\eqref{sufficesfirst}, 
we note
that $J_{g_{a,M}}^{{\rm main}, n}[\Psi] \cdot {\rm n}_{\mathcal{S}}$ is a quadratic expression in $\Psi$ and derivatives
of $\Psi$.  We may thus bound~\eqref{sufficesfirst} as a sum of terms analogous to~\eqref{rewrittensecond} and~\eqref{othertermoverhere}
and argue exactly as above.

By our above remarks, the proposition follows.
\end{proof}

\subsubsection{The degeneration functions $\chi_n$ and $\tilde\chi_n$ and the generalised bulk coercivity property of~$K^{{\rm main}, n}_{g_{a,M}}$ in $r_+\le r\le R_{\rm freq}$
on the image of~$P_n$} 
\label{nowforthebulk}

We now turn to the issue of (degenerate) bulk coercivity in $r_+ \le r \le R_{\rm freq}$.

We can already deduce of course
from our definitions above that for a general function  $\Psi$,
\[
K^{{\rm main}, n}_{g_{a,M}} [\Psi] =0 
\]
in $\mathcal{D}_n$.
For general $\Psi$, however, we will not have good bulk properties
in $\{ r_+\le r\le R_{\rm freq}\} \setminus \mathcal{D}_n$. 
In this section, we will show that $K^{{\rm main}, n}_{g_{a,M}}$ has  nonnegativity properties
(and in fact suitable \underline{degenerate} coercivity properties), modulo error terms,
\emph{when restricted to the image of the associated projection $P_n$.}
As announced already at the beginning of Section~\ref{generalisedcoerc}, the results
of the present section will in particular
require frequency analysis based on Carter's full separation, in order
to correctly  exploit the allowed error terms. Thus, the proof of Theorem~\ref{globalgenbulkcoercprop} below
will only be given in Appendix~\ref{carterestimatesappend}, to which all uses of Carter's separation are deferred.

The spacetime degeneration of the coercive control provided by our
bulk current $K^{{\rm main},n}_{g_{a,M}}$ will be captured by a function $\chi_n(r)$. 
We define $\chi_n$
to be a nonnegative function
with the following properties: 
For the  generalised superradiant frequency ranges $n$, we simply set $\chi_n=1$ identically. 
For the nonsuperradiant frequency ranges $n$, 
we define $\chi_n$ to be suitably regular so that
\begin{itemize}
\item 
$\chi_n=0$ in a set whose interior contains $[r_{n,1},r_{n,2}]$,
\item
 $\chi_n=1$ for  $r\le  r'_{n,1}$ and $r\ge  r_{n,2}'$, where
 we recall these parameters from~\eqref{defofenlarged},
\item
$ |\chi_n' |\lesssim \sqrt \chi_n$.
\end{itemize}
We define finally an auxiliary nonnegative function
$\tilde\chi_n$ to be such that $\tilde\chi_n=0$ in $(r_{n,1},r_{n,2})$ but
$\tilde\chi_n=1$ in the support of $\chi_n$. (Again, for generalised superradiant
$n$, we define $\tilde\chi_n=1$ identically.)
The function $\chi_n$ will appear in the definition of $f_n$ and $y_n$ appearing
in the 
currents~\eqref{willtakethefollowingformnonsuperbefore}.

We may now state  the generalised bulk coercivity statement which is at the
heart of our argument: 
\begin{theorem}[Global generalised bulk coercivity on the image of $P_n$]
\label{globalgenbulkcoercprop}
We may choose $b_{\rm trap}>0$, $\gamma_{\rm elliptic}>0$
 and $b_{\rm elliptic}>0$ sufficiently small in~\eqref{betterdefhere},~\eqref{rellipticdef} 
 and~\eqref{elllipticdefinitionreg},
respectively, so that
the following holds.

Let $E>0$ be arbitrary and let $e_{\rm red}>0$ be sufficiently small  in the definitions~\eqref{willtakethefollowingformnonsuperbefore}
and~\eqref{willtakethefollowingformsuperbefore}.  Then we may choose $R_{\rm freq}$
and all functions appearing in these definitions so that the following is true.

Let $\psi$ be a smooth function on $\mathcal{R}(\tau_0,\tau_1)$
and, for $0\le n\le N$,  define $\psi_{n,{\bf k}}$ by~\eqref{thenpackets}.  Then
\begin{eqnarray}
\nonumber
\int_{\mathcal{M}\cap \{r_+ \le r \le R_{\rm fixed}\} } K^{{\rm main}, n}_{g_{a,M}}  [\psi_{n,{\bf k}}] 
&\ge  &c   \int_{\mathcal{M}\cap \{r_+ \le r \le R_{\rm fixed}\} }  \chi_ n r^{-2}  \left( | L\psi_{n,{\bf k}} |^2 + |\underline{L}\psi_{n,{\bf k}} | ^2 
+|\slashed\nabla \psi_{n,{\bf k}}|^2   \right)   \\
\label{globalstatementfirst}
&&\quad   -C \, {}^{\rm elp}_{\scalebox{.6}{\mbox{\tiny{\boxed{+1}}}}}\,\Xkminusone(\tau_0,\tau_1)[\psi]-C \, {}^{\natural}\Xkminusone(\tau_0,\tau_1)[\psi]
-C\, {}^{\chi}\Xzerokminusone(\tau_0,\tau_1)[\psi].
\end{eqnarray}
\end{theorem}
Note in particular the presence of ${}^{\rm elp}_{\scalebox{.6}{\mbox{\tiny{\boxed{+1}}}}}\,\Xkminusone(\tau_0,\tau_1)[\psi]$ on the right hand side.
\begin{proof}
See Appendix~\ref{fundcoercivityapp}.
\end{proof}

We may now localise the above to appropriate 
subregions $\mathcal{R}(\tilde\tau_1,\tilde\tau_2)\subset
\mathcal{R}(\tau_0,\tau_1)$ using the pseudodifferential calculus:

\begin{proposition}[Localised generalised bulk coercivity on the image of $P_n$]
Under the assumptions of Theorem~\ref{globalgenbulkcoercprop}, the following
is moreover true:

For
all  $\tau_0+1\le \tilde\tau_0 \le \tilde\tau_1-1 \le\tilde\tau_1\le  \tau_1-1$  
we have
\begin{eqnarray}
\nonumber
\int_{\mathcal{R}(\tilde\tau_0,\tilde\tau_1) \cap \{r_+\le r \le R_{\rm freq} \} } K^{{\rm main}, n}_{g_{a,M}} [\psi_{n,{\bf k}}] 
&\ge &  c  \int_{\mathcal{R}(\tilde\tau_0,\tilde\tau_1)  \cap \{r_+ \le r \le R_{\rm fixed} \} } \chi_ n r^{-2} \left( | L\psi_{n,{\bf k}} |^2 + | \underline{L}\psi_{n,{\bf k}} |^2 
+|\slashed\nabla \psi_{n,{\bf k}}|^2   \right)       \\
\label{genbulkcoer}
&&\qquad- C\epsilon_{\rm cutoff}\, \sup_{\tau\in[\tau_0,\tau_1]} \Ezerok(\tau) [\psi] -C\, {}^{\chi\natural}_{\scalebox{.6}{\mbox{\tiny{\boxed{+1}}}}}\, \Xzerokminusone(\tau_0,\tau_1)  [\psi]
\end{eqnarray}
 where $k:=|{\bf k}|$, where the constant $c$ and the constant $C$ multiplying $\epsilon_{\rm cutoff}$ may 
be chosen independently of $\epsilon_{\rm cutoff}$.
\end{proposition}

\begin{remark}
Note that the statement of the proposition is indeed compatible with  the conventions on the dependence of generic constants on $\epsilon_{\rm cutoff}$ 
from Section~\ref{noteonconstants}. 
In contrast, the constant multiplying
${}^{\chi\natural}_{\scalebox{.6}{\mbox{\tiny{\boxed{+1}}}}}\, \Xzerokminusone(\tau_0,\tau_1)[\psi]$ will depend unfavourably on $\epsilon_{\rm cutoff}$.
Again we remark that the first term on line~\eqref{genbulkcoer} will be a top order error term 
in our later argument, and we must  thus exploit the smallness of $\epsilon_{\rm cutoff}$. 
See already the proof of Proposition~\ref{toporderestimateprop}.
\end{remark}

\begin{proof}
In view of~\eqref{globalstatementfirst} shown in
Theorem~\ref{globalgenbulkcoercprop}, given $\tilde\tau_1\le \tilde\tau_2-2$,
we define 
again
$\widetilde\psi_{n,{\bf k}}$ by~\eqref{psiwithatildehere} with the 
$\epsilon_{\rm cutoff}$-dependent cutoff $\tilde\chi_{\tilde\tau_0,\tilde\tau_1}$, 
and we have that~\eqref{globalstatementfirst} holds replacing
$\psi_{n,{\bf k}}$ with $\widetilde\psi_{n,{\bf k}}$ and $\tau_0$, $\tau_1$ with 
$\tilde\tau_0$, $\tilde\tau_1$, and where the constant $c$ (but not the constants
$C$!) may be chosen independently of $\epsilon_{\rm cutoff}$.

Now note that by~\eqref{weseeeasilythatthis}, 
we have
\[
 {}^{\rm elp}_{\scalebox{.6}{\mbox{\tiny{\boxed{+1}}}}}\, \Xkminusone(\tilde\tau_0,\tilde\tau_1)+  {}^{\natural}\Xkminusone(\tilde\tau_0,\tilde\tau_1)
 +
 {}^{\chi}\Xzerokminusone(\tilde\tau_0,\tilde\tau_1)\lesssim 
{}^{\chi\natural}_{\scalebox{.6}{\mbox{\tiny{\boxed{+1}}}}}\, \Xzerokminusone(\tau_0,\tau_1),
\]
where we remark however that the implicit constant has unfavourable dependence on $\epsilon_{\rm cutoff}$.

To prove~\eqref{genbulkcoer} 
from~\eqref{globalstatementfirst} applied to $\widetilde\psi_{n,{\bf k}}$, we must estimate the
following terms:
\begin{equation}
\label{termonetoestimate}
 \int_{\mathcal{M}\cap \{r_+ \le r \le R_{\rm fixed}\} }  
\chi_n r^{-2}  \left| \, \iota_{\mathcal{R}(\tilde\tau_0,\tilde\tau_1)}( |L\psi_{n,{\bf k}}|^2+|\underline{L}\psi_{n, {\bf k}}|^2
+ |\nablaslash\psi_{n,{\bf k}}|^2)
- (|L\widetilde{\psi}_{n,{\bf k}}|^2
+ |\underline{L}\widetilde{\psi}_{n, {\bf k}}|^2
+|\nablaslash\widetilde{\psi}_{n,{\bf k}}|^2) \,
\right|    \, ,
\end{equation}
\begin{equation}
\label{termtwotoestimate}
\int_{\mathcal{M}\cap \{r_+ \le r \le R_{\rm fixed}\} }
 \left| \,  \iota_{\mathcal{R}(\tilde\tau_0,\tilde\tau_1)} K^{{\rm main}, n}_{g_{a,M}} [\psi_{n,{\bf k}}] 
-  K^{{\rm main}, n}_{g_{a,M}} [\widetilde{\psi}_{n,{\bf k}}] \, \right|   .
\end{equation}

We estimate first~\eqref{termonetoestimate}.
We may  clearly estimate this by the sum of 
\begin{align}
\nonumber
& \int_{\mathcal{R}(\tilde \tau_0,\tilde \tau_1)\cap \{r_+ \le r \le R_{\rm fixed}\} }  
\chi_n  r^{-2} \left( \,  |L\psi_{n,{\bf k}}-L\widetilde{\psi}_{n, {\bf k}} | \,
|L\psi_{n,{\bf k}}+{L}\widetilde \psi_{n, {\bf k}}|
+  |\underline L\psi_{n,{\bf k}}-\underline L\widetilde{\psi}_{n, {\bf k}} | \,
|\underline L\psi_{n,{\bf k}}+{\underline L}\widetilde \psi_{n, {\bf k}}| 
\right.
\\
&
\label{differencesusualstory}
\qquad\qquad\qquad\qquad\left.
+
|\nablaslash\psi_{n,{\bf k}}-\nablaslash\widetilde{\psi}_{n, {\bf k}} | \, |
\nablaslash\psi_{n,{\bf k}}+\nablaslash \psi_{n, {\bf k}} |
 \,
\right )  \, 
\end{align}
and
\begin{equation}
\label{bythesumofthisalso}
 \int_{(\mathcal{M}\setminus\mathcal{R}(\tilde\tau_0,\tilde\tau_1))\cap \{r_+ \le r \le R_{\rm fixed}\} }  
\chi_n r^{-2}  (|L\widetilde{\psi}_{n,{\bf k}}|^2
+ |\underline{L}\widetilde{\psi}_{n, {\bf k}}|^2
+|\nablaslash\widetilde{\psi}_{n,{\bf k}}|^2) .
\end{equation}

To estimate~\eqref{differencesusualstory}, we will apply Cauchy--Schwarz to the products putting a small constant on the second
factors. The quadratic term arising from the second factors can be absorbed then by the main terms we are comparing,
as in the proof of Proposition~\ref{asintheproofofthisprop}.
Thus, to estimate~\eqref{differencesusualstory}, it suffices to bound the quadratic term arising
from the first factors, namely:
\begin{equation}
\label{againareduction}
 \int_{\mathcal{R}(\tilde \tau_0,\tilde \tau_1)\cap \{r_+ \le r \le R_{\rm fixed}\} }  
\chi_n r^{-2} \left( \,  |L\psi_{n,{\bf k}}-L\widetilde{\psi}_{n, {\bf k}} |^2 
+  |\underline L\psi_{n,{\bf k}}-\underline L\widetilde{\psi}_{n, {\bf k}} |^2 
+
|\nablaslash\psi_{n,{\bf k}}-\nablaslash\widetilde{\psi}_{n, {\bf k}} |^2
\right )  \, .
\end{equation}

We may bound~\eqref{againareduction} by 
\[
\eqref{againareduction} \lesssim
\| Q \chi_{\tau_0,\tau_1}\psi \|^2_{H^{0,1}(\mathcal{M}\cap \{r_+ \le r \le R_{\rm freq}\}  \setminus \mathcal{D}_n)}, 
\]
where
\begin{eqnarray*}
Q&=& \tilde\chi_{\tilde\tau_0,\tilde\tau_1}( P_n\mathring{\mathfrak{D}}^{\bf k} \chi_{\tau_0,\tau_1}  - P_n\mathring{\mathfrak{D}}^{\bf k} \tilde\chi^2_{\tilde\tau_0,\tilde\tau_1})\\
&=& \tilde\chi_{\tilde\tau_0,\tilde\tau_1} P_n
\mathring{\mathfrak{D}}^{\bf k}(\chi_{\tau_0,\tau_1}-\tilde\chi^2_{\tilde\tau_0,\tilde\tau_1})
\\
&=& P_n \mathring{\mathfrak{D}}^{\bf k}  \tilde{\chi}_{\tilde\tau_0,\tilde\tau_1} (\chi_{\tau_0,\tau_1}-\tilde\chi^2_{\tilde\tau_0,\tilde\tau_1})
+  [  \tilde\chi_{\tilde\tau_0,\tilde\tau_1} ,P_n
\mathring{\mathfrak{D}}^{\bf k}]
(\chi_{\tau_0,\tau_1}-\tilde\chi^2_{\tilde\tau_0,\tilde\tau_1})
\\
&=&Q_1+Q_2 (\chi_{\tau_0,\tau_1}-\chi^2_{\tilde\tau_0,\tilde\tau_1}),
\end{eqnarray*}
where $Q_1$ is a pseudodifferential operator of order $k$  and $Q_2$ is a
pseudodifferential operator of order $k-1$.

Examining the precise form of $Q_1$, we see from Proposition~\ref{mixedsobolevboundednessprop} that
\[
\| Q_1 \chi_{\tau_0,\tau_1}\psi\|^2_{H^{0,1}(\mathcal{M}\cap \{r_+\le r\le R_{\rm freq}  \} \setminus \mathcal{D}_n)}
\lesssim \epsilon_{\rm cutoff} \sup_{\tau_0\le \tau \le \tau_1} \Ezerok(\tau)  +  {}^{\natural}\Xkminusone[\psi],
\]
where the implicit constant for the first term (but not the second!)~may be chosen
independent of $\epsilon_{\rm cutoff}$.
On the other hand,  using  Proposition~\ref{finalpseudostatement} we may rewrite 
\[
Q_2=Q_3+Q_4
\] 
where $Q_3$ is of order $k-1$ 
supported in frequency space entirely in $\mathcal{F}_n$,
while $Q_4$ is of order $k-2$. It follows from Proposition~\ref{forpseudoerrors} that
\[
\| Q_3 \chi_{\tau_0,\tau_1} \psi\|^2_{H^{0,1}(\mathcal{M}\cap \{r_+\le r\le R_{\rm freq}  \}  \setminus \mathcal{D}_n)}
\lesssim   {}^{\natural}\Xkminusone[\psi] + {}^\chi\Xzerokminusone[\psi] 
\]
while
\[
\| Q_4 \chi_{\tau_0,\tau_1}\psi \|^2_{H^{0,1}(\mathcal{M}\cap \{r_+\le r\le R_{\rm freq}  \})}
 \lesssim   {}^\chi\Xzerokminusone[\psi]  .
\]
(Here all constants have unfavourable dependence on $\epsilon_{\rm cutoff}$.)

We may estimate thus~\eqref{againareduction} by
\begin{equation}
\label{veoautopouektimoumetelika}
\eqref{againareduction} \lesssim
 \epsilon_{\rm cutoff} \sup_{\tau_0\le \tau \le \tau_1} \Ezerok(\tau)
 +   {}^{\natural}\Xkminusone[\psi] + {}^\chi\Xzerokminusone[\psi] ,
 \end{equation}
 where the implicit constant multiplying $ \epsilon_{\rm cutoff} \sup_{\tau_0\le \tau \le \tau_1} \Ezerok(\tau)$ may be taken independent of $\epsilon_{\rm cutoff}$.
 By our above remarks, this allows us to estimate~\eqref{differencesusualstory} by~\eqref{veoautopouektimoumetelika}
 and terms that we may absorb into the bulk integral on the right hand side of~\eqref{genbulkcoer}. 

To estimate~\eqref{bythesumofthisalso} on the other hand, 
we remark that this may be estimated 
\[
\eqref{bythesumofthisalso}\lesssim
\| Q \chi_{\tilde\tau_0,\tilde\tau_1}\psi \|^2_{H^{0,1}(\mathcal{M}\cap \{r_+ \le r \le R_{\rm freq}\}  \setminus \mathcal{D}_n)}, 
\]
where
\[
Q=(1-\tilde\chi_{\tilde\tau_0+\epsilon_{\rm cutoff},\tilde\tau_1-\epsilon_{\rm cutoff}}) P_n \mathring{\mathfrak{D}}^{\bf k} \tilde\chi_{\tilde\tau_0,\tilde\tau_1},
\]
and this may again be estimated by the right hand side of~\eqref{veoautopouektimoumetelika} 
by decomposing $Q$  appropriately as above.

Since we have estimated both~\eqref{againareduction} and~\eqref{bythesumofthisalso}, 
we have thus indeed estimated~\eqref{termonetoestimate} in the desired manner.

To obtain the proposition, it remains then to estimate~\eqref{termtwotoestimate}.
We argue similarly.
Let us first note the following general bound, which uses some more information
about our current:
\begin{lemma}[General boundedness properties of $K^{{\rm main}, n}_{g_{a,M}}$]
\label{generalboundednesspr}
For a general function $\psi$, we have 
\begin{equation}
\label{toboundinfactindividual}
\left|\, K^{{\rm main}, n}_{g_{a,M}} [\psi]\,\right| \lesssim (\chi_n)^{\frac12} r^{-1} \left( |L\psi|^2+|\underline{L}\psi|^2 
+|\nablaslash\psi|^2 \right)  + \tilde\chi_n r^{-3}  |\psi|^2
\end{equation}
where $\chi_n$ and $\tilde\chi_n$ are as defined 
above. (In particular the right hand side of~\eqref{toboundinfactindividual} 
vanishes identically in $\mathcal{D}_n$.)
In fact the stronger statement holds that all terms on the right hand side
of~\eqref{generalKdef} can be individually bounded 
in norm by the right hand side of~\eqref{toboundinfactindividual}.
\end{lemma}
\begin{proof}[Proof of Lemma~\ref{generalboundednesspr}]
This is clear from the form of the definition of the currents~\eqref{willtakethefollowingformnonsuperbefore}--\eqref{willtakethefollowingformsuperbefore} in  Section~\ref{fundcoercivityapp} and the definition of $\chi_n$ and $\tilde\chi_n$. Note that when expanding these currents
in the form~\eqref{generalKdef} there will be pure first order terms, mixed terms and
zeroth order terms. The first order terms can be bounded by the first term on the
right hand side of~\eqref{toboundinfactindividual} 
while the zeroth order terms can be bounded by the second term.
The mixed terms can be bounded by the sum of the two terms using 
Cauchy--Schwarz (in fact with $\chi_n$
replacing $(\chi_n)^{\frac12}$).
\end{proof}

Note that $K^{{\rm main}, n}_{g_{a,M}}[\psi]$ is a sum of quadratic terms, including
purely first order terms, mixed terms and zeroth order terms and the above
lemma bounds all of these terms  individually in norm.

The contribution of all zeroth order terms in~\eqref{termtwotoestimate}
can be bounded by sums of terms of the form
\begin{equation}
\label{againbybounded}
 \int_{\mathcal{M}\cap \{r_+ \le r \le R_{\rm fixed}\} }   \tilde\chi_n \left(|LQ\chi_{\tau_0,\tau_1} \psi|^2 + |\underline LQ\chi_{\tau_0,\tau_1}\psi|^2 + |\nablaslash Q\chi_{\tau_0,\tau_1}\psi|^2\right)
\end{equation}
where $Q$ is a pseudodifferential operator of order $k-1$ whose  symbol has
frequency support contained in the support of  $P_n$. 
Thus, Proposition~\ref{forpseudoerrors} applies to bound this term by
${}^{\natural}\Xkminusone(\tau_0,\tau_1)$ which is of course 
in turn estimated by~${}^{\chi\natural}_{\scalebox{.6}{\mbox{\tiny{\boxed{+1}}}}}\,\Xpkminusone(\tau_0,\tau_1)$ appearing on the right
hand side of~\eqref{genbulkcoer}.

The contribution of all mixed terms in~\eqref{termtwotoestimate}
can be handled by applying Cauchy--Schwarz, throwing a small constant
on the resulting first order terms and a large constant on the zeroth order terms,
and as in the proof of Lemma~\ref{generalboundednesspr}, noting that one can replace $(\chi_n)^{\frac12}$
with $\chi_n$. 
The zeroth order terms are bounded again by~\eqref{againbybounded}, while the first order terms,
for a sufficiently small constant, in view of the $\chi_n$ can be absorbed in
the main term being estimated.

Thus, it suffices to estimate the difference of purely first order quadratic terms
contained in~\eqref{termtwotoestimate}.  
Using Lemma~\ref{generalboundednesspr}, we may bound
the difference of the purely first order quadratic terms by the sum of
\eqref{differencesusualstory} and~\eqref{bythesumofthisalso},
where $(\chi_n)^{\frac12}$ is replaced by $\chi_n$.
The arguments above still apply  then 
and  these terms are again estimated by~\eqref{veoautopouektimoumetelika}
and a term which can be absorbed in the bulk integral on the right hand side of~\eqref{genbulkcoer}.
This yields the desired estimate on~\eqref{termtwotoestimate}, and thus completes the proof. 
\end{proof}

\subsubsection{Bulk coercivity properties of $K_{g_{a,M}}^{{\rm main},n}$  in the regions $r_0\le r\le r_+$ and
$r\ge R_{\rm freq}$}
\label{bulkcoerfar}

Concerning the region $r_0\le r\le r_+$ within the Kerr event horizon, we may essentially infer bulk coercivity
softly from the statements already shown, after restriction of $r_0$:

\begin{proposition}[Bulk coercivity of $K^{{\rm main},n}_{g_{a,M}}$ in $r_0\le r\le r_+$]
\label{notbcgsc}
Under the assumptions of Theorem~\ref{globalgenbulkcoercprop}, 
after choosing $r_0$ sufficiently close to $r_+$, then
for a general function $\Psi$ and all $0\le n\le N$, we have in the region $r_0\le r\le r_+$
the coercivity statement
\[
K^{{\rm main},n}_{g_{a,M}}  [\Psi ] \ge c(|L\Psi|^2 +|\underline L \Psi|^2 +|\nablaslash \Psi|^2)
- C |\Psi|^2.
\]
\end{proposition}

\begin{proof}
After choosing $r_0$ sufficiently close to $r_+$, 
this follows softly from the analogous statement in $r_+\le r \le r_{\rm close}$ shown in the
proof of Theorem~\ref{globalgenbulkcoercprop} in Appendix~\ref{fundcoercivityapp}, independently
of the precise extension of the currents to $r<r_+$.
\end{proof}

Recall that  $K^{{\rm main},n}_{g_{a,M}}$ is independent of $n$
in the region $r\ge R_{\rm freq}$, and we write 
 $K^{{\rm main},n}_{g_{a,M}} = K^{{\rm main}}_{g_{a,M}}$. 
We have the following bulk coercivity statement:

\begin{proposition}[Bulk coercivity of $K^{{\rm main}}_{g_{a,M}}$ in $r\ge R_{\rm freq}$]
\label{farawaybulkcoercpropmain}
Under the assumptions of Theorem~\ref{globalgenbulkcoercprop},
for a general function $\Psi$, then in the region $r\ge R_{\rm freq}$, we have
the coercivity statement
\[
K^{{\rm main}}_{g_{a,M}} [\Psi ] \gtrsim 
r^{-3} ( |L\Psi|)^2 + | \underline L\Psi|+  |\nablaslash\Psi|^2
+ r^{-1} |\Psi|^2 ).
\]
\end{proposition}

\begin{proof}
This follows from Proposition~\ref{coercoftheelements} of Appendix~\ref{fundcoercivityapp}.
\end{proof}

\subsection{The currents $J{}_{g_{a,M}}^{\rm redder}$, $J{}_{g_{a,M}}^{\rm total}$ and improved coercivity properties}
 \label{absorbingetcsec}

Given $\delta\le p\le 2-\delta$ or $p=0$, we define
in the region $r\ge R_{\rm freq}$ the current
\[
\Jp{}_{g_{a,M}}^{{\rm total}}[\psi] : = J_{g_{a,M}}^{{\rm main}}[\psi]  +
e_{\rm total}\zeta \Jp{}^{\rm far}_{g_{a,M}} [\psi],
\]
where $\Jp^{\rm far}$ is the current defined in Appendix B of~\cite{DHRT22},
$\zeta$ is the cutoff function of Section~\ref{commutationcommutation}
and $e_{\rm total}>0$ is a sufficiently small
parameter.

We also define the current
\[
J_{g_{a,M}}^{\rm redder} [\psi]:= \hat{\hat{\zeta} }J^{\varsigma(k)}_{g_{a,M}}[\psi],
 \]
where $J^{\varsigma(k)}{}_{g_{a,M}}[\psi]$ is the current from
Section~3.5 of~\cite{DHRT22} and
$\hat{\hat\zeta}$ is a cutoff function
such that $\hat{\hat\zeta}(r) =1$ for $r\le r_1+2\frac{r_2-r_1}3$
and $\hat{\hat\zeta}(r)=0$ for $r\ge r_2$
(cf.~the cutoff function $\hat\zeta$
 of Section~\ref{commutationcommutation}).

Recall that these currents have improved
coercivity properties, allowing us in particular
to absorb top order commutation errors  generated by $\mathfrak{D}^{\bf k}\psi$ 
for $k_i\ne 0$ for $i\ge 2$.

We first consider $J^{\rm redder}_{g_{a,M}}$ and the region near the horizon.
Concerning the boundary, we have
\begin{proposition}[Boundary coercivity of $J^{\rm redder}$ on $\mathcal{S}$ and nonnegativity
on $\Sigma_\tau$]
\label{redderbound}
We have
\[
J_{g_{a,M}}^{\rm redder}[\psi] \cdot {\rm n}_{\mathcal{S}}\ge c(r_0)(|L\psi|^2+|\underline{L}\psi|^2
+|\nablaslash\psi|^2),
\]
and
\begin{equation}
\label{wehaveinfactcoercclosefirsttime}
J_{g_{a,M}}^{\rm redder}[\psi]\cdot {\rm n}_{\Sigma(\tau) }\ge 0 .
\end{equation}
\end{proposition}
\begin{proof}
This is of course immediate from Proposition 3.4.2 of~\cite{DHRT22}.
We have in fact a stronger coercivity statement in place of~\eqref{wehaveinfactcoercclose}
for $r\le r_1$ say, but we shall not make direct use of that. We note that~\eqref{wehaveinfactcoercclosefirsttime} in fact
vanishes for $r\ge r_2$.
\end{proof}
Concerning the bulk, we have
\begin{proposition}[Improved bulk coercivity of $K^{\rm redder}$ near $\mathcal{S}$]
\label{improvcoercprop}
For all $r_2>r_1>r_+$ sufficiently close to $r_+$,
we have for all functions $\Psi$ the coercivity
\[
 K{}^{{\rm redder}}_{g_{a,M}}  [\Psi] \gtrsim |L\Psi|^2+|\underline{L}\Psi|^2+|\nablaslash\Psi|^2
 \qquad\text{\ in\ } r_0\le r\le r_1+2\frac{r_2-r_1}3
\]
and for any $k\ge 1$,  there exists a positive real weight function $\alpha({\bf k})$ defined
on $|{\bf k}| =k$ 
 such that, for all solutions $\psi$
of the homogeneous $\Box_g\psi=0$,
\begin{align}
\nonumber
\sum_{|{\bf k}| = k}
\int_{\mathcal{R}(\tilde\tau_0,\tilde\tau_1)\cap \{r_0\le r \le r_1+2\frac{r_2-r_1}3 \} }  K{}^{{\rm redder}}_{g_{a,M}} 
[\alpha({\bf k}) \mathfrak{D}^{\bf k} \psi]
+ {\rm Re} ( H{}^{\rm redder}_{g_{a,M}}[\psi] \cdot \overline{ \Box_{g_{a,M}}\alpha({\bf k}) \mathfrak{D}^{\bf k}\psi})\\
\label{giveitanamenow}
  \ge
\int_{\mathcal{R}(\tilde\tau_0,\tilde\tau_1)\cap \{r_0\le r \le r_1+2\frac{r_2-r_1}3 \} }
\frac13 
 K{}^{\rm redder}_{g_{a,M}} 
[\alpha({\bf k}) \mathfrak{D}^{\bf k} \psi]
   -C\, {}^\chi\Xzerokminusone(\tilde\tau_0,\tilde\tau_1)[\psi].
\end{align}
\end{proposition}
\begin{proof}
This is immediate from Proposition 3.4.2 of~\cite{DHRT22} and
the elliptic estimates of Proposition~\ref{ellproplargea}. (We note that the second term
in the integral on the left hand side is supported only in the region $r\le r_1+\frac{r_2-r_1}2$
in view of the definition of the cutoff~$\hat\zeta$ in Section~\ref{commutationcommutation}.)
See the proof of Proposition~3.6.4 of~\cite{DHRT22} for the more general inhomogeneous case.
\end{proof}

We now turn to $\Jp_{g_{a,M}}^{\rm total}$. Concerning the boundary, we have
\begin{proposition}[Improved boundary coercivity in $r\ge R_{\rm freq}$]
\label{largercoercivityboundary}
Let $E>0$ be sufficiently large in the definitions~\eqref{willtakethefollowingformnonsuperbefore}
and~\eqref{willtakethefollowingformsuperbefore}. For a general function $\psi$, we have  for $r\ge R_{\rm freq}$, and $\delta \le p\le 2-\delta$ the boundary coercivity
\[
\Jp{}_{g_{a,M}}^{{\rm total}}[\psi] \cdot  {\rm n}_{\Sigma(\tau)} \gtrsim   r^p| r^{-1}L(r\psi)|^2 +r^{\frac{p}2}|L\psi|^2 +|\slashed\nabla \psi|^2
+r^{\frac{p}2-2}|\psi|^2  +\iota_{r\le R}| \underline{L}\psi|^2
\]
while for $p=0$ we have
\[
\Jzero{}_{g_{a,M}}^{{\rm total}}[\psi] \cdot  {\rm n}_{\Sigma(\tau)} \gtrsim  |L\psi|^2 +|\slashed\nabla \psi|^2
+r^{-2}|\psi|^2 +\iota_{r\le R}| \underline{L}\psi|^2.
\]
\end{proposition}
\begin{proof}
This follows from Sections B.1 and B.2 of~\cite{DHRT22}.
\end{proof}
Concerning the bulk, we have
\begin{proposition}[Improved bulk coercivity  in $r\ge R_{\rm freq}$]
\label{impbulcoerpropfar}
Let $E>0$ be sufficiently large in the definitions~\eqref{willtakethefollowingformnonsuperbefore}
and~\eqref{willtakethefollowingformsuperbefore}.
For $e_{\rm total}>0$ sufficiently small, we have for $r\ge R_{\rm freq}$
and all functions $\Psi$ the coercivity statements
\begin{align*}
 \Kp{}^{{\rm total}}_{g_{a,M}}[\Psi]
&\gtrsim  r^{p-1} ( 
 (r^{-1}|Lr\Psi|)^2 + | L\Psi|+|\nablaslash\Psi|^2) + r^{-1-\delta}  |\underline L \Psi|^2
+ r^{p-3} |\Psi|^2  \qquad \delta\le p \le 2-\delta,\\
 \Kzero{}^{{\rm total}}_{g_{a,M}}[\Psi]
&\gtrsim  r^{-1-\delta} ( 
 | L\Psi|+|\nablaslash\psi|^2 + |\underline L \Psi|^2)
+ r^{-3-\delta } |\Psi|^2  \qquad p=0,
\end{align*}
 and, for any $k\ge 1$, 
the  $\alpha({\bf k})$ of Proposition~\ref{improvcoercprop} may be chosen so that
moreover, 
 for all solutions of the homogeneous equation $\Box_{g_{a,M}}\psi=0$,  we have
\begin{align}
\nonumber
\int_{\mathcal{R}(\tilde\tau_0,\tilde\tau_1)\cap \{r \ge  R/4\} } 
\sum_{|{\bf k}| = k}
  \Kp{}^{{\rm total}}_{g_{a,M}}  
[\alpha({\bf k}) \mathfrak{D}^{\bf k} \psi]
+{\rm Re} (\Hap{}^{{\rm total}}_{g_{a,M}}[\psi] \cdot\overline{ \Box_{g_{a,M}}\alpha({\bf k}) \mathfrak{D}^{\bf k}\psi }) \\
\label{improvedfarawaycoercivity}
\ge
\int_{\mathcal{R}(\tilde\tau_0,\tilde\tau_1)\cap \{r \ge R/4  \} } 
\frac13 \sum_{|{\bf k}| = k}
  \Kp{}^{{\rm total}}_{g_{a,M}}  
[\alpha({\bf k}) \mathfrak{D}^{\bf k} \psi]    -C\, {}^\chi\Xpkminusone(\tilde\tau_0,\tilde\tau_1)[\psi].
\end{align}
\end{proposition}
\begin{proof}
This follows from Section 3.6.6 of~\cite{DHRT22}, where the more general
inhomogeneous version is obtained. (We recall that in a region $r\ge R_k$ (depending on $k$, however, with $R_k\to \infty$ as $k\to \infty$!),~\eqref{improvedfarawaycoercivity} follows in fact from a pointwise coercivity bound, while 
for the region $R/2\le r\le R_k$ one must
exploit the elliptic estimate~\eqref{ellipticestimnoYspacetime}.)
\end{proof}

\begin{proposition}[Boundary coercivity at $\mathcal{I}^+$]
\label{boundcoercatIplus}
Let $E>0$ be sufficiently large in the definitions~\eqref{willtakethefollowingformnonsuperbefore}
and~\eqref{willtakethefollowingformsuperbefore}.
For a general function $\psi$ on $\mathcal{R}(\tau_0,\tau_1)$,
suppose 
\begin{equation}
\label{supposeitholds}
\int_{\Sigma(\tau_0)\cap \{r\ge R\} } \Jp{}_{g_{a,M}}^{{\rm total}}[\psi]  \cdot {\rm n} <\infty, 
\qquad
\int_{\mathcal{R}(\tau_0,\tau_1) \cap \{r\ge R\} } |\Hap^{\rm total}[\psi]\Box_{g_{a,M}}\psi |<\infty .
\end{equation}
Then we may integrate the divergence identity of $\Jp{}_{g_{a,M}}^{\rm total}[\psi]$ globally
in $\mathcal{R}(\tau_0,\tau_1)\cap \{r\ge R\}$ with  a well defined additional boundary term
 at $\mathcal{I}^+$ which satisfies the coercivity property:
\begin{equation}
\label{throwitaway}
\int_{\mathcal{I}^+(\tau_0,\tau_1)}
\Jp{}_{g_{a,M}}^{{\rm total}}[\psi]  \cdot {\rm n}_{\mathcal{I}^+}
\ge 0.
\end{equation}
\end{proposition}

\begin{proof}
This follows from a straightforward limiting argument from the energy identity~\eqref{energyidentity} on a compact region. The left hand side of~\eqref{throwitaway}
may be interpreted as $\lim_{v\to \infty} \int_{C_{v}\cap \mathcal{R}(\tau_0,\tau_1)}  
\Jp{}_{g_{a,M}}^{{\rm total}}[\psi]  \cdot {\rm n}_{C_v}$, whose coercivity follows from
Section~B.2  of~\cite{DHRT22}.
\end{proof}
We will only use Proposition~\ref{boundcoercatIplus} to throw away the boundary term~\eqref{throwitaway} 
in the energy identity.
See already~Proposition~\ref{toporderestimateprop}.

\subsection{Fixing the parameters}
\label{fixersection}

Proposition~\ref{boundcoercatIplus}, which constrains $E$, provides in fact the final constraint on our parameters. We now choose $e_{\rm red}>0$ sufficiently small to satisfy all the smallness
assumptions and
 $E>0$ sufficiently large so as to satisfy all largeness assumptions of the
various propositions of this section. We fix then 
$R_{\rm freq}$ from Theorem~\ref{globalgenbulkcoercprop}, $e_{\rm total}$ from Proposition~\ref{impbulcoerpropfar}, and $r_0<r_+$
sufficiently close to $r_+$ satisfying the closeness requirements of the propositions 
of this section.

We fix $r_1$ so as to satisfy the closeness of Proposition~\ref{improvcoercprop} 
as well as the constraints
$r_1<r_{\rm pot}$ where $r_{\rm pot}$ is a parameter appearing in Appendix~\ref{carterestimatesappend} (fixed at the beginning of Section~\ref{fixingthecovering})
and we fix now also $r_+<r_1<r_2<r_{\rm pot}$. 

We fix finally $R$ such that $R/4\ge R_{\rm freq}$.

We recall that the parameters
$b_{\rm trap}>0$, $\gamma_{\rm elliptic}$, $b_{\rm elliptic}>0$ of~\eqref{betterdefhere} and~\eqref{elllipticdefinitionreg}
are fixed by the statement of Theorem~\ref{globalgenbulkcoercprop}
(see already Section~\ref{fixparamsheremust}).

All $r$-parameters from Table~\ref{firsttable}, 
as well as the current parameters $e_{\rm red}$, $E$, $e_{\rm total}$,
the frequency range parameters $b_{\rm trap}$, $\gamma_{\rm ellitpic}$, $b_{\rm elliptic}$ 
and the functions appearing in  definitions~\eqref{willtakethefollowingformnonsuperbefore}
and~\eqref{willtakethefollowingformsuperbefore}  (whose discussion is of course
deferred to  Appendix~\ref{carterestimatesappend}),
 can now be viewed as chosen depending only on $a$ and $M$. 
In particular, when applied in Sections~\ref{quasilinearsectionlargea}--\ref{proofofmainthelargea}, all the estimates of the present section have generic constants $C$ which can be considered to
depend only on $a$, $M$ (and $k$ and $\delta$ where appropriate) and possibly on $\epsilon_{\rm cutoff}$ (and an additional constant $\epsilon_{\rm redder}$ to be introduced, which both still must be chosen in the context of Proposition~\ref{toporderestimateprop}). Note again, however, the special
 convention of 
Section~\ref{noteonconstants} regarding the dependence of generic constants $c$ and $C$ on 
$\epsilon_{\rm cutoff}$, $\epsilon_{\rm redder}$ (in particular when occurring in the product expressions $C\epsilon_{\rm cutoff}$).  (We will thus have to continue to track the dependence of those constants explicitly in order
to eventually absorb such terms.)

\section{Quasilinear equations: preliminaries and applications of the estimates and currents}
\label{quasilinearsectionlargea}

We now at last may turn to our study of the quasilinear equation~\eqref{theequationzero}.

We first review in Section~\ref{generalclassconsideredrecalled} the general class of nonlinearities
considered from~\cite{DHRT22}.
In Section~\ref{LWPsec}, we will review local well posedness for~\eqref{theequationzero} 
including useful Cauchy-type stability 
statements.  
We then introduce in Section~\ref{additionalenergysec} the basic smallness assumption
on $\psi$, 
the extension $\widehat\psi$ and
the additional top order energies which will play a central role in our argument. 
In Section~\ref{covariantones},
we will promote the currents of Section~\ref{generalisedcoerc} and~\ref{absorbingetcsec}
to covariant currents with respect to $g(\psi,x)$  and infer stability properties of the coercivity statements,
up to suitable error terms. As outlined in Section~\ref{toporderenerintro}, we then apply in Section~\ref{toporderidentity} these currents
at top order to~\eqref{theequationzero},  to obtain a basic energy estimate with various
error terms on the right hand side.
 We estimate all terms on the right hand side  in Section~\ref{estimforinhomog} in turn from
 our energies,    to obtain finally 
 a closed system of estimates in Section~\ref{puttogethersec}.

\subsection{The class of equations from~\cite{DHRT22}}
\label{generalclassconsideredrecalled}

The class of equations considered will be exactly 
as in Section~4.1 and 4.7 of~\cite{DHRT22}, specialised
to the case where $g_{(0,x)} = g_{a,M}$ with $|a|<M$. 
(We in particular again  allow for now also the case $a=0$.)

We review here:
We will consider  quasilinear equations of the form below:
\begin{equation} 
\label{theequationforlargea}
\Box_{g(\psi, x)} \psi = N^{\mu\nu}(\psi,x)\partial_\mu\psi\,\partial_\nu \psi,
\end{equation}
where we recall
\begin{equation}
\label{thenonlinearityfunctions}
g: \mathbb R\times \mathcal{M} \to T^*\mathcal{M}\otimes T^*\mathcal{M}, \qquad
N:\mathbb R\times \mathcal{M}\to T \mathcal{M}\otimes T \mathcal{M}
\end{equation}
satisfy $\pi\circ g(\psi,x)=x$, $\pi\circ N(\psi, x)=x$
with  $\pi : T \mathcal{M}\otimes T \mathcal{M}\to \mathcal{M}$, $\pi : T^*\mathcal{M}\otimes T^*\mathcal{M}\to \mathcal{M}$
the canonical projections, and 
 such that 
 \[
 g(0,x)=g_{a,M}(x)
 \] 
 for all $x$ while $g(\cdot,x)=g_{a,M}$ for $r(x)\ge R/2$, and $N$ and $g$ are smooth, and, for each $k$,
$\partial_{{\bf x}}^{{\beta}}\partial_\xi^s g^{\mu\nu}(\xi, x)$ and 
$\partial_{{\bf x}}^{{\beta} } \partial_\xi^s N^{\mu\nu}(\xi, x)$ are
uniformly bounded for all $|\beta|+s\le k$, all $r\le R$ and $|\xi|\le 1$, where
$\partial_{{\bf x}}^{\beta}$ is expressed in multi-index notation
with respect to the ambient Cartesian coordinates
recalled in  Section~\ref{subextremalkerrsec}.

Our additional assumption on $N$ will only depend on the region $r\ge 8R/9$, so we will need some notation to denote the restriction of energies, etc., to this region.
Given $r_*\ge r_0$ and  $v$,  and a $\psi$ defined on $\mathcal{R}(\tau_0,\tau_1)$, by 
\[
  \Xpk_{r_*,v}(\tau_0,\tau_1)[\psi], \quad \Xzerok_{r_*,v}(\tau_0,\tau_1)[\psi], \quad \Xzeroplusk_{r_*,v}(\tau_0,\tau_1)[\psi]
 \]
we shall mean  the expressions defined as in Section~\ref{masterenergiessec}, but where all domains of integration are restricted to the region $\mathcal{R}(\tau_0,\tau_1,v)\cap \{r\ge r_*\}$. 

With this, we may now recall  our additional assumption on $N$  capturing null structure (Assumption~4.7.1 of~\cite{DHRT22}):

\begin{assumption}[Null structural assumption for semilinear terms~\cite{DHRT22}]
\label{largeanullcondassumption}
There exists a $k_{\rm null}$, and for all $k\ge k_{\rm null}$, 
 an $\varepsilon_{\rm null}>0$, such that the following holds for all $\tau_0\le \tau_1$, $v$.

Let $\psi$ be a smooth function defined on $\mathcal{R}(\tau_0,\tau_1,v)\cap \{ r\ge 8R/9\}$ 
 satisfying the bound
 \begin{equation}
 \label{basicbootstrapinnullcond}
 \Xzerolesslessk_{8R/9, v}(\tau_0,\tau_1)[\psi] \leq \varepsilon
 \end{equation}
  for
some $0<\varepsilon\leq \varepsilon_{null}$. 
Then for all $\delta\le p\le 2-\delta$,
\begin{align}
\nonumber
\int_{\mathcal{R}(\tau_0,\tau_1,v)\cap \{r\ge R\}}&\sum_{|{\bf k}|\le k}(|r^pr^{-1}L(r{\mathfrak{D}}^{\bf k}\psi)|+|L{\mathfrak{D}}^{\bf k}\psi|
+|\underline{L}{\mathfrak{D}}^{\bf k}\psi|+|r^{-1}\slashed\nabla{\mathfrak{D}}^{\bf k}\psi|+ |r^{-1}{\mathfrak{D}}^{\bf k}\psi| )  |{\mathfrak{D}}^{\bf k}
( N^{\alpha\beta}(\psi,x)\partial_\alpha\psi\,\partial_\beta \psi) |\\
\nonumber
&\qquad\qquad\qquad\qquad\qquad
+ |\mathfrak{D}^{\bf k} ( N^{\alpha\beta}(\psi,x)\partial_\alpha\psi\,\partial_\beta \psi)|^2 \\
\label{nullcondassump}
&\lesssim
 \Xpk_{\frac{8R}9,v}(\tau_0,\tau_1) \sqrt{\Xzerolesslessk_{\frac{8R}9,v}(\tau_0,\tau_1)} + \sqrt{\Xpk_{\frac{8R}9,v}(\tau_0,\tau_1)}\sqrt{\Xzerok_{\frac{8R}9,v}(\tau_0,\tau_1)}\sqrt{\Xplesslessk_{\frac{8R}9,v}(\tau_0,\tau_1)},
\end{align}
while, corresponding to $p=0$ we have
\begin{align}
\nonumber
\int_{\mathcal{R}(\tau_0,\tau_1,v)\cap \{r\ge R\}}&\sum_{|{\bf k}|\le k}(|r^{-1}L(r{\mathfrak{D}}^{\bf k}\psi)|
+|L{\mathfrak{D}}^{\bf k}\psi|
+|\underline{L}{\mathfrak{D}}^{\bf k}\psi|+|r^{-1}\slashed\nabla{\mathfrak{D}}^{\bf k}\psi|+ |r^{-1}{\mathfrak{D}}^{\bf k}\psi| )  |{\mathfrak{D}}^{\bf k}
( N^{\alpha\beta}(\psi,x)\partial_\alpha\psi\,\partial_\beta \psi) |\\
\nonumber
&\qquad\qquad\qquad\qquad\qquad
+ |\mathfrak{D}^{\bf k} ( N^{\alpha\beta}(\psi,x)\partial_\alpha\psi\,\partial_\beta \psi)|^2 \\
\label{nullcondassumptwo}
 &\lesssim 
 \Xzeroplusk_{\frac{8R}9,v} (\tau_0,\tau_1) \sqrt{ \Xzeropluslesslessk_{\frac{8R}9,v}(\tau_0,\tau_1)}.
\end{align}
\end{assumption}

We have formulated our null structural assumption directly in terms
of the necessary estimates for maximum generality. We note that, as
shown in Appendix C.2 of~\cite{DHRT22}, a sufficient condition
for the above in the general subextremal Kerr case is that  $N(\xi ,x)$ satisfies
\begin{align}
\nonumber
	\sup_{\vert \xi \vert \leq 1, r\ge R} \sum_{\vert \textbf{k} \vert + s \leq k} \sum_{A,B=1,2}
	&|\Dk \partial_{\xi}^{s} 
	(r N^{uu})|
	+
	|\Dk \partial_{\xi}^{s}  N^{uv}|
	+
	|\Dk \partial_{\xi}^{s}  N^{vv}|\\
	 \label{assumponNzero}
	&+
	|\Dk \partial_{\xi}^{s}  (r N^{Au})|
	+
	|\Dk \partial_{\xi}^{s} (r N^{Av})|
	+
	|\Dk \partial_{\xi}^{s} (r^2 N^{AB})|
	\lesssim
	D_k
	,
\end{align}
where the $A$, $B$ indices refer to coordinates  $\theta,  \phi$, and $D_k$ are
arbitrary constants. This is thus a strict generalisation of the classical null 
condition of~\cite{KlNull} and the class of semilinearities on Kerr considered 
in~\cite{MR3082240}.

\begin{remark}
The work~\cite{DHRT22} always considered real-valued $\psi$. 
Since for convenience, 
the formalism of the present paper allows complex valued functions to accommodate
Fourier projections, one may without additional complication consider at the outset also
complex-valued solutions 
of~\eqref{theequationforlargea},
i.e.~replace $\mathbb R$ by $\mathbb C$ in~\eqref{thenonlinearityfunctions}.
In that case, it is natural to allow a more general expression
\[
N_1^{\mu\nu}(\psi,x)\partial_\mu \psi \,\partial_\nu\psi 
+ N_2^{\mu\nu}(\psi,x)\overline{\partial_\mu \psi} \,\partial_\nu\psi 
+N_3^{\mu\nu}(\psi,x)\partial_\mu \psi\, \overline{\partial_\nu\psi }
+N_4^{\mu\nu}(\psi,x)\overline{\partial_\mu\psi}\,\overline{ \partial_\nu\psi }
 \]
 on the right hand side of~\eqref{theequationforlargea}, where
$N_i :\mathbb C\times \mathcal{M}\to T \mathcal{M}\otimes T \mathcal{M} \otimes \mathbb C$.
We have that the results of this paper will apply as long as each of the expressions
\[
N_1^{\mu\nu}(\psi,x)\partial_\mu \psi\, \partial_\nu\psi,
\,\,\,N_2^{\mu\nu}(\psi,x)\overline{\partial_\mu \psi} \,\partial_\nu\psi,\ldots
\]
individually
satisfy~\eqref{nullcondassump} and~\eqref{nullcondassumptwo}. 
Note that a sufficient condition for this would be that
each of the $N_i$ satisfy~\eqref{assumponNzero}. 
\end{remark}

\subsection{Local well posedness and Cauchy-type stability}
\label{LWPsec}

We will review here ``local'' well posedness and  Cauchy-type stability results.
(We emphasise that as these statements are local with respect to the parameter $\tau$ and the hypersurfaces
$\Sigma(\tau)$ asymptote to null infinity $\mathcal{I}^+$,  they are actually what
are sometimes referred to as \emph{semi-global} statements, and their proofs use already the
generalised  null 
structure Assumption~\ref{largeanullcondassumption}.)

We define a smooth initial data set on $\Sigma(\tau_0)$
to be a pair $(\uppsi, \uppsi')$ where $\uppsi$ is  a function on
$\Sigma(\tau_0)$ smooth on 
$\Sigma(\tau_0)\cap\{r\le R\}$ and smooth on $\Sigma(\tau_0)\cap \{r \ge R\}$, and 
$\uppsi'$ is a smooth function on $\Sigma(\tau_0)\cap \{r\le R\}$,
such that moreover  there exists a smooth a function $\Psi$ on $\mathcal{R}(\tau_0,\tau_1)$
for some $\tau_1$ such that $\Psi|_{\Sigma(\tau_0)}=\uppsi$, and $n\Psi|_{\Sigma(\tau_0)\cap\{r\le R\}}=\uppsi'$,
where $n$ denotes the normal to $\Sigma(\tau_0)\cap\{r\le R\}$.

We recall the energy norm
\begin{equation}
\label{tobedfined}
 \Epk[\uppsi,\uppsi']
\end{equation}
from~\cite{DHRT22} defined  by $\Epk[\uppsi,\uppsi']= \Epk(\tau_0)[\psi]$,
where the spacetime derivatives of $\psi$ contained in the latter 
are interpreted in terms of the obvious computation
that arises from setting
 \begin{equation}
 \label{toattainthatdata}
 \psi|_{\Sigma(\tau_0)} =\uppsi, \qquad n\psi|_{\Sigma(\tau_0)\cap \{r\le R\}} = \uppsi' 
 \end{equation}
 and imposing the equation~\eqref{theequationforlargea} formally. 
Note of course that the expression~\eqref{tobedfined} may well be infinite.

\begin{proposition}[Local well posedness (cf.~\cite{DHRT22}) and Cauchy stability]
\label{localexistencelargea}
Let $(\mathcal{M},g_{a,M})$ be the sub-extremal Kerr manifold of 
Section~\ref{subextremalkerrsec} for $|a|<M$, and 
consider equation~\eqref{theequationforlargea} satisfying the assumptions
of Section~\ref{generalclassconsideredrecalled}.
Fix either $p=0$ or 
$\delta \le p\le 2-\delta$.

There exists a positive integer $k_{\rm loc}\ge 4$ such that the following holds. Let $k\ge k_{\rm loc}$.
Then there exists a positive real constant $C>0$ sufficiently large,
a positive real parameter $\varepsilon_{\rm loc}>0$ sufficiently small, and a decreasing positive function $\tau_{\rm exist}:(0,\varepsilon_{\rm loc})\to \mathbb R$
such that for all smooth initial data $(\uppsi, \uppsi')$ on $\Sigma(\tau_0)$ such that
\[
 \Epk[\uppsi,\uppsi'] \leq \varepsilon_0
 \leq \varepsilon_{\rm loc},
\]
there exists a smooth solution~$\psi$ of~\eqref{theequationforlargea}  
in $\mathcal{R}(\tau_0,\tau_1)$ for $\tau_1:=\tau_0+\tau_{\rm exist}$ attaining
the data~\eqref{toattainthatdata} and satisfying
\begin{equation}
\label{andthisboundedness}
 \Xpk(\tau_0,\tau_1) [\psi] \leq C \varepsilon_0.
\end{equation}
Moreover, $\Sigma(\tau_0)$ is a past Cauchy hypersurface for $\mathcal{R}(\tau_0,\tau_1)$
and
for all $\tau_0\le \tau\le \tau_1$,
 any other smooth $\widetilde\psi$ defined on $\mathcal{R}(\tau_0,\tau)$ satisfying~\eqref{theequationforlargea} in
 $\mathcal{R}(\tau_0,\tau)$ and attaining the initial data, i.e.~satisfying~\eqref{toattainthatdata} (with $\widetilde{\psi}$
replacing $\psi$), coincides with the restriction of $\psi$, i.e.~$\widetilde\psi=\psi|_{\mathcal{R}(\tau_0,\tau)}$.

We have moreover that for $\psi$ above, the finiteness~\eqref{supposeitholds} holds
with $\mathfrak{D}^{\bf k}\psi$ replacing
$\psi$ for all $|{\bf k}|\le k$,
 and thus the
 limiting quantity at $\mathcal{I}^+$  on the left hand side of~\eqref{throwitaway}
is well defined 
(again with $\mathfrak{D}^{\bf k}\psi$ replacing
$\psi$) and the nonnegativity~\eqref{throwitaway} holds.
The energy identity corresponding to such $\Jp_{g_{a,M}}^{\rm total}[\mathfrak{D}^{\bf k}\psi]$
may thus be integrated globally in $\mathcal{R}(\tau_0,\tau_1)$ with this limiting quantity as additional boundary term.

We have in addition the following propagation of higher order regularity and/or higher weighted estimates. 
Given $p$, $k\ge k_{\rm loc}$, $\varepsilon_0\le \varepsilon_{\rm loc}$, $\tau_0$, $\tau_1$  as above, 
$2-\delta \ge p' \ge  p$ if $p\ge \delta$  and either $2-\delta\ge p'\ge \delta$ or $p'=0$ if $p=0$, and  $k'\ge k$, then there exists a constant $C(k',\tau_1-\tau_0)$ such that
for all $\psi$ as above, 
\[
\Xpprimekprime(\tau_0,\tau_1)[\psi]
 \leq C(k',\tau_1-\tau_0) \, \Epprimekprime(\tau_0)[\psi].
\]

We have the following weak global uniqueness statement for
the case $\tau>\tau_1$.  Suppose $\tau>\tau_1$ and we are given a smooth solution
$\widetilde\psi$  on $\mathcal{R}(\tau_0,\tau)$ with the same data such that
we assume moreover  $\Xpk(\tau_0,\tau)[\widetilde\psi] \leq C \varepsilon_0$. Then
$\widetilde\psi$ is a smooth extension of $\psi$ and any other solution 
$\widetilde{\widetilde\psi}$ on
 $\mathcal{R}(\tau_0,\tau)$ with the same initial data 
 moreover satisfies $\widetilde{\widetilde\psi}=\widetilde\psi$.

Finally, we have the following Cauchy stability statement:
Suppose
$\tau_0<\tau<\infty$ and $\widetilde\psi$ is a solution of~\eqref{theequationforlargea}
on $\mathcal{R}(\tau_0,\tau)$ with initial data 
$(\widetilde\uppsi,\widetilde\uppsi')$ satisfying 
\[
 \Epk[\widetilde\uppsi,\widetilde\uppsi'] \leq \varepsilon_0/2, \qquad 
\Xpk(\tau_0,\tau)[\widetilde\psi] \leq C \varepsilon_0/2, 
\]
\[
 \Epkplusone[\widetilde\uppsi,\widetilde\uppsi'] \leq D<\infty,
 \qquad 
  \Epkplusone[\widetilde{\widetilde\uppsi},\widetilde{\widetilde\uppsi}'] \leq D<\infty,
\]
then  there exists an $\varepsilon_{\rm Cauchy}(D) >0$
such that if 
$ \Epk[\widetilde\uppsi-\widetilde{\widetilde\uppsi},\widetilde\uppsi'-\widetilde{\widetilde\uppsi}'] \leq \varepsilon_{\rm Cauchy}$, 
then $\widetilde{\widetilde\psi}$ exists on $\mathcal{R}(\tau_0,\tau)$ and satisfies
 \[
  \Epk[\widetilde{\widetilde\uppsi},\widetilde{\widetilde{\uppsi}}'] \leq \varepsilon_0,\qquad 
\Xpk(\tau_0,\tau)[\widetilde{\widetilde\psi}] \leq C \varepsilon_0.
\]
\end{proposition}

The above Cauchy stability statement will be useful  for the continuity argument applied 
in fixed spacetime
slabs, which for technical convenience will be formulated as continuity \emph{with respect to
data} and not with respect to time.

\subsection{Basic assumptions on $\psi$, the auxiliary $\widehat\psi$ and comparison for top order
energies in a fixed $\mathcal{R}(\tau_0,\tau_1)$} 
\label{additionalenergysec}

In the rest of this section, we fix $\tau_0+2\le  \tau_1$ and consider
a solution~$\psi$ of~\eqref{theequationforlargea}  on $\mathcal{R}(\tau_0,\tau_1)$
satisfying the assumptions of Section~\ref{generalclassconsideredrecalled}
with $a\ne 0$. (We will discuss the $a=0$ case  and the issue
of uniform dependence of our estimates on $a$ in Remark~\ref{aequalszeroremark}.)

We first introduce our basic smallness assumption in Section~\ref{basicassumptionofsmall} (necessary already  for the stability of energy current coercivity to be shown 
in Section~\ref{covariantones}).
We then define in Section~\ref{auxsec} an auxiliary $\widehat\psi$ extending $\psi$ in the future by solving the linear homogeneous wave equation.  This extension will be useful technically. 
Finally, we will define our top order energies in Section~\ref{toporderergiescomp} 
and give some elementary
comparison properties.

\subsubsection{Basic smallness assumption}
\label{basicassumptionofsmall}

We will need to assume an additional smallness assumption on $\psi$. 
For many of the results of this section it would be sufficient to assume that
\[
|\psi|_{C^1(\mathcal{R}(\tau_0,\tau_1)\cap \{r\le R/2\} )} \leq \varepsilon,
\]
which, in view of our assumptions on $g(\psi,x)$ from Section~\ref{generalclassconsideredrecalled}, would imply
\begin{equation}
\label{isaconseq}
|g-g_{a,M}|_{C^1(\mathcal{R}(\tau_0,\tau_1)\cap \{r\le R/2\} )} \lesssim \varepsilon.
\end{equation}
Importantly, for $\varepsilon>0$ sufficiently small, this already implies that
the spacelike character of $\Sigma(\tau)\cap\{r\le R/2\}$ and $\mathcal{S}(\tau_0,\tau_1)$ is retained with respect
to $g$,
and that all induced volume forms with respect to both metrics are comparable.
(Recall that $g(\psi, x)= g_{a,M}$ for $r\ge R/2$.)

For convenience, let us make already the stronger assumption 
\begin{equation}
\label{newbasicbootstrap}
\Xzerolesslessk(\tau_0,\tau_1)[\psi] \leq\varepsilon
\end{equation}
for suitably small $\varepsilon$, 
and for a large $k$, large enough such that $\lesslessk$ is itself sufficiently large
such that~\eqref{isaconseq} is a consequence
of~\eqref{newbasicbootstrap} and the Sobolev estimate of Proposition~\ref{largeasobolevforfunctions}.

\subsubsection{The auxiliary $\widehat\psi$ and the definition of $\widehat\psi_{n,{\bf k}}$}
\label{auxsec}

Given $\psi$ as above, 
we define $\widehat\psi$ as follows.
First, let us define
\[
\hat{g}:= g(\chi^2_{\tau_1}\psi, x), \qquad
\hat{F}:=N(\partial \chi^2_{\tau_1} \psi,  \chi^2_{\tau_1} \psi, x) 
\]
where $\chi_{\tau_1}: =1$ for $\tau\le \tau_1-1$ and $\chi_{\tau_1}: =0$ for $\tau\ge \tau_1$,
and so that $\chi_{\tau_1}(\tau-\tau_1)= \chi_{\tau_1'}(\tau-\tau_1')$.
Now let $\widehat\psi$ be the unique solution in $\mathcal{R}(\tau_0,\infty)$ 
of the initial value problem for
the linear equation
\[
\Box_{\hat{g}}\widehat \psi= \hat{F},
\]
with initial data on $\Sigma(\tau_0)$ coincinding with the data of $\psi$.
Note that in the region where $\chi_{\tau_1}=1$, we have $\widehat\psi=\psi$.
while in the region $\chi_{\tau_1}=0$ then $\widehat\psi$ solves the linear homogeneous
equation $\Box_{g_{a,M}}\widehat\psi=0$.

Finally, we fix 
\[
\widehat\tau_1:=\tau_1+L_{\rm long}[\psi],
\]
for an $L_{\rm long}$ (depending on $\psi$!)~to be determined later (immediately before Remark~\ref{Ldetermineafter}),  and
we define
$\widehat\psi_{n, {\bf k}}$  by~\eqref{thenpackets}, but replacing $\psi$ with $\widehat\psi$
and replacing $\chi_{\tau_0,\tau_1}$ with $\chi_{\tau_0,\widehat{\tau}_1}$, i.e.
\begin{equation}
\label{newpsinkdef}
\widehat\psi_{n,{\bf k}} : = P_n( \mathring{\mathfrak{D}}^{\bf k} \chi^2 _{\tau_0,\widehat{\tau}_1} \widehat\psi ).
\end{equation}

 Let us note that
 \begin{proposition}
 \label{whenoneextends}
With $\psi$ and $\widehat\psi$ as above, we have
 \[
  {}^{\chi}\Xpk(\tau_0,\widehat\tau_1)[\widehat\psi] 
 \lesssim
 {}^{\chi}\Xpk(\tau_0,\tau_1)[\psi] ,
 \qquad
 \sup_{\tau_0\le \tau\le \widehat\tau_1}\Epk(\tau) [\widehat\psi] 
 \lesssim \sup_{\tau_0\le \tau\le \tau_1}\Epk (\tau) [\psi] , 
 \]
 \[
   {}^{\chi}\Xpk(\tau_0,\widehat\tau_1)[\widehat\psi] +
  {}^{\natural}\Xk(\tau_0,\widehat\tau_1)[\widehat\psi] 
 \lesssim
  {}^{\chi}\Xpk(\tau_0,\tau_1)[\psi]+
 {}^{\natural}\Xk(\tau_0,\tau_1)[\psi] ,\qquad
{}^{\chi\natural}_{\scalebox{.6}{\mbox{\tiny{\boxed{+1}}}}}\Xpk (\tau_0,\widehat\tau_1)[\widehat\psi] 
 \lesssim
{}^{\chi\natural}_{\scalebox{.6}{\mbox{\tiny{\boxed{+1}}}}}\Xpk(\tau_0,\tau_1)[\psi] ,
 \]
 \[
\Xpkminusone(\tau_0,\widehat\tau_1)[\widehat\psi] 
 \lesssim
\Xpk(\tau_0,\tau_1)[\psi] ,\qquad 
 \Xpkminusone{}^*(\tau_0,\widehat\tau_1)[\widehat\psi] 
 \lesssim \Xpk{}^*(\tau_0,\tau_1)[\psi],
 \]
where the implicit constants are in particular \underline{independent} of the choice of $L_{\rm long}$.
 \end{proposition}
 \begin{proof}
 This follows easily from Theorem~\ref{refinedblackboxoneforkerr}.
 \end{proof}

 In view of the above, we shall primarily use the quantities relating to $\psi$ on the right
 hand side of estimates.
 In particular, ${}^{\chi}\Xpk(\tau_0,\tau_1)$, etc., will always denote
 ${}^{\chi}\Xpk(\tau_0,\tau_1)[\psi]$, unless indicated otherwise.

\begin{remark}
As discussed in Section~\ref{toporderenerintro},
the reason for introducing the future extension $\widehat\psi$ is that our top order estimates will require going to a late time in the future in order to estimate an earlier time.
See already Lemma~\ref{lemmaformainestimate}. Note that by~\cite{partiii}, it follows
that $\widehat\psi$
decays to zero at fixed $r$ as $\tau\to \infty$.
In principle, one could work directly with $\widehat\psi$ cut off only near $\tau_0$,
 but note that this decay
would generically only be at some fixed inverse power law rate.
The additional cutoff at $\widehat{\tau}_1$ is simply to ensure that $\chi^2 _{\tau_0,\widehat{\tau}_1}\widehat\psi$ lies in the Schwartz class.
\end{remark}

\begin{remark}
\label{independenceofLlong}
We emphasise that $L_{\rm long}=L_{\rm long}[\psi]$ depends on the solution $\psi$ and is thus 
\underline{not} to be considered a parameter in the
sense of Section~\ref{noteonconstants}. 
Thus, in particular, the implicit constants in the inequalities in the above 
proposition and all inequalities to follow will \underline{always} be independent of  $L_{\rm long}$ (and $\widehat{\tau}_1=\tau_1+L_{\rm long}$), even though the precise values
of~\eqref{newpsinkdef} of course depend on  $L_{\rm long}$.
\end{remark}

 \subsubsection{Top order energies and a comparison theorem}
 \label{toporderergiescomp}

In addition to the energies of~\cite{DHRT22} reviewed in Section~\ref{reviewedenergynotations}, 
we will need
an additional energy to replace the $\mathfrak{E}$ energies of~\cite{DHRT22}.
These energies have been discussed already informally in Section~\ref{toporderenerintro}.

 Recall from Section~\ref{generalisedcoerc}
 that the boundary terms of $\Jp^{{\rm total}, n}$ on $\Sigma(\tau)$  enjoy coercivity properties
 in the region $\widetilde{\mathcal{D}}_n \cup \{r \ge R_{\rm freq} \}$
 whereas the bulk term $\Kp^{{\rm total},n}$ enjoy coercivity properties in 
 $ \{r \le R_{\rm freq} \} \setminus  \widetilde{\mathcal{D}}_n$.

(Recall always that for generalised superradiant ranges, we have 
$\widetilde{\mathcal{D}}_n =\emptyset$.)

In anticipation of the coercive terms in our energy identities, 
we define 
\begin{align}
\label{firstfancydef}
\Efancypk{}_{{\rm bdry}}(\tau)[\widehat\psi] &:=\sum_n \sum_{|{\bf k}|=k}
 \int_{\widetilde{\mathcal{D}}_n \cap \Sigma(\tau)} |L\widehat\psi_{n, {\bf k}}|^2
+|\underline{L}\widehat\psi_{n, {\bf k}}|^2  +|\nablaslash\widehat\psi_{n, {\bf k}}|^2  \\
\nonumber
&\qquad\qquad
+  \int_{\{r \ge R_{\rm freq} \} \cap \Sigma(\tau)} r^p| r^{-1}L(r{\mathfrak{D}}^{\bf k}\widehat\psi)|^2 +r^{\frac{p}2}|L\mathfrak{D}^{\bf k} \widehat\psi|^2 +|\slashed{\mathfrak{D}}^{\bf k}\nabla \widehat \psi|^2
+r^{\frac{p}2-2}|{\mathfrak{D}}^{\bf k} \widehat\psi|^2  +\iota_{r\le R}| \underline{L}{\mathfrak{D}}^{\bf k} \widehat\psi|^2 ,\\
\label{secondfancydef}
\Efancypk{}_{{\rm bulk}}(\tau)[\widehat\psi] &: =\sum_n \sum_{|{\bf k}|=k} \int_{(\{r_0\le r\le R_{\rm freq}\} \setminus \widetilde{\mathcal{D}}_n )\cap  \Sigma(\tau)} |L \widehat\psi_{n, {\bf k}}|^2 
+|\underline{L} \widehat\psi_{n, {\bf k}}|^2  +|\nablaslash \widehat\psi_{n, {\bf k}}|^2 \\
\nonumber
&\qquad +
\int_{\{r\ge R_{\rm freq}\} \cap\Sigma(\tau) }
r^{p-1} ( 
 |r^{-1}Lr{\mathfrak{D}}^{\bf k} \widehat\psi|^2 + | L{\mathfrak{D}}^{\bf k}\widehat\psi|+|\nablaslash{\mathfrak{D}}^{\bf k} \widehat\psi|^2) + r^{-1-\delta}  |\underline L {\mathfrak{D}}^{\bf k} \widehat\psi|^2
+ r^{p-3} |{\mathfrak{D}}^{\bf k} \widehat\psi|^2\\
\nonumber
&\qquad
+\epsilon_{\rm redder}  \int_{ \{r \le r_1\} \cap\Sigma(\tau) } |L
\mathfrak{D}^{\bf k} \widehat\psi|^2 
+|\underline{L}\mathfrak{D}^{\bf k} \widehat\psi|^2  +|\nablaslash \mathfrak{D}^{\bf k} \widehat\psi|^2,\\
\label{fancydef}
\Efancypk (\tau)[\widehat\psi] &:= \Efancypk{}_{\rm bdry}(\tau)[\widehat\psi]  +\Efancypk{}_{\rm bulk}(\tau) [\widehat\psi],
\end{align}
where $\epsilon_{\rm redder}>0$ is a sufficiently small parameter to be chosen later.
We note finally that we are borrowing here the $\mathfrak{E}$ notation from~\cite{DHRT22},
despite not tracking
exact constants (i.e.~these energies are not defined by the fluxes themselves), to emphasise that (as in~\cite{DHRT22}) the estimates for these energies will be derived
by applying an energy identity with
respect to the operator $\Box_{g(\psi,x)}$.
Note in particular, however, that as these do not represent exact fluxes,  the volume forms
above are here for convenience taken with respect to $g_{a,M}$, and not $g$,
unlike in~\cite{DHRT22}.

Let us note finally that the above definitions~\eqref{firstfancydef} and~\eqref{secondfancydef} are 
global: They depend in particular also on the choice of fixed $\mathcal{R}(\tau_0,\tau_1)$
one is considering, and moreover, they also depend of course
on $L_{\rm long}$ determining   $\widehat\tau_1$ and thus $\widehat\psi$.

Because the above energies are defined only with the commutation operators $\mathring{\mathfrak{D}}$ and $\mathfrak{D}$, they are not pointwise
coercive as they do not span derivatives in all directions. 
Nonetheless,  we have the following:
\begin{proposition}
\label{comparisonenergies}
Let $k\ge k_{\rm prelim}$ for $k_{\rm prelim}$ sufficiently large, and let
$\psi$ be a smooth solution of~\eqref{theequationforlargea}  on $\mathcal{R}(\tau_0,\tau_1)$
satisfying~\eqref{newbasicbootstrap} for sufficiently small $\varepsilon$. 
Then, defining $\widehat\psi_{n,{\bf k}}$
by~\eqref{newpsinkdef} for arbitrary $L_{\rm long}$, and defining $\Efancypk(\tau)$ by~\eqref{fancydef},
 we 
have for all $\tau_0\le \tau\le \widehat\tau_1-1$, the inequality
\begin{equation}
\label{fromthisfirstpoint}
\Epk(\tau)[\widehat\psi]
 \lesssim \Epk(\tau_0)[\psi] +  {}^\chi \Xpkminusone(\tau_0,\tau_1)[\psi] + 
\Efancypk (\tau)[\widehat\psi]
\end{equation}
and for all $\tau_0+1\le \tau \le \widehat\tau_1-1$, the inequality
\begin{equation}
\label{secondineqhere}
{\Ezerok}{}_{r\le r_1}[\widehat\psi]+
{}^\chi \Epminusonek'(\tau)[\widehat\psi] \lesssim \Efancypk{}_{{\rm bulk}}(\tau)[\widehat\psi] + {}^\chi \Epminusonekminusone'(\tau)[\widehat\psi].
\end{equation}
\end{proposition}
\begin{proof}
Inequality~\eqref{fromthisfirstpoint} follows from the triangle inequality, the Sobolev inequality of 
Proposition~\ref{largeasobolevforfunctions} and the elliptic estimates of 
Proposition~\ref{ellproplargea} (cf.~Corollaries 4.3.3 and 4.5.3 of~\cite{DHRT22})
and  Proposition~\ref{whenoneextends}.

Inequality~\eqref{secondineqhere} on the other hand follows simply from the triangle inequality.
\end{proof}

We also have a reverse inequality
\begin{proposition}
\label{reversecomparison}
Under the assumptions of Proposition~\ref{comparisonenergies}, for all 
$\tau_0\le \tau \le \widehat\tau_1-1$, we have
\begin{equation}
\label{reversehere}
\int_{\tau+\frac12}^{\tau+1}
\Efancypk (\tau')[\widehat\psi] d\tau'  \lesssim 
\Epk(\tau) [\widehat\psi]+   \Xpkminustwo(\tau_0,\tau_1)[\psi].
\end{equation}
\end{proposition}
\begin{proof}
We have
\[
\int_{\tau+\frac12}^{\tau+1}
\Efancypk (\tau')[\widehat\psi]
\lesssim 
\int_{\tau}^{\tau+\frac32}
  \Epk(\tau)[\widehat\psi] +   \Xpkminustwo(\tau_0,\widehat\tau_1)[\widehat\psi]
\]
from the properties of the pseudodifferential calculus (Proposition~\ref{mixedsobolevboundednessprop} and~\ref{smoothingprop}),
whence~\eqref{reversehere} then follows from Proposition~\ref{whenoneextends}
and Proposition~\ref{localexistencelargea}, provided 
$k_{\rm loc} \le \lesslessk$.
\end{proof}

We note that the $\lesssim$ in the Proposition~\ref{comparisonenergies} depends (unfavourably!) on
the yet-to-be-chosen parameter $\epsilon_{\rm  redder}$ introduced in~\eqref{secondfancydef}.
The following proposition will provide an $\epsilon_{\rm redder}$-independent
statement, which will be necessary for absorbing some top order $\epsilon_{\rm redder}$-error terms.

\begin{proposition}
Under the assumptions of Proposition~\ref{comparisonenergies}, we have
the additional statement
\begin{equation}
\label{additionalforredder}
\int_{\tau_1+2}^{\widehat\tau_1-1}  \Ezerok_{r_1\le r\le r_2}(\tau')[\widehat\psi] d\tau'
\lesssim
\inf_{\tau_1+1\le \tau' \le \tau_1+2}
\mathring{\Ezerok}(\tau')[\widehat\psi]
\end{equation}
where we recall  $\mathring{\Ezerok}(\tau)[\psi]$
is defined in~\eqref{othercommutatorsenergy}.
\end{proposition}
\begin{proof}
Note that $\widehat\psi$ satisfies the homogeneous wave equation $\Box_{g_{a,M}}\hat\psi=0$
in $\mathcal{R}(\tau_1+1,\widehat\tau_1-1)$.
The statement then  follows from the $k=0$ case of Theorem~\ref{blackboxoneforkerr} after applying
commutation by $\mathring{\mathfrak{D}}^{\bf k}$ and 
the elliptic estimate~\eqref{ellproplargea}.
\end{proof}

As the above proposition is proven without reference to the energies~\eqref{secondfancydef},
the constant above is in particular independent of $\epsilon_{\rm redder}$. 
As a corollary of the above and the previous, we obtain an $\epsilon_{\rm redder}$-independent
estimate (involving  $\Efancypk{}_{{\rm bulk}}$ on the right hand side) of the following form:

\begin{proposition}
\label{mautokavoumeabsorb}
Under the assumptions of Proposition~\ref{comparisonenergies},
we have
\[
\int_{\tau_0+1}^{\widehat\tau_1-1}  \Ezerok_{r_1\le r\le r_2}(\tau')[\widehat\psi] d\tau'
\lesssim \int_{\tau_0+2}^{\tau_1+2} \Efancypk{}_{{\rm bulk}}(\tau')[\widehat\psi] d\tau' 
+ {}^\chi \Xpkminusone(\tau_0,\tau_1)[\psi] 
+ \Epk(\tau_0)[\psi] + \inf_{\tau_1+1\le  \tau' \le \tau_1+2}  \Efancypk (\tau')[\widehat\psi]
\]
where the implicit constant  in $\lesssim$ is independent of the choice of
$\epsilon_{\rm redder}$ in~\eqref{secondfancydef}.
\end{proposition}

\begin{proof}
We note first that for $\tau_0+1\le \tau \le \widehat\tau_1-1$,
we have the following variant of~\eqref{fromthisfirstpoint}
\begin{equation}
\label{variantwithgeom}
\mathring{\Ezerok}(\tau')[\widehat\psi]
\lesssim  \Efancypk (\tau')[\widehat\psi] +  {}^\chi \Xpkminusone(\tau_0,\tau_1)[\psi]
\end{equation}
and
\begin{equation}
\label{variantwithgeomloc}
\mathring{\Ezerok}_{r_1\le r\le r_2} (\tau')[\widehat\psi]
\lesssim  \Efancypk_{\rm bulk} (\tau')[\widehat\psi] +  {}^\chi \Xpkminusone(\tau_0,\tau_1)[\psi]
\end{equation}
where the constant implicit in $\lesssim$ is independent of $\epsilon_{\rm redder}$.
Integrating in $\tau'$ we obtain
\[
\int_{\tau_0+2}^{\tau_1+2}  \Ezerok_{r_1\le r\le r_2}(\tau')[\widehat\psi] d\tau'
\lesssim \int_{\tau_0+2}^{\tau_1+2}   \Efancypk{}_{{\rm bulk}}(\tau')[\widehat\psi] d\tau' 
 + {}^\chi \Xpkminusone(\tau_0,\tau_1)[\psi] 
\]
where the constant implicit in $\lesssim$ is again 
 independent of $\epsilon_{\rm redder}$.
Thus 
the desired statement follows combining the above and~\eqref{additionalforredder}
applied to~\eqref{variantwithgeom}.
\end{proof}

\subsection{The covariant currents $J^{{\rm total},n}_g$, $K^{{\rm total},n}_g$ 
and stability of generalised coercivity properties}
\label{covariantones}

Let us now promote our currents of Section~\ref{generalisedcoerc} to covariant currents 
with respect to a general Lorentzian metric $g$ on $\mathcal{M}$, i.e.~let
us define
\begin{equation}
\label{torefertothecurrents}
J^{{\rm main},n}_{g}[\Psi],\qquad  K^{{\rm main}, n}_{g}[\Psi], \qquad H^{{\rm main}, n}_{g}[\Psi]
\end{equation}
by making our currents covariant in $g$, either directly in twisted form or equivalently by 
reexpressing the currents 
$J^{{\rm main},n}_{g_{a,M}}[\Psi]$,
$K^{{\rm main}, n}_{g_{a,M}}[\Psi]$, $H^{{\rm main}, n}_{g_{a,M}}[\Psi]$
with respect to the 
$V_n$, $w_n$, $q_n$, $\varpi_n$  parametrisation 
and replacing $g_{a,M}$ with $g$ in
 the formulas~\eqref{generalJdef}--\eqref{generalHdef}.

We similarly now define
\[
\Jp^{\rm total}_{g}[\Psi],\qquad  \Kp^{\rm total}_{g}[\Psi], \qquad \Hap^{{\rm total}}_{g}[\Psi]
\]
and
\[
\Jp^{\rm redder}_{g}[\Psi],\qquad  \Kp^{\rm redder}_{g}[\Psi], \qquad \Hap^{{\rm redder}}_{g}[\Psi]
\]
by promoting to covariant the currents defined in Section~\ref{absorbingetcsec}.

We have the following:
\begin{proposition}[Stability of coercivity properties with allowed errors]
\label{stabofcoercivproposition}
Let $g$ be a metric on $\mathcal{R}(\tau_0,\tau_1)$
such that $|g-g_{a,M}|_{C^1} \lesssim \varepsilon$ and $g=g_{a,M}$ for $r\ge R/2$, for $a\ne 0$,
and let~\eqref{torefertothecurrents}  be the currents defined above.

Then for $\varepsilon>0$ sufficiently small, 
we have that $\mathcal{S}(\tau_0,\tau_1)$ and $\Sigma(\tau)\cap \{r\le R\}$ for all $\tau_0\le \tau\le \tau_1$
are spacelike with respect to $g$, and
for a general function $\Psi$, the following
boundary and bulk coercivity statements hold in the region $r\ge R_{\rm freq}$:
\begin{align}
\label{boundcoer}
\Jp{}_{g}^{{\rm total}}[\Psi]  \cdot {\rm n}_{\Sigma(\tau)} \gtrsim &  r^p| r^{-1}L(r\Psi)|^2 +r^{\frac{p}2}|L\Psi|^2 +|\slashed\nabla \Psi|^2
+r^{\frac{p}2-2}|\Psi|^2  +\iota_{r\le R}| \underline{L}\Psi|^2, & \delta\le p\le 2-\delta, \\
\Jzero{}_{g}^{{\rm total}}[\Psi]  \cdot {\rm n}_{\Sigma(\tau)} \gtrsim & |L\Psi|^2 +|\slashed\nabla \Psi|^2
+r^{-2}|\Psi|^2 +\iota_{r\le R}| \underline{L} \Psi |^2, & p=0, \\
\label{bulkcoercivps}
 \Kp{}^{{\rm total}}_{g}[\Psi]
\gtrsim &  r^{p-1} ( 
 (r^{-1}|Lr \Psi |)^2 + | L \Psi |+|\nablaslash \Psi |^2) + r^{-1-\delta}  |\underline L \Psi |^2
+ r^{p-3} | \Psi |^2,   & \delta\le p\le 2-\delta, \\
\label{bulkcoercivpequalszero}
 \Kzero{}^{{\rm total}}_{g}[\Psi]
\gtrsim &  r^{-1-\delta} ( 
 | L\Psi|+|\nablaslash \Psi |^2 + |\underline L \Psi |^2)
+ r^{-3-\delta} | \Psi |^2, &p=0 .
\end{align}

On $\Sigma(\tau) \cap \widetilde{\mathcal{D}}_n$, we have, for a general function $\Psi $,
the boundary coercivity statement
\begin{equation}
\label{coercboundatDn}
J{}_{g}^{{\rm main}, n}[\Psi]  \cdot {\rm n}_{\Sigma(\tau)} \gtrsim 
|L \Psi | ^2 + |\underline L \Psi |^2 +|\nablaslash \Psi |^2  +| \Psi |^2.
\end{equation}

On $\mathcal{S}$ we have, for a general function $\Psi $,
 the  boundary coercivity statement
\begin{equation}
\label{newcoercbound}
J^{\rm redder}_g [\Psi] \cdot {\rm n}_{\mathcal{S}} \ge c(r_0) (|L \Psi |^2+|\underline{L} \Psi |^2
+|\nablaslash \Psi |^2)
\end{equation}
and on $\Sigma(\tau)$ the positivity statement
\begin{equation}
\label{wehaveinfactcoercclose}
J_{g}^{\rm redder}[\psi]\cdot {\rm n}_{\Sigma(\tau) }\ge 0 .
\end{equation}

For a general function $\Psi$, we moreover
have
 in the region $r_0\le r\le r_1+2\frac{r_2-r_1}3$ the bulk coercivity statement
 \begin{equation}
\label{newcoercbulk}
 K{}^{{\rm redder}}_{g}  [\Psi] \gtrsim |L\Psi|^2+|\underline{L}\Psi|^2+|\nablaslash\Psi|^2,
 \end{equation}
 and for a solution $\psi$ of~\eqref{theequationforlargea} in $\mathcal{R}(\tau_0,\tau_1)$
 with $g=g(\psi,x)$
under the assumption~\eqref{newbasicbootstrap}, the bulk coercivity:
\begin{align}
\nonumber
\sum_{|{\bf k}| = k}
\int_{\mathcal{R}(\tilde\tau_0,\tilde\tau_1)\cap \{r_0 \le r \le  r_1+2\frac{r_2-r_1}3 \} }  K{}^{{\rm redder}}_{g} 
[\alpha({\bf k}) \mathfrak{D}^{\bf k} \psi]
+{\rm Re} ( H{}^{\rm redder}_{g}[\psi] \cdot \overline{ \Box_{g_{a,M}}\alpha({\bf k}) \mathfrak{D}^{\bf k}\psi} )   \\
\label{newcoercbulkimproved}
  \ge
\int_{\mathcal{R}(\tilde\tau_0,\tilde\tau_1)\cap \{r_0 \le r \le r_1+2\frac{r_2-r_1}3 \} }
\frac14
 K{}^{\rm redder}_{g} 
[\alpha({\bf k}) \mathfrak{D}^{\bf k} \psi]
   -C\, {}^\chi\Xzerokminusone(\tilde\tau_0,\tilde\tau_1)[\psi],
\end{align}
where $\alpha({\bf k})$ are as in Proposition~\ref{improvcoercprop},
and a similar estimate holds for $\widehat\psi$ for $\tau_0+1\le \tilde\tau_0\le \tilde\tau_1-2\le
\widehat\tau_1-3$ and $g=g(\chi^2_{\tau_1}\psi,x)$.

For a solution 
 $\psi$ of~\eqref{theequationforlargea} in $\mathcal{R}(\tau_0,\tau_1)$
under the assumption~\eqref{newbasicbootstrap}, we have the bulk coercivity:
\begin{align}
\nonumber
\int_{\mathcal{R}(\tilde\tau_0,\tilde\tau_1)\cap \{r \ge  R/4\} } 
\sum_{|{\bf k}| = k}
  \Kp{}^{{\rm total}}_{g}  
[\alpha({\bf k}) \mathfrak{D}^{\bf k} \psi]
+{\rm Re}( \Hap{}^{{\rm total}}_{g}[\psi] \cdot \overline{ \Box_{g_{a,M}}\alpha({\bf k}) \mathfrak{D}^{\bf k}\psi})  \\
\label{improvedfarawaycoercivitymetricg}
\ge
\int_{\mathcal{R}(\tilde\tau_0,\tilde\tau_1)\cap \{r \ge R/4  \} } 
\frac14 \sum_{|{\bf k}| = k}
  \Kp{}^{{\rm total}}_{g}  
[\alpha({\bf k}) \mathfrak{D}^{\bf k} \psi]    -C\, {}^\chi\Xpkminusone(\tilde\tau_0,\tilde\tau_1),
\end{align}
and a similar estimate holds for $\widehat\psi$ for  $\tau_0+1\le \tilde\tau_0\le \tilde\tau_1-2\le
\widehat\tau_1-3$ and $g=g(\chi^2_{\tau_1}\psi,x)$ and $g= g(\chi^2_{\tau_1}\psi,x)$.

Finally, applied to the projections $\widehat\psi_{n,{\bf k}}$, 
 we have
for all  $\tau_0+1\le \tilde\tau_0 \le \tilde\tau_1-2 \le \widehat\tau_1-3$ the coercivity statements
\begin{eqnarray}
\nonumber
\int_{\mathcal{S}(\tilde\tau_0,\tilde\tau_1)} {J_{g}^{{\rm main},n}}[\widehat\psi_{n,{\bf k}}] \cdot {\rm n}_{\mathcal{S}}  
&\ge&
\int_{\mathcal{S}(\tilde\tau_0,\tilde\tau_1)} c(r_0) ( | L \widehat\psi_{n,{\bf k}} |^2 + | \underline{L} \widehat\psi_{n,{\bf k}} |^2 
+|\slashed\nabla \widehat\psi_{n,{\bf k}}|^2 ) \\
\nonumber
&&\qquad
-C\, \Ezerok_{\mathcal{S}} (\tilde\tau_0-\epsilon_{\rm cutoff},\tilde\tau_0+\epsilon_{\rm cutoff}) 
-C\, \Ezerok_{\mathcal{S}} (\tilde\tau_1-\epsilon_{\rm cutoff},\tilde\tau_1+\epsilon_{\rm cutoff})\\
\label{haserrortermstoo}
&&\qquad - C\,{}^{\chi\natural}_{\scalebox{.6}{\mbox{\tiny{\boxed{+1}}}}}\,\Xpkminusone(\tau_0,\tau_1)[\psi],
\end{eqnarray}
\begin{eqnarray}
\label{bulkcoerc}
\int_{\mathcal{R}(\tilde\tau_0,\tilde\tau_1) \cap \{r_0\le r\le R_{\rm fixed}\} } 
K^{{\rm main},n}_g[\widehat\psi_{n, {\bf k}}]
&\ge & c  \int_{\mathcal{R}(\tilde\tau_0,\tilde\tau_1) \cap \{r_0\le r\le R_{\rm fixed}\}} 
 \chi_ n \left( | L \widehat\psi_{n,{\bf k}} |^2 + | \underline{L} \widehat\psi_{n,{\bf k}} | ^2 
+|\slashed\nabla \widehat \psi_{n,{\bf k}}|^2  \right) \,\,\,\,\,\,\,\,\,\,\,\,\,  \\
\label{bulkcoerclintwo}
&&\qquad -
C \epsilon_{\rm cutoff}\, \sup_{\tau\in[\tau_0,\tau_1]} \Ezerok(\tau)[\psi] 
 -C\,{}^{\chi\natural}_{\scalebox{.6}{\mbox{\tiny{\boxed{+1}}}}}\,\Xpkminusone(\tau_0,\tau_1)[\psi] \\
&&\qquad
\label{bulkcoerclinthree}
- \varepsilon  C \int_{\tilde\tau_0}^{\tilde\tau_1}\Efancypk(\tau)[\widehat\psi] d\tau
\end{eqnarray}
where $k:=|{\bf k}|$, and where the convention of Section~\ref{noteonconstants} 
regarding
dependence of constants on $\epsilon_{\rm cutoff}$ has been adhered to.

Here the normals and volume forms are with respect to $g$.
Under the assumption~\eqref{newbasicbootstrap},
if $g=g(\chi^2_{\tau_1}\psi,x)$ as in Section~\ref{generalclassconsideredrecalled}, 
then $\varepsilon$ on the right hand side of 
line~\eqref{bulkcoerclinthree} may be replaced by $\sqrt{\Xzerolesslessk(\tau_0,\tau_1)}$.
\end{proposition} 

\begin{proof}
The purely pointwise statements all follow 
immediately from the covariant nature of the currents and the $C^1$ closeness
of $g$ and~$g_{a,M}$. 
To obtain~\eqref{newcoercbulkimproved}, we have also used the 
more general inhomogeneous version of~\eqref{giveitanamenow} and the
elliptic estimates
of Proposition~\ref{ellproplargea}  while
to obtain~\eqref{improvedfarawaycoercivitymetricg}, we have used  the more general inhomogeneous version
of~\eqref{improvedfarawaycoercivity}, the elliptic estimates
of Proposition~\ref{ellproplargea} 
and Assumption~\ref{largeanullcondassumption} (see~\cite{DHRT22} for details). For the inequalities
for $\widehat\psi_{n, {\bf k}}$, note that 
we have also used Proposition~\ref{whenoneextends}.
\end{proof}

\subsection{The top order energy identity and estimates}
\label{toporderidentity}

In this section we will apply our currents at top order to solutions of~\eqref{theequationforlargea}
with $g=g(\chi^2_{\tau_1}\psi, x)$ to obtain the following statement:
   
\begin{proposition}[Top order energy estimate]
\label{toporderestimateprop}
Let  $k$ be sufficiently large, the parameters $\epsilon_{\rm cutoff}$  of Section~\ref{Schwartzandcutoffs} and $\epsilon_{\rm redder}$ of~\eqref{secondfancydef} be sufficiently small,
 and $\varepsilon>0$ sufficiently small,
 let $\psi$ be a solution of the quasilinear 
equation~\eqref{theequationforlargea}
(satisfying the assumptions of Section~\ref{generalclassconsideredrecalled} with $a\ne 0$)
 in $\mathcal{R}(\tau_0,\tau_1)$
and assume the bound~\eqref{newbasicbootstrap}.

Then there exists an $L_{\rm long}[\psi]\ge 1$ (depending  in fact  
only on ${}^\chi \Xpk(\tau_0,\tau_1)[\psi]$) such that
defining $\widehat\tau_1:=\tau_1+L_{\rm long}$, and the associated
$\widehat\psi$, etc.,
the following is true.
For $p=0$ or $0<\delta\le p \le 2-\delta$,  it follows that
\begin{equation}
\label{topordergiatwra}
{}^\rho \Xpk(\tau_0,\tau_1) [\psi]
\lesssim
\Epk(\tau_0) [\psi] + \mathcal{Y}(\tau_0,\widehat\tau_1)[\widehat\psi],
\end{equation}
where $\mathcal{Y}(\tau_0,\widehat\tau_1)$ is defined in~\eqref{mathcalYdefinition}.
\end{proposition}
\begin{remark}
Again we emphasise that the implicit constants are in particular independent of $L_{\rm long}[\psi]$.
\end{remark}

\begin{proof} 
Let $\widehat\psi_{n, {\bf k}}$ be defined as in Section~\ref{auxsec} by~\eqref{newpsinkdef}
for $L_{\rm long}$ to be determined.
The expression $g$ below will denote $g=g(\chi^2_{\tau_1} \psi,x)$.
It follows that
\begin{equation}
\label{eachsatisfies}
\Box_{g(\chi^2_{\tau_1}\psi,x)}  \widehat\psi_{n,{\bf k}} = F_{n,{\bf k}} := F^{\rm cutoff}_{n,{\bf k}} 
+F^{\rm commu}_{n,{\bf k}}
+F^{\rm pseudo}_{n,{\bf k}}
+F^{\rm semi}_{n,{\bf k}}
\end{equation}
where $F^{\rm cutoff}_{n,{\bf k}}$ represent ``linear'' terms from the time cutoffs
\begin{equation}
\label{exprwithcutoffs}
F^{\rm cutoff}_{n,{\bf k}} :=P_n ( \mathring{\mathfrak{D}}^{\bf k}  (\Box_g \chi^2_{\tau_0,\widehat\tau_1} ) \widehat\psi  + 2 \mathring{\mathfrak{D}}^{\bf k} (\nabla_g\chi^2_{\tau_0,\widehat\tau_1}\cdot \nabla_{g}\widehat\psi)),
\end{equation}
$F^{\rm commu}_{n,{\bf k}}$ represents  the
nonlinear commutation terms
\[
F^{\rm commu}_{n,{\bf k}}:= P_n [\mathring{\mathfrak{D}}^{\bf k}, \Box_{g(\chi^2_{\tau_1}\psi_,x)}-\Box_{g_{a,M}}] 
 \chi_{\tau_0,\widehat\tau_1}^2 \psi,
\]
$F^{\rm pseudo}_{n,{\bf k}}$  represent nonlinear pseudodifferential commutation terms
\[
F^{\rm pseudo}_{n,{\bf k}}:=
[P_n ,\Box_{g(\chi^2_{\tau_1}\psi,x)}]\chi^2_{\tau_0,\widehat\tau_1} \mathring{\mathfrak{D}}^{\bf k}\widehat\psi=
 [P_n ,\Box_{g(\chi^2_{\tau_1}\psi,x)} -\Box_{g_{a,M}}] \mathring{\mathfrak{D}}^{\bf k} \chi^2_{\tau_0,\widehat\tau_1}\widehat\psi
\]
and $F^{\rm semi}_{n,{\bf k}}$ represents nonlinear terms arising from the semilinear expression
\[
F^{\rm semi}_{n,{\bf k}}:= 
P_n( \mathring{\mathfrak{D}}^{\bf k} \chi^2_{\tau_0,\widehat\tau_1} (N(\partial \chi^2_{\tau_1} \psi, \chi^2_{\tau_1} \psi,x))).
\]

On the other hand, we have
\begin{equation}
\label{otherformulahere}
\Box_{g(\chi^2_{\tau_1}\psi,x)} \mathfrak{D}^{\bf k} \widehat\psi
= F_{{\bf k}} := 
F^{\rm lincom}_{{\bf k}}
+F^{\rm nonlincom}_{{\bf k}}
+F^{\rm semi}_{{\bf k}}
\end{equation}
where
$F^{\rm lincom}_{{\bf k}}$ represents the ``linear'' commutation terms 
\begin{equation}
\label{thelincomterms}
F^{\rm lincom}_{\bf k}:=  [{\mathfrak{D}}^{\bf k}, \Box_{g_{a,M}}]  \widehat\psi,
\end{equation}
$F^{\rm nonlincom}_{\bf k}$ represents the nonlinear commutation terms
\[
F^{\rm nonlincom}_{\bf k}:=  [\mathfrak{D}^{\bf k}, \Box_{g(\chi^2_{\tau_1}\psi_,x)}-\Box_{g_{a,M}}] 
\psi,
\]
and $F^{\rm semi}_{{\bf k}}$ represents the semilinear terms
\[
F^{\rm semi}_{\bf k}:=
\mathfrak{D}^{\bf k} N(\partial \chi^2_{\tau_1} \psi,  \chi^2_{\tau_1} \psi, x) .
\]

We apply now the energy identity associated to the current
$\Jp^{{\rm main},n}_g[\widehat\psi_{n,{\bf k}}]$ of~\eqref{torefertothecurrents}
where $g=g(\chi^2_{\tau_1}\psi,x)$:
\begin{equation}
\label{nonlineardividentity}
\nabla^\mu_g  \Jp_g^{{\rm main},n}{}_
\mu[\widehat\psi_{n,{\bf k}}] =
\Kp_g^{{\rm main},n} [  \widehat\psi_{n,{\bf k}}]  +{\rm Re}\left( \Hap_g^{{\rm main},n}[\widehat\psi_{n,{\bf k}}]  \overline{F_{n,{\bf k}}} \right).
\end{equation}
Dropping now the $g$ subscript,
integrating~\eqref{nonlineardividentity} in an arbitrary 
$\mathcal{R}(\tilde\tau_0,\tilde\tau_1)\cap \{   r_0\le r\le  R_{\rm freq}\}$
for $\tau_0+3\le \tilde\tau_0+2 \le\tilde\tau_1\le \widehat\tau_1-1$ we obtain 
\begin{align*}
\int_{\Sigma(\tilde \tau_1)\cap\{r\le R_{\rm freq}\}} \Jp^{{\rm main},n }[\widehat\psi_{n, {\bf k}}]\cdot {\rm n} +\int_{\mathcal{S}(\tilde\tau_0,\tilde \tau_1)} \Jp^{{\rm main},n }[\widehat\psi_{n, {\bf k}}] \cdot {\rm n}
 +\int_{\mathcal{R}(\tilde\tau_0,\tilde\tau_1)\cap \{r\le R_{\rm freq}\} } \Kp^{{\rm main},n}[\widehat\psi_{n, {\bf k}}]\\
 =  \int_{\Sigma(\tilde\tau_0)\cap \{r\le R_{\rm freq}\} } \Jp^{{\rm main},n}[\widehat\psi_{n, {\bf k}}] \cdot {\rm n}- \int_{\mathcal{R}(\tilde\tau_0,\tilde\tau_1) \cap \{r\le R_{\rm freq} \} } {\rm Re} \left(\Hap^{{\rm main},n}[\widehat\psi_{n, {\bf k}}] \overline{F_{n,{\bf k}}}\right)\\
 + \int_{\mathcal{R}(\tilde\tau_0,\tilde\tau_1)\cap \{r =R_{\rm freq} \} }
 J^{ {\rm main},n }[\widehat\psi_{n, {\bf k}}]  \cdot {\rm n}_{\{r=R_{\rm freq}\}},
\end{align*}
where the currents  and volume forms are all with respect to $g=g(\chi^2_{\tau_1}\psi,x)$.
We now sum the above identity over $n$ and $|{\bf k}|=k$ with coefficients $\alpha({\bf k})$
from Propositions~\ref{improvcoercprop} and~\ref{impbulcoerpropfar}
(here we may interpret these coefficients as applying to the $\mathring{\mathfrak{D}}$ operators
as  $\alpha(k_1,k_2):= \alpha(k_1,k_2,0,\ldots)$)
and add it to the following identity
summed over ${\bf k}$ (now corresponding to the $\mathfrak{D}^{\bf k}$ operators) with $|{\bf k}|=k$, with coefficients $\alpha({\bf k})$:
\begin{align}
\nonumber
\int_{\Sigma(\tilde \tau_1)\cap\{r\ge R_{\rm freq}\}} \Jp^{ {\rm total} }[ \mathfrak{D}^{\bf k} \widehat\psi] 
 \cdot {\rm n}+\int_{\mathcal{I}^+(\tilde\tau_0,\tilde \tau_1)} \Jp^{ {\rm total} }[ \mathfrak{D}^{\bf k} \widehat \psi]  \cdot {\rm n}
 +\int_{\mathcal{R}(\tilde\tau_0,\tilde\tau_1)\cap \{r\ge R_{\rm freq}\} } \Kp^{ {\rm total} }[ \mathfrak{D}^{\bf k} \widehat\psi]\\
 \nonumber
+  \int_{\mathcal{R}(\tilde\tau_0,\tilde\tau_1) \cap \{r\le R_{\rm freq} \} } {\rm Re} \left( \Hap^{ {\rm total}}[ \mathfrak{D}^{\bf k} \widehat \psi] \overline{F^{\rm lincom}_{\bf k}}  \right)\\
 \label{tobeaddedto}
 =  \int_{\Sigma(\tilde\tau_0)\cap \{r\ge R_{\rm freq}\} } \Jp^{ {\rm total} }[ \mathfrak{D}^{\bf k}
 \widehat \psi]  \cdot {\rm n}- \int_{\mathcal{R}(\tilde\tau_0,\tilde\tau_1) \cap \{r\le R_{\rm freq} \} } {\rm Re} \left( \Hap^{ {\rm total}}[ \mathfrak{D}^{\bf k} \widehat \psi] (\overline{F^{\rm nonlincom}_{\bf k}}
 +\overline{F^{\rm semi}_{\bf k}} ) \right)\\
 \nonumber
 - \int_{\mathcal{R}(\tilde\tau_0,\tilde\tau_1)\cap \{r =R_{\rm freq} \} }
 \Jp^{ {\rm total} }[\mathfrak{D}^{\bf k} \widehat \psi]  \cdot {\rm n}_{\{r=R_{\rm freq}\}}.
\end{align}
Recalling the parameter $\epsilon_{\rm redder}$ from~\eqref{secondfancydef}, we finally
add $\epsilon_{\rm redder}>0$ times the identity, summed with $\alpha({\bf k})$ coefficients:
\begin{align}
\nonumber
\int_{\Sigma(\tilde \tau_1)\cap\{r\le r_2\}} J^{ {\rm redder} }[ \mathfrak{D}^{\bf k} \widehat\psi] 
 \cdot {\rm n}
 +\int_{\mathcal{S}(\tilde\tau_0,\tilde \tau_1)} J^{ {\rm redder} }[ \mathfrak{D}^{\bf k} \widehat \psi] 
  \cdot {\rm n}
 +\int_{\mathcal{R}(\tilde\tau_0,\tilde\tau_1)\cap \{r\le r_2\} } K^{ {\rm redder} }[ \mathfrak{D}^{\bf k} \widehat\psi]\\
 \nonumber
  + \int_{\mathcal{R}(\tilde\tau_0,\tilde\tau_1) \cap \{r\le r_2 \} } {\rm Re} \left(H^{ {\rm redder}}[ \mathfrak{D}^{\bf k} \widehat \psi] \overline{F^{\rm lincom}_{\bf k}}\right)
 \\
 \label{tobeaddedtoredder}
 =  \int_{\Sigma(\tilde\tau_0)\cap \{r\le r_2\} } J^{ {\rm redder} }[ \mathfrak{D}^{\bf k}
 \widehat \psi]
 \cdot {\rm n} - \int_{\mathcal{R}(\tilde\tau_0,\tilde\tau_1) \cap \{r\le r_2 \} } {\rm Re} \left(H^{ {\rm redder}}[ \mathfrak{D}^{\bf k} \widehat \psi] (\overline{F^{\rm nonlincom}_{\bf k}}+\overline{F^{\rm semi}_{\bf k}}) \right).
\end{align}
(Note that~\eqref{tobeaddedtoredder} does not contain any other boundary terms in view
of the support of the current $J^{\rm redder}$.)

The resulting summed identity yields the following inequality:
\begin{align}
\label{mainestimate}
\Efancypk{}_{{\rm bdry}}(\tilde\tau_1)&+ \int_{\tilde\tau_0}^{\tilde\tau_1}\Efancypk{}_{\rm bulk} (\tau') 
  \lesssim \Efancypk{}_{\rm bdry}(\tilde\tau_0)
  +\Efancypk{}_{\rm bulk}(\tilde\tau_0) + \Efancypk{}_{\rm bulk}(\tilde\tau_1) \\
   \label{reddertermhere}
 &\quad
 + \epsilon_{\rm redder}\int_{\tilde\tau_0}^{\tilde\tau_1}  \Ezerok_{r_1\le r\le r_2}(\tau')[\widehat\psi] d\tau'
 \\
    &\quad
  \label{fromcoerc0}
+ \epsilon_{\rm cutoff}\, \sup_{\tau\in[\tau_0,\tau_1]} \Ezerok(\tau)
 \\
  &\quad
  \label{fromcoerc}
+  \,{}^{\chi\natural}_{\scalebox{.6}{\mbox{\tiny{\boxed{+1}}}}}\, \Xpkminusone(\tau_0,\tau_1) 
+\sqrt{\Xzerolesslessk(\tau_0,\tau_1)} \int_{\tilde\tau_0}^{\tilde\tau_1}\Efancypk(\tau)[\widehat\psi] d\tau
\\
&
\label{oldterms}
\quad
+ \sum_{|{\bf k}|=k} \int_{\mathcal{R}(\tilde\tau_0,\tilde\tau_1)\cap \{r\ge R_{\rm freq}\}} |\Hap^{{\rm total}}[\mathfrak{D}^{\bf k} \widehat\psi] 
F^{\rm nonlincom}_{{\bf k}}|
  + |\Hap^{{\rm total}}[\mathfrak{D}^{\bf k} \widehat\psi] F^{\rm semi}_{{\bf k}}|
\\
\label{oldtermsredder}
&\quad
+\epsilon_{\rm redder} \sum_{|{\bf k}|=k} \int_{\mathcal{R}(\tilde\tau_0,\tilde\tau_1)\cap \{r_0 \le r \le r_1 \}} |H^{{\rm redder}} [\mathfrak{D}^{\bf k} \widehat\psi]  F^{\rm nonlincom}_{{\bf k}}|
  + |H^{{\rm redder}}[\mathfrak{D}^{\bf k} \widehat\psi] F^{\rm semi}_{{\bf k}}|
 \\
\label{oneinhomogtermhere}
  &\quad 
  +
\sum_{n, |{\bf k}|=k} \left|  \int_{\mathcal{R}(\tilde\tau_0,\tilde\tau_1)\cap \{r\le R_{\rm freq}\}} 
 {\rm Re} \left(H^{{\rm main}, n}[\widehat\psi_{n, {\bf k}}] \overline{F^{\rm cutoff}_{n,{\bf k}}}\right)\right|\\
  \label{twoinhomogtermshere}
  &\quad  
+ \sum_{n, |{\bf k}|=k}   \left| \int_{\mathcal{R}(\tilde\tau_0,\tilde\tau_1)\cap \{r\le R_{\rm freq}\}}  {\rm Re} \left(H^{{\rm main}, n}[\widehat\psi_{n, {\bf k}}] \overline{F^{\rm commu}_{n,{\bf k}} }\right)\right|  \\
\label{thepseudoinhomgthershere}
&\quad
   +\sum_{n, |{\bf k}|=k} \left|  \int_{\mathcal{R}(\tilde\tau_0,\tilde\tau_1)\cap \{r\le R_{\rm freq}\}}  {\rm Re} \left(H^{{\rm main}, n}[\widehat\psi_{n, {\bf k}}]\overline{F^{\rm pseudo}_{n,{\bf k}}}\right)   \right|  \\
   \label{andonemoreinhomogtermhere}
     &\quad  
+  \sum_{n, |{\bf k}|=k}  \left|  \int_{\mathcal{R}(\tilde\tau_0,\tilde\tau_1)\cap \{r\le R_{\rm freq}\}}  {\rm Re} \left(H^{{\rm main}, n}[\widehat\psi_{n, {\bf k}}] \overline{F^{\rm semi}_{n,{\bf k}}}\right)\right|
  \\
  \label{lowerorderboundarterms}
  &\quad  
  +\mathring{\Ezerok}_{\mathcal{S}} (\tilde\tau_0-\epsilon_{\rm cutoff},\tilde\tau_0+\epsilon_{\rm cutoff}) [
  \widehat\psi]
+ \mathring{\Ezerok}_{\mathcal{S}} (\tilde\tau_1-\epsilon_{\rm cutoff},\tilde\tau_1+\epsilon_{\rm cutoff}) [\widehat\psi]
 \\
  \label{cancellingfluxes}
   &\quad 
   +\left|  \int_{\mathcal{R}(\tilde\tau_0,\tilde\tau_1)\cap \{r =R_{\rm fixed} \} }
\left(\ J^{{\rm total}} \left[\sum_{n, |{\bf k}|=k} \widehat\psi_{n, {\bf k}}-\mathfrak{D}^{\bf k} \widehat\psi\right]\right)  \cdot {\rm n}_{\{r=R_{\rm freq}\}}\right|.
\end{align}

The terms on the left hand side arise from the coercivity properties
of Proposition~\ref{stabofcoercivproposition}, the definitions of the energies of Section~\ref{additionalenergysec}, and the closeness~\eqref{isaconseq} which follows from~\eqref{newbasicbootstrap}.
In particular, we have used the improved coercivity~\eqref{newcoercbulkimproved} 
and~\eqref{improvedfarawaycoercivitymetricg}
to
already absorb the linear commutation terms~\eqref{thelincomterms} in the bulk.
 (Note that
 applying these statements generates error terms on the right hand side discussed in connection to
 lines~\eqref{reddertermhere}--\eqref{fromcoerc} and~\eqref{lowerorderboundarterms} below.)

The terms on the right hand side arise as follows:

The terms on line~\eqref{mainestimate} arise from bounding the boundary terms in the region
not coercively controlled by energy.

The term on line~\eqref{reddertermhere} arises from the failure of coercivity of $K^{\rm redder}$ in
the support of the cutoff $\nabla \hat{\hat\zeta}$. The convention of Section~\ref{noteonconstants}
applies here to the implicit constant multiplying $\epsilon_{\rm redder}$.

The terms on lines~\eqref{fromcoerc0}--\eqref{fromcoerc}
arise from the error terms~\eqref{bulkcoerclintwo}--\eqref{bulkcoerclinthree} 
in the coercivity estimate of $K^{{\rm main},n}$. Note that we have used here also Proposition~\ref{whenoneextends}.
Again, the convention of Section~\ref{noteonconstants}
applies here to the implicit constant multiplying $\epsilon_{\rm cutoff}$.

 The terms on lines~\eqref{oldterms} and~\eqref{oldtermsredder}
arise from the inhomogeneous term in~\eqref{tobeaddedto} and~\eqref{tobeaddedtoredder}
while 
the terms on lines~\eqref{oneinhomogtermhere}--\eqref{andonemoreinhomogtermhere} 
arise from the four inhomogeneous terms
in~\eqref{eachsatisfies}.

Now we may replace the terms on 
lines~\eqref{oneinhomogtermhere}--\eqref{andonemoreinhomogtermhere} 
with the sum
\begin{align}
\label{oneinhomogtermheresmooth}
  &\quad 
\sum_{n, |{\bf k}|=k} \left|  \int_{\mathcal{M}\cap \{r\le R_{\rm freq}\}} 
\tilde\chi_{\tilde\tau_0,\tilde\tau_1}^4
 {\rm Re} \left(H^{{\rm main}, n}[\widehat\psi_{n, {\bf k}}] \overline{F^{\rm cutoff}_{n,{\bf k}}}\right)\right|\\
  \label{twoinhomogtermsheresmooth}
  &\quad  
+ \sum_{n, |{\bf k}|=k}   \left| \int_{\mathcal{M}\cap \{r\le R_{\rm freq}\}}  
\tilde\chi_{\tilde\tau_0,\tilde\tau_1}^4
{\rm Re} \left(H^{{\rm main}, n}[\widehat\psi_{n, {\bf k}}] \overline{F^{\rm commu}_{n,{\bf k}} }\right)\right|  \\
\label{thepseudoinhomgthersheresmooth}
&\quad
   +\sum_{n, |{\bf k}|=k} \left|  \int_{\mathcal{M}\cap \{r\le R_{\rm freq}\}} 
  \tilde\chi_{\tilde\tau_0,\tilde\tau_1}^4
    {\rm Re} \left(H^{{\rm main}, n}[\widehat\psi_{n, {\bf k}}]\overline{F^{\rm pseudo}_{n,{\bf k}}}\right)   \right|  \\
   \label{andonemoreinhomogtermheresmooth}
     &\quad  
+  \sum_{n, |{\bf k}|=k}  \left|  \int_{\mathcal{M}\cap \{r\le R_{\rm freq}\}} 
\tilde\chi_{\tilde\tau_0,\tilde\tau_1}^4
 {\rm Re} \left(H^{{\rm main}, n}[\widehat\psi_{n, {\bf k}}] \overline{F^{\rm semi}_{n,{\bf k}}}\right)\right|
&
\end{align}
and an additional term on line~\eqref{fromcoerc0} and an additional
$  \,{}^{\chi\natural}_{\scalebox{.6}{\mbox{\tiny{\boxed{+1}}}}}\, \Xpkminusone(\tau_0,\tau_1)$ 
 term on line~\eqref{fromcoerc}, 
where again the implicit constants are independent of $\epsilon_{\rm cutoff}$
except for the constant multiplying the $  \,{}^{\chi\natural}_{\scalebox{.6}{\mbox{\tiny{\boxed{+1}}}}}\, \Xpkminusone(\tau_0,\tau_1)$ term.

To see this, say for~\eqref{oneinhomogtermheresmooth},
we note that~\eqref{oneinhomogtermheresmooth} differs from~\eqref{oneinhomogtermhere}
by a term which may be estimated by
\[
\sum_{n, |{\bf k}|=k}  \int_{\mathcal{M}\cap \{r\le R_{\rm freq}\}}  \left|
\tilde{\tilde\chi}^4
 {\rm Re} \left(H^{{\rm main}, n}[\widehat\psi_{n, {\bf k}}] \overline{F^{\rm cutoff}_{n,{\bf k}}}\right)\right|
\]
where $\tilde{\tilde\chi}$ is a smooth cutoff supported in an $\epsilon_{\rm cutoff}$-neighbourhood
of $\tau=\tilde\tau_0$ and $\tau=\tilde\tau_1$.
We may rewrite this schematically as a sum of terms 
\begin{equation}
\label{schemsumofterms}
 \int_{\mathcal{M}\cap \{r\le R_{\rm freq}\}}  \left|\tilde{\tilde\chi}^2
\partial  P_n \mathring{\mathfrak{D}}^{\bf k}  ( \chi_{\tau_0, \widehat\tau_1}^2  \widehat\psi)  \cdot 
 \tilde{\tilde\chi}^2  P_n
 \mathring{\mathfrak{D}}^{\bf k}  \left( \nabla_g \chi_{\tau_0,\widehat\tau_1} \cdot \nabla_g  \chi_{\tau_0,\widehat\tau_1} \widehat \psi\right) \right|
\end{equation}
for $|{\bf k}|=k$ plus lower order terms. We apply Cauchy--Schwarz. Consider say the term arising from the first factor of~\eqref{schemsumofterms}. 
Applying Proposition~\ref{commutatorestimate}, 
we may rewrite  
\[
\tilde{\tilde\chi}^2 \partial_{x^i} P_n \mathring{\mathfrak{D}}^{\bf k}=
\tilde{\tilde\chi}^2 P_n \mathring{\mathfrak{D}}^{\bf k}\partial_{x^i} 
   = \tilde{\tilde{\chi}} P_n \mathring{\mathfrak{D}}^{\bf k}   \tilde{\tilde\chi}\partial_{x^i} + 
\tilde{\tilde\chi} Q_1\partial_{x^i}=\tilde{\tilde \chi}P_n \mathring{\mathfrak{D}}^{\bf k}   \tilde{\tilde\chi} \partial_{x^i}+ Q_1\tilde{\tilde\chi}\partial_{x^i} +
Q_2 \partial_{x^i}
\]
where $Q_1$ is  of order $k-1$ and $Q_2$ is of order $k-2$.
We see by Proposition~\ref{mixedsobolevboundednessprop} 
that the resulting terms can be bounded by
\begin{equation}
\label{includesthelattertwo}
C\epsilon_{\rm cutoff}\, \sup_{\tau\in[\tau_0,\tau_1]} \Ezerok(\tau)[\psi]
+ C \sup_{\tau\in[\tau_0,\tau_1]} \Ezerokminusone(\tau)[\psi]
+ C\Xzerokminustwo(\tau_0,\tau_1)[\psi],
\end{equation}
where we have used also Proposition~\ref{whenoneextends},  
and where the first $C$ is independent of
$\epsilon_{\rm cutoff}$ whereas the other $C$'s depend unfavourably on $\epsilon_{\rm cutoff}$.
Thus this term  can indeed be estimated as claimed by the term on line~\eqref{fromcoerc0} and
the first term on line~\eqref{fromcoerc}, which includes in its definition the latter two terms of~\eqref{includesthelattertwo}.
The term resulting from the second factor in~\eqref{schemsumofterms} can be estimated similarly (it is in fact smoothing though we need not make use of this here), as can all other lower order terms.
The terms~\eqref{twoinhomogtermsheresmooth}--\eqref{andonemoreinhomogtermheresmooth} 
can be shown to replace~\eqref{twoinhomogtermshere}--\eqref{andonemoreinhomogtermhere} 
up to the
error terms described above by a similar computation. (The latter resulting error terms have in fact additional 
smallness as they are nonlinear, but it is not necessary to use this.)

The terms on line~\eqref{lowerorderboundarterms} appear with a constant independent
of $\epsilon_{\rm cutoff}$ and
arise to control the highest order error terms in the boundary coercivity~\eqref{haserrortermstoo}
(given that the first term of~\eqref{fromcoerc} has already been included on the right hand side). 
Regarding the terms on line~\eqref{lowerorderboundarterms},
note that 
 for $\epsilon_{\rm cutoff}$ sufficiently small, 
we may bound
\begin{align}
\nonumber
\mathring{\Ezerok}_{\mathcal{S}} &(\tilde\tau_0-\epsilon_{\rm cutoff},\tilde\tau_0+\epsilon_{\rm cutoff}) [\widehat\psi]
+ \mathring{\Ezerok}_{\mathcal{S}} (\tilde\tau_1-\epsilon_{\rm cutoff},\tilde\tau_1+\epsilon_{\rm cutoff})[\widehat\psi] \\
\label{fromthegeometryjustifiedhere}
&\lesssim \mathring{\Ezerok}{}_{r\le r_1}(\tilde\tau_0-\epsilon_{\rm cutoff})[\widehat\psi]	
+  \mathring{\Ezerok}{}_{r\le r_1}(\tilde\tau_1-\epsilon_{\rm cutoff})[\widehat\psi]
+ \epsilon_{\rm cutoff} \sqrt{ \varepsilon}\, \sup_{\tau\in[\tau_0,\tau_1]} \Ezerok(\tau)
\\
\label{providedweadd}
&\lesssim 
\Efancypk{}_{\rm bulk}(\tilde\tau_0-\epsilon_{\rm cutoff})[\widehat\psi]  +
\Efancypk{}_{\rm bulk}(\tilde\tau_1-\epsilon_{\rm cutoff})[\widehat\psi] 
+\epsilon_{\rm cutoff} \sqrt{ \varepsilon}\, \sup_{\tau\in[\tau_0,\tau_1]} \Ezerok(\tau)
+ {}^\chi \Xzerokminusone(\tau_0,\tau_1) .
\end{align}
Here, inequality~\eqref{fromthegeometryjustifiedhere} 
holds from the local existence Proposition~\ref{localexistencelargea}, 
localised to the future domain of  dependence of $\Sigma(\tau)\cap \{r\le r_1\}$ with respect to
the metric $g$, where we note that
for $\epsilon_{\rm cutoff}$ and $\varepsilon$ sufficiently small, this domain of dependence
includes $\mathcal{S}(\tau,\tau+2\epsilon_{\rm cutoff})$. 
(More specifically, the estimate follows by applying the divergence identity of 
$J^{{\rm redder}}_{\Sigma(\tau)}(\mathring{\mathfrak{D}}^{\bf k}\widehat\psi)$
and estimating the error bulk
term by $\epsilon_{\rm cutoff}\sqrt{\varepsilon}  \sup_{\tau\in[\tau_0,\tau_1]} \Ezerok(\tau)$, where
the $\epsilon_{\rm cutoff}$ comes from the dimensions of the bulk and
the $\sqrt{\varepsilon}$ comes from the estimate of the commutator terms.)
The inequality~\eqref{providedweadd} 
on the other
hand follows from inequality~\eqref{variantwithgeomloc} (appearing in the proof
of Proposition~\ref{mautokavoumeabsorb}).
We note that the implicit constants in $\lesssim$ of~\eqref{fromthegeometryjustifiedhere} and~\eqref{providedweadd} are all independent of both
$\epsilon_{\rm redder}$ and
 $\epsilon_{\rm cutoff}$.

The final term~\eqref{cancellingfluxes} arises by combining the boundary terms on
$\{r=R_{\rm fixed}\}$ recalling that the currents are independent of $n$ there.
This term in fact \underline{vanishes} however in view of the restriction on~$\tilde\tau_0$ and~$\tilde\tau_1$ so as for  $\chi_{\tau_0,\widehat{\tau}_1}=1$ in the domain of integration.

Thus the estimate holds with only~\eqref{mainestimate}--\eqref{oldtermsredder} 
and~\eqref{oneinhomogtermhere}--\eqref{andonemoreinhomogtermhere} 
on the right hand side provided
that we add the first two terms on the right hand side of~\eqref{providedweadd} to the right
hand side of line~\eqref{mainestimate}. 

We now define
\begin{equation}
\label{mathcalYdefinition}
\mathcal{Y}(\tau_0, \widehat \tau_1)[\widehat\psi] := \sum\sup_{\tau_0+1\le \tilde\tau_0 \le
\tilde\tau_0+2\le \tilde\tau_1\le \widehat\tau_1-1} \text{\rm terms\ on\ lines\ \eqref{fromcoerc}--\eqref{oldtermsredder}}
\text{\rm\ and\ on\ lines\ \eqref{oneinhomogtermheresmooth}--\eqref{andonemoreinhomogtermheresmooth}},
\end{equation}
where in the above definition, so as to be able to consistently apply our convention related
 to multiplication by $\epsilon_{\rm cutoff}$ and $\epsilon_{\rm redder}$,
we keep on the right hand side of~\eqref{mathcalYdefinition} all constants implicit in the $\lesssim$  of~\eqref{mainestimate} 
with
unfavourable dependence on $\epsilon_{\rm cutoff}$ and $\epsilon_{\rm redder}$ together with their correesponding ``terms''.
We finally define
\begin{equation}
\label{mathcalYprimedef}
\mathcal{Y}'(\tau_0, \widehat \tau_1)[\widehat\psi] :=  \mathcal{Y}(\tau_0, \widehat \tau_1)[\widehat\psi]+
  \epsilon_{\rm redder}\int_{\tau_0+1}^{\widehat\tau_1-1}  \Ezerok_{r_1\le r\le r_2}(\tau')[\widehat\psi] d\tau'
+ \epsilon_{\rm cutoff}\, \sup_{\tau\in[\tau_0,\tau_1]} \Ezerok(\tau)[\psi].
\end{equation}

We may thus write the estimate we have obtained as the statement that
for all $\tau_0+1\le \tilde\tau_0 \le \tilde\tau_1-2\le \widehat\tau_1-3$,
the following holds:
\begin{align}
\nonumber
\Efancypk{}_{{\rm bdry}}(\tilde\tau_1)&+ \int_{\tilde\tau_0}^{\tilde\tau_1}\Efancypk{}_{\rm bulk} (\tau') 
  \lesssim \Efancypk{}_{\rm bdry}(\tilde\tau_0)
  +\Efancypk{}_{\rm bulk}(\tilde\tau_0) + \Efancypk{}_{\rm bulk}(\tilde\tau_1)\\
  \label{edwhektimhsh}
 & + \Efancypk{}_{\rm bulk}(\tilde\tau_0-\epsilon_{\rm cutoff})  +
\Efancypk{}_{\rm bulk}(\tilde\tau_1-\epsilon_{\rm cutoff}) 
 + \mathcal{Y}'(\tau_0,\widehat\tau_1).
\end{align}
In view of our comment after~\eqref{mathcalYdefinition}, we 
emphasise that  the implicit constant in $\lesssim$ is itself independent of $\epsilon_{\rm redder}$
and $\epsilon_{\rm cutoff}$, while all unfavourable dependence on $\epsilon_{\rm redder}$ and 
$\epsilon_{\rm redder}$ in the estimate is restricted to some of the terms in~\eqref{mathcalYdefinition} appearing in $\mathcal{Y}'(\tau_0,\widehat\tau_1)$, 
with the bad constant incorporated into the definition of~\eqref{mathcalYdefinition}.

To make use of the above estimate, let us first state a simple general calculus lemma.
\begin{lemma}
\label{lemmaformainestimate}
Let $B>1$ and $H\ge 0$ be constants, let $\tau_0^*+8\le \widehat\tau_1$ be arbitrary 
and let $F(\tau)$,  $G(\tau)$ be two nonnegative 
functions on $[\tau_0^*, \widehat\tau_1]$
satisfying
\begin{equation}
\label{ithere}
F(\tilde \tau_1) + \int_{\tilde\tau_0}^{\tilde\tau_1} G(\tau')d\tau'  \le B ( F(\tilde\tau_0) +G(\tilde\tau_0) +G(\tilde\tau_1) +G(\tilde\tau_0-\epsilon_{\rm cutoff}) +G(\tilde\tau_1-\epsilon_{\rm cutoff}) + H  )
\end{equation}
for all $\tau_0^*+1\le \tilde\tau_0\le\tilde\tau_0+2\le \tilde\tau_1\le\widehat\tau_1-1$.

 Then, given $\gamma>0$ and $\epsilon>0$,
there exists an $L_{\rm lem}=L_{\rm lem}(B,\gamma, \epsilon)\ge 4$ 
depending only on $B$, $\gamma$ and $\epsilon$,
such that if 
\begin{equation}
\label{slowgrowthhere}
\int_{\tau}^{\tau+\frac12} G(\tau' )d\tau' \le \gamma
\end{equation}
for all $\tau \in [\tau_0^*+1,\widehat\tau_1-2]$,
then in fact 
\begin{equation}
\label{secondlemmaconclusion}
\int_{\tau^*_0+1}^{\hat\tau_1-L_{\rm lem}+3} G(\tau') \le 
2B(  F(\tau_0^*+1)+G(\tau_0^*+1) +G(\tau_0^*+1 -\epsilon_{\rm cutoff})   +H) +\epsilon
\end{equation}
and
\begin{equation}
\label{lemmaconclusion}
\inf_{\tau\le \tau' \le \tau+1} F(\tau')+G(\tau') \le  14B(  F(\tau_0^*+1)+G(\tau_0^*+1) +G(\tau_0^*+1 -\epsilon_{\rm cutoff})   +H) +\epsilon
\end{equation}
for all $\tau \in [\tau_0^*+1,\widehat\tau_1-L_{\rm lem}]$, provided
$\widehat\tau_1-\tau_0^* \ge L_{\rm lem}+4$.
\end{lemma}

\begin{proof}[Proof of lemma]
By the pigeonhole principle applied to the left hand side of~\eqref{slowgrowthhere},
there exists a $\widehat\tau_{\rm final} \in [\widehat\tau_1-3/2,\widehat\tau_1-5/4]$
such that
\[
 G(\widehat\tau_{\rm final})+ G(\widehat\tau_{\rm final}-\epsilon_{\rm cutoff})
 \leq 8\gamma.
\]

By the pigeonhole principle applied to the integral on the left hand side of~\eqref{ithere} with
$\tilde\tau_0=\tau_0^*+1$ and $\tilde\tau_1=\widehat\tau_{\rm final}$, 
given $L_{\rm lem}\ge 4$ arbitrary such that $\widehat\tau_1-\tau_0^*\ge L_{\rm lem}+4$,
there 
exists a $\tilde\tau_{\rm final} \in [\widehat\tau_1-L_{\rm lem}+4, \widehat\tau_1] $ 
such that
\begin{eqnarray*}
G(\tilde\tau_{\rm final}-\epsilon_{\rm cutoff})+
G(\tilde\tau_{\rm final})   &\le&  \frac{2B}{L_{\rm lem}}  (F(\tau_0^*+1) +
 G(\widehat\tau_{\rm final})+ G(\widehat\tau_{\rm final}-\epsilon_{\rm cutoff}) \\
 &&\qquad \qquad
 +G(\tau_0^*+1)+G(\tau_0^*+1-\epsilon_{\rm cutoff})+ H ) \\
&\le&  \frac{2B}{L_{\rm lem}} (  F(\tau_0^*+1)+G(\tau_0^*+1) +G(\tau_0^*+1 -\epsilon_{\rm cutoff})   +H) +\frac{16B\gamma}{L_{\rm Lem}},
\end{eqnarray*}
where for the second inequality we have used~\eqref{slowgrowthhere}.

If we require $L_{\rm lem}$ sufficiently large such that
 $|\frac{2B}{L_{\rm lem}}| < 1$ and $\left|\frac{32B^3\gamma}{L_{\rm lem}}\right| <\epsilon$, we  that
\begin{equation}
\label{applytotheright}
G(\tilde\tau_{\rm final}-\epsilon_{\rm cutoff})+
 G(\tilde\tau_{\rm final}) \le   F(\tau_0^*+1)+G(\tau_0^*+1) +G(\tau_0^*+1 -\epsilon_{\rm cutoff})   +H
 +\frac{\epsilon}{2B^2}.
 \end{equation}

Now  
we may reapply~\eqref{ithere} in $[\tau_0^*+1, \tilde\tau_{\rm final}]$. We immediately
obtain~\eqref{secondlemmaconclusion}. Moreover, 
given arbitrary $ \tau_0^*+2\le \tau\le \widehat\tau_1-L_{\rm lem}$, 
and using now~\eqref{applytotheright} 
and the pigeonhole principle it follows
that there exists a  $\tau-1\le \tau'\le \tau$ such that 
\begin{equation}
\label{finalforG}
G(\tau')+G(\tau'-\epsilon_{\rm cutoff}) \le 
6 B (  F(\tau_0^*+1)+G(\tau_0^*+1) +G(\tau_0^*+1 -\epsilon_{\rm cutoff})   +H)+\frac{\epsilon}{2B} .
\end{equation}
Hence, again applying~\eqref{ithere} now in $[\tau_0+1, \tau']$, it follows
that
\begin{equation}
\label{finalforF}
F(\tau') \le 7B (  F(\tau_0^*+1)+G(\tau_0^*+1) +G(\tau_0^*+1 -\epsilon_{\rm cutoff})   +H)+\frac{\epsilon}2 .
\end{equation}
Estimates~\eqref{finalforG} and~\eqref{finalforF} yield~\eqref{lemmaconclusion}.
\end{proof}

We now  apply the lemma to 
\begin{equation}
\label{withthesedefs}
F(\tau ): =\Efancypk{}_{{\rm bdry}}(\tau) [\widehat\psi], 
 \qquad G(\tau): = \Efancypk{}_{{\rm bulk}}(\tau)[\widehat\psi],     , \qquad   
 H: =\mathcal{Y}'(\tau_0,\widehat\tau_1)[\widehat\psi],
\end{equation}
where $L_{\rm long}$ and thus $\widehat\tau_1$ and $\widehat\psi$ will be determined immediately below
and where $\tau_0^*\in [\tau_0,\tau_0+1]$ is arbitrary.
Note that the estimate~\eqref{edwhektimhsh} yields~\eqref{ithere},
 where we take  $B$ to be the constant implicit in the $\lesssim$ of~\eqref{edwhektimhsh}.
 Recall that the constant implicit in the $\lesssim$ of~\eqref{edwhektimhsh} is a
 general constant in the sense  of Section~\ref{noteonconstants} and thus 
 in particular,
  $B$  is independent of $L_{\rm long}$ defining $\widehat\tau_1$
 and $\widehat\psi$ (cf.~Remark~\ref{independenceofLlong}).
 Moreover, by our remarks above,
 $B$ is independent of $\epsilon_{\rm redder}$
and $\epsilon_{\rm cutoff}$.
 Let us note that by our assumptions on existence of $\psi\in \mathcal{R}(\tau_0,\tau_1)$
  postulated in
 Proposition~\ref{toporderestimateprop}, 
 we have the (non-quantitative) finiteness statement
 \begin{equation}
 \label{nonquantitativefiniteness}
  {}^{\chi}\Xpk(\tau_0,\tau_1)[\psi]<\infty.
  \end{equation}
Let us note that by Proposition~\ref{whenoneextends} and~\ref{reversecomparison},
 there exists a $\gamma=\gamma[\psi]<\infty$, depending in fact only on the value of
 ${}^{\chi}\Xpk(\tau_0,\tau_1)[\psi]$, such that~\eqref{slowgrowthhere} holds, independently
 of the choice of $L_{\rm long}$ defining $\widehat\tau_1$ and $\widehat\psi$. Finally, let  us set
 \begin{equation}
 \label{epsgivenby}
 \epsilon:=  \Epk(\tau_0) [\psi].
 \end{equation}
 This determines a sufficiently large $L_{\rm lem}=L_{\rm lem}[\psi]$ by the lemma, where the $\psi$-dependence
 arises from the $\psi$-dependence of our chosen $\gamma$ and $\epsilon$.
 
 We may finally 
 choose $L_{\rm long}$ in the definition of $\widehat\tau_1$ and $\widehat\psi$
to be precisely $L_{\rm long}[\psi]= L_{\rm lem}$ given  by the lemma with $B$, $\gamma=\gamma[\psi]$ and $\epsilon=\epsilon[\psi]$ determined above.

The statement~\eqref{lemmaconclusion}
 of the lemma thus holds for all $\tau_0^*\in [\tau_0,\tau_0+1]$, and given
 moreover the relation $\widehat\tau_1 =\tau_1+L_{\rm long}$,.~we obtain that
 \begin{align*}
  \int_{\tau_0+2}^{\tau_1+2} \Efancypk{}_{{\rm bulk}}(\tau')[\widehat\psi] d\tau' 
  &\lesssim  \Efancypk(\tau^*_0+1) 
+ \Efancypk(\tau^*_0+1-\epsilon_{\rm cutoff} ) +  \mathcal{Y}'(\tau_0,\widehat\tau_1)+\epsilon\\
\inf_{\tau\le \tau' \le \tau+1} \Efancypk(\tau') 
&\lesssim  \Efancypk(\tau^*_0+1) 
+ \Efancypk(\tau^*_0+1-\epsilon_{\rm cutoff} ) +  \mathcal{Y}'(\tau_0,\widehat\tau_1)+\epsilon
\end{align*}
for all  $\tau \in [\tau_0^*+1,\tau_1]$, with $\epsilon$ given by~\eqref{epsgivenby}.
 
 \begin{remark}
 \label{Ldetermineafter}
 We emphasise that~\eqref{slowgrowthhere}  applied  with $\gamma$ as above 
 is essentially~\eqref{nonquantitativeintheintro} of the introduction (averaged over a unit time length).
 The assumption is indeed nonquantitative because $\gamma$ depends on
 the nonquantitative~\eqref{nonquantitativefiniteness}. 
 \end{remark}

It follows by a little bit of averaging (integrating~\eqref{lemmaconclusion} in $\tau_0^*\in [\tau_0,\tau_0+1]$) and applying Proposition~\ref{reversecomparison} 
and the local existence Proposition~\ref{localexistencelargea} (for $\varepsilon$ sufficiently small in the bound~\eqref{newbasicbootstrap}, provided that $k_{\rm loc}\le \lesslessk$), 
that for all $\tau_0+2\le \tau\le \tau_1$
\begin{align*}
\inf_{\tau\le \tau' \le \tau+1} \Efancypk(\tau')+
  \int_{\tau_0+2}^{\tau_1+2} \Efancypk{}_{{\rm bulk}}(\tau')[\widehat\psi] d\tau' 
& \lesssim \int_{\tau_0}^{\tau_0+1} \Efancypk(\tau^*_0+1) d\tau_0^* +  \mathcal{Y}'(\tau_0,\widehat\tau_1)+\epsilon \\
&\lesssim \Epk(\tau_0) +  \mathcal{Y}'(\tau_0,\widehat\tau_1)+\epsilon \\
& \lesssim \Epk(\tau_0) +  \mathcal{Y}'(\tau_0,\widehat\tau_1)
\end{align*}
where we have plugged in~\eqref{epsgivenby} in the last inequality.

From the definition~\eqref{mathcalYprimedef} we have thus for all $\tau_0+2\le \tau\le \tau_1$,
\[
\inf_{\tau\le \tau' \le \tau+1} \Efancypk(\tau') 
+ \int_{\tau_0+2}^{\tau_1+2} \Efancypk{}_{{\rm bulk}}(\tau')[\widehat\psi] d\tau' 
 \lesssim \Epk(\tau_0) +  \mathcal{Y}(\tau_0,\widehat\tau_1)
+
 \epsilon_{\rm redder}\int_{\tau_0+1}^{\widehat\tau_1-1}  \Ezerok_{r_1\le r\le r_2}(\tau')[\widehat\psi] d\tau'+ \epsilon_{\rm cutoff}\, \sup_{\tau\in[\tau_0,\tau_1]} \Ezerok(\tau).
\]
Tracking the dependence of various constants implicit in $\lesssim$
on $\epsilon_{\rm redder}$ in the propositions we have applied, the 
 reader can check that
the convention of Section~\ref{noteonconstants} still applies to both the constant multiplying 
$\epsilon_{\rm redder}$ and the constant multiplying $\epsilon_{\rm cutoff}$ above.

It follows that we may apply Proposition~\ref{mautokavoumeabsorb}, noting that
the implicit constant in the $\lesssim$ in that proposition also is independent of 
$\epsilon_{\rm redder}$, to absorb the term multiplying $\epsilon_{\rm redder}$, provided we choose
$\epsilon_{\rm redder}>0$ sufficiently small.
We obtain thus
\begin{equation}
\label{wewillrevistsoon}
\inf_{\tau\le \tau' \le \tau+1} \Efancypk(\tau') 
+ \int_{\tau_0+2}^{\tau_1+2} \Efancypk{}_{{\rm bulk}}(\tau')[\widehat\psi] d\tau' 
 \lesssim \Epk(\tau_0) +  \mathcal{Y}(\tau_0,\widehat\tau_1)
 + \epsilon_{\rm cutoff}\, \sup_{\tau\in[\tau_0,\tau_1]} \Ezerok(\tau).
\end{equation}

Whence, it follows from~\eqref{fromthisfirstpoint} of Proposition~\ref{comparisonenergies}
that 
\[
\inf_{\tau\le \tau' \le \tau+1} \Epk(\tau') \lesssim \Epk(\tau_0) +  \mathcal{Y}(\tau_0,\widehat\tau_1)
 + \epsilon_{\rm cutoff}\, \sup_{\tau\in[\tau_0,\tau_1]} \Ezerok(\tau),
\]
for all $\tau_0+2\le \tau\le \tau_1$. 
But now an additional application of Proposition~\ref{localexistencelargea}  yields 
that 
\[
\sup_{\tau_0\le \tau\le \tau_1}  \Epk(\tau') \lesssim \Epk(\tau_0) +   \mathcal{Y}(\tau_0,\widehat\tau_1)
 + \epsilon_{\rm cutoff}\, \sup_{\tau\in[\tau_0,\tau_1]} \Ezerok(\tau),
\]
whence, since in writing the above we have adhered to the convention of Section~\ref{noteonconstants} regarding
constants multiplying $\epsilon_{\rm cutoff}$, 
we may absorb the last term on the right
hand side for sufficiently small $\epsilon_{\rm cutoff}$ 
to obtain
\begin{equation}
\label{inviewofagain}
\sup_{\tau_0\le \tau\le \tau_1}  \Epk(\tau') \lesssim \Epk(\tau_0) +   \mathcal{Y}(\tau_0,\widehat\tau_1).
\end{equation}

Revisiting now~\eqref{wewillrevistsoon}, we obtain that we may estimate the bulk
\[
\int_{\tau_0+2}^{\tau_1+2}\Efancypk{}_{\rm bulk} (\tau) d\tau \lesssim
\Epk(\tau_0)  +
\mathcal{Y}(\tau_0, \widehat\tau_1) .
\]
Using now inequality~\eqref{secondineqhere} of Proposition~\ref{comparisonenergies}, 
 it follows that
\begin{equation}
\label{willpromotesoon}
\int_{\tau_0+2}^{\tau_1+2} {}^\rho \Epminusonek' [\widehat\psi]\lesssim  \Epk(\tau_0)  +
\mathcal{Y}(\tau_0, \widehat\tau_1) .
\end{equation}
Again, using Proposition~\ref{localexistencelargea} and in view of~\eqref{inviewofagain},
we promote~\eqref{willpromotesoon} to the statement
\begin{equation}
\label{proagwgn}
\int_{\tau_0}^{\tau_1} {}^\rho \Epminusonek'[\psi] \lesssim  \Epk(\tau_0)  +
\mathcal{Y}(\tau_0, \widehat\tau_1) 
\end{equation}
and similarly for the boundary terms.

From~\eqref{inviewofagain} and~\eqref{proagwgn} it is now elementary to obtain similar bounds for the additional flux terms in 
the definition of~${}^\rho\Xpk(\tau_0,\tau_1)$. 
We obtain thus 
\[
{}^\rho\Xpk(\tau_0,\tau_1)
\lesssim  \Epk(\tau_0)  +
\mathcal{Y}(\tau_0, \widehat\tau_1) 
\]
which is precisely the desired~\eqref{topordergiatwra}.
\end{proof}

In what follows, we will consider both $\epsilon_{\rm redder}>0$  and
$\epsilon_{\rm cutoff}>0$ to have been fixed to satisfy the smallness necessary for  the above proof.

\subsection{Estimates for the inhomogeneous terms}
\label{estimforinhomog}
In this section we estimate the inhomogeneous terms 
contained in~$\mathcal{Y}(\tau_0, \widehat\tau_1)$
from our basic energies.

\subsubsection{Estimates for the linear terms arising from the cutoff}
For the linear terms arising from the cutoff, we  have:
\begin{proposition}
\label{fromthecutoffbounding}
The term in line~\eqref{oneinhomogtermheresmooth} may be estimated:
\begin{equation}
\label{tobeabsorbedlater}
 \left| \int_{\mathcal{M}\cap\{r\le R_{\rm freq}\}}
 \tilde\chi_{\tilde\tau_0,\tilde\tau_1}^4
  {\rm Re} \left(H^{{\rm main},n}[\widehat\psi_{n, {\bf k}}] \overline{F^{\rm cutoff}_{n,{\bf k}}}\right)\right|
 \lesssim   
 \Xzerokminustwo(\tau_0, \tau_1) .
  \end{equation}
\end{proposition}
\begin{proof}
From~\eqref{exprwithcutoffs}, one sees that the
 highest order terms here schematically take the form
\begin{equation}
\label{thetermslooklike}
\partial  P_n \mathring{\mathfrak{D}}^{\bf k}  ( \chi_{\tau_0, \widehat\tau_1}^2  \widehat\psi)  \cdot 
 \tilde\chi_{\tilde\tau_0,\tilde\tau_1}^4 P_n
 \mathring{\mathfrak{D}}^{\bf k}  \left( \nabla_g \chi_{\tau_0,\widehat\tau_1} \cdot \nabla_g  \chi_{\tau_0,\widehat\tau_1} \widehat \psi\right).
\end{equation}
We note that by the disjoint
support properties of $ \tilde\chi_{\tilde\tau_0,\tilde\tau_1}^4 $ and
 $\nabla_g\chi^2_{\tau_0, \widehat\tau_1}$ and Proposition~\ref{smoothingprop}
it follows that, defining $Q$ by 
\[
Q\Psi :=
 \tilde\chi_{\tilde\tau_0,\tilde\tau_1}^4  P_n\mathring{\mathfrak{D}}^{\bf k} ((\nabla_g \chi_{\tau_0, \widehat\tau_1}) \cdot \Psi ),
\]
$Q$ is $(t^*,\phi^*)$-smoothing, which we may view as applied to $\Psi:= \nabla_g (\chi_{\tau_0,\widehat\tau_1}
\widehat\psi)$.
We may integrate~\eqref{thetermslooklike} 
in the region $r_0\le r\le R_{\rm freq}$ globally in $t^*$ and $\phi^*$ and integrate 
by parts two say  $\mathring{\mathfrak{D}}$
off from the first factor~\eqref{thetermslooklike}. In view of the  factor
$  \tilde\chi_{\tilde\tau_0,\tilde\tau_1}^4$,
there are no boundary terms.
The resulting operator on the second factor
\[
\tilde{Q}\Psi :=
\mathring{\mathfrak{D}}\mathring{\mathfrak{D}} \tilde\chi_{\tilde\tau_0,\tilde\tau_1}^4  P_n \mathring{\mathfrak{D}}^{\bf k} ((\nabla_g \chi_{\tau_0, \widehat\tau_1})\cdot \Psi )
\]
is again $(t^*,\phi^*)$-smoothing by Proposition~\ref{smoothingprop} and satisfies for instance
\begin{equation}
\label{whyforinstance}
\|\tilde{Q} \Psi\|_{L^2(
\mathcal{M}\cap \{ r_0\le r\le R_{\rm freq}\})}
\lesssim\left (1 +   \sqrt{\,\,\Xzerolesslessk(\tau_0,\tau_1) } \right)  \|\Psi\|_{H^{k-3, 0}(
\mathcal{M}\cap \{ r_0\le r\le R_{\rm freq}\})}.
\end{equation}
Upon application of Cauchy--Schwarz, 
both arising terms can thus be bounded by
\[
\Xzerokminusthree(\tau_0,\widehat\tau_1)[\widehat\psi] \lesssim \Xzerokminustwo(\tau_0,\tau_1)
\]
(where for the above inequality we have appealed to Proposition~\ref{whenoneextends}).

The lower order terms in~\eqref{thetermslooklike} can of course be estimated similarly. 
\end{proof}

\subsubsection{Nonlinear terms arising from commuting with $\mathring{\mathfrak{D}}^{\bf k}$}
\label{Dkcommutingtermsnonlinearestimate}

Recall that we have absorbed all linear commutation terms in the bulk,
so all remaining commutation terms are nonlinear.
\begin{proposition}
\label{nonlincomprophereref}
The term in line~\eqref{twoinhomogtermsheresmooth} may be estimated:
\begin{eqnarray}
\label{addnumbertopropstatement}
\left| \int_{\mathcal{M}\cap\{r\le R_{\rm freq}\}}
 \tilde\chi_{\tilde\tau_0,\tilde\tau_1}^4
 {\rm Re} \left(
H^{{\rm main},n}[\widehat\psi_{n, {\bf k}}] \overline{F^{\rm commu}_{n,{\bf k}}}\right)\right|
 &\lesssim&
 \sqrt{\tau_1-\tau_0} \,\,
{}^\rho\Xzerok(\tau_0,\tau_1) \sqrt{\,\,\Xzerolesslessk(\tau_0,\tau_1) } .
 \end{eqnarray}
\end{proposition}
\begin{proof}
The highest order terms are of the form
\[
 \tilde\chi_{\tilde\tau_0,\tilde\tau_1}^4 \partial P_n \mathring{\mathfrak{D}}^{\bf k}  ( \chi^2_{\tau_0,\widehat\tau_1} \widehat \psi) \cdot
P_n [\mathring{\mathfrak{D}}^{\bf k}, \Box_{g(\chi_{\tau_1}^2\psi_,x)}-\Box_{g_{a,M}}]  \chi^2_{\tau_0, \widehat\tau_1}\widehat\psi.
\]
In view of the presence of $\chi_{\tau_1}$ we may rewrite this as
\[
 \tilde\chi_{\tilde\tau_0,\tilde\tau_1}^4 \partial P_n \mathring{\mathfrak{D}}^{\bf k}  ( \chi^2_{\tau_0,\widehat\tau_1}  \widehat\psi) \cdot
P_n [\mathring{\mathfrak{D}}^{\bf k}, \Box_{g(\chi_{\tau_1}^2\psi_,x)}-\Box_{g_{a,M}}]  \chi^2_{\tau_0, \tau_1+1}\widehat\psi
\]
and thus as
\begin{align}
\label{firsttermandthusas}
 \tilde\chi_{\tilde\tau_0,\tilde\tau_1}^4\partial P_n \mathring{\mathfrak{D}}^{\bf k}  ( \chi^2_{\tau_0,\tau_1+3}  \widehat\psi) \cdot
P_n [\mathring{\mathfrak{D}}^{\bf k}, \Box_{g(\chi_{\tau_1}^2\psi_,x)}-\Box_{g_{a,M}}]  \chi^2_{\tau_0, \tau_1+1}
\widehat\psi\\
\label{secondtermandthusas}
+ \tilde\chi_{\tilde\tau_0,\tilde\tau_1}^4 \chi^2_{\tau_0,\tau_1+2}  \partial P_n \mathring{\mathfrak{D}}^{\bf k}  ( \chi^2_{\tau_1+3,\widehat\tau_1}  
\widehat\psi) \cdot
P_n [\mathring{\mathfrak{D}}^{\bf k}, \Box_{g(\chi_{\tau_1}^2\psi_,x)}-\Box_{g_{a,M}}]  \chi^2_{\tau_0, \tau_1+1}
\widehat\psi\\
\label{thirdtermandthusas}
+ 
 \partial P_n \mathring{\mathfrak{D}}^{\bf k}  ( \chi^2_{\tau_1+3,\widehat\tau_1} \widehat \psi) \cdot
\tilde\chi_{\tilde\tau_0,\tilde\tau_1}^4\chi^2_{\tau_1+2,\widehat \tau_1}  P_n [\mathring{\mathfrak{D}}^{\bf k}, \Box_{g(\chi_{\tau_1}^2\psi_,x)}-\Box_{g_{a,M}}]  \chi^2_{\tau_0, \tau_1+1} \widehat\psi.
\end{align}

For the term~\eqref{firsttermandthusas}, 
we apply Cauchy--Schwarz, to estimate by:
\[
\epsilon\|  \partial P_n \mathring{\mathfrak{D}}^{\bf k}  ( \chi_{\tau_0,\tau_1+3}^2 \widehat \psi)  \|_{L^2(\mathcal{M}\cap\{r\le R_{\rm freq}\})}^2 +\epsilon^{-1} \| P_n[\mathring{\mathfrak{D}}^{\bf k} ,\Box_{g(\chi^2_{\tau_1}\psi,x)} -\Box_{g_{a,M}}]
\chi_{\tau_0,\tau_1+1}^2\widehat\psi\|_{L^2(\mathcal{M}\cap\{r\le R_{\rm freq}\} )}^2,
\]
and thus, applying Proposition~\ref{mixedsobolevboundednessprop}, by
\[
\epsilon\|   \mathring{\mathfrak{D}}^{\bf k}  ( \chi_{\tau_0,\tau_1+3}^2 \widehat \psi)  \|_{H^1(\mathcal{M}\cap\{r\le R_{\rm freq}\})}^2 +\epsilon^{-1} \| [\mathring{\mathfrak{D}}^{\bf k} ,\Box_{g(\chi^2_{\tau_1}\psi,x)} -\Box_{g_{a,M}}]
\chi_{\tau_0,\tau_1+1}^2 \widehat\psi\|_{L^2(\mathcal{M}\cap\{r\le R_{\rm freq}\} )}^2.
\]
Now by the Sobolev estimate~\eqref{largeasobolev} of Proposition~\ref{largeasobolevforfunctions}, we may estimate
 \begin{align}
 \nonumber
 \| [\mathring{\mathfrak{D}}^{\bf k} ,\Box_{g(\chi^2_{\tau_1}\psi,x)} -\Box_{g_{a,M}}]
\chi_{\tau_0,\tau_1+1}^2\widehat \psi\|_{L^2(\mathcal{M}\cap\{r\le R_{\rm freq}\} )}^2
&\lesssim \int_{\tau_0}^{\tau_1+1}  \Ezerominusoneminusdeltalesslessk' (\tau') 
  \Ezerok(\tau')  d\tau' 
\\ 
\label{anumberwhereweneedone}
&\lesssim \Xzerolesslessk(\tau_0,\tau_1)
\sup_{\tau_0\le \tau' \le \tau_1}  \Ezerok(\tau') 
\end{align}
(where we are also using Proposition~\ref{whenoneextends}),
whereas  we may similarly estimate
\[
\|   \mathring{\mathfrak{D}}^{\bf k}  ( \chi_{\tau_0,\tau_1+3}^2 \widehat\psi)  \|_{H^1(\mathcal{M}\cap\{r\le R_{\rm freq}\})}^2 \lesssim (\tau_1-\tau_0) 
\sup_{\tau_0\le \tau' \le \tau_1}  \Ezerok(\tau') .
\] 

If $\Xzerolesslessk(\tau_0,\tau_1)=0$, then $\psi$ vanishes identically and the statement
of our Proposition is trivially true. Otherwise, we choose 
$\epsilon :=(\tau_1-\tau_0)^{-\frac12}\sqrt{\,\,\Xzerolesslessk(\tau_0,\tau_1) }$ in the above, to obtain
that the term~\eqref{firsttermandthusas} is 
indeed estimated by the right hand side of~\eqref{addnumbertopropstatement}.

For the term~\eqref{secondtermandthusas}, we notice that the first factor may be rewritten
as
\[
\partial Q \chi_{\tau_1+3,\widehat\tau_1} \widehat \psi
\]
where $Q$ is a smoothing operator satisfying in particular say
$\|Q\Psi\|_{H^{0,1}(\mathcal{M}\cap \{r\le R_{\rm freq}\}) }\lesssim \|\Psi\|_{H^{k-2,1}(\mathcal{M}\cap \{r\le R_{\rm freq}\})}$.
Applying Cauchy--Schwarz as before, the above term is
bounded by say
\[
 \Xzerokminusthree{}^*(\tau_0,\widehat\tau_1)[\widehat\psi] \lesssim
 \Xzerokminustwo{}^*(\tau_0,\tau_1)[\psi]  \lesssim 
 \Ezerok(\tau_0) ,
\]
while the second term is bounded as before by~\eqref{anumberwhereweneedone}.
Here we have used Proposition~\ref{whenoneextends}. 
Choosing say
$\epsilon :=\sqrt{\,\,\Xzerolesslessk(\tau_0,\tau_1) }$, these two terms
are bounded by the right hand side of~\eqref{addnumbertopropstatement}.

For the term~\eqref{thirdtermandthusas}, 
upon integration with respect to $t^*$ and $\phi^*$,
we may integrate by parts two $\mathring{\mathfrak{D}}$ from the first to the second factor. 
There are no boundary terms in view of the cutoffs.
The second factor becomes
\begin{equation}
\label{secondfactorbecomes}
\mathring{\mathfrak{D}}\mathring{\mathfrak{D}} 
\tilde\chi_{\tilde\tau_0,\tilde\tau_1}^4
\chi^2_{\tau_1+2,\widehat \tau_1}  P_n 
 [\mathring{\mathfrak{D}}^{\bf k}, \Box_{g(\chi_{\tau_1}^2\psi_,x)}-\Box_{g_{a,M}}] 
 \chi^2_{\tau_0,\tau_1+1}\widehat\psi.
\end{equation}
We expand the commutator as 
\[
\mathring{\mathfrak{D}}\mathring{\mathfrak{D}} \tilde\chi_{\tilde\tau_0,\tilde\tau_1}^4\chi^2_{\tau_1+2,\widehat \tau_1}  
P_n \mathring{\mathfrak{D}}^{\bf k}(\Box_{g(\chi_{\tau_1}^2\psi_,x)}-\Box_{g_{a,M}})
 \chi^2_{\tau_0,\tau_1+1}\widehat\psi
 - \mathring{\mathfrak{D}}\mathring{\mathfrak{D}}\tilde\chi_{\tilde\tau_0,\tilde\tau_1}^4 \chi^2_{\tau_1+2,\widehat \tau_1}  
P_n (\Box_{g(\chi_{\tau_1}^2\psi_,x)}-\Box_{g_{a,M}}) \mathring{\mathfrak{D}}^{\bf k}
 \chi^2_{\tau_0,\tau_1+1}\widehat\psi
\]
which we may express as
\[
\mathring{\mathfrak{D}}\mathring{\mathfrak{D}} \tilde\chi_{\tilde\tau_0,\tilde\tau_1}^4\chi^2_{\tau_1+2,\widehat \tau_1}  
P_n \mathring{\mathfrak{D}}^{\bf k}  \chi_{\tau_0,\tau_1+1}
S  \chi_{\tau_0,\tau_1+1}
\widehat\psi +
\mathring{\mathfrak{D}}\mathring{\mathfrak{D}}\tilde\chi_{\tilde\tau_0,\tilde\tau_1}^4 \chi^2_{\tau_1+2,\widehat \tau_1}  
P_n  \chi_{\tau_0,\tau_1+1}
\tilde{S}\mathring{\mathfrak{D}}^{\bf k}  \chi_{\tau_0,\tau_1+1} \widehat\psi
\]
where $S$ and $\tilde{S}$ are differential operators of order two.

Setting
\[
Q:= \mathring{\mathfrak{D}}\mathring{\mathfrak{D}} \tilde\chi_{\tilde\tau_0,\tilde\tau_1}^4\chi^2_{\tau_1+2,\widehat \tau_1}  P_n 
\mathring{\mathfrak{D}}^{\bf k}
 \chi _{\tau_0, \tau_1+1} ,  \qquad
 \tilde{Q}: = \mathring{\mathfrak{D}}\mathring{\mathfrak{D}} \tilde\chi_{\tilde\tau_0,\tilde\tau_1}^4\chi^2_{\tau_1+2,\widehat \tau_1}  
P_n  \chi_{\tau_0,\tau_1+1},
 \]
these are both smoothing operators
and  we may rewrite the 
term as
\[
QS  \chi_{\tau_0,\tau_1+1} \widehat\psi + \tilde{Q}\tilde{S} \mathring{\mathfrak{D}}^{\bf k}  \chi_{\tau_0,\tau_1+1} \widehat\psi.
\]

Now we note that 
\[
\|S \chi_{\tau_0,\tau_1+1} \widehat\psi \|_{L^2(\mathcal{M}\cap\{r_0\le r\le R_{\rm freq} \} )} \lesssim
 \sqrt{\,\,\Xzerolesslessk(\tau_0,\tau_1) }
  \|\chi_{\tau_0,\tau_1+1}\widehat\psi \|_{H^{0,2}(\mathcal{M}\cap\{r_0\le r\le R_{\rm freq} \} )}
\]
while  say
\[
\|\tilde{S}\mathring{\mathfrak{D}}^{\bf k} \chi_{\tau_0,\tau_1+1} \widehat\psi \|_{H^{-2,0}(\mathcal{M}\cap\{r_0\le r\le R_{\rm freq} \} )} \lesssim 
 \sqrt{\,\,\Xzerolesslessk(\tau_0,\tau_1) }
  \|\chi_{\tau_0,\tau_1+1}\widehat\psi\|_{H^{k-2,2}(\mathcal{M}\cap\{r_0\le r\le R_{\rm freq} \} )}.
\]

For $Q$ it is enough to apply
\[
\| Q\Psi\|_{L^2((\mathcal{M}\cap \{ r_0 \le r\le R_{\rm fixed} )} \lesssim
 \,
 \| \Psi\|_{L^2(\mathcal{M}\cap \{ r_0 \le r \le R_{\rm fixed}\})} ,
\]
while for $\tilde{Q}$ we apply 
\[
\| \tilde{Q}\Psi\|_{H^{0,0}((\mathcal{M}\cap \{ r_0 \le r\le R_{\rm fixed} )} \lesssim
 \| \Psi\|_{H^{-2,0}(\mathcal{M}\cap \{ r_0 \le r \le R_{\rm fixed}\})} .
\]
It follows that
\begin{eqnarray*}
\| QS  \chi_{\tau_0,\tau_1+1} \widehat\psi \|^2 + \| \tilde{Q}\tilde{S} \mathring{\mathfrak{D}}^{\bf k}  \chi_{\tau_0,\tau_1+1} \widehat\psi\|^2 
&\lesssim& \Xzerolesslessk(\tau_0,\tau_1)   \, \|\chi_{\tau_0,\tau_1+1}\widehat\psi\|^2_{H^{k-2,2}(\mathcal{M}\cap\{r_0\le r\le R_{\rm freq} \} )}\\
&\lesssim&  \Xzerolesslessk(\tau_0,\tau_1)  \, \Xzerokminustwo^*(\tau_0, \tau_1)\\
&\lesssim&  \Xzerolesslessk(\tau_0,\tau_1)  \, {}^\rho \Xzerok(\tau_0, \tau_1).
\end{eqnarray*}

We may thus estimate as before~\eqref{thirdtermandthusas} with Cauchy--Schwarz (factored
as described above with~\eqref{secondfactorbecomes} as the second factor) and
$\epsilon :=\sqrt{\,\,\Xzerolesslessk(\tau_0,\tau_1) }$, estimating now both
terms by  the right hand side of~\eqref{addnumbertopropstatement}.
\end{proof}

\subsubsection{Terms arising from $P_n$: pseudodifferential commutation estimates}

\begin{proposition}
The term in line~\eqref{thepseudoinhomgthersheresmooth} may be estimated:
\begin{eqnarray}
\label{willaddanumberhere}
\left|  \int_{\mathcal{M}\cap \{r\le R_{\rm freq}\}} 
\tilde\chi_{\tilde\tau_0,\tilde\tau_1}^4
 {\rm Re} \left(
H^{{\rm main},n}[\widehat\psi_{n, {\bf k}}] \overline{F^{\rm pseudo}_{n,{\bf k}}}\right)
\right|  &\lesssim&
 \sqrt{\tau_1-\tau_0}\,\,
{}^\rho \Xzerok(\tau_0,\tau_1)  \sqrt{\,\,\Xzerolesslessk(\tau_0,\tau_1) }.
 \end{eqnarray}
 \end{proposition}

\begin{proof}
The highest order terms take the form
\[
\tilde\chi_{\tilde\tau_0,\tilde\tau_1}^4\partial P_n \mathring{\mathfrak{D}}^{\bf k}  ( \chi_{\tau_0,\tau_1} ^2 \widehat\psi)  \cdot
 [P_n ,\Box_g(\psi,x) -\Box_{g_{a,M}}] \mathring{\mathfrak{D}}^{\bf k} \chi_{\tau_0,\tau_1}^2 \widehat\psi
\]
which we may again partition as
\begin{align}
\label{firsttermandthusasagain}
\tilde\chi_{\tilde\tau_0,\tilde\tau_1}^4\partial P_n \mathring{\mathfrak{D}}^{\bf k}  ( \chi^2_{\tau_0,\tau_1+3}  \widehat\psi) \cdot
[P_n, \Box_{g(\chi_{\tau_1}^2\psi_,x)}-\Box_{g_{a,M}}] \mathring{\mathfrak{D}}^{\bf k}  \chi^2_{\tau_0, \tau_1+1}
\widehat\psi\\
\label{secondtermandthusasagain}
+\tilde\chi_{\tilde\tau_0,\tilde\tau_1}^4 \chi^2_{\tau_0,\tau_1+2}  \partial P_n \mathring{\mathfrak{D}}^{\bf k}  ( \chi^2_{\tau_1+3,\widehat\tau_1}  
\widehat\psi) \cdot
[P_n, \Box_{g(\chi_{\tau_1}^2\psi_,x)}-\Box_{g_{a,M}}] \mathring{\mathfrak{D}}^{\bf k} \chi^2_{\tau_0, \tau_1+1}
\widehat\psi\\
\label{thirdtermandthusasagain}
+ \partial P_n \mathring{\mathfrak{D}}^{\bf k}  ( \chi^2_{\tau_1+3,\widehat\tau_1} \widehat \psi) \cdot
\tilde\chi_{\tilde\tau_0,\tilde\tau_1}^4\chi^2_{\tau_1+2,\widehat \tau_1}  [P_n, \Box_{g(\chi_{\tau_1}^2\psi_,x)}-\Box_{g_{a,M}}] 
\mathring{\mathfrak{D}}^{\bf k} \chi^2_{\tau_0, \tau_1+1} \widehat\psi.
\end{align}

Let us consider first the term~\eqref{firsttermandthusasagain}.
By Cauchy--Schwarz, we may bound this by
\[
\epsilon\|  \partial P_n \mathring{\mathfrak{D}}^{\bf k}  ( \chi_{\tau_0,\tau_1+3}^2 \widehat \psi)  \|_{L^2(\mathcal{M}\cap\{r\le R_{\rm freq}\})}^2 +\epsilon^{-1} \| [P_n ,\Box_{g(\chi^2_{\tau_1}\psi,x)} -\Box_{g_{a,M}}]
\mathring{\mathfrak{D}}^{\bf k}
\chi_{\tau_0,\tau_1+1}^2\widehat\psi\|_{L^2(\mathcal{M}\cap\{r\le R_{\rm freq}\} )}^2
\]
whence 
\begin{equation}
\label{whence}
\epsilon\|  \mathring{\mathfrak{D}}^{\bf k}  ( \chi_{\tau_0,\tau_1+3} ^2 \widehat\psi)  \|_{H^1(\mathcal{M}\cap\{r\le R_{\rm freq}\})}^2 +\epsilon^{-1} \| [P_n ,\Box_{g(\psi,x)} -\Box_{g_{a,M}}] 
\mathring{\mathfrak{D}}^{\bf k}
\chi_{\tau_0,\tau_1+1}^2 \widehat\psi\|_{L^2(\mathcal{M}\cap\{r\le R_{\rm freq}\} )}^2.
\end{equation}

Now applying Proposition~\ref{fornonlinearpseudo} to the operators $P_n$ and 
$(\Box_{g(\psi,x)}-\Box_{g_{a,M}}) \mathring{\mathfrak{D}}^{\bf k}\chi_{\tau_0,\tau_1}$,  
we obtain
\begin{align}
\nonumber
 \big\|  [P_n , &\Box_{g(\chi^2_{\tau_1}\psi,x)} -\Box_{g_{a,M}}] \mathring{\mathfrak{D}}^{\bf k} \chi_{\tau_0,\tau_1+1} \Psi \big\|_{L^2(\mathcal{M}\cap\{r\le R_{\rm freq}\})} \\
 \nonumber
 &\lesssim \sqrt{\Xzerolesslessk(\tau_0,\tau_1)}\, \left( \sup_{\tau\in \mathbb R} \| \chi_{\tau,\tau+2} \mathring{\mathfrak{D}}^{\bf k} \chi_{\tau_0,\tau_1+1} \Psi \|_{H^{-1,2}(\mathcal{M}\cap\{r\le R_{\rm freq}\})} + \|  \mathring{\mathfrak{D}}^{\bf k} \chi_{\tau_0,\tau_1+1} \Psi \|_{H^{-3,2}(\mathcal{M}\cap\{r\le R_{\rm freq}\})}  \right) \\
\label{nowweclaimthat}
 & \lesssim \sqrt{\Xzerolesslessk(\tau_0,\tau_1)}\,\left(  \sup_{\tau\in \mathbb R}  \|  \chi_{\tau,\tau+2}\chi_{\tau_0,\tau_1+1}  \Psi \|_{H^{k-1,2}(\mathcal{M}\cap\{r\le R_{\rm freq}\})} +  \|  \chi_{\tau_0,\tau_1+1}  \Psi \|_{H^{k-3,2}(\mathcal{M}\cap\{r\le R_{\rm freq}\})} \right) ,
\end{align}
where $\Psi= \chi_{\tau_0,\tau_1+1}\widehat\psi$.

Plugging~\eqref{nowweclaimthat} into~\eqref{whence} we obtain 
\begin{equation}
\label{intheabovehere}
\epsilon(\tau_1-\tau_0)\sup_{\tau_0\le \tau' \le \tau_1}   \Ezerok(\tau')[\psi]+
\epsilon^{-1}\Xzerolesslessk(\tau_0,\tau_1)[\psi]\cdot  
 {}^\rho \Xzerok(\tau_0,\tau_1)[\psi],
\end{equation}
where we have used also Proposition~\ref{whenoneextends}.
Now, as before, we remark that 
if $\Xzerolesslessk(\tau_0,\tau_1)=0$ then $\psi=0$ identically and the proposition trivially follows. Otherwise, choose $\epsilon= (\tau_1-\tau_0)^{-\frac12} \sqrt{\Xzerolesslessk(\tau_0,\tau_1)}$ in~\eqref{intheabovehere}
and we see that both terms are estimated by~\eqref{willaddanumberhere}.

The remaining terms~\eqref{secondtermandthusasagain} and~\eqref{thirdtermandthusasagain} 
can be similarly estimated as in Section~\ref{Dkcommutingtermsnonlinearestimate}
by the right hand side of~\eqref{willaddanumberhere} 
and lower order terms may obviously also be estimated similarly.
\end{proof}

\subsubsection{Terms arising from the semilinearity in $r\le R_{\rm freq}$}

\begin{proposition}
\label{fromthesemihereprop}
The terms on line~\eqref{andonemoreinhomogtermheresmooth} may be estimated:
\[
\left| \int_{\mathcal{M}\cap\{r\le R_{\rm freq}\}} 
\tilde\chi_{\tilde\tau_0,\tilde\tau_1}^4
 {\rm Re} \left(
H^{ {\rm main},n }[\widehat\psi_{n, {\bf k}}] \overline{F^{\rm semi}_{n,{\bf k}}}\right)\right|
 \lesssim
 \sqrt{ \tau_1-\tau_0 }
\sup_{\tau_0\le \tau' \le \tau_1}  \Ezerok(\tau') \sqrt{\,\,\Xzerolesslessk(\tau_0,\tau_1) }.
 \]
\end{proposition}
\begin{proof}
The highest order terms in the integrand on the left hand side take the form
\[
\tilde\chi_{\tilde\tau_0,\tilde\tau_1}^4\partial P_n \mathring{\mathfrak{D}}^{\bf k}  ( \chi_{\tau_0,\widehat\tau_1} ^2 \widehat\psi)  \cdot
 P_n\mathring{\mathfrak{D}}^{\bf k} \chi^2_{\tau_0,\widehat\tau_1} N(\partial \chi^2_{\tau_1} \psi,  \chi^2_{\tau_1} \psi, x).
\]

We once again partition as
\begin{align}
\label{firsttermandthusasyetagain}
\tilde\chi_{\tilde\tau_0,\tilde\tau_1}^4
\partial P_n \mathring{\mathfrak{D}}^{\bf k}  ( \chi^2_{\tau_0,\tau_1+3}  \widehat\psi) \cdot
P_n  \mathring{\mathfrak{D}}^{\bf k} \chi^2_{\tau_0,\widehat\tau_1} N(\partial \chi^2_{\tau_1} \psi,  \chi^2_{\tau_1} \psi, x) \\
\label{secondtermandthusasyetagain}
+ \tilde\chi_{\tilde\tau_0,\tilde\tau_1}^4\chi^2_{\tau_0,\tau_1+2}  \partial P_n \mathring{\mathfrak{D}}^{\bf k}  ( \chi^2_{\tau_1+3,\widehat\tau_1}  
\widehat\psi) \cdot
P_n \mathring{\mathfrak{D}}^{\bf k} \chi^2_{\tau_0,\widehat\tau_1} N(\partial \chi^2_{\tau_1} \psi,  \chi^2_{\tau_1} \psi, x)\\
\label{thirdtermandthusasyetagain}
+ \partial P_n \mathring{\mathfrak{D}}^{\bf k}  ( \chi^2_{\tau_1+3,\widehat\tau_1} \widehat \psi) \cdot
\tilde\chi_{\tilde\tau_0,\tilde\tau_1}^4\chi^2_{\tau_1+2,\widehat \tau_1} P_n \mathring{\mathfrak{D}}^{\bf k} \chi^2_{\tau_0,\widehat\tau_1} N(\partial \chi^2_{\tau_1} \psi,  \chi^2_{\tau_1} \psi, x).
\end{align}

We consider the first term~\eqref{firsttermandthusasyetagain}. 
We may argue as before
via Cauchy--Schwarz to estimate by:
\[
\epsilon \| P_n \mathring{\mathfrak{D}}^{\bf k}  ( \chi_{\tau_0,\tau_1+3} ^2 \widehat\psi)\|^2_{L^2(\mathcal{M}\cap \{r\le R_{\rm freq} \})} +\epsilon^{-1}
\|  P_n \mathring{\mathfrak{D}}^{\bf k}  \chi^2_{\tau_0,\widehat\tau_1} N(\partial \chi^2_{\tau_1} \psi,  \chi^2_{\tau_1} \psi, x)\|^2_{L^2(\mathcal{M}\cap \{r\le R_{\rm freq} \})}
\]
and thus simply by the $L^2$ mapping properties of $P_n$ (Proposition~\ref{mixedsobolevboundednessprop} with $s=0=k=j$)
\[
\epsilon \|  \mathring{\mathfrak{D}}^{\bf k}  ( \chi_{\tau_0,\tau_1+3} ^2 \psi)\|^2_{L^2(\mathcal{M}\cap \{r\le R_{\rm freq} \})} +\epsilon^{-1}
\|   \mathring{\mathfrak{D}}^{\bf k} \chi_{\tau_0,\tau_1+3}^2   N(\partial \chi^2_{\tau_1} \psi,  \chi^2_{\tau_1} \psi, x)  \|^2_{L^2(\mathcal{M}\cap \{r\le R_{\rm freq} \})}
\]
and hence
\begin{equation}
\label{andhenceedw}
\epsilon (\tau_0-\tau_1) \sup_{\tau_0\le \tau' \le \tau_1} \Ezerok(\tau')
+ \epsilon^{-1} \Xzerolesslessk(\tau_0,\tau_1) \,  \sup_{\tau_0\le \tau' \le \tau_1} \Ezerok(\tau'),
\end{equation}
where we have used the Sobolev inequality of Proposition~\ref{largeasobolevforfunctions}.
Again, noting that if $\Xzerolesslessk(\tau_0,\tau_1)=0$ then $\psi=0$ and the desired
statement is trivially true, the term~\eqref{firsttermandthusasyetagain} is estimated
applying~\eqref{andhenceedw} with $\epsilon = (\tau_1-\tau_0)^{-\frac12}\sqrt{ \Xzerolesslessk(\tau_0,\tau_1)}$.

The terms~\eqref{secondtermandthusasyetagain} and~\eqref{thirdtermandthusasyetagain} are 
again estimated as in Section~\ref{Dkcommutingtermsnonlinearestimate}.  

Lower order terms are of course estimated similarly.
\end{proof}
\subsubsection{Terms in $r\ge R_{\rm freq}$ and terms from $H^{\rm redder}$}

The terms in the region $r\ge R_{\rm freq}$ can be estimated exactly as in~\cite{DHRT22}.

We first estimate the terms in the subregion $R_{\rm freq}\le r\le R$:
\begin{proposition}
\label{insertareferencetorefer}
Restricted to  $R_{\rm freq}\le r\le R$, the terms on line~\eqref{oldterms} may be estimated:
\begin{align}
\nonumber
  \int_{\mathcal{R}(\tilde\tau_0,\tilde\tau_1)\cap \{R_{\rm freq} \le r\le R\}} |\Hap^{{\rm total}}[\mathfrak{D}^{\bf k} \widehat\psi] 
  \cdot [\mathfrak{D}^{\bf k}, \Box_{g(\chi^2_{\tau_1}\psi_,x)}-\Box_{g_{a,M}}] \psi |
  + |\Hap^{{\rm total}}[\mathfrak{D}^{\bf k} \widehat\psi] \cdot  \mathfrak{D}^{\bf k} N(\partial \chi^2_{\tau_1} \psi,  \chi^2_{\tau_1} \psi, x) |\\
  \label{thatrighthandsidealready}
  \lesssim  \sqrt{\tau_1-\tau_0 }
\sup_{\tau_0\le \tau' \le \tau_1}  \Ezerok(\tau') \sqrt{\,\,\Xzerolesslessk(\tau_0,\tau_1) }.
\end{align}
\end{proposition}
\begin{proof}
This follows as in~\cite{DHRT22}. Note that we could also of course  estimate in this region by
the ``better'' quantity
\[
\sqrt{\Xzerolesslessk(\tau_0,\tau_1)}\, {}^\rho\Xzerok(\tau_0,\tau_1),
\]
but there is no gain as we already must include the right hand 
side of~\eqref{thatrighthandsidealready} 
on the right hand side of our estimates.
\end{proof}

We now estimate the terms in the subregion $r\ge R$. 
Recall that $(\Box_{g(\chi^2_{\tau_1}\psi_,x)}-\Box_{g_{a,M}}) \psi $ vanishes  here so it is only the semilinear
term which must be estimated. 
It is here of course that the null structure 
Assumption~\ref{largeanullcondassumption} enters:
 \begin{proposition}
 \label{lastinhomogprop}
 Restricted to  $r\ge R$, the terms on line~\eqref{oldterms} may be estimated:
\begin{align*}
 & \int_{\mathcal{R}(\tilde\tau_0,\tilde\tau_1)\cap \{r\ge R\}} 
 |\Hap^{{\rm total}}[\mathfrak{D}^{\bf k} \widehat\psi] \cdot   \mathfrak{D}^{\bf k} N(\partial \chi^2_{\tau_1} \psi,  \chi^2_{\tau_1} \psi, x)| \\
 &\qquad\qquad
 \lesssim
  \Xpk_{\frac{8R}9}(\tau_0,\tau_1) \sqrt{\Xzerolesslessk_{\frac{8R}9}(\tau_0,\tau_1)}  +\sqrt{\Xpk_{\frac{8R}9}(\tau_0,\tau_1)}\sqrt{\Xzerok_{\frac{8R}9}(\tau_0,\tau_1)}\sqrt{\Xplesslessk_{\frac{8R}9}(\tau_0,\tau_1)} \, .
\end{align*}
\end{proposition}
\begin{proof}
As in~\cite{DHRT22}, this
 follows immediately from the properties of the current $\Hap^{\rm total}$ and
from Assumption~\ref{largeanullcondassumption}.
\end{proof}

Clearly now, the terms on line~\eqref{oldtermsredder} 
arising from $H^{\rm redder}$
can be estimated exactly as in Proposition~\ref{insertareferencetorefer},
i.e.~we have
\begin{proposition}
\label{insertareferencetoreferredder}
The terms on line~\eqref{oldtermsredder} may be estimated:
\begin{align}
\nonumber
  \int_{\mathcal{R}(\tilde\tau_0,\tilde\tau_1)\cap \{r_0 \le r\le r_1\}} |H^{{\rm redder}}[\mathfrak{D}^{\bf k} \widehat\psi] 
  \cdot [\mathfrak{D}^{\bf k}, \Box_{g(\chi^2_{\tau_1}\psi_,x)}-\Box_{g_{a,M}}] \widehat\psi |
  + |H^{{\rm redder}}[\mathfrak{D}^{\bf k}\widehat \psi] \cdot  \mathfrak{D}^{\bf k} N(\partial \chi^2_{\tau_1} \psi,  \chi^2_{\tau_1} \psi, x)  |\\
  \label{thatrighthandsidealreadyredder}
  \lesssim  \sqrt{\tau_1-\tau_0} 
\sup_{\tau_0\le \tau' \le \tau_1}  \Ezerok(\tau') \sqrt{\,\,\Xzerolesslessk(\tau_0,\tau_1) }.
\end{align}
\end{proposition}

\subsubsection{An alternative estimate for the nonlinear terms}
Finally, we note the following alternative estimate for the non-linear terms:
\begin{proposition}
\label{onlyuseforboundednesstoporder}
Let $\tau_0=\tilde\tau_0\le \tilde\tau_1\le \cdots\le \tilde\tau_q=\tau_1$ be an arbitrary partition
of the interval $[\tau_0,\tau_1]$ such that $\tilde\tau_i-\tilde\tau_{i-1}\ge 2$.
We may replace the right hand side of the estimates
in Propositions~\ref{nonlincomprophereref}--\ref{fromthesemihereprop} 
with the expression
\begin{equation}
\label{theexpressiontobereplaced}
{}^\rho \Xzerok(\tau_0,\tau_1)[\psi] \sqrt{\Xzerolesslessk(\tau_0,\tau_1)[\psi] }+
\sum_{i=0}^{q-1}\sqrt{\tilde\tau_{i+1}-\tilde\tau_i}\,\,
{}^\rho \Xzerok(\tilde\tau_i,\tilde\tau_{i+1})[\psi]\sqrt{\Xzerolesslessk(\tilde\tau_i,\tilde\tau_{i+1})[\psi] }
\end{equation}
while we may replace the right  hand side of Proposition~\ref{insertareferencetorefer}
and~\ref{insertareferencetoreferredder} 
with the $\sum_{i=0}^{q-1}$
of the analogous expression with $(\tilde\tau_0,\tilde\tau_1)$ replaced by
$(\tilde\tau_{i},\tilde\tau_{i+1})$. 
\end{proposition}
\begin{proof}
The second statement is clear from the fact that no pseudodifferential operators are applied
and thus the statement arises from a local estimate in physical space. The first statement,
on the other hand follows as before from pseudodifferential commutator estimates, 
using now also the dependence of the constant $C$ in Proposition~\ref{smoothingprop}
on the distance between the supports of the cutoffs.
\end{proof}

We shall only use Proposition~\ref{onlyuseforboundednesstoporder}
to obtain the  statement in the last sentence of Proposition~\ref{largeafinal},  
which in turn will only be used to
replace  top order slow polynomial growth with uniform boundedness  in the statement
of our main theorem.

\subsection{Putting it all together: the final estimates}
\label{puttogethersec}

We may now state a closed system of estimates involving our master
energies:

\begin{proposition} 
\label{largeafinal}
Fix $0\ne |a|<M$, and let $(\mathcal{M},g_{a,M})$ be the sub-extremal Kerr manifold of 
Section~\ref{subextremalkerrsec}  
and consider equation~\eqref{theequationforlargea} satisfying the assumptions
of Section~\ref{generalclassconsideredrecalled}.

Then for all $k\ge 4$ sufficiently large,  there exist
constants $C>0$ (implicit in the $\lesssim$ below)
and $\varepsilon_{\rm prelim}>0$, 
 such that the following is true.

Let $\tau_0+2 \le \tau_1$ and 
let $\psi$ be a solution of~\eqref{theequationforlargea} in $\mathcal{R}(\tau_0,\tau_1)$ 
satisfying~\eqref{newbasicbootstrap} 
 for some $0<\varepsilon\le \varepsilon_{\rm prelim}$.
Then for $\delta \le p \le 2-\delta$, one has the estimates
\begin{align}
\nonumber
 {}^\rho \Xpk(\tau_0,\tau_1)  \nonumber
  &\lesssim \Epk(\tau_0) +
{}^{\chi\natural}_{\scalebox{.6}{\mbox{\tiny{\boxed{+1}}}}}\, \Xzerokminusone(\tau_0,\tau_1)  \\
\nonumber
&\qquad+
 \Xpk_{\frac{8R}9}(\tau_0,\tau_1) \sqrt{\Xzerolesslessk_{\frac{8R}9}(\tau_0,\tau_1)}  +\sqrt{\Xpk_{\frac{8R}9}(\tau_0,\tau_1)}\sqrt{\Xzerok_{\frac{8R}9}(\tau_0,\tau_1)}\sqrt{\Xplesslessk_{\frac{8R}9}(\tau_0,\tau_1)}
\\
 \label{toporderherelastterm}
 &\qquad +\sqrt{\tau_1-\tau_0} \,\,{}^\rho \Xzerok(\tau_0,\tau_1) \sqrt{\, \Xzerolesslessk(\tau_0,\tau_1) } \, ,
 \end{align}
\begin{align}
\nonumber
{}^{\chi\natural}_{\scalebox{.6}{\mbox{\tiny{\boxed{+1}}}}}\, \Xpkminusone(\tau_0,\tau_1) &\lesssim \Epkminusone(\tau_0) +
 \Xpkminusone_{\frac{8R}9}(\tau_0,\tau_1) \sqrt{\Xzerolesslessk_{\frac{8R}9}(\tau_0,\tau_1)} 
 + \sqrt{\Xpkminusone_{\frac{8R}9}(\tau_0,\tau_1)}\sqrt{\Xzerokminusone_{\frac{8R}9}(\tau_0,\tau_1)}
 \sqrt{\Xplesslessk_{\frac{8R}9}(\tau_0,\tau_1)}
 \\
\label{notetheextraterm}
 &\qquad
+\sqrt{\tau_1-\tau_0}\sup_{\tau_0\le \tau'\le \tau_1}  \Ezerok(\tau') \sqrt{\, \Xzerolesslessk(\tau_0,\tau_1) },
\end{align}
while for $p=0$, one has the estimates
\begin{align}
\label{needanameforthis}
 {}^\rho \Xzerok(\tau_0,\tau_1) &\lesssim \Ezerok(\tau_0) +
{}^{\chi\natural}_{\scalebox{.6}{\mbox{\tiny{\boxed{+1}}}}}\, \Xzerokminusone(\tau_0,\tau_1)+
 \Xzeroplusk_{\frac{8R}9} (\tau_0,\tau_1) \sqrt{ \Xzeropluslesslessk_{\frac{8R}9}(\tau_0,\tau_1)}
+  \sqrt{\tau_1-\tau_0} \,\,{}^\rho \Xzerok(\tau_0,\tau_1) \sqrt{\, \Xzerolesslessk(\tau_0,\tau_1) } , \\
\label{ditto}
{}^{\chi\natural}_{\scalebox{.6}{\mbox{\tiny{\boxed{+1}}}}}\, \Xzerokminusone(\tau_0,\tau_1) &\lesssim \Ezerokminusone(\tau_0) +
 \Xzeroplusk_{\frac{8R}9} (\tau_0,\tau_1) \sqrt{ \Xzeropluslesslessk_{\frac{8R}9}(\tau_0,\tau_1)}
+
\sqrt{\tau_1-\tau_0}\sup_{\tau_0\le \tau'\le \tau_1}  \Ezerok(\tau') \sqrt{\, \Xzerolesslessk(\tau_0,\tau_1) }
.
\end{align}

Moreover, we may replace the expression on lines~\eqref{toporderherelastterm} 
and~\eqref{notetheextraterm} 
with~\eqref{theexpressiontobereplaced} (and similarly in the $p=0$ identities).
\end{proposition}

\begin{proof}
Inequalities~\eqref{toporderherelastterm} and~\eqref{needanameforthis} follow by applying Proposition~\ref{toporderestimateprop},
plugging in the bounds of 
Propositions~\ref{fromthecutoffbounding}--\ref{insertareferencetoreferredder}
of Section~\ref{estimforinhomog}.

The remaining inequalities~\eqref{notetheextraterm},~\eqref{ditto}
follow as in~\cite{DHRT22}.

The last sentence of the statement of the proposition follows from the alternative
bound of Proposition~\ref{onlyuseforboundednesstoporder}.
\end{proof}

\begin{remark}
\label{aequalszeroremark}
By the results of~\cite{DHRT22}, in the case $a=0$, the above Proposition holds as stated where 
${}^{\chi\natural}_{\scalebox{.6}{\mbox{\tiny{\boxed{+1}}}}}\, \Xpkminusone(\tau_0,\tau_1)$ is
replaced by ${}^{\chi}\, \Xpkminusone(\tau_0,\tau_1)$ for all $p$ as above.
In fact, by the results of~\cite{DHRT22},
the Proposition holds as stated 
with this substitution for $|a|\le a_{\rm small}$ for some $a_{\rm small}(M)>0$, 
with constants independent of $a$ in this range.
\end{remark}

\section{The main theorem: global existence and stability for quasilinear waves on Kerr in the full subextremal range
 $|a|<M$}
\label{largeamaintheoremsection}
We may now state the precise version of the  main result of this paper.

\begin{theorem}[Global existence and stability]
\label{largeamaintheorem}
Fix $|a|<M$, and let $(\mathcal{M},g_{a,M})$ be the sub-extremal Kerr manifold of 
Section~\ref{subextremalkerrsec}, 
and consider equation~\eqref{theequationforlargea} with $g(\psi, x)$, $N(\partial\psi,\psi,x)$
satisfying the assumptions
of Section~\ref{generalclassconsideredrecalled}.

There exists a positive integer $k_{\rm global}\ge k_{\rm local}$ sufficiently large,
such that, given $k\ge k_{\rm global}$, there exists a positive $0<\varepsilon_{\rm global}<\varepsilon_{\rm loc}$ sufficiently small,
and a positive  constant $C>0$ sufficiently large, so that the following holds.

Fix $\tau_0=1$
and consider as in Proposition~\ref{localexistencelargea}
initial data $(\uppsi,\uppsi')$ on $\Sigma(\tau_0)$ for~\eqref{theequationforlargea} satisfying
\[
\Eonek[\uppsi,\uppsi'] +  \Etwominusdelkminustwo\,[\uppsi,\uppsi']
 \leq\varepsilon_0\leq \varepsilon_{\rm global}.
\]
Then the $\psi$ given by Proposition~\ref{localexistencelargea} 
extends to a unique globally defined solution of~\eqref{theequationforlargea} in $\mathcal{R}(\tau_0,\infty)$,
satisfying the estimates
\begin{equation}
\label{basicestimatelargea}
{}^\rho\Xonek(\tau_0,\tau)[\psi]+
{}^\rho\Xtwominusdelkminustwo(\tau_0,\tau)[\psi]\leq  C\varepsilon_0,
\end{equation}
for all $\tau\ge \tau_0$, in particular 
\begin{equation}
\label{boundednessstafin}
\Eonek(\tau)[\psi]+
\Etwominusdelkminustwo(\tau)[\psi] \leq  C\varepsilon_0.
\end{equation}

The solution will satisfy moreover additional power law decay  estimates, for instance~\eqref{powerlawdecaystatements}.
\end{theorem}

\begin{remark}
\label{aquestionofuniformity}
For fixed $M>0$, 
choosing the ambient structure of Section~\ref{coordsheresec} 
to depend smoothly on $a\in (-M,M)$,
it follows that the constants $\varepsilon_{\rm global}$ and $C$ above can be chosen uniformly for all $a\in [-a_0,a_0]$ for any $0<a_0<M$,
i.e.~we may write $C=C(a_0, M)$, $\varepsilon_{\rm global}=\varepsilon_{\rm gobal}(a_0,M)$.
This will follow because  $C$ and $\varepsilon_{\rm global}$ will ultimately depend on the constants
of Proposition~\ref{largeafinal}, and these constants may
be chosen to depend smoothly on $a$ for $0<|a|<M$,
while for $|a|\le a_{\rm small}$  we may apply Remark~\ref{aequalszeroremark}.
Note, on the other hand, that we have no control over $C(a_0,M)$ as $a_0\to M$, even
for the linear wave equation.
\end{remark}

\section{Proof of global existence and decay via iteration on consecutive spacetime slabs}
\label{proofofmainthelargea}

We will prove in this section Theorem~\ref{largeamaintheorem},
following a version of the scheme of Section 6 
of~\cite{DHRT22} discussed already in Section~\ref{dyadicintro}  of the introduction.

We begin in Section~\ref{revisedestimatehierarchy} by reinterpreting our main estimates of Section~\ref{puttogethersec} as a suitable
$p$-hierarchy. We then infer a global existence statement on fixed ${\rm L}$-slabs in 
Section~\ref{globular}.  
In Section~\ref{pigeonsheretoo}, we show that by strengthening the initial data assumptions on such
a slab, we may obtain better 
estimates at the top of the slab via a pigeonhole argument, leading to an iterating hierarchy. 
This allows us in Section~\ref{simpleiterationproof} 
to deduce global existence, boundedness and decay  using an elementary iteration.

In what follows, we recall that Proposition~\ref{largeafinal} holds also for the $a=0$ case
with the notational substitution of Remark~\ref{aequalszeroremark}, and that this substitution may
also be applied in the case $|a|\le a_{\rm small}(M)$.

\subsection{A revised estimate hierarchy}
\label{revisedestimatehierarchy}

We first reinterpret Proposition~\ref{largeafinal} in the form of a $p$-hierarchy of
estimates.

\begin{proposition}
\label{revisedhierarchy}
Let $(\mathcal{M},g_{a,M})$ be the sub-extremal Kerr manifold of 
Section~\ref{subextremalkerrsec} for $|a|<M$, 
and consider equation~\eqref{theequationforlargea} satisfying the assumptions
of Section~\ref{generalclassconsideredrecalled}.

Let $k$ be sufficiently large.
There exist constants $C>0$ implicit in the $\lesssim$ notation 
and an $\varepsilon_{\rm boot}>0$ sufficiently small 
such that the following is true.

Consider a region $\mathcal{R}(\tau_0,\tau_1)$
and a $\psi$ solving~\eqref{theequationforlargea} on  $\mathcal{R}(\tau_0,\tau_1)$, satisfying 
moreover~\eqref{newbasicbootstrap} and $0< \varepsilon\le \varepsilon_{\rm boot}$.
Let us assume moreover that 
\[
\tau_1= \tau_0+{\rm L}
\]
for some arbitrary ${\rm L}\ge 2$.
We have the following hierarchy of inequalities:
\begin{eqnarray}
 \label{largeapweightedtopsummed}
{}^\rho\Xtwominusdelk  +  {}^{\chi\natural}_{\scalebox{.6}{\mbox{\tiny{\boxed{+1}}}}}\,\Xzerokminusone 
&\lesssim& \Etwominusdelk(\tau_0)  +   {}^\rho\Xtwominusdelk\sqrt{\Xzerolesslessk} + \sqrt{{}^\rho\Xtwominusdelk}\sqrt{ {}^\rho\Xzerok}
\sqrt{\Xtwominusdeltalesslessk}+
 {}^\rho\Xzerok\sqrt{\Xzerolesslessk}\sqrt{\rm L} \, , 
  \\
\label{largeapminusoneweightedtopsummed}
{}^\rho \Xonek   +  {}^{\chi\natural}_{\scalebox{.6}{\mbox{\tiny{\boxed{+1}}}}}\,\Xzerokminusone    &\lesssim& \Eonek(\tau_0) 
     + {}^\rho \Xonek
\sqrt{\Xonelesslessk} 
 +  {}^\rho\Xzerok\sqrt{\Xzerolesslessk}\sqrt{{\rm L}}   \,  ,  \\
 \label{largeanonpweightedtopsummed}
{}^\rho \Xzerok + {}^{\chi\natural}_{\scalebox{.6}{\mbox{\tiny{\boxed{+1}}}}}\, \Xzerokminusone  &\lesssim& 
\Ezerok(\tau_0)
+ \left( \,{}^\rho \Xzerok + (\,{}^\rho \Xzerok)^{\frac{1-\delta}{1+\delta}} (\, {}^\rho \Xonek)^{\frac{2\delta}{1+\delta}} \right)
\sqrt{\Xzerolesslessk+ (\,\Xzerolesslessk)^{\frac{1-\delta}{1+\delta}} (\, \Xonelesslessk)^{\frac{2\delta}{1+\delta}}}+  {}^\rho\Xzerok\sqrt{\Xzerolesslessk}\sqrt{{\rm L}}   \,   .
\end{eqnarray}
\end{proposition}

\begin{proof}

From Proposition~\ref{largeafinal} we immediately deduce the following inequalities:
\begin{eqnarray}
 \label{largeapweightedtop}
{}^\rho\Xtwominusdelk
&\lesssim& \Etwominusdelk(\tau_0)  +  {}^{\chi\natural}_{\scalebox{.6}{\mbox{\tiny{\boxed{+1}}}}}\,\Xzerokminusone  +   {}^\rho\Xtwominusdelk\sqrt{\Xzerolesslessk} + \sqrt{{}^\rho\Xtwominusdelk}\sqrt{ {}^\rho\Xzerok}
\sqrt{\Xtwominusdeltalesslessk}+
 {}^\rho\Xzerok\sqrt{\Xzerolesslessk}\sqrt{{\rm L}} \, , 
  \\
  \label{largeapweightednottop}
 {}^{\chi\natural}_{\scalebox{.6}{\mbox{\tiny{\boxed{+1}}}}}\,\Xtwominusdelkminusone
&\lesssim& \Etwominusdelkminusone(\tau_0)  +    {}^\rho\Xtwominusdelkminusone\sqrt{\Xzerolesslessk} + \sqrt{{}^\rho\Xtwominusdelkminusone}\sqrt{ {}^\rho\Xzerokminusone}
\sqrt{\Xtwominusdeltalesslessk}+
 {}^\rho\Xzerok\sqrt{\Xzerolesslessk}\sqrt{{\rm L}} 
  , 
  \\
\label{largeapminusoneweightedtop}
{}^\rho \Xonek&\lesssim& \Eonek(\tau_0) 
  +  {}^{\chi\natural}_{\scalebox{.6}{\mbox{\tiny{\boxed{+1}}}}}\,\Xzerokminusone        + {}^\rho \Xonek
\sqrt{\Xonelesslessk} 
 +  {}^\rho\Xzerok\sqrt{\Xzerolesslessk}\sqrt{{\rm L}}   \,  ,  \\
 \label{largeapminusoneweightednottop}
{}^{\chi\natural}_{\scalebox{.6}{\mbox{\tiny{\boxed{+1}}}}}\,\Xonekminusone&\lesssim& \Eonekminusone(\tau_0) 
  +    {}^\rho \Xonekminusone
\sqrt{\Xzerolesslessk} 
 +  {}^\rho\Xzerok\sqrt{\Xzerolesslessk}\sqrt{{\rm L}} 
    ,  \\
\label{largeanonpweightedtop}
{}^\rho \Xzerok &\lesssim& 
\Ezerok(\tau_0)
 +{}^{\chi\natural}_{\scalebox{.6}{\mbox{\tiny{\boxed{+1}}}}}\, \Xzerokminusone 
+ \left( \,{}^\rho \Xzerok + (\,{}^\rho \Xzerok)^{\frac{1-\delta}{1+\delta}} (\, {}^\rho \Xonek)^{\frac{2\delta}{1+\delta}} \right)
\sqrt{\Xzerolesslessk+ (\,\Xzerolesslessk)^{\frac{1-\delta}{1+\delta}} (\, \Xonelesslessk)^{\frac{2\delta}{1+\delta}}}+  {}^\rho\Xzerok\sqrt{\Xzerolesslessk}\sqrt{{\rm L}}   \,   , 
 \\
\label{largeanonpweightednottop}
{}^{\chi\natural}_{\scalebox{.6}{\mbox{\tiny{\boxed{+1}}}}}\,\Xzerokminusone &\lesssim& 
\Ezerokminusone(\tau_0)
+  \left( \,{}^\rho \Xzerokminusone + (\,{}^\rho \Xzerokminusone)^{\frac{1-\delta}{1+\delta}} (\, {}^\rho \Xonekminusone)^{\frac{2\delta}{1+\delta}} \right)
\sqrt{\Xzerolesslessk+ (\,\Xzerolesslessk)^{\frac{1-\delta}{1+\delta}} (\, \Xonelesslessk)^{\frac{2\delta}{1+\delta}}}
+  {}^\rho\Xzerok\sqrt{\Xzerolesslessk}\sqrt{{\rm L}}
,
\end{eqnarray}
where we have also applied the interpolation estimate~\eqref{largeainterpolationstatementbulk}
to the $p=0$ estimates, and we have estimated 
\[
 \Xpk_{\frac{8R}9}(\tau_0,\tau_1) \lesssim {}^\rho \Xpk(\tau_0,\tau_1) , {\rm\ etc.}
\]
We now sum pairwise, where we weight the second summand with a sufficiently large multiple to absorb the second term on the right hand side of the first summand.
\end{proof}

\subsection{Global existence in an ${\rm L}$-slab}
\label{globular}

It is now  elementary to infer a global existence result on ${\rm L}$-slabs.

\begin{proposition}
\label{givesoneinaslab}
Let $(\mathcal{M},g_{a,M})$ be the sub-extremal Kerr manifold of 
Section~\ref{subextremalkerrsec} for $|a|<M$, 
and consider equation~\eqref{theequationforlargea} satisfying the assumptions
of Section~\ref{generalclassconsideredrecalled}.

Let $k\ge k_{\rm loc}$ be sufficiently large. 
Then  there exists a positive constant $\varepsilon_{\rm slab}\leq\varepsilon_{\rm loc}$ and a constant $C>0$ implicit in the sign $\lesssim$ below such that the following  is true.

Given arbitrary ${\rm L}\ge 2$,
 $\tau_0\ge 0$,  $0<\varepsilon_0\leq \varepsilon_{\rm slab}$ and initial data $(\uppsi,\uppsi')$ 
 on $\Sigma(\tau_0)$ as in Proposition~\ref{localexistencelargea}, satisfying moreover
\begin{equation}
\label{assumptionondataherelargea}
\Eonekminusfour(\tau_0)[\psi]	 \leq \varepsilon_0, 
\qquad  \Ezerokminussix(\tau_0)[\psi]	 \leq \varepsilon_0{\rm L}^{-1}, 
\end{equation}
then the unique solution of Proposition~\ref{localexistencelargea} 
obtaining the data can be extended to 
a $\psi$
defined on the entire  spacetime slab $\mathcal{R}(\tau_0,\tau_0+{\rm L})$
satisfying the equation~\eqref{theequationforlargea}
and
 the estimates
 \begin{equation}
 \label{newformatestimateslargea}
 {}^\rho  \Xonekminusfour [\psi]	+{}^{\chi\natural}_{\scalebox{.6}{\mbox{\tiny{\boxed{+1}}}}}\,\Xonekminusfive [\psi]	\lesssim \varepsilon_0,
\qquad {}^\rho \Xzerokminussix [\psi]	 +{}^{\chi\natural}_{\scalebox{.6}{\mbox{\tiny{\boxed{+1}}}}}\, \Xzerokminusseven[\psi]	 \lesssim \varepsilon_0{\rm L}^{-1}.
 \end{equation}
\end{proposition}

\begin{remark}
We have labelled the highest order energy in the above by $k-4$ instead of $k$ 
to be consistent with the way this proposition will be used in the section to follow.
\end{remark}

\begin{proof}
It will be convenient to do a continuity argument \emph{in the data} rather than \emph{in time}, i.e.~let us replace
the data with $s\uppsi$, $s\uppsi'$ for $s\in [0,1]$.  Let us denote the resulting solution
$\psi_s$.   Note that the statement of the proposition trivially holds for $\psi_0$ since
$\psi_0=0$. 
Given $\varepsilon>0$ sufficiently small, let us consider the set $S$ consisting of all $s\in [0,1]$ such that the
solution $\psi_s$  exists in $\mathcal{R}(\tau_0,\tau_1)$ and satisfies 
\begin{equation}
\label{assumptionondataherelargeaboot}
\Eonelesslessk(\tau)[\psi]	 \leq \varepsilon, 
\qquad  \Ezerolesslessk(\tau)[\psi]	 \leq \varepsilon {\rm L}^{-1}
\end{equation}
for all $\tau_0\le \tau\le \tau_1$.
Note that~\eqref{assumptionondataherelargeaboot} in particular implies 
assumption~\eqref{newbasicbootstrap}.
We have just shown that $0\in S$ and thus $S$ is nonempty, 
and we note that by the Cauchy stability statement of Proposition~\ref{localexistencelargea}, 
the set  $S$ is manifestly closed. 
We will show below that the statement of the proposition holds for $\psi_s$ if $s\in S$,
and that moreover, for $s\in S$,~\eqref{assumptionondataherelargeaboot} holds with the improved constant
$\varepsilon/2$ in place of $\varepsilon$.
It then follows again by the Cauchy stability statement of 
Proposition~\ref{localexistencelargea} that the set $S$ is open, and thus coincides
with $[0,1]$.

Since~\eqref{assumptionondataherelargeaboot} (and thus also~\eqref{newbasicbootstrap}) holds, if $k$ is sufficiently large we may 
indeed apply~\eqref{largeapminusoneweightedtopsummed} of Proposition~\ref{revisedhierarchy}
for $k-4$, absorb (using the smallness of $\varepsilon$) the last two terms on the right hand side of~\eqref{largeapminusoneweightedtopsummed} into the first term on the left hand side,
 and obtain the first inequality of~\eqref{newformatestimateslargea} for $\psi_s$. Then, 
we may apply~\eqref{largeanonpweightedtopsummed} 
 for $k-6$  and obtain the second inequality of~\eqref{newformatestimateslargea} for $\psi_s$,
 again using the smallness of $\varepsilon$ and the ``bootstrap'' assumption~\eqref{assumptionondataherelargeaboot} to absorb the nonlinear terms.
This yields the statement of the Proposition for $\psi_s$, including~\eqref{newformatestimateslargea}.
Finally, for sufficiently small $\varepsilon_{\rm slab}\ll \varepsilon$,
then the inequalities~\eqref{newformatestimateslargea}
improve the constant of~\eqref{assumptionondataherelargeaboot} 
to say $\varepsilon/2$.
We thus indeed obtain openness,
and the Proposition follows as stated. 
\end{proof}

\subsection{The pigeonhole argument}
\label{pigeonsheretoo}

The key to our iterative argument for global existence and decay is the following
pigeonhole argument:

\begin{proposition}
\label{pigeonsforlargea}
Fix $\beta>0$ such that $\beta+\delta <1$.

Under the assumptions of Proposition~\ref{givesoneinaslab}, 
there exists a constant $C>0$, implicit in the notation $\lesssim$ below,
a  parameter $\alpha_0\ge 3$
and, for all $\alpha\ge \alpha_0$, a parameter~$\hat{\varepsilon}_{\rm slab}(\alpha)$
such that for all $0<\hat\varepsilon_0 \le\hat{\varepsilon}_{\rm slab}(\alpha)$ the following  holds.

Let us assume in addition to~\eqref{assumptionondataherelargea} that we have
\begin{equation}
\label{newiterative}
\Eonek(\tau_0) \leq \hat\varepsilon_0 {\rm L}^{\beta}, \qquad
\Etwominusdelkminustwo(\tau_0) \leq \hat\varepsilon_0 {\rm L}^{\beta}, \qquad
\Ezerokminustwo(\tau_0)\leq\hat\varepsilon_0\alpha^{2}  {\rm L}^{-1+\beta}, \qquad
 \Eonekminusfour(\tau_0) \leq \hat\varepsilon_0 \alpha^4 {\rm L}^{-1+\delta+\beta},  \qquad
 \Ezerokminussix(\tau_0) \leq \hat\varepsilon_0 \alpha^8 {\rm L}^{-2+\delta+ \beta}.
\end{equation}

Then the solution $\psi$ of~\eqref{theequationforlargea}
on $\mathcal{R}(\tau_0,\tau_0+{\rm L})$ given by Proposition~\ref{givesoneinaslab} 
satisfies the additional estimates
\begin{equation}
\label{addedherenotfancyiii}
\sup_{\tau_0\le \tau\le \tau_0+{\rm L}} \, \Eonek(\tau) \leq \hat\varepsilon_0  
\alpha^{\beta} {\rm L}^{\beta}, \qquad  
\sup_{\tau_0\le \tau\le \tau_0+{\rm L}} \,  \Etwominusdelkminustwo(\tau)   \leq  \hat\varepsilon_0 
 \alpha^{\beta} {\rm L}^\beta,
\end{equation}
\begin{equation}
\label{kiallobasicnewformiiiagain}
{}^\rho \Xonek+{}^{\chi\natural}_{\scalebox{.6}{\mbox{\tiny{\boxed{+1}}}}}\,\Xonekminusone \lesssim \hat\varepsilon_0 {\rm L}^\beta, \qquad
{}^\rho \Xtwominusdelkminustwo+{}^{\chi\natural}_{\scalebox{.6}{\mbox{\tiny{\boxed{+1}}}}}\, \Xtwominusdelkminusthree \lesssim \hat\varepsilon_0 {\rm L}^\beta,
\end{equation}
\begin{equation}
\label{akomakiallobasicnewformiiiagain}
{}^\rho \Xzerokminustwo+{}^{\chi\natural}_{\scalebox{.6}{\mbox{\tiny{\boxed{+1}}}}}\,\Xzerokminusthree \lesssim\hat\varepsilon_0\alpha^{2} {\rm L}^{-1+\beta}, \qquad
{}^\rho \Xonekminusfour+{}^{\chi\natural}_{\scalebox{.6}{\mbox{\tiny{\boxed{+1}}}}}\,\Xonekminusfive \lesssim\hat\varepsilon_0\alpha^4 {\rm L}^{-1+\delta+\beta}, 
\qquad
{}^\rho \Xzerokminussix+{}^{\chi\natural}_{\scalebox{.6}{\mbox{\tiny{\boxed{+1}}}}}\,\Xzerokminusseven \lesssim\hat\varepsilon_0\alpha^{8} {\rm L}^{-2+\delta+\beta}.
\end{equation}

 Moreover, 
for all times $\tau'$ with ${\rm L}\ge \tau'-\tau_0 \ge {\rm L}/2$, we have that
\begin{eqnarray}
\label{bterallolargea}
\Ezerokminustwo(\tau') &\leq&   \hat\varepsilon_0 \alpha {\rm L}^{-1+\beta},\\
\label{bterjustlargea}
\Eonekminusfour(\tau') & \leq&  \hat\varepsilon_0 \alpha^{2} {\rm L}^{-1+\delta+\beta},\\
\label{btteronelargea}
\Ezerokminussix(\tau') &\leq & \hat\varepsilon_0 \alpha^4 {\rm L}^{-2+\delta+\beta}.
\end{eqnarray}
\end{proposition}

\begin{proof}
We note first that  for $\varepsilon_{\rm slab}(\alpha)$ sufficiently small, 
then~\eqref{newiterative} implies assumption~\eqref{assumptionondataherelargea} of 
Proposition~\ref{givesoneinaslab}. 

We thus have that the solution indeed exists in $\mathcal{R}(\tau_0,\tau_0+{\rm L})$ and satisfies the bounds~\eqref{newformatestimateslargea}, whence it follows
that $\Xzerolesslessk \lesssim \varepsilon  {\rm L}^{-1}$, while $\Xonelesslessk\lesssim \varepsilon$.

Turning to the hierarchy of Proposition~\ref{revisedhierarchy}, 
it follows from~\eqref{largeapminusoneweightedtopsummed}
that the first inequality of~\eqref{kiallobasicnewformiiiagain} holds. 
It now follows from~\eqref{largeanonpweightedtopsummed} that 
the first inequality of~\eqref{akomakiallobasicnewformiiiagain} holds.  
It now follow from~\eqref{largeapweightedtopsummed} that the
second inequality of~\eqref{kiallobasicnewformiiiagain} holds.

Applying~\eqref{largeapminusoneweightedtopsummed} now with $k-4$, it now follows
that the second inequality of~\eqref{akomakiallobasicnewformiiiagain} holds, while from~\eqref{largeanonpweightedtopsummed} applied now with $k-6$,
we obtain the final inequality of~\eqref{akomakiallobasicnewformiiiagain}.

Note that it now follows in particular that
\[
 \Xzerolesslessk \lesssim \varepsilon  {\rm L}^{-2+\delta+\beta}.
\]

To obtain~\eqref{bterallolargea}--\eqref{btteronelargea},  
we apply the pigeonhole principle to 
\[
\int_{\tau_0+{\rm L}/4}^{\tau_0+{\rm L}/2}  \left( {\rm L}^{-\beta}\, \Ezerokminustwo'(\tau') + {\rm L}^{-\beta}\,\,  \Eoneminusdelkminusfour' (\tau')
+\alpha^{-4} {\rm L}^{1-\delta-\beta} \,\Ezerokminussix' (\tau')\right)d\tau' \lesssim \hat\varepsilon_0 
\]
which, upon addition, follows from the  inequalities of the 
estimates~\eqref{kiallobasicnewformiiiagain}--\eqref{akomakiallobasicnewformiiiagain} 
already shown. 
Recalling from~\eqref{largeahigherorderfluxbulkrelation} that  we have

\[
\Ezerokminustwo' \gtrsim \Ezerokminustwo, \qquad \Eoneminusdelkminusfour' \gtrsim \Eoneminusdelkminusfour, \qquad
\Ezerokminussix' \gtrsim \Ezerokminussix, 
\]
we obtain 
that there exists $\tau''  \in [\tau_0,\tau_0+{\rm L}/2]$, whose precise value depends on the solution, such that
\begin{eqnarray}
\label{withimplicitconstoneiiinew}
\Ezerokminustwo(\tau'') &\lesssim& \hat{\varepsilon}_0 {\rm L}^{-1+\beta},\\
\label{withimplicitconstoneiii}
\Eoneminusdelkminusfour(\tau'') &\lesssim& \hat{\varepsilon}_0  {\rm L}^{-1+\beta},\\
\label{withimplicitconstwoiii}
\Ezerokminussix(\tau'')	&\lesssim&
\hat{\varepsilon}_0 \alpha^4 {\rm L}^{-2+\delta+\beta}.
\end{eqnarray}
Now in view of the interpolation estimate~\eqref{largeainterpolationstatementnewback} of Proposition~\ref{largeainterpolationprop}, 
we have
\begin{equation}
\label{reftosecondfactor}
\Eonekminusfour(\tau'')\lesssim  \left( \,\, \Eoneminusdelkminusfour(\tau'') \right)^{1-\delta} 
 \left( \,\, \Etwominusdelkminusfour(\tau'')\right)^{\delta}  
\leq \left( \,\, \Eoneminusdelkminusfour(\tau'') \right)^{1-\delta} 
 \left( \,\, \sup_{\tau_0\le \tau\le \tau_0+{\rm L}} \Etwominusdelkminusfour(\tau) \right)^{\delta}  
\end{equation}
 and thus
 \begin{equation}
 \label{afterinterpolationhereiii}
\Eonekminusfour(\tau'')  \lesssim \hat{\varepsilon}_0   {\rm L}^{-1+\delta+\beta} ,
\end{equation}
where we have used~\eqref{withimplicitconstoneiii} and
the estimate for the second factor on the right hand side of~\eqref{reftosecondfactor} contained in the second inequality of~\eqref{kiallobasicnewformiiiagain}.

Now
we apply~\eqref{largeapminusoneweightedtopsummed} 
and~\eqref{largeanonpweightedtopsummed}  again, with $\tau''$ in place of $\tau_0$,  
using~\eqref{withimplicitconstoneiiinew},~\eqref{afterinterpolationhereiii} and \eqref{withimplicitconstwoiii} to bound the initial data,
to obtain that 
for all $\tau_0+{\rm L}\ge \tau'\ge \tau_0+{\rm L}/2$, we have
\begin{align}
\label{withimplicitconstonelateriiinewhere}
 \Ezerokminustwo(\tau) &\lesssim  \hat{\varepsilon}_0 {\rm L}^{-1+\beta},\\
\label{withimplicitconstonelateriii}
 \Eonekminusfour(\tau) &\lesssim  \hat{\varepsilon}_0  {\rm L}^{-1+\delta+\beta}, \\
\label{withimplicitconsttwolateriii}
\Ezerokminussix(\tau) &\lesssim 
\hat{\varepsilon}_0 \alpha {\rm L}^{-2+\delta+\beta}.
\end{align}

Thus, in view of the requirement $\alpha \ge \alpha_0 \ge 3$, for sufficiently large $\alpha_0$ we may 
absorb the constants implicit in $\lesssim$ by explicit constants of our choice by adding extra positive $\alpha$
powers to the right hand side of~\eqref{withimplicitconstonelateriii} and~\eqref{withimplicitconsttwolateriii}.
In this way,  we obtain the specific 
estimates~\eqref{bterallolargea},~\eqref{bterjustlargea} and~\eqref{btteronelargea}. 
In the same way, we also obtain the specific constant of the  estimate
of~\eqref{addedherenotfancyiii}
 which will be convenient in our scheme.

\end{proof}
\subsection{The proof of Theorem~\ref{largeamaintheorem}:  iteration in consecutive spacetime slabs}
\label{simpleiterationproof}

It is now elementary to deduce Theorem~\ref{largeamaintheorem}:

We  fix $\alpha\ge \alpha_0$  so that the statement of Proposition~\ref{pigeonsforlargea} holds,
for instance $\alpha:=\alpha_0$.
(Note that we shall no longer track the dependence of constants and parameters on $\alpha$ since it is now considered fixed.
Thus implicit constants depending on the choice of $\alpha$ will from now on be incorporated in the notations $\sim$ and  $\lesssim$. The $\alpha$ factors in exact inequalities $\leq$ will of course still
necessarily be denoted!)

For $i=0,1,\ldots$, we define 
\begin{equation}
\label{tauiLidef}
\tau_i:= \alpha^i, \qquad 
{\rm L}_i:= \alpha^{i+1}-\alpha^{i}= \alpha^i(\alpha-1)\ge 2.
\end{equation}
Note that 
\begin{equation}
\label{notethisherenowyes}
{\rm L}_{i+1}=\alpha {\rm L}_{i}.
\end{equation}

For $0<\varepsilon_0 \le \varepsilon_{\rm global}$ and $\varepsilon_{\rm global}$ sufficiently small,
we have by assumption
\begin{equation}
\label{afouedw}
\Eonek(\tau_0) \leq\varepsilon_0, \qquad
\Etwominusdelkminustwo(\tau_0)\leq\varepsilon_0.
\end{equation}

In particular, for $\varepsilon_{\rm global}$ sufficiently small we have
\begin{equation}
\label{newiterativeinproof}
\Eonek(\tau_i) \leq \hat\varepsilon_0 {\rm L}_i^{\beta}, \qquad
\Etwominusdelkminustwo(\tau_i) \leq \hat\varepsilon_0 {\rm L}_i^{\beta}, \qquad
\Ezerokminustwo(\tau_i) \leq\hat\varepsilon_0\alpha^{2}  {\rm L}_i^{-1+\beta}, \qquad
 \Eonekminusfour(\tau_i) \leq \hat\varepsilon_0 \alpha^4 {\rm L}_i^{-1+\delta+\beta},  \qquad
 \Ezerokminussix(\tau_i) \leq \hat\varepsilon_0 \alpha^8 {\rm L}_i^{-2+\delta+ \beta}
\end{equation} 
 for $i=0$ (for which $\tau_0=1$, $L_0=\Xi-1$), and where by the smallness of $\varepsilon_{\rm global}$, these  imply also~\eqref{assumptionondataherelargea}.
 
 Now suppose~\eqref{newiterativeinproof} hold for  some $i\ge 0$ where $\tau_i$,
 ${\rm L}_i$ are defined as in~\eqref{tauiLidef}. Note that~\eqref{assumptionondataherelargea} then also holds.
Proposition~\ref{pigeonsforlargea} applies to yield that the solution extends to 
$\mathcal{R}(\tau_i, \tau_i+{\rm L}_i)=\mathcal{R}(\tau_i,\tau_{i+1})$ and satisfies
\begin{equation}
\label{withintheslabsproof}
\Eonek(\tau) \lesssim \varepsilon_0 {\rm L}_i^{\beta}, \qquad
\Etwominusdelkminustwo(\tau) \lesssim  \hat\varepsilon_0 {\rm L}_i^{\beta}, \qquad
\Ezerokminustwo(\tau) \lesssim \hat\varepsilon_0 {\rm L}_i^{-1+\beta}, \qquad
 \Eonekminusfour(\tau) \lesssim \hat\varepsilon_0 {\rm L}_i^{-1+\delta+\beta},  \qquad
 \Ezerokminussix(\tau) \lesssim \hat\varepsilon_0  {\rm L}_i^{-2+\delta+ \beta}
\end{equation} 
for all $\tau \in [\tau_i,\tau_{i+1}]$,
as well as the estimates at $\tau=\tau_{i+1}$ with precise constant~\eqref{addedherenotfancyiii}, \eqref{bterallolargea}--\eqref{btteronelargea}. Examining the powers of $\alpha$ in these latter estimates,
 it follows that~\eqref{newiterativeinproof} holds
for $i+1$, in view of~\eqref{notethisherenowyes}.

It follows by induction that the solution extends to $\mathcal{R}(\tau_0,\infty)$ 
and satisfies~\eqref{withintheslabsproof} with $\tau$ replacing ${\rm L}_i$, for all $\tau\ge \tau_0$,
so in particular
we have the higher order inverse power law decay statements
\begin{equation}
\label{powerlawdecaystatements}
\Ezerokminustwo(\tau) \lesssim \hat\varepsilon_0 \tau^{-1+\beta}, \qquad
 \Eonekminusfour(\tau) \lesssim \hat\varepsilon_0 \tau ^{-1+\delta+\beta},  \qquad
 \Ezerokminussix(\tau) \lesssim \hat\varepsilon_0  \tau^{-2+\delta+ \beta}.
\end{equation}

For each $\tau_i$, we may revisit the estimates~\eqref{toporderherelastterm} 
and~\eqref{notetheextraterm}  of Proposition~\ref{largeafinal}
but now in the (non-dyadic!)~region  $\mathcal{R}(\tau_0,\tau_i)$, both
for $k$ and $p=1$
and for $k-2$ and $p=2-\delta$,
using the alternative right hand side~\eqref{theexpressiontobereplaced}
referred to in the last sentence of the statement of the proposition, where we 
apply~\eqref{theexpressiontobereplaced} to the partition of $(\tau_0,\tau_i)$ 
by $\tau_j$ for $j=0,\ldots, i$.  
Plugging in the bounds already shown, we may now bound the arising terms~\eqref{theexpressiontobereplaced}
by 
$ \Epk(\tau_0)$.
Summing~\eqref{toporderherelastterm} 
and~\eqref{notetheextraterm} 
as before, we deduce
\[
{}^\rho\Xonek(\tau_0,\tau_i)[\psi] +{}^{\chi\natural}_{\scalebox{.6}{\mbox{\tiny{\boxed{+1}}}}}(\tau_0,\tau_i)[\psi]\,\Xzerokminusone \lesssim \Eonek(\tau_0)[\psi],
\]
\[
{}^\rho\Xtwominusdelkminustwo(\tau_0,\tau_i)[\psi] +{}^{\chi\natural}_{\scalebox{.6}{\mbox{\tiny{\boxed{+1}}}}}\,\Xzerokminusthree(\tau_0,\tau_i)[\psi] \lesssim \Etwominusdelk(\tau_0)[\psi]
\]
(and thus the boundedness statements~\eqref{boundednessstafin}).

This completes the proof of Theorem~\ref{largeamaintheorem}.

\appendix

\section{Estimates in Carter frequency space}
\label{carterestimatesappend}

In this section, we collect all results which are proven using fixed frequency
estimates.  Although deferred to the last section of the paper, the constructions here should really be viewed
as the main content of our argument.

We first derive in Section~\ref{propofpotsections} 
some basic properties of the potential $V$ arising in Carter's separation,
in particular deducing those fundamental properties in the superradiant and non-superradiant
regimes that we shall use for our later constructions.
We shall then review in Section~\ref{translationsection} our separated
current formalism
from~\cite{partiii} and provide a translation to the physical space twisted
currents of~\cite{DHRT22}. 
Section~\ref{refinedapp} then proves our elliptic refinement (Proposition~\ref{refinedproposition}).
Section~\ref{fundcoercivityapp} proves  our global (degenerate) bulk coercivity (Theorem~\ref{globalgenbulkcoercprop}).  
Finally, Section~\ref{fundcoercivityboundapp} proves our global boundary coervicity (Proposition~\ref{globalboundarySpositivity}).

\subsection{Properties of the potential $V$ and determination of 
the frequency covering~$\mathcal{F}_n$}
\label{propofpotsections}

We review in this section the properties of the potential $V$ arising from Carter's separation
and use these to fix 
our covering $\mathcal{F}_n$ of frequency space.

In Section~\ref{basicdecompo}, we will decompose the potential $V$ into its high frequency $V_0$ and remainder part $V_1$ and introduce also the $\gamma$-dependent part $V_\gamma$.
In Section~\ref{potentialparameters} 
we shall review  some basic facts  about $V_0$ from~\cite{partiii} valid for
all frequencies. 
In Section~\ref{propssuper} we shall deduce some properties special to the (generalised) superradiant regime, while
in Section~\ref{propsnonsuper} we shall deduce properties special to the non-superradiant regime.
(These properties can be viewed as refinements of properties shown in~\cite{partiii}.)
Finally, in Section~\ref{fixingthecovering}, on the basis of the above properties,
we shall fix our covering~$\mathcal{F}_n$  of frequency space announced in
Section~\ref{projectiondef}.

\subsubsection{The high-frequency potential $V_0$ and the $\gamma$-dependent piece $V_\gamma$}
\label{basicdecompo}
Recall the potential $V(\omega, m, \ell, r)$ defined in~\eqref{Vdefbody} appearing in the radial 
ode~\eqref{radialodewithpot}, where $\Lambda$ is given by
$\Lambda_{\ell m}^{(a\omega)} := \lambda^{(a\omega)}_{m\ell} + a^2\omega^2$
and $\lambda^{(a\omega)}_{m\ell}$ represent the eigenvalues of~\eqref{operatorhere}
parametrised by $\ell\ge |m|$.

In practice, we shall  typically think of $\Lambda$ as an independent parameter, subject to the inequalities~\eqref{proptyofLambda},
~\eqref{proptyofLambdatwo},  and consider~\eqref{Vdefbody} to define
an expression $V(\omega, m ,\Lambda, r)$. Recall that a triple 
$(\omega, m, \Lambda)\in \mathbb R\times\mathbb Z\times\mathbb R$
satisfying~\eqref{proptyofLambda} and~\eqref{proptyofLambdatwo} is known as an \emph{admissible frequency triple}. 

We decompose
\begin{equation}
\label{decomposeit}
V(\omega, m, \Lambda, r) =V_0(\omega, m, \Lambda, r)+V_1(r)
\end{equation}
into the \emph{high frequency potential} given by 
\begin{equation}
\label{defofVzero}
V_0(\omega, m, \ell, r) = 
\frac{4Mram\omega-a^2m^2+\Delta\Lambda}{(r^2+a^2)^2}
\end{equation}
and the \emph{remainder potential} given  by
\begin{equation}
\label{defofVone}
V_1(r) = \frac{\Delta}{(r^2+a^2)^4}\left(a^2\Delta +2Mr(r^2-a^2)\right).
\end{equation}
Note that $V_1\ge 0$.

Let us decompose $V_0$ further as
\[
V_0=V_\gamma  +\Delta\Lambda(r^2+a^2)^{-2}
\]
where we define $\gamma= \frac{\omega}{am}$, if $m\ne 0$ and $a\ne 0$,
\[
V_\gamma =\frac{a^2 m^2 ( 4Mr  \gamma - 1 )}{(r^2 +a^2)^2} .
\]
If $m=0$ or $a=0$ we set $V_\gamma=0$.
If $\omega\ne 0$,  $m\ne 0$, may rewrite $V_\gamma$  as
\[
V_\gamma = \omega^2 \gamma^{-1} \frac{(4Mr-\gamma^{-1})}{(r^2+a^2)^2}.
\]

Note that  for $\gamma=(2Mr_+)^{-1}$,
\[
V_\gamma(r_+) = \omega^2 (2M r_+)^2 (r_+^2+a^2)^{-2}  = \omega^2,
\]
and more generally
\begin{equation}
\label{useit}
V_{\gamma}(r_+) \le \omega^2
\end{equation} 
with equality iff $\gamma= (2Mr_+)^{-1}$.

\begin{remark}
In what follows we recall from Section~\ref{Carterscompleteseparation} the definition of the rescaled coordinate $r^*$ given by
\[
\frac{dr^*}{dr} = \frac{r^2+a^2}{\Delta} , \qquad r^*(3M)= 0,
\]
and our notation
\begin{equation}
\label{reallymeanit}
' := \frac{d}{dr^*}.
\end{equation}
We note that functions of $r\in (r_+,\infty)$ can be thought of as functions of $r^*\in (-\infty,\infty)$
and we will implicitly go back and forth between these two parametrisations  of the domain 
without further comment. Thus we emphasise that in writing $V'(r)$, etc.,
we really do mean the derivative defined by~\eqref{reallymeanit}. When referring to the $r$-derivative $\frac{d}{dr}$, we shall
always use Leibniz notation. Note that these operators  are of course related by
\[
\frac{d}{dr^*} = \frac{\Delta}{r^2+a^2} \frac{d}{dr}.
\]
\end{remark}

\subsubsection{The parameters $r_{\rm pot}$ and $R_{\rm pot}$ and general properties of $V_0$} 
\label{potentialparameters}

In this section we will review some basic properties of $V_0$, valid for all admissible frequencies
$(\omega, m, \Lambda)$, which can be read off from
the analysis in~\cite{partiii}. (We shall infer more refined properties localised to
fixed $\gamma$ in the two sections that follow, according to whether $\gamma$
is (generalised) superradiant or non-superradiant.)

We have the following proposition which collects various facts about $V_0$ shown in~\cite{partiii}:

\begin{proposition}[General properties of $V_0$~\cite{partiii}]
\label{genpropsofV0}
Let $|a|< M$.
Let $(\omega, m, \Lambda)$ be an admissible frequency triple.
The potential $V_0(\omega,m,\Lambda, r)$ is either strictly decreasing on $(r_+,\infty)$
or has a unique interior local maximum $r_{\rm max}^0(\omega, m, \Lambda)$ 
at which 
\begin{equation}
\label{secondderiv}
-\frac{d^2}{dr^2} V_0 \gtrsim \Lambda
\end{equation}
Moreover,   when $r_{\rm max}^0$ exists we have
$r_{\rm max}^0(\omega,m ,\Lambda)\le 5M$.

The potential may also have a single interior local minimum at some $r_+<r_{\rm min}^0<r_{\rm max}^0$. 
In this case, $V_0'<0$ for $r_+<r<r^0_{\rm min}$, $V'_0>0$ for $ r^0_{\rm min}<r <r^0_{\rm max}$
and $V'_0<0$ for $r>r^0_{\rm max}$.

If $r^0_{\rm min}(\omega, m, \Lambda)$ exists,  then $\frac{d}{dr}V_0(r_+)<0$, and
$\gamma >  (2M\sqrt{M^2-a^2})^{-1} $ and thus $(\omega, m)$ is in particular non-superradiant.
In this case, we have
\begin{equation}
\label{Vminlvaluebound}
V_0 (r)-\omega^2  \gtrsim \omega^2
\end{equation}
for all $r\in [r_+, r_{\rm min}]$.
\end{proposition}
\begin{remark}
We remind the reader that in accordance with our conventions, under the assumptions given above
for which~\eqref{secondderiv} and~\eqref{Vminlvaluebound}  hold, 
the constants implicit in $\gtrsim$ are 
in particular independent of the frequency triple $(\omega, m, \Lambda)$.
This will apply in our subsequent use of $\lesssim$, $\gtrsim$  in what follows.
\end{remark}

Let us first introduce two parameters $r_{\rm pot}$ and $R_{\rm pot}$ which will satisfy
\begin{equation}
\label{potssatisfying}
r_2<r_{\rm pot}<2.01M<5M<R_{\rm pot}<2R_{\rm pot}< R_{\rm freq}
\end{equation}
where $R_{\rm freq}$ will be determined in Section~\ref{nonsuperradiantregimecurrentdef}.

The parameters $r_{\rm pot}$ and $R_{\rm pot}$ are directly related to properties
of $V_0$ and will be further constrained by the propositions in Sections~\ref{propssuper} 
and~\ref{propsnonsuper}
as well as the following statement which may be easily inferred:

\begin{proposition}
Let $|a|<M$.
If $r_{\rm pot}>r_+$ is sufficiently close to $r_+$ then $r_{\rm pot}<r^0_{\rm max}$ 
for all admissible triples $(\omega, m, \Lambda)$ for which $r^0_{\rm max}$ exists and the one-sided
bound
\begin{equation}
\label{onesidedwhereneg}
-V'_0(r) \lesssim \Delta r^{-2}  |m\omega|
\end{equation}
 holds in $[r_+,r_{\rm pot}]$
for all admissible frequencies $(\omega, m, \Lambda)$, as well as the bound
\begin{equation}
\label{forusewithredshift}
\left(\frac{(r^2+a^2)^2}{(r-r_+)} (V-V(r_+))\right)'  \gtrsim \Delta (\Lambda+ 1).
\end{equation}

If $R_{\rm pot}$ is sufficiently large then
$-V_0' \gtrsim \Lambda r^{-3}$, $ -(rV_0)' \gtrsim \Lambda r^{-2}$
and in fact
\begin{equation}
\label{Vprimefar}
-V' \gtrsim \Lambda r^{-3}, \qquad  -(rV)' \gtrsim \Lambda r^{-2} +r^{-3}, \qquad  V \lesssim r^{-1}V'
\end{equation}
in $r\ge R_{\rm pot}$.
\end{proposition}

Let us note that the constraint provided by the above proposition
requires us to take $r_{\rm pot}\to r_+$ as $|a|\to M$. For more discussion of the properties of $V$ in the extremal case $|a|=M$ see~\cite{ritamode, gajicazimuthal}.

\subsubsection{Properties of $V_0$ in the superradiant regime: supperadiant frequencies
are not trapped}
\label{propssuper}

In~\cite{partiii}, a fundamental role was played by the remark that superradiant frequencies
are never trapped, more precisely, the maximum of $V_0$ is quantitatively
greater than $\omega^2$ for all admissible superradiant
frequency triples $(\omega, m,\Lambda)$.

Examination of the proof of~\cite{partiii} reveals that something slightly stronger
is true, namely that given~$\gamma$,
there is a fixed domain in $r$ for which $V_0$ is quantitatively
greater than $\omega^2$ for \emph{all} superradiant frequency triples $(\omega, m, \Lambda)$
such that $\omega/am =\gamma$. We show this fact here.

We have the following:
\begin{proposition}[Common elliptic region for fixed-$\gamma$ generalised superradiant frequencies]
\label{superradpotentialprop}
Let $0\ne |a|<M$.
There exists an open interval 
\begin{equation}
\label{gensuperdef}
I_{\rm gensuper}:= (-\epsilon,(2Mr_+)^{-1}+\epsilon),
\end{equation}
defined for some sufficiently small $\epsilon>0$, such that for all $r_{\rm pot}>r_+$
sufficiently close to $r_+$ and all $R_{\rm pot}$ sufficiently large the following holds.

Let $\gamma \in I_{\rm gensuper}$. Then there exists a nonempty  interval
\[
[r_{\gamma,1},r_{\gamma,2}]\subset (r_{\rm pot}, R_{\rm pot})
\]
and sufficiently small parameters $b_\gamma>0$, $b'_\gamma>0$,
 depending continuously on $\gamma$, such that
for all admissible $(\omega ,m\ne 0 ,\Lambda)$ with $\omega/am = \gamma$, then
\begin{align}
\label{alsothishere}
\{r>r_+: V_0 -\omega^2 < b_\gamma \omega^2, \quad V'_0> 0 \} & \subset  (r_+,r_{\gamma,1}),\\
\label{andalsothishere}
\{r>r_+: V_0 -\omega^2 < b_\gamma\omega^2, \quad V'_0<0 \}  & \subset  (r_{\gamma,2},\infty),\\
\label{tobeproven}
\{r>r_+: V_0 -\omega^2 \ge b_\gamma\omega^2\}  & \supset  [r_{\gamma,1},r_{\gamma,2}],\\
\label{supernottrapped}
\{r_+<r< 2R_{\rm pot} : |V_0 -\omega^2| \le b_\gamma \omega^2,\quad |V'_0|\le b'_{\gamma}\Delta r^{-2} \omega^2\} & =  \emptyset ,
\end{align}
as well as the bound
\begin{equation}
\label{aswellasthis}
 V'_0 \gtrsim \Delta r^{-2} ( \Lambda+\omega^2) 
 {\rm\ in\ }  (r_+,r_{\rm pot}].
\end{equation}

For
 $|\gamma -(2Mr_+)^{-1}|<\epsilon$, we may choose  $r_{\gamma,1}$, $r_{\gamma,2}$ 
 independently of $\gamma$.

Moreover, for $|\gamma|< \epsilon$, we 
we may choose $r_{\gamma,1}$, $r_{\gamma,2}$ independently of $\gamma$ 
satisfying
\[
[r_{\gamma,1},r_{\gamma,2}]\subset  (4M, 6M)
\]
and the following additional bounds hold:
\begin{align}
\label{wehavethis}
 V'_0 \gtrsim \Delta r^{-2}( \Lambda+\omega^2)  &{\rm\ in\ } (r_+,2.1M),\\
\label{andthistoo}
V_0 \gtrsim   r^{-2}( \Lambda +\omega^2)  &{\rm\ in\ } [2.01M, \infty),\\
\label{canaddthisalsonow}
 V_0 \gtrsim   \frac12 \omega^2  &{\rm\ in\ } [2.01M, 1.2R_{\rm pot}].
\end{align}
\end{proposition}

\begin{proof}
We first show~\eqref{tobeproven} and the properties of $[r_{\gamma,1}, r_{\gamma,2}]$.

Consider first the case
\[
0< \gamma <  \frac{1}{2Mr_+}.
\]

Let us note that from~\eqref{proptyofLambda} we immediately have
\[
V_0 \ge   \frac{4Mr  \gamma^{-1}  - \gamma^{-2} + \Delta a^{-2} \gamma^{-2} }{(r^2+a^2)^2} \omega^2.
\]

We estimate $V_0 -\omega^2$  at $r_\gamma:= \gamma^{-1}/2M$, noting the relation
\[
(r^2+a^2)^2 =\Delta (r^2+2Mr+a^2) +4M^2r^2,
\]
as follows:
\begin{eqnarray*}
V_0 (r_\gamma) -\omega^2  &\ge& \frac{ \gamma^{-2} + \Delta a^{-2} \gamma^{-2} -(r^2+a^2)^2 }{(r^2+a^2)^2} \omega^2  \\
&=&  \frac{ 4M^2r^2  + 4M^2 r^2 \Delta a^{-2} -(r^2+a^2)^2 }{(r^2+a^2)^2} \omega^2 \\
&=&  \frac{ 4M^2r^2  + 4M^2 r^2 \Delta a^{-2} -\Delta(r^2+2Mr+a^2) -4M^2r^2 }{(r^2+a^2)^2} \omega^2\\
&=& \frac{  4M^2 r^2 a^{-2} -(r^2+2Mr+a^2)  }{(r^2+a^2)^2} \Delta \omega^2\\
&\ge& \frac{(M-a)(M+a)}{a^2 (r^2+a^2)^2 }\Delta\omega^2.
\end{eqnarray*}

Now since  by~\eqref{potssatisfying} and Proposition~\ref{genpropsofV0} it follows that 
$V'_0< 0 $ for $r\ge  5M$, let us define
\[
\hat{r}_\gamma := \min \{ 5M, r_\gamma\}.
\]
We may now choose a suitably small interval $[r_{\gamma,1},r_{\gamma,2}]$ containing
$\hat{r}_\gamma$, depending continuously on $\gamma$, so that~\eqref{tobeproven} holds
for sufficiently small $b_\gamma>0$.

Note that if $0<\gamma<\epsilon$, then $\hat{r}_\gamma= 5M$. It is clear that
one may in fact choose $[r_{\gamma,1},r_{\gamma,2}]$ to be a fixed interval around
$5M$ for all $|\gamma|<\epsilon$,  say contained in $(4M, 6M)$, 
and~\eqref{tobeproven} holds for sufficiently small $b_\gamma>0$.

It remains only to consider the case  $|\gamma -(2Mr_+)^{-1}|<\epsilon$.
For this, notice that $r_\gamma=r_+$ for $\gamma= \frac{1}{2Mr_+}$.
But, for such $\gamma$
\[
\frac{d}{dr} V_0 (r_+) = (r^2+a^2)^2 (4M\gamma^{-1}) >0.
\]
It follows immediately that for sufficiently small $\epsilon$,
and all $|\gamma -(2Mr_+)^{-1}|<\epsilon$,
there exists an interval $r_{1,\gamma},r_{2,\gamma}$
satisfying~\eqref{tobeproven} for sufficiently small $b_\gamma>0$, which may in fact be taken independent of $\gamma$,
 with $r_{1,\gamma}>r_{\rm pot}$,
for $r_{\rm pot}>r_+$ sufficiently close to $r_+$.

Let us recall that if $\epsilon>0$ is sufficiently small, then by 
Proposition~\ref{genpropsofV0}, there is no $r^0_{\rm min}$  for $V_0$ 
for $(\omega, m, \Lambda)$ with $\gamma\in I_{\rm gensuper}$.

Finally, note that~\eqref{alsothishere} and~\eqref{andalsothishere} follow 
from~\eqref{tobeproven} and the uniqueness
of an interior local maximum  $V_0$.

Note that if~\eqref{alsothishere}--\eqref{tobeproven} 
hold for some $b_\gamma>0$, they hold for any smaller $b_\gamma>0$. 
Now, given $b'_\gamma>0$ sufficiently small, then since there is no interior local maximum
such that $V_0-\omega^2=0$, by our computation at $r_+$, and by the behaviour
of $V_0$ as $\Lambda\to \infty$, it follows by compactness that there exists a $b_\gamma$
satisfying in addition~\eqref{supernottrapped}.

The property~\eqref{aswellasthis} follows easily from direct computation for $\epsilon>0$
sufficiently small.

To show~\eqref{wehavethis}, we note that for $\gamma=0$,
\[
V_\gamma'  \gtrsim m^2 r^{-4}
\]
while
\[
(\Delta (r^2+a^2)^{-2} )'  \gtrsim \Delta 
\] 
in say $(r_+,2.1M)$. The result follows by continuity for 
$|\gamma|<\epsilon$ for sufficiently small $\epsilon>0$.

Similarly, to show~\eqref{andthistoo} and~\eqref{canaddthisalsonow}, let us first note that for $\gamma=0$, by~\eqref{proptyofLambda}
we have
\[
V_0 \ge \frac{(\Delta -a^2)m^2 }{(r^2+a^2)^2}.
\]
The statements~\eqref{andthistoo} and~\eqref{canaddthisalsonow} for sufficiently small $\epsilon$
then follows by continuity in $\gamma$ since $\Delta-a^2 \gtrsim 1$ in $r\ge 2.01M$.

\end{proof}

\subsubsection{Properties of  $V_0$ in the non-superradiant regime}
\label{propsnonsuper}
In this section, we turn to the non-superradiant regime.

Assume $0\ne |a|<M$.
Let us fix an open 
\begin{equation}
\label{Inonsuperdefhere}
I_{\rm nonsuper}\subset \mathbb R\setminus [0,(2Mr_+)^{-1}] 
\end{equation}
with two connected components   such that $0$ and $(2Mr_+)^{-1}$ are not limit points
of $I_{\rm nonsuper}$ and 
\begin{equation}
\label{ikavopoiouvauto}
I_{\rm nonsuper} \cup I_{\rm gensuper}=\mathbb R,
\end{equation}
where $I_{\rm gensuper}$ is defined by~\eqref{gensuperdef}.

For $\gamma \in I_{\rm nonsuper}$, we define $r^0_\gamma$ as follows:

If 
\begin{equation}
\label{ifthisistrue}
\max \gamma^{-1}   \frac{4Mr-\gamma^{-1}}{(r^2+a^2)^2}  >1,
\end{equation}
then $V_\gamma$ obtains a maximum at a unique $r$-value.
This is then $r^0_\gamma$.
Note that
\[
V_0(\omega, m , \Lambda)(r^0_\gamma) -\omega^2 \ge V_\gamma(r^0_\gamma)- \omega^2 
\ge b(\gamma)\omega^2.
\]

Otherwise,
there is a unique $\Xi^2_\gamma \in (0,\infty)$ such that, for any
$\omega\ne0$, $m\ne0 $ with $\omega/am=\gamma$, setting $\Lambda (\omega,\gamma):=\Xi_\gamma^2\omega^2$,  then $r^0_{\rm max}(\omega, m, \Lambda(\omega, \gamma))$ exists and
\[
V_{\rm max}:=V_0  (r^0_{\rm max} (\omega,m, \Lambda ) ) =\omega^2 .
\]
In this case then, let us denote $r^0_\gamma: = r^0_{\rm max}(\omega, m , \Lambda(\omega, \gamma))$.
Note this is indeed independent of the choice of $\omega, m$ such that $\omega/am=\gamma$.
We have that there exists $\epsilon>0$
such that for all $\gamma\in  I_{\rm nonsuper}$, 
\[
R_{\rm pot} -\epsilon> r^0_{\gamma} >   r_{\rm pot}+\epsilon
\]
provided that $R_{\rm pot}$ is sufficiently large and $r_{\rm pot}>r_+$ is sufficiently
close to $r_+$.

Let us note finally that under this definition $r^0_\gamma$ depends continuously
on $\gamma$ in $I_{\rm nonsuper}$, and in the case~\eqref{ifthisistrue}, then setting
$\Xi^2_\gamma:=0$, $\Lambda:=0$, then
$r^0_\gamma= r^0_{\rm max}(\omega, m, \Lambda=0) $ for any $\omega\ne0$, $m\ne 0$ with  $\omega/am=\gamma$. (Note that $(\omega\ne 0, m\ne 0, \Lambda=0)$ does not constitute an admissible triple, but we may of course nonetheless plug such a $\Lambda$ into the expression~\eqref{defofVzero}.)

\begin{proposition}[$\gamma$-dependent localisation of trapping for non-superradiant frequencies]
\label{nonsuperradpotentialprop}
Let $0\ne |a|<M$.
For all
$r_{\rm pot}>r_+$ sufficiently close to $r_+$  and $R_{\rm pot}$ sufficiently large,
the following holds:

Let $\gamma\in I_{\rm nonsuper}$. Given any $\epsilon>0$ sufficiently small, and defining
\begin{equation}
\label{epsiloninthisdef}
r_{\gamma,1}:=r^0_\gamma-\epsilon, \qquad r_{\gamma,2}:=r^0_\gamma+\epsilon,
\end{equation}
we have 
\[
[r_{\gamma,1}, r_{\gamma,2} ] \subset (r_{\rm pot}, R_{\rm pot})
\]
and 
there exist sufficiently small parameters  $b_\gamma>0$, $b'_\gamma>0$, depending 
on $\epsilon$
and continuously
on $\gamma$,
such that for all admissible $(\omega\ne0 , m \ne 0, \Lambda)$ with $\omega/am=\gamma$, the following
inclusions hold:
\begin{align}
\label{onesethere}
\{ r>r_+ : |V_0 -\omega^2| <b_\gamma \omega^2,\quad V' _0> 0 \}
&\subset (r_+  , r_{\gamma,2} )  ,  \\
\label{anothersethere}
\{ r>r_+ : |V_0 -\omega^2| <b_\gamma \omega^2,\quad V'_0 < 0  \}
&\subset (r_{\gamma,1},\infty) , \\
\label{athirdsethere}
\{ r_+<r<2R_{\rm pot}: |V_0 -\omega^2| < b_\gamma\omega^2,\quad |V'_0| <  b'_\gamma \Delta r^{-2}\omega^2  \}
&\subset (r_{\gamma,1}, r_{\gamma,2}) ,
\end{align}
as well as
\begin{equation}
\label{V0boundsLambda}
\omega^2 +V_0  \gtrsim \Delta \Lambda   \text{\ in\ }\{ r_+\le r \le 1.2R_{\rm pot}\} \cap \{V_0\ge \omega^2\} .
\end{equation}

Moreover, there exist $b_{\infty}>0$, $b'_{\infty}>0$ sufficiently small 
and $\gamma_\infty$ sufficiently large, and an interval 
$[r_{\infty, 1}, r_{\infty,2}]\subset (2.1M, R_{\rm pot})$
such that 
if $|\gamma|\ge \gamma_\infty$ or $m=0$, then~\eqref{onesethere}--\eqref{athirdsethere} hold 
with $r_{\infty, 1}$ replacing $r_{\gamma,1}$ and $r_{\infty, 2}$ replacing $r_{\gamma,2}$
and with $b_{\infty}$, $b'_{\infty}$ in place of $b_\gamma$, $b'_\gamma$.
\end{proposition}

\begin{proof}
Given $\gamma\in I_{\rm nonsuper}$ and given
 $\epsilon>0$ and $r_{\gamma,1}$, $r_{\gamma,2}$ defined by~\eqref{epsiloninthisdef}, 
 the relations~\eqref{onesethere}--\eqref{athirdsethere} 
 hold for all admissible
$(\omega\ne0, m\ne0, \Lambda = \Xi^2_\gamma \omega^2)$ with $\omega/am =\gamma$
as long as $b_\gamma$ is sufficiently small,
regardless of the size of $b'_\gamma$, if $\Xi^2_\gamma\ne 0$. (In the case~\eqref{ifthisistrue}
when $\Xi^2_\gamma=0$, 
note that for $b_\gamma$ and $b'_{\gamma}$ sufficiently small, then the set on the left hand side of~\eqref{athirdsethere}
is empty, and the inclusions~\eqref{onesethere}--\eqref{anothersethere} hold.)
On the other hand, for fixed $\omega\ne0, m\ne0$ with $\omega/am=\gamma$, the set of $\Lambda$ for which the above relations all hold for $(\omega, m, \Lambda)$ is open,
so they hold for
 \begin{equation}
 \label{Lambdasatisfying}
 \alpha \Xi^2_\gamma\omega^2 < \Lambda <  \alpha^{-1}\Xi^2_\gamma\omega^2
 \end{equation}
for some $\alpha<1$. Moreover, by homogeneity in $\omega^2$, it follows that~\eqref{onesethere}--\eqref{athirdsethere} in fact hold for all admissible $(\omega\ne0, m\ne 0, \Lambda)$ 
with $\omega/am=\gamma$ with $\Lambda$ satisfying~\eqref{Lambdasatisfying}.

 Note that if the inclusions \eqref{onesethere}--\eqref{athirdsethere}
 hold for some $b_\gamma$, $b'_\gamma>0$, 
then  they trivially continue to hold if we take $b_\gamma>0$, $b'_\gamma>0$ smaller. 
But clearly if $\Lambda\le \alpha \Xi^2_\gamma\omega^2$, then, after restriction of $b_\gamma$,
the inclusions~\eqref{onesethere}--\eqref{athirdsethere} hold trivially, because $\{ r: |V_0-\omega^2|<b_\gamma\omega^2\} =\emptyset$.

So it suffices to show that, possibly by restricting $b_\gamma$ and $b'_\gamma$ further,~\eqref{onesethere}--\eqref{athirdsethere} hold also for 
\begin{equation}
\label{sufficestoshownow}
\Lambda\ge \alpha^{-1}\Xi^2_\gamma\omega^2.
\end{equation}

Now, recall from~\eqref{Vminlvaluebound},
that for all admissible $(\omega, m, \Lambda)$, for $b_\gamma>0$
sufficiently small, $V_{\rm min} \le (1-2b_\gamma)\omega^2$.
We claim that for $b_\gamma>0$ sufficiently small, we have that for all
admissible $(\omega, m, \Lambda)$ satisfying also~\eqref{sufficestoshownow},
the following holds:
\begin{equation}
\label{wheresmax}
V_{\rm max}(\omega, m, \Lambda) \ge(1+3b_\gamma)  \omega^2.
\end{equation}
This follows because 
\begin{equation}
\label{similarlyto}
V_{\rm max}(\omega,m, \Lambda)- V_{\rm max}(\omega, m, \Xi^2_\gamma \omega^2)
\ge (\Lambda-\Xi^2_\gamma\omega^2) h
\end{equation}
where $h=\Delta(r^2+a^2)^{-2} (r_{\rm max}^0(\omega, m, \Lambda=\Xi^2_\gamma\omega^2))\ge b$, since $r_{\rm max}^0(\omega, m, \Lambda=\Xi^2_\gamma\omega^2)$
lies in  $(r_{\rm pot}, R_{\rm pot})$.

On the other hand, given $b_\gamma>0$, we 
may then choose $b'_\gamma$ sufficiently small (depending on $b_\gamma$), such that
for all admissible $(\omega, m, \Lambda)$ with $\omega/am=\gamma$, we have
\begin{equation}
\label{toargue}
\{ r_+<r<2R_{\rm pot} :|V'_0| < b'_\gamma\Delta r^{-2}\omega^2 \} \subset
\{|V_0-V^0_{\rm max} |\le b_\gamma \omega^2\}\cup \{|V_0-V^0_{\rm min}| \le b_\gamma \omega^2\}. 
\end{equation}
We see~\eqref{toargue} as follows:
Let us write $\Lambda =\beta\omega^2$ and write, for $\omega/am=\gamma$,
\[
V_0(\omega, m , \Lambda= \beta\omega^2 ,r) = \mathcal{V}_{\beta, \gamma}(r)\omega^2
\]
\[
\frac{d}{dr} V_0=\frac{d}{dr}\mathcal{V}_{\beta, \gamma}(r)\omega^2.
\]
Note that $\mathcal{V}_{\beta,\gamma}(r)$ satisfies $\frac{d}{dr}\mathcal{V}_{\beta, \gamma}(r_+) \ne 0$ and its critical points and values
coincide with those of $\omega^2 V_0(\omega, m , \Lambda = \beta\omega^2, r)$.
For fixed $\beta\ge \gamma^{-1}$ (admissibility), and more generally for $\beta\in X$ for any compact subset $X\subset [ \gamma^{-1} , \infty)$, it follows by compactness that given $\hat{b}>0$, 
then if  $\hat{b}'>0$ is sufficiently small, depending on $\hat{b}$, $\gamma$ and $X$,  then 
\begin{equation}
\label{alsoholds}
\left\{ r_+<r<2R_{\rm pot}:  \left|\frac{d}{dr} \mathcal{V}_{\beta,\gamma}\right| < \hat{b}' \right\} \subset
\{| \mathcal{V}_{\beta,\gamma} -\mathcal{V}^{\rm max} _{\beta,\gamma} |\le \hat{c}\}\cup \{| \mathcal{V}_{\beta,\gamma} -\mathcal{V}^{\rm min} _{\beta,\gamma} |\le \hat{c}\}. 
\end{equation}
But we notice that given $\hat{b}>0$, for $\hat{b}'>0$ sufficiently small,~\eqref{alsoholds}  also holds for
all $\beta \gg \alpha^{-1}\Xi^2_\gamma$. To see this note that there is a unique $r_{\rm tr} \in [r_+,2R_{\rm pot}]$
such that
\[
\frac{d}{dr}\left[ \Delta(r^2+a^2)^{-2} \right](r_{\rm tr}) =0,
\]
and the left hand side of~\eqref{alsoholds} 
is contained in the set $\{ r : |r - r_{tr}| \le C(\gamma) \beta^{-1}\}$ for some $C(\gamma)$. On the other hand, we claim
that on $\{ r : |r - r_{tr}| \le C(\gamma) \beta^{-1}\}$ we have 
$| \mathcal{V}_{\beta,\gamma} -\mathcal{V}^{\rm max} _{\beta,\gamma} |\le C(\gamma) \beta^{-1}\cdot\beta^2 =C(\gamma)\beta^{-1}$. Thus, given any $\hat{b}$ we indeed have~\eqref{alsoholds} also for 
 $\beta \gg \alpha^{-1}\Xi^2_\gamma$ for all $\hat{b}'>0$ sufficiently small (depending on $\hat{b}$
 and $\gamma$).
 Having established that~\eqref{alsoholds} for all $\beta\ge \gamma^{-1}$, we note 
 that~\eqref{toargue} follows for all admissible $(\omega, m, \Lambda)$ with $\omega/am=\gamma$  by re-expressing in terms of $V_0$ and expressing $\frac{d}{dr}$
 in terms of ${}'=\frac{d}{dr^*}$.

It follows that given $b_\gamma>0$ 
sufficiently small, we may choose $b'_\gamma>0$ sufficiently small so that 
the left hand side of~\eqref{athirdsethere} is empty for all admissible $(\omega, m, \Lambda)$ 
with $\omega/am =\gamma$ and $\Lambda$ satisfying~\eqref{sufficestoshownow}.

Thus, in view of our comments above~\eqref{athirdsethere} indeed
holds for all admissible $(\omega, m, \Lambda)$ with $\omega/am=\gamma$, provided
that $b_\gamma>0$ and $b'_\gamma>0$ are chosen sufficiently small.

To show the remaining inclusions~\eqref{onesethere}--\eqref{anothersethere} in the case~\eqref{sufficestoshownow}, 
let us first note that given $b_\gamma>0$, 
there exists a $0<\epsilon_0<\epsilon$, depending on $b_\gamma$, such that
for all admissible $(\omega, m, \Lambda)$  with $\omega/am=\gamma$ and
$\Lambda \ge \alpha^{-1}\Xi^2_{\gamma}\omega^2$, we have:
\begin{equation}
\label{tobereferredto}
V_0 > (1+b_\gamma)\omega^2 \text{\ on\ }(r^0_\gamma-\epsilon_0, r^0_\gamma+\epsilon_0).
\end{equation}  
(This follows from similar considerations to those leading to~\eqref{similarlyto}.)

Recall now that $V_0(r_+) <\omega^2$ for all admissible frequencies with $\gamma\in I_{\rm nonsuper}$.
Appealing to the bound of $\gamma$ away from $\frac{a}{2Mr_+}$ in the definition of  $I_{\rm nonsuper}$, 
let us require $b_\gamma>0$ 
to be sufficiently small such that say $V_0(r_+)<(1-3b_\gamma)\omega^2$ 
for all admissible frequencies with
 $\gamma\in I_{\rm nonsuper}$.
Since $V_0(r)\to \infty$ as $r\to \infty$, it follows from 
Proposition~\ref{genpropsofV0} and~\eqref{wheresmax}, 
that for all admissible $(\omega, m, \Lambda)$ with $\omega/am = \gamma$ and
$\Lambda$ satisfying~\eqref{sufficestoshownow},  given $x\in (c,c)$, there exist unique $r$-values $r_{\uparrow}(\omega, m, \Lambda, x)<r_{\downarrow}(\omega, m, \Lambda,x)$ such that
\[
\{r>r_+: V_0(\omega, m, \Lambda) =(1+x) \omega^2\}=\{r_{\uparrow}, r_{\downarrow}\}.
\]
Moreover, it follows from~\eqref{tobereferredto} that   $r_{\uparrow}\le r^0_\gamma+\epsilon_0$, $r_{\downarrow}\ge r^0_\gamma-\epsilon_0$ and
\[
V_0'(\omega, m, \Lambda, r_\uparrow)>0, \qquad V_0'(\omega, m, \Lambda, r_\downarrow)<0.
\]
The relations~\eqref{onesethere}--\eqref{anothersethere} now follow immediately upon varying $x$.

The bound~\eqref{V0boundsLambda} follows immediately from the form of the expression~\eqref{defofVzero} defining $V_0$
(where we use that $m^2\lesssim \omega^2$, noting that this depends on the fact that $0$
is not a limit point of $I_{\rm nonsuper}$).

The final statement of the proposition concerning $\gamma_{\infty}$ or the case $m=0$ follows easily. 
\end{proof}

\begin{remark}
Note that the frequencies for which $\gamma$ satisfies~\eqref{ifthisistrue} 
are in fact never trapped,
and could also have been merged with the generalised superradiant frequencies
in our treatment later.
\end{remark}

\subsubsection{Determining the frequency space covering $\mathcal{F}_n$}
\label{fixingthecovering} 
In what follows, let us fix $2.01M>r_{\rm pot}>r_+$, $R_{\rm pot}> 5M$ satisfying all the constraints of the previous 
section.

Recall the intervals $I_{\rm gensuper}$, $I_{\rm nonsuper}$ defined by~\eqref{gensuperdef},~\eqref{Inonsuperdefhere} and satisfying~\eqref{ikavopoiouvauto}.
The following proposition now determines  covering intervals $I_n$ that in 
turn determine the $(\omega, m)$ frequency-space covering $\mathcal{F}_n$
of Section~\ref{projectiondef}. 

\begin{proposition}
\label{determiningthecoveringprop}
Let $0\ne |a|<M$.
There exists a finite covering of $\mathbb R$
\[
\mathbb R = \cup_{n=1}^N I_n,
\]
 by connected open intervals $I_n\subset \mathbb R$, $n=1,\ldots, N$,
 a collection 
 \[
 \alpha_n\in \mathbb R,
 \]  
 a collection of nonempty intervals 
\begin{equation}
\label{intervalsherereference}
[r_{n,1},r_{n,2}]\subset (r_{\rm pot}, R_{\rm pot}),
\end{equation}
 and parameters $b_{\rm pot}>0$, $b'_{\rm pot}>0$
such that 
the following properties hold:
\begin{itemize}
\item
$0\in I_1$, $(2Mr_+)^{-1}\in I_2$, the collection
$I_3, \ldots , I_{N_s}$ are contained in $I_{\rm strictsuper}:= (\epsilon/2, (2Mr_+)^{-1}-\epsilon/2)$,
and
\[
\cup_{n=1}^{N_s} I_n= I_{\rm gensuper}.
\]
 \item
For all $n=1,\ldots, N_s$, all $\gamma\in I_n$ and
all admissible $(\omega, m\ne 0, \Lambda)$ with $\omega/am  = \gamma$, we have
\begin{align}
\label{superonesetherebundled}
\{ r>r_+ :  V_0 -\omega^2  <b_{\rm pot} \omega^2,\quad V' _0> 0 \}
&\subset (r_+,  r_{n,1} ) , \\
\label{superanothersetherebundled}
\{ r>r_+ :  V_0 -\omega^2  < b_{\rm pot} \omega^2,\quad V'_0 < 0  \}
&\subset ( r_{n,2},\infty), \\
\label{whatwehavehere}
\{ r> r_+: V_0 -\omega^2 \ge b_{\rm pot}  \omega^2 \} & \supset  [r_{n,1}, r_{n,2}],\\
\label{itsemptycompletely}
\{ r_+<r<2R_{\rm pot}: |V_0 -\omega^2| \le b_{\rm pot}  \omega^2 , \quad |V'_0| \le b'_{\rm pot} \Delta  r^{-2} \omega^2 \} &=\emptyset .
\end{align} 
We set $\alpha_n= \frac{a}{2Mr_+}$. We have moreover
\begin{align}
\label{rstatementonpotentialforallsuper}
V'_0 \gtrsim \Delta r^{-2}(  \Lambda +\omega^2 ) &{\rm\ in\ } (r_+,  r_{\rm pot}].
\end{align}
If $n\ne 1$, we have 
\begin{equation}
\label{officialLambdaboundfromVzero}
( \omega^2 +V_0 ) \gtrsim  \Delta \Lambda    \text{\ in\ }\{ r_+\le r \le 1.2R_{\rm pot}\} \cap \{V_0\ge \omega^2\} .
\end{equation}
\item
 In the case of $n=1$, we have the additional statements
$[r_{1,1},r_{2,1}]\subset (4M,6M)$, 
\begin{align}
\label{strongerstatementonpotentialone}
V'_0 &\gtrsim   \Delta r^{-2} (  \Lambda+\omega^2 ) \} &{\rm\ in\ } (r_+, 2.1M],\\
\label{Vzerocontrols}
V_0 &\gtrsim    r^{-2}( \Lambda +\omega^2)  &{\rm\ in\ } [2.01M, \infty),\\
\label{ellipticregionfornequalzero}
 V_0 &>  \frac12 \omega^2  &{\rm\ in\ } [2.01M, 1.2R_{\rm pot}].
\end{align}
\item
For all $n=1,\ldots, N_s$,  all
$\gamma \in I_n$ and all admissible $(\omega, m\ne 0, \Lambda)$ with,
$\omega/am=\gamma$, we have 
\begin{equation}
\label{finalchinatpropertyone}
\chi_{\natural}(r, \omega, m, \Lambda)=1,
\end{equation}
provided that the parameter  $b_{\rm trap}>0$ in~\eqref{betterdefhere} is
sufficiently small. 
\item
The collection $I_{N_s+1},\ldots , I_{N}$ satisfies
\[
\cup_{n=N_s+1}^{N-1}  I_n  =I_{\rm nonsuper}.
\]
\item
For all $n\in N_s+1,\ldots ,N$,   all $\gamma\in I_n$ and all 
admissible $(\omega, m\ne0 , \Lambda)$ with $\omega/am=\gamma$
\begin{align}
\label{onesetherebundled}
\{ r>r_+ : |V_0 -\omega^2| <b_{\rm pot}  \omega^2,\quad V' _0> 0 \}
&\subset (r_+, r_{n,2}) , \\
\label{anothersetherebundled}
\{ r>r_+ : |V_0 -\omega^2| <b_{\rm pot}  \omega^2,\quad V'_0 < 0  \}
&\subset ( r_{n,1},\infty) , \\
\label{athirdsetherebundled}
\{ r_+<r<2R_{\rm pot}: |V_0 -\omega^2| < b_{\rm pot} \omega^2,\quad |V'_0| <  b'_{\rm pot} \Delta r^{-2}\omega^2  \}
&\subset (r_{n,1}, r_{n,2}).
\end{align}
We have moreover~\eqref{officialLambdaboundfromVzero}.
For $n=N-1$, then $-\infty$ is a limit point of $I_n$ 
and~\eqref{onesetherebundled}--\eqref{athirdsetherebundled} hold also 
for $m=0$, $a\omega<0$, while for
$n=N$ then $\infty$ is a limit point of $I_n$ and~\eqref{onesetherebundled}--\eqref{athirdsetherebundled} hold also 
for $m=0$, $a\omega>0$.
Defining the spacetime region $\mathcal{L}_n =\{r _{n,1} < r < r_{n,2}\}$, then
there exists an $\alpha_n$ 
satisfying
\[
0\le  \alpha_n \le  \frac{a}{2Mr_+}
\]
such that 
\begin{equation}
\label{timelikeherereflink}
g(T+\alpha_n\Omega_1,  T+ \alpha_n \Omega_1)\le b<0
\end{equation}
in $\mathcal{L}_n$.   For $n=N-1$ and $n=N$, we may in fact take
$\alpha_n=0$.
\item
For all $n=N_s+1,\ldots, N$,  all $\gamma\in I_n$ and all 
admissible $(\omega, m\ne0 , \Lambda)$ with $\omega/am=\gamma$, we have 
\begin{equation}
\label{finalchinatpropertytwo}
\chi_{\natural}(r, \omega, m, \Lambda)\gtrsim 1
\end{equation}
in $\{r_+< r\le r_{n,1}\} \cup \{ r\ge r_{n, 2}\}$ (where the constant implicit in $\gtrsim$ is in particular independent of frequency),
and similarly for $n=N-1$ if $m=0$, $a\omega<0$ and for $n=N$ if $m=0$, $a\omega>0$,
provided that the parameter  $b_{\rm trap}>0$ in~\eqref{betterdefhere} is
sufficiently small. 
\end{itemize}
\end{proposition}
\begin{proof}
This follows easily from Propositions~\ref{superradpotentialprop}
and~\ref{nonsuperradpotentialprop} and compactness considerations.
To ensure properties~\eqref{finalchinatpropertyone} and~\eqref{finalchinatpropertytwo}
we must require $b_{\rm trap}>0$ sufficiently small in~\eqref{betterdefhere}.
\end{proof}

\begin{remark}
We note that the we could in fact combine the frequency ranges $I_{N-1}$ and $I_N$,
(and thus $\mathcal{F}_{N-1}$ and $\mathcal{F}_N$) as our currents to be defined
in Section~\ref{translationsection} 
will in fact be identical in these two ranges). We have kept them apart
simply so that the $I_n$ are all connected intervals.
\end{remark}

\subsection{Fixed frequency currents of~\cite{partiii, kerrscattering} 
and their relation to twisted physical
space currents of~\cite{DHRT22}}
\label{translationsection}

Let $u$ be a solution of the inhomogeneous ode
\begin{equation}
\label{theinhomoghere}
u'' + (\omega^2-V)u = G
\end{equation}
where $V$ is the potential defined by~\eqref{translationsection}.

We recall the fixed Carter frequency currents of~\cite{partiii,kerrscattering}:
\begin{eqnarray*}
{\text{Q}}^f [u] &=& f(|u'|^2 + (\omega^2-V)|u|^2) + f' {\rm Re} (u'\bar u) -\frac12 f'' |u|^2, \\
{\text{Q}}^h[u] &=& h Re (u' \bar u) -\frac12 h' |u|^2 , \\
{\text{Q}}^y[u] &=& y ( |u'|^2 +(\omega^2 -V) |u|^2 ), \\
{\text{Q}}^{\hat z}[u] &=& \hat{z}|u'-i\omega u|^2  - \hat{z} V|u|^2, \\
{\text{Q}}^z_{\rm red} [u] &=&  z |u' + i(\omega-\frac{a}{2Mr_+} m) u|^2 -z(V-V(r_+))|u|^2 ,\\
{\text{Q}}^T[u] &=& \omega {\rm Im}(u' \bar u) , \\
{\text{Q}}^{\Omega_1} [u] &=& -m{\rm Im} (u' \bar u) .  
\end{eqnarray*}
Here $f(r)$, $h(r)$, $y(r)$, $\hat{z}(r)$ and $z(r)$ are functions.

From~\eqref{theinhomoghere} we have the identities
\begin{eqnarray}
\label{thefidentity}
({\text{Q}}^f [u])' =& 2f' |u'|^2 -fV' |u|^2 -\frac12 f''' |u|^2 &+{\rm Re} (2f \bar G u' +f' \bar G u)  , \\
\label{thehidentity}
({\text{Q}}^h[u])' =& h (|u'|^2 + (V-\omega^2 ) |u|^2 ) -\frac12 h'' |u|^2& + h {\rm Re}(u \bar G) ,\\
\label{theyidentity}
({\text{Q}}^y[u])' =& y' (|u'|^2 +(\omega^2-V) |u|^2 ) - y V' |u|^2 &+2 y {\rm Re}(u' \bar G) ,\\
({\text{Q}}^{\hat{z}}[u] )' = & \hat{z}' |u-i\omega u|^2 - (\hat{z}V)' |u|^2 &+2\hat{z}{\rm Re} \left(G\overline{u'-i\omega u}\right) ,\\
\label{redidentity}
({\text{Q}}^z_{\rm red}[u] )' = &   z' |u' + i(\omega-\frac{a}{2Mr_+} m) u|^2 &\\
\nonumber
& -(z(V-V(r_+)))' |u|^2&
+2z{\rm Re}\left(G\overline{u'+i(\omega-\frac{a}{2Mr_+} m)u}\right), \\
({\text{Q}}^T[u] )'=& 0&+ \omega {\rm Im}(G \bar u), \\
({\text{Q}}^{\Omega_1}[u])'=&0&-{\rm Im} (G \bar u),  \\
\nonumber
({\text{Q}}^{Z}[u])' =&0&+ \omega {\rm Im}(G\bar u)\\
&& - \frac{a}{2Mr_+}m {\rm Im}(G\bar u),\\
\label{Killingcutoffidentityfirst}
({\text{Q}}^T[u]+ \qquad&&\\
\chi_{{\rm Killing},n }\alpha_n {\text{Q}}^{\Omega_1}[u] ) ' =& 
\nonumber 
-\chi'_{{\rm Killing},n}\alpha_n m{\rm Im}(u'\bar u)&+( \omega - m\alpha_n\chi_{{\rm Killing},n}){\rm Im} (G\bar u),\\
\nonumber
( {\text{Q}}^T[u]+ \qquad&&\\
\frac{\chi_{{\rm Killing},n }(2Mr_+ \alpha_n/a)}{(1-2Mr_+ \alpha_n/a) } {\text{Q}}^{Z}[u] ) ' =& 
\label{Killingcutoffidentity}
\frac{\chi'_{{\rm Killing},n}(2Mr_+ \alpha_n/a)}{ (1-2Mr_+ \alpha_n/a) } (\omega- \frac{a}{2Mr_+} m){\rm Im}(u'\bar u)&
+\omega{\rm Im}(G\bar u) \\
\nonumber 
&&
+ \frac{\chi_{{\rm Killing},n} ( \omega - \frac{a}{2Mr_+} m)}{(1-2Mr_+\alpha_n/a)}{\rm Im} (G\bar u),
\end{eqnarray}
where we have replaced some letters from the original~\cite{partiii} so as not to conflict with the notation of the present
paper.

Let us define the ``bulk'' quantities~$\mathfrak{P}^f[u]$ by the first column on the right
hand side of~\eqref{thefidentity}--\eqref{Killingcutoffidentity}, i.e.
\[
\mathfrak{P}^f [u]:=2f' |u'|^2 -fV' |u|^2 -\frac12 f''' |u|^2,
\]
etc.

The above identities translate to physical space identities as follows:
If $f$, $h$, $y$, $\hat{z}$, $z$ are independent of frequency, then 
the above identities~\eqref{thefidentity}--\eqref{Killingcutoffidentity} 
all correspond to the divergence identity of physical space currents.
We will denote these ``corresponding currents'' as
\begin{equation}
\label{correspondcurrents}
\tilde{J}^f[\psi],\quad  
\tilde{J}^h[\psi], \quad \tilde{J}^y[\psi],  \quad
\tilde{J}^{\hat{z}}[\psi], \quad \tilde{J}^z_{\rm red}[\psi], 
\quad -\tilde{J}^T[\psi],\quad  -\tilde{J}^{\Omega_1}[\psi],
\quad  -\tilde{J}^{Z}[\psi],
\end{equation}
and these satisfy divergence identities of the form
\begin{equation}
\label{twisteddivergence}
\nabla^\mu \tilde{J}^f_{\mu} [\psi] = \tilde{K}^f [\psi]  +{\rm Re}( \tilde{H}^f [\psi]\overline{\Box_g \psi}),
\end{equation}
etc.
(The unusual choice of sign in $-\tilde{J}^T[\psi]$, $-\tilde{J}^{\Omega_1}[\psi]$
$-\tilde{J}^{Z}[\psi]$
is to retain consistency with~\cite{partiii}.)

The above currents~\eqref{correspondcurrents} correspond to twisted currents as described in Appendix A.2 of~\cite{DHRT22},
with twisting function $\beta := (r_{BL}^2+a^2)^{-1/2}$.
Recall the twisted derivative operator 
\begin{equation}
\label{nablatwisted}
\tilde\nabla_\mu \psi  := \beta \nabla_\mu( \beta^{-1}\psi )
\end{equation}
and the twisted energy momentum tensor:
\begin{equation}
\label{twistedenergymom}
\tilde{T}_{\mu\nu}[\psi]  := {\rm Re} (\tilde\nabla_\mu \psi \overline{\tilde\nabla_\nu\psi }) -\frac12 g_{\mu\nu}\left( {\rm Re} (\tilde\nabla^\alpha \psi\overline{ \tilde\nabla_\alpha \psi} )
+(r^2+a^2\cos^2\theta)^{-1}(r^2+a^2)^2\Delta^{-1} V_1  |\psi|^2 \right)
\end{equation}
where $V_1$ is defined in~\eqref{defofVone}.

The current 
\[
\tilde{J}^f_\mu[\psi] := \tilde{T}_{\mu\nu}[\psi] X^\nu + f' {\rm Re}(\psi  \overline{ \tilde\nabla_\mu \psi}) -\frac12 |\psi|^2 \nabla_\mu f'
\]
where $X=2 f (r^*)\partial_{r^*}$. (It thus corresponds to what was referred to as
$\tilde{J}^{X}_\mu +\tilde{J}_\mu^{aux, f'}$ in~\cite{DHRT22}.)
We then have
\begin{equation}
\label{introducingthecurrentyinappendix}
\tilde{J}^y_{\mu}[\psi]:= \tilde{T}_{\mu\nu}[\psi] X^\nu ,
\end{equation}
where $X=2y(r^*)\partial_{r^*}$, 
\[
\tilde{J}^h_{\mu}[\psi]= 
 h {\rm  Re} ( \psi \overline{\tilde{\nabla}_\mu \psi })- \frac{1}{2} |\psi|^2 \nabla_\mu h
\]
(the above corresponding to what was referred to as $\tilde{J}^{aux, h}$ in~\cite{DHRT22}),
\[
\tilde{J}^{\hat z}_\mu [\psi] := \tilde{T}_{\mu\nu}[\psi] X^\nu, \qquad 
\]
where $X =2\hat{z}(r^*)(\partial_{r^*} + \partial_{t^*})$
\begin{equation}
\label{reddefinitionsecretlyX}
{\tilde{J}^z_{\rm red}}_{\, \mu} [\psi] := \tilde{T}_{\mu\nu}[\psi] X^\nu, \qquad 
\end{equation}
where $X =2z(r^*)(\partial_{r^*} - \partial_{t^*} -\frac{a}{Mr_+} \partial_{\phi^*})$,
\begin{equation}
\label{twisteddefT}
\tilde{J}^T_{\mu}[\psi]: = \tilde{T}_{\mu\nu}[\psi] T^\nu ,
\end{equation}
\begin{equation}
\label{twisteddefOmega}
\tilde{J}^{\Omega_1}_{\mu}[\psi]: = \tilde{T}_{\mu\nu}[\psi] (\Omega_1)^\nu .
\end{equation}

Let us note that if $f(r)$, $y(r)$ are smooth functions in $r\in [r_+,\infty)$, then $X=2f(r^*)\partial_{r_*}$
or $X=2y(r^*)\partial_{r^*}$ 
extend to smooth vector fields beyond the horizon $\mathcal{H}^+$ to the full region
$r\ge r_0$,
while if $\Delta z(r)$ is smooth in $r\in [r_+,\infty)$, then a similar statement holds concerning
$X=2z(r^*)(\partial_{r^*} - \partial_{t^*} -\frac{a}{Mr_+} \partial_{\phi^*})$.
Given currents defined with such $f$, $y$, $z$ in $r\ge r_+$, we may define
arbitrary smooth extensions to $r\ge r_0$. 

\begin{remark}
\label{explicitchoices}
Since we shall restrict to $r_0$ close to $r_+$,
then Propositions~\ref{bcgsc} 
and~\ref{notbcgsc} of Section~\ref{generalisedcoerc}  
will hold
independently of the choice of the smooth extensions to the region $r_0\le r\le r_+$. Consequently, in the sections to follow 
we shall only make explicit choices in the region $r\ge r_+$. 
\end{remark}

The bulk currents $\tilde{K}^f[\psi]$, $\tilde{K}^y[\psi]$, etc., in~\eqref{twisteddivergence}
can then be read off from the bulk currents $\mathfrak{P}^f[u]$, $\mathfrak{P}^y[u]$ of the first
column of the right hand side~\eqref{thefidentity}--\eqref{Killingcutoffidentity} 
by making the following replacements:
\begin{align}
\label{translationone}
|u|^2 &\mapsto |\psi|^2 \\
|u '|^2 & \mapsto | \tilde\nabla_{r^*}\psi |^2 \\
\omega   m  |u|^2 &\mapsto {\rm Re} ( \tilde\nabla_t \psi  \overline{ \tilde\nabla_\phi   \psi} )    \\
m^2|u|^2  &\mapsto |\tilde\nabla_\phi \psi|^2  \\
\Lambda |u|^2  &\mapsto  |\tilde\nabla_\theta \psi|^2  +\sin^{-2}\theta |\tilde\nabla_\phi \psi|^2
+a^2 (\cos^2\theta+1) |\tilde\nabla_t \psi|^2 \\
\label{finaltranslation}
m{\rm Im} (u' \bar u) &\mapsto {\rm Re} ( \tilde\nabla_\phi \psi \overline{ \tilde\nabla_{r^*} \psi})
\end{align}
and multiplying everything by 
\[
\frac{(r^2+a^2)^2}{(r^2+a^2\cos^2 \theta)\Delta}  .
\]
For example, we have
\begin{align*}
\tilde{K}^f[\psi] &=\Biggl( 2f'| \tilde\nabla_{r^*}\psi|^2   \\ 
&\quad\quad-
f\left(
\left(
\frac{4Mra}{(r^2+a^2)^2}
\right)'
{\rm Re} ( \tilde\nabla_t \psi  \overline{ \tilde\nabla_\phi   \psi} ) -\left(\frac{a^2}{(r^2+a^2)^2}\right)'|\tilde\nabla_\phi\psi|^2\right. \\
&\quad\quad\qquad \left.
+\left(\frac{\Delta}{r^2+a^2}\right)' (  |\tilde\nabla_\theta \psi|^2  +\sin^{-2}\theta |\tilde\nabla_\phi \psi|^2
+a^2 (\cos^2\theta+1) |\tilde\nabla_t \psi|^2) \right) \\
&\quad\quad-fV_1'|\psi|^2 -\frac12 f'''|\psi|^2\Biggl) \frac{(r^2+a^2)^2}{(r^2+a^2\cos^2 \theta)\Delta} .
\end{align*}

Finally, we note that by Plancherel relations of Section~\ref{Plancherelsection}, 
we have the following:
If $u(r,\omega,m,\ell):= (r^2+a^2 )^{1/2}\mathfrak{C}\circ \mathfrak{F}_{RL}[\Psi]$,
then
\begin{equation}
\label{corollaryofPlancherel}
\int_{-\infty}^{R_{\rm freq}^*}\int_{-\infty}^{\infty} \sum_{m, \ell}  \mathfrak{P}^{f}[u](r,\omega,m,\ell)  d\omega  dr^*
= \int_{r_+\le r\le R_{\rm freq}}  \tilde{K}^{f} [\Psi] ,
\end{equation}
etc.

\subsection{The elliptic refinement: Proof of Proposition~\ref{refinedproposition}} 
\label{refinedapp}

In this section, we prove Proposition~\ref{refinedproposition}.
Recall the parameters $b_{\rm elliiptic}$  and $\gamma_{\rm elliptic}$ from Section~\ref{ellipticimprovement},
and the definition~\eqref{rellipticdef} of $r_{\rm elliptic}$.

The crux of the proof is the following fixed frequency estimate:
\begin{proposition}
\label{thisisthecrux}
Let $|a|<M$,  $(\omega, m, \Lambda)$ be admissible frequencies and let
$u$ satisfy~\eqref{theinhomoghere}. Given $b_{\rm elliptic}>0$ and $\gamma_{\rm elliptic}>0$, then we have the estimate
\begin{equation}
\label{ellipticestimatefixedfreq}
\int_{r^*_{\rm elliptic}(\omega, m) }^{(1.2R_{\rm pot})^*}  \Delta r^{-2} \iota_{V_0-\omega^2\ge b_{\rm elliptic}\omega^2}  (V_0 - \omega^2)^2  |u|^2
\lesssim \int_{(\max\{ .99r_{\rm elliptic}, r_+\})^*}^{(1.3R_{\rm pot})^*}  \Delta r^{-2} \iota_{V_0-\omega^2 \ge b_{\rm elliptic} \omega^2/2} \left( (\omega^2 +\Lambda) |u|^2 +
|G|^2 \right),
\end{equation}
where the constant implicit in $\lesssim$ depends on the choice of $b_{\rm elliptic}$ and
$\gamma_{\rm elliptic}$.
\end{proposition}

\begin{proof}
Since our implicit constants are independent of frequency,
we may assume without loss of generality $\omega\ne0$.

Let $\xi(s)$ be function such that $\xi(s)=s$  if $s\ge b_{\rm elliptic}\omega^2$, $\xi (s) =0$ for $s\le b_{\rm elliptic}\omega^2/2 $ and 
$0\le \frac{d}{ds}\xi \le 4 $, $0\le \frac{d^2}{ds^2}\xi \le 8\omega^{-2}$, and $\xi(s)\le s$ for $s\ge 0$.  
If $r_{\rm elliptic}= r_+$, let $\chi_{\rm pot}$ be a cutoff function which $=1$ for $r\le 1.2R_{\rm pot}$
and $=0$ for $r\ge 1.3R_{\rm pot}$.
If $r_{\rm elliiptic}= 2.03M$ then let  $\chi_{\rm pot}$ be a cutoff function which $=1$
in  $2.03M\le r\le 1.2R_{\rm pot}$ and $0$ if $r\le .99r_{\rm elliptic}$ or $r\ge 1.3R_{\rm pot}$.

We will consider the current ${\text{Q}}^h[u]$ of Section~\ref{translationsection}, choosing: 
\[
h:=\xi (V_0- \omega^2 ) \Delta  r^{-2} \chi_{{\rm pot}}(r) .
\]

We compute 
\[
h' = \frac{d}{ds}\xi (V_0-\omega^2)V_0' \Delta r^{-2}\chi_{\rm pot}(r) 
+  \xi (V_0- \omega^2 ) \left[ \Delta  r^{-2} \chi_{{\rm pot}}\right]'
\]
and
\[
h'' = \frac{d^2}{ds^2}\xi (V_0-\omega^2)(V_0')^2  \Delta r^{-2}\chi_{\rm pot}(r) 
+  \frac{d}{ds}\xi (V_0-\omega^2)V_0''  \Delta r^{-2}\chi_{\rm pot}(r) 
+   \xi (V_0- \omega^2 ) \left[ \Delta  r^{-2} \chi_{{\rm pot}}\right]''.
\]
We note that
\[
|h'| \lesssim |V_0'|\Delta r^{-2}   \lesssim  (\Lambda  +\omega^2)\Delta^2 r^{-4}
\iota_{\max\{.99 r_{\rm elliptic}, r_+\}\le r\le 1.3 R_{\rm pot}}
\]
while 
\begin{align}
\nonumber
|h'' |  &\lesssim  \left|\frac{d^2}{ds^2}\xi (V_0-\omega^2)\right| (\Lambda^2 +\omega^4)\Delta^3r^{-10}\chi_{{\rm pot}}(r) 
+(\Lambda +\omega^2) r^{-8} \Delta^2 \chi_{\rm pot}
+\omega^{2} \Delta r^{-4} \iota_{\max\{.99 r_{\rm elliptic}, r_+\}\le r\le 1.3 R_{\rm pot}} \\
\label{labelforthisline}
 &\lesssim \omega^{-2} \iota_{\omega^2/2\le (V_0-\omega^2)\le \omega^2} (\Lambda^2 +\omega^4)\Delta^3r^{-10}\chi_{{\rm pot}}(r) \\
 \nonumber
 &\qquad
+(\Lambda +\omega^2) r^{-8} \Delta^2 \chi_{\rm pot}
+\omega^{2} \Delta r^{-4} \iota_{\max\{.99 r_{\rm elliptic}, r_+\}\le r\le 1.3 R_{\rm pot}}.
\end{align}

Now, we claim that for all admissible frequencies $(\omega, m, \Lambda)$ for which $\omega\ne0$,
\begin{equation}
\label{yetanotherclaimhere}
\iota_{\omega^2/2\le (V_0-\omega^2)\le \omega^2} \omega^{-2} \Lambda^2 \Delta^3r^{-10}\chi_{{\rm pot}}(r)  \lesssim \Delta^2 \Lambda  r^{-6} \chi_{\rm pot} .
\end{equation}
To see this, note that if $a=0$ or $m=0$ or $|\omega/am|\ge \gamma_{\rm elliptic}$ or $\gamma:=\omega/am \in I_n$ for $n\ne 1$, then
\[
\Delta \Lambda r^{-4} \lesssim V_0 + r^{-2} \omega^2 
\]
and thus
\[
\iota_{\omega^2/2\le (V_0-\omega^2)\le \omega^2} \Delta \Lambda r^{-4} \lesssim \omega^2,
\]
whence~\eqref{yetanotherclaimhere} indeed follows.
On the other hand, if $|\omega/am |< \gamma_{\rm elliptic}$ and $n=1$, 
then, since $\chi_{\rm pot}$ is supported in $r\ge 2.01M$, we have
by~\eqref{Vzerocontrols}
\[
\chi_{\rm pot}(r) \Delta \Lambda  r^{-4} \lesssim  \chi_{\rm pot}(r) ( V_0+ r^{-2} m^2) \lesssim  
 \chi_{\rm pot}(r)  V_0
\]
and thus
\[
\iota_{\omega^2/2\le (V_0-\omega^2)\le \omega^2}
\chi_{\rm pot}(r) \Delta \Lambda  r^{-4} \lesssim  
\iota_{\omega^2/2\le (V_0-\omega^2)\le \omega^2}  \chi_{\rm pot}(r)  V_0
\lesssim \chi_{\rm pot}(r) \omega^2,
\]
whence~\eqref{yetanotherclaimhere} again follows.

Integrating~\eqref{thehidentity}, and using the above bounds for $|h'|$ and $|h''|$
(where we apply~\eqref{yetanotherclaimhere} 
to bound the first term on the right hand side of~\eqref{labelforthisline})
we obtain
\begin{align*}
&\int_{-\infty}^{\infty}   \Delta  r^{-2}  \chi_{\rm pot}  \xi(V_0-\omega^2)  
(|u'|^2 + (V_0-\omega^2)|u|^2 ) \\
&\qquad\lesssim
 \int_{  \{  V\ge b_{\rm elliptic}\omega^2/2 \} \cap \{  r\le 1.3R_{\rm pot} \} } 
  \Delta r^{-4}(\Lambda+ \omega^2) |u|^2 + 
   \Delta  r^{-2} \chi_{\rm pot} \xi(V_0-\omega^2)   | u\bar {G} |,
\end{align*}
and thus,  since $\xi^2(V_0-\Omega) \le \xi(V_0-\Omega)\cdot (V_0-\Omega)$,
we may absorb the $u$ from the $|u\bar{G}|$ into the left hand side to obtain
\begin{align*}
&\int_{-\infty}^{\infty}   \Delta  r^{-2} \chi_{\rm pot} \xi(V_0-\omega^2)  
(|u'|^2 + (V_0-\omega^2)|u|^2 ) \\
&\qquad\lesssim
 \int_{  \{  V_0\ge b_{\rm elliptic}\omega^2/2 \} \cap \{  r\le 1.3R_{\rm pot} \} } 
  \Delta r^{-4}(\Lambda+ \omega^2 +1 ) |u|^2 +  \Delta  r^{-2} \chi_{\rm pot}  |\bar {G} |^2.
\end{align*}
Now the left hand side manifestly controls 
\[
\int_{-\infty}^{\infty}  \Delta  r^{-2}  \chi_{\rm pot}\iota  (V_0-\omega^2)(V_0-\omega^2)|u|^2 
\]
whence~\eqref{ellipticestimatefixedfreq} follows immediately.
\end{proof}

Let $\psi$ now be as in the statement of Proposition~\ref{refinedproposition}, 
and apply the above to 
\[
u:=  \mathfrak{C}\circ\mathfrak{F}_{BL}[\mathring{\mathfrak{D}}^{\bf k} \chi^2_{\tau_1,\tau_2} \psi ],
\]
for all $|{\bf k}|\le k$. 
We may partition the term $G=G_1+G_2$ where $G_1$ arises from the right hand side $F$ and
$G_2$ arises from the cutoff. 

Summing over ${\bf k}$, 
Proposition~\ref{refinedproposition} then follows immediately from Proposition~\ref{thisisthecrux}, the Plancherel formulae of Section~\ref{Plancherelsection},
the definitions of ${}^{\rm elp}_{\scalebox{.6}{\mbox{\tiny{\boxed{+1}}}}}\Xk(\tau_0,\tau_1)$ and ${}^{\natural}\Xk(\tau_0,\tau_1)$ 
and the fact that
\begin{equation}
\label{givethisalabeltoo}
\iota_{V_0-\omega^2 \ge b_{\rm elliptic} \omega^2/2} \lesssim \chi_{\natural},
\end{equation}
provided that $b_{\rm trap}>0$ in~\eqref{betterdefhere} is taken sufficiently small,
depending on $b_{\rm elliptic}$
(where we recall the definition of $\chi_{\natural}$ from Section~\ref{capturingprecise}).

Let us now fix $\gamma_{\rm elliptic}>0$ such that
\begin{equation}
\label{choiceofgammazero}
[-\gamma_{\rm elliptic}, \gamma_{\rm elliptic}]\subset I_1 \setminus \cup_{n\ne 1} I_n
\end{equation}
 where
$I_1$ is the interval of Proposition~\ref{determiningthecoveringprop}. 
(This thus determines $r_{\rm elliptic}(\omega, m)$ of~\eqref{rellipticdef}.)

\begin{remark}
\label{inviewofelliptic}
In view of the above proof, we may now also quickly explain how to
immediately obtain Proposition~\ref{moreprecise} from the result stated in~\cite{partiii}.
It suffices to obtain a fixed frequency estimate for all
 $(\omega, m, \Lambda)\in \mathcal{G}_\natural$ where the degeneration function is
replaced by $\chi_{\natural}$ and  with error term controlled by the quantity actually
estimated in~\cite{partiii}, which
at fixed frequency gives an estimate that degenerates at  the maximum $r_{\rm max}(\omega, m ,\Lambda)$
of the full potential $V(\omega, m, \Lambda)$.
One may assume without loss of generality that $ \omega^2_{\rm large}\gg 1$
in the definition of $\mathcal{G}_\natural $. 
Let us note  that for $\omega^2_{\rm large}\gg 1$, then 
there is an $\epsilon>0$ independent of frequency in the
range $\mathcal{G}_\natural$ such that 
$r^0_{\rm max}-\max \{ r^0_{\rm min},r_+\} >2\epsilon$ and $|r_{\rm max}-r^0_{\rm max}|<\epsilon/2$,
while $|\hat{V}_{\rm max}-V_{\rm max}|<\epsilon\omega^2$. (Hhere $\hat{V}_{\rm max}$ denotes the
maximum value of the full potential $V$. Recall that in the notation of the present paper,
$V_{\rm max}$ denotes the maximum value of $V_0$.)
We note then that if $V_{\rm max}\ge (1+b_{\rm trap})\omega^2$, we remove the
degeneration at $r_{\rm max}$ by 
applying the estimate~\eqref{thehidentity} 
arising from  $h=\xi \Delta r^{-2} \chi_{r^0_{\rm max}}$ (where we replace $b_{\rm elliptic}$ with $b_{\rm trap}$
in the definition of $\xi$) and where $\chi_{r^0_{\rm max}}$ equals one in an $\epsilon$ neighbourhood
of $r^0_{\rm max}$ and $0$ outside a $2\epsilon$ neighbourhood. 
On the other hand,
if $V_{\rm max} \le (1-b_{\rm trap})\omega^2$, we remove the
degeneration at $r_{\rm max}$ 
by applying the estimate~\eqref{theyidentity} to 
$y = (r-r^0_{\rm max}) \chi_{r^0_{\rm max}}$. 
Note that where $\chi_{r^0_{\rm max}}=1$ all terms  in the bulk are good modulo lower
order terms because
we also have there  $-V_0'(r-r^0_{\rm max})\ge 0$, in view of the properties
of $V_0$ and the fact that by our choice
of $\epsilon$ we have there $r\ge r^0_{\rm min}$.
All error terms are absorbed by the
fact that we have already nondegenerate control from the precise result stated
in~\cite{partiii} for all quantities in the $(\epsilon, 2\epsilon)$ $r$-annulus around $r^0_{\rm max}$.
\end{remark}

Let us also note the following helpful inequality for later:
\begin{equation}
\label{helpful}
\iota_{r\le r_1} \iota_{\rm elliptic} (V_0 + \omega^2)|u|^2
 \lesssim \Delta ( \Lambda  + \omega^2 +m^2)|u|^2
\end{equation}
independent of frequency.  The prefactor $\Delta$ on the right hand side arises in front of all frequencies, not just~$\Lambda$, because
of~\eqref{useit}.

\subsection{The global bulk coercivity: Proof of Theorem~\ref{globalgenbulkcoercprop}}
\label{fundcoercivityapp}

In this section, we prove
Theorem~\ref{globalgenbulkcoercprop}.
Recall from the statement of the theorem 
that we are given $E>0$ arbitrary and $e_{\rm red}>0$ sufficiently small.
(The smallness constraints on $e_{\rm red}>0$ will appear in the course of the proof.)

In the non-superradiant regime,
 our currents will take the form
\begin{eqnarray}
\nonumber
J^{{\rm main}, n}[\psi] &:=& \tilde{J}^{y_n}[\psi] +\tilde{J}^{f_n}[\psi] + \tilde{J}^{f_{\rm fixed}} [\psi] +\tilde{J}^{y_{\rm fixed}}[\psi] +\tilde{J}^{\hat z}[\psi]+e_{\rm red} \tilde{J}^z_{\rm red}[\psi]\\
\label{willtakethefollowingformnonsuper}
&&\qquad +E
\left (   \tilde{J}^T[\psi]  +
\chi_{{\rm Killing},n}\frac{(2Mr_+ \alpha_n/a) }{(1 -2Mr_+ \alpha_n/a)}\tilde{J}^{Z}[\psi] ) \right)
\end{eqnarray}
while in the superradiant regime
our currents will take the form
\begin{eqnarray}
\nonumber
J^{{\rm main}, n}[\psi] &:=& \tilde{J}^{y_n}[\psi] +\tilde{J}^{f_n}[\psi] + \tilde{J}^{f_{\rm fixed}} [\psi] +\tilde{J}^{y_{\rm fixed}}[\psi] +\tilde{J}^{\hat z}[\psi]+e_{\rm red} \tilde{J}^z_{\rm red}[\psi]\\
\label{willtakethefollowingformsuper}
&&\qquad 
+E( \tilde{J}^T[\psi] +\chi_{{\rm Killing},n} \alpha_n \tilde{J}^{\Omega_1}[\psi]  ).
\end{eqnarray}
Here,
$\alpha_n$ is the parameter~\eqref{alphanbound},
and the functions $f_n$, $y_n$,  $\hat{z}$, $z$, 
$\chi_{n, {\rm Killing}}$ are to be described below.

The functions will have the property that for $r\ge R_{\rm freq}$,  where $R_{\rm freq}$
is a paremeter that will depend on $E$,
we will have $\chi_{n, {\rm Killing}}=0$, $f_n=0$  and 
\begin{equation}
\label{ynindependencefar}
y_n(r^*)=B_{\rm indep} 
\end{equation}
for some constant $B_{\rm indep}>0$ (to be chosen later) independent  of $n$ indexing
the frequency range. (In fact,~\eqref{ynindependencefar} will hold for all $R\ge 1.2 R_{\rm pot}$.)
We will have $z(r^*)=0$  for $r\ge r_2$.
Thus, in particular,
 for $r\ge R_{\rm freq}$, 
 we have that
 \[
 J^{{\rm main}, n}[\psi] =J^{\rm main}[\psi]
 \] 
 is
independent of~$n$.

We will explain the choice of our $n$-independent functions in
Section~\ref{firstchoiceoffetc} and prove coercivity in certain regions.
We will then treat the non-superradiant ranges first in 
Section~\ref{nonsuperradiantregimecurrentdef}
followed by the generalised 
superradiant ranges in Section~\ref{superradiantregimecurrentdef}.
We will fix all remaining
parameters in our multiplier constructions in Section~\ref{fixparamsheremust} and 
finally, we will complete the proof in Section~\ref{completingtheproofsectionhere}. 

\subsubsection{The choice of $f_{\rm fixed}$, $y_{\rm fixed}$, $\hat{z}$ and $z$}
\label{firstchoiceoffetc}

The functions $f_{\rm fixed}$, $y_{\rm fixed}$, $\hat{z}$ and $z$ will be independent of $n$. They
are chosen as follows.

We define $f_{\rm fixed}(r)$ so that
\begin{equation}
\label{fdefstuff}
f_{\rm fixed} =0 \text{\ for\ } r\le R_{\rm pot}, \quad f'_{\rm fixed} \ge 0\text{\ for\ }r\ge R_{\rm pot}, \quad f\ge 1, f'_{\rm fixed} = b_{\rm tail} r^{-2}, \, f'''_{\rm fixed} = 6b_{\rm tail} r^{-4} \text{\ for\ } r\ge 1.1R_{\rm pot},
\end{equation}
and we define $y_{\rm fixed}$ by
\begin{equation}
\label{yfixeddefsutff}
y_{\rm fixed} =0 \text{\ for\ } r\le R_{\rm pot}, \quad y_{\rm fixed} ' \ge 0\text{\ for\ }r\ge R_{\rm pot}, \quad y_{\rm fixed}'= b_{\rm tail} r^{-3}\text{\ for\ } r \ge  1.1R_{\rm pot}
\end{equation}
where $b_{\rm tail}>0$ is a small parameter which will be fixed immediately
below in Proposition~\ref{coercoftheelements}.

We define $\hat{z}$, depending on $R_{\rm freq}$, so that
\begin{align}
\label{hatzdefstuffone}
\hat{z}=0 \text{\ for\ } r\le R_{\rm pot}, \,\, \hat{z}' \ge 0\text{\ for\ } r\ge R_{\rm pot}, \,\, \hat{z}'=0 \text{\ for\ } r\ge R_{\rm freq},\\
\label{hatzdefsufftwo}
\hat{z} = r  \text{\ for\ } R_{\rm pot}\le r\le .9R_{\rm freq}, \\
\label{addthisonejust}
0\le \hat{z}' \le  r'  \text{\ for\ } 9R_{\rm pot}\le r\le R_{\rm freq}.
\end{align}

Finally, we define 
$z$ so that
\begin{equation}
\label{zdefstuff}
{z}=0 \text{\ for\ } r\ge r_2, \qquad {z}= - \frac{(r^2+a^2)}{r-r_+} \text{\ for\ } r\le r_1 .
\end{equation}
Note that by our comments in Section~\ref{translationsection}, the above indeed
define currents which smoothly extend to the horizon $r=r_+$.

We have

\begin{proposition}[Coercivity properties of the $n$-independent elements of the current]
\label{coercoftheelements}
Given $R_{\rm freq}\ge 2R_{\rm pot}$, let $f_{\rm fixed}$, $y_{\rm fixed}$, $\hat{z}$ and $z$ be defined as above
with a fixed $b_{\rm tail}>0$ sufficiently small. 
Then, in the region $1.1R_{\rm pot} \le r \le.9 R_{\rm freq}$ we have
\begin{equation}
\mathfrak{P}^{f_{\rm fixed}} [u] +\mathfrak{P}^{y_{\rm fixed}} [u]
+\mathfrak{P}^{\hat z}[u] 
\label{freqindepestimate}
  \gtrsim r^{-2} ( |u'|^2 + \omega^2|u|^2 + \Lambda |u|^2 )
					+ |u'-i\omega u|^2  +r^{-3}|u|^2
\end{equation}
while for $r\ge .9 R_{\rm freq}$ we have
\begin{equation}
\mathfrak{P}^{f_{\rm fixed}} [u] + \mathfrak{P}^{y_{\rm fixed}} [u]
+\mathfrak{P}^{\hat z}[u] 
\label{farfaroutfreqindepestimate}
  \gtrsim  r^{-2} ( |u'|^2 +r^{-1}  \omega^2|u|^2 + r^{-1} \Lambda |u|^2  +r^{-2}|u|^2).
  \end{equation}
 The inequalities remain true when further restricted to the region $r\ge 1.2R_{\rm pot}$ upon the addition of $\mathfrak{P}^{y_n} [u]$ to the right hand
 side for all values of $B_{\rm indep}\ge 0$ in~\eqref{ynindependencefar} (recall this too is independent of $n$).
\end{proposition}
\begin{proof}
This follows immediately from the formulae of Section~\ref{translationsection}, where
we have used~\eqref{Vprimefar}. (Note that one controls lower order error terms because~\eqref{Vprimefar} refers to $V$ and not $V_0$.)
\end{proof}

Let us note already that Proposition~\ref{coercoftheelements} translates into bulk coercivity statements
for ${\tilde{K}}^{y_{\rm fixed}}+ {\tilde{K}}^{f_{\rm fixed}}+{\tilde{K}}^{\hat{z}}+{\tilde{K}}^{y_n}$
and thus in particular gives the proof of Proposition~\ref{farawaybulkcoercpropmain}.

Let us also note:
\begin{proposition}[Coercivity properties of the redshift current]
\label{redshiftcoercpropfreq}
In $r_+\le r\le r_1$ we have,
\begin{equation}
\label{mredcoerc}
\mathfrak{P}^{z}_{\rm red}  [u] \gtrsim  
\Delta^{-1}  |u' + i(\omega-\frac{a}{2Mr_+} m) u|^2+
  \Delta  (\Lambda |u|^2 +|u|^2).
\end{equation}
\end{proposition}
\begin{proof}
This follows immediately from~\eqref{redidentity} and the definition~\eqref{zdefstuff}
in view of~\eqref{forusewithredshift} and the constraint $r_1<r_{\rm pot}$.
\end{proof}

We define $\chi_{\rm red}$ a cutoff such that $\chi_{\rm red}=1$ for $r\le r_2$
and $\chi_{\rm red}=0$ for $r\ge r_1$.

\subsubsection{The non-superradiant ranges}
\label{nonsuperradiantregimecurrentdef}

Given a nonsuperradiant frequency range indexed by $n$, 
let us recall the degeneration function $\chi_n$ defined in Section~\ref{nowforthebulk}, 
which vanishes in $[r_{n,1}, r_{n,2}]$,  and satisfies
$\{\chi_n\ne 1\} \subset  (r'_{n,1} ,r'_{n,2})$,
where $[r_{n,1},r_{n,2}]\subset (r'_{n,1} ,r'_{n,2})\subset  (r_{\rm pot}, R_{\rm pot})$.
Recall also
the auxiliary $\tilde\chi_n$, which  also
vanishes in $[r_{n,1},r_{n,2}]$ but  which equals $1$ in the support of $\chi_n$.
Recall finally the $n$-independent functions
$f_{\rm fixed}$, $y_{\rm fixed}$, $\hat{z}$, $z$ defined in Section~\ref{firstchoiceoffetc}.

The crux of the proof for non-superradiant frequency ranges $\mathcal{F}_n$
will be the following fixed frequency statement.

\begin{proposition}[Non-superradiant degenerate bulk coercivity]
\label{bulkfornonsuperduper}

Let $\mathcal{F}_n$ be a nonsuperradiant range.
Let  $B_{\rm indep}>0$ be sufficiently
large. Let the parameter
$b_{\rm elliptic}>0$ from~\eqref{elllipticdefinitionreg} be sufficiently small (with smallness requirement depending on $B_{\rm indep}$).
Then there exist functions   $y_n(r)$,  $f_n(r)$ 
 with $y_n(r)=B_{\rm indep}$ for $r\ge 1.1R_{\rm pot}$,
 and $f_n(r)=0$ for $r\ge R_{\rm pot}$, and given arbitrary 
 $E>0$, a function  $\chi_{n,{\rm Killing}}(r)$ satisfying
   $\chi_{n,{\rm Killing}}(r)=0$ for $r\ge R_{\rm freq}$ 
such that the following holds for all sufficiently small $e_{\rm red}\ge0$:
For all  admissible frequencies $(\omega, m, \Lambda)$ such that
$(\omega, m) \in \mathcal{F}_n$, we have in the region
$r_+ \le r\le R_{\rm freq}$ the coercivity bound
\begin{align}
\nonumber
\mathfrak{P}^{y_n}[u] + \mathfrak{P}^{f_n} [u] +\mathfrak{P}^{f_{\rm fixed}} [u]+\mathfrak{P}^{y_{\rm fixed}}[u] +\mathfrak{P}^{\hat z}[u]+e_{\rm red} \mathfrak{P}^z_{\rm red}[u]
\\
\nonumber
-E\, \mathfrak{P}^T[u]
-E\, \frac{(2Mr_+ \alpha_n/a) }{(1 -2Mr_+ \alpha_n/a)} \mathfrak{P}^{\chi_{n,{\rm Killing}} Z }[u]   \\
\nonumber
  \ge c\Delta r^{-4} \chi_n ( |u'|^2 +r^{-1} \omega^2|u|^2 + r^{-1} \Lambda |u|^2 )
  +ce_{\rm red}\chi_{\rm red} \Delta^{-1}  |u' + i(\omega-\frac{a}{2Mr_+} m) u|^2\\
  \label{fixedfrequencycoerciveestimate}
					-C\Delta r^{-6} \tilde\chi_n |u|^2 
					- C\Delta \iota_{\rm elliptic} (V_0 + \omega^2)|u|^2.
\end{align}

Moreover, for $\epsilon_{\rm red}>0$ sufficiently small, there exists
an $r_{\rm close}(e_{\rm red})>r_+$ such that   
in the region $r_+\le r \le r_{\rm close}$, inequality~\eqref{fixedfrequencycoerciveestimate}
in fact holds without the final term on the right hand side.
\end{proposition}

\begin{remark}
\label{justforcf}
It is specifically Proposition~\ref{bulkfornonsuperduper} 
which constrains $R_{\rm freq}$ to be large depending on $E$.
We emphasise that, according to our conventions, the constants $c$, $C$ on the right hand side
of~\eqref{fixedfrequencycoerciveestimate} will potentially depend on
our yet to be chosen parameters $b_{\rm elliptic}$, $B_{\rm indep}$ and
 $E$  (though they will in fact not depend on $e_{\rm red}$).
\end{remark}

\begin{proof}
We will first prove~\eqref{fixedfrequencycoerciveestimate} 
in the case $e_{\rm red}=0$. The statement for small 
enough positive $e_{\rm red}>0$ then follows immediately from
Proposition~\ref{redshiftcoercpropfreq}, given that we may absorb all negative terms
in $\mathfrak{P}^{z}_{\rm red} $ away from the horizon into the terms already controlled.
Note that indeed, we may remove the last term from inequality~\eqref{fixedfrequencycoerciveestimate} 
in some region 
 $r_+\le r \le r_{\rm close}(e_{\rm red})$, since in this region
 we have $\Delta \iota_{\rm elliptic}(V_0+\omega^2)
 \lesssim \Delta^2( \Lambda^2 +\omega^2+m^2)$ by~\eqref{helpful}.

In what follows thus, we consider only the case $e_{\rm red}=0$.

Recall the degeneration function $\chi_n$ defined in Section~\ref{nowforthebulk}.
For all $B_{\rm indep}$ sufficiently large, we define $y_n(r^*)$ to be a function such 
that $y_n(r)$ smoothly extends to $r=r_+$ and such that the following are satisfied:
\begin{align}
\label{ypropertiesone}
y_n=0 \quad\text{in\ }\{\chi_n=0\}, \\
\label{ypropertiestwo}
y'_n \ge \chi_n\Delta/r^2 + B_{\rm large} |y_n| \Delta/r^2 \quad  \text{for all\ } r\in[r_+ 
,1.1R_{\rm pot}],\\
\label{ypropertiesthree}
y_n'\ge 0   \quad  \text{for all\ } r\ge 1.1R_{\rm pot}, \qquad y_n'=0 \quad   \text{for all\ } r\ge 1.2R_{\rm pot},\\
\label{makesame}
y_n = B_{\rm indep} \text{\ for all\ }r\ge 1.2R_{\rm pot}, \qquad B_{\rm indep}>0\text{\ constant independent of\ }n,\\
\label{alsoboundheremakeexplicit}
0\le y_n' \le  C(B_{\rm large}, B_{\rm indep})\chi_n  \quad  \text{for all\ } r\in[r_+ 
,1.1R_{\rm pot}],
\end{align}
where $B_{\rm large}\gg 1$ is to be determined below (depending only on $B_{\rm indep}$ and $b_{\rm elliptic}$).
We define $f_n(r^*)$ to be
\begin{equation}
\label{fndefnonsuper}
f_n = 0 {\rm\ in\ } \{\chi_n=0\},
\qquad f_n' = \chi_n^{\frac12} \text{\  in\ } \{\chi_n\le b_{\rm small}\},
\qquad
f_n= 0\text{\  in\ } \{\chi_n\ge 2b_{\rm small} \},
\end{equation}
where $b_{\rm small}>0$ is a  sufficiently small parameter which we immediately fix so that the support of $f$
is completely contained in the enlarged region $(r'_{n,1}, r'_{n,2})$.
Note that if $\chi_n$ vanishes to sufficiently high order, then $|f_n| \gtrsim \chi_n$
in $ \{\chi_n\le b_{\rm small}\}$.

The function $y_n$ will moreover be independent of the choice of $B_{\rm indep}$ in the region $r_+\le r\le 1.1R_{\rm pot}$. It is the necessity of taking $B_{\rm large}$ in~\eqref{ypropertiestwo} 
which will constrain $B_{\rm indep}$
to be sufficiently large, in view of the condition that $y'_n\ge 0$. We note that our assumptions
imply in particular that $y_n\le B_{\rm indep}$.

Thus, let us remark already that by the above definition, the
 $|u'|^2$ term on the right hand side of~\eqref{fixedfrequencycoerciveestimate} 
 is manifestly controlled in the region $r_+\le r \le 1.1R_{\rm pot}$ by 
$ \mathfrak{P}^{y_n} [u]$ and in the region $r \ge 1.1R_{\rm pot}$ by
$\mathfrak{P}^{f_{\rm fixed}} [u] $.
(In the former region, the only competing term comes from $\mathfrak{P}^{f_n}[u]$
in the region  $ \{b_{\rm small} \le \chi_n\le 2b_{\rm small}\}$. If $B_{\rm large}$
is sufficiently large, this can be absorbed in the term coming from $ \mathfrak{P}^{y_n} [u]$.)

Note that requiring $\chi'_{{\rm Killing}, n}\le 0$, we have
\begin{eqnarray}
\label{seehere}
- \mathfrak{P}^T[u]
- \frac{(2Mr_+ \alpha_n/a) }{(1 -2Mr_+ \alpha_n/a)} \mathfrak{P}^{\chi_{n,{\rm Killing}} Z }[u]  
&=&
-\frac{\chi'_{{\rm Killing},n}(2Mr_+ \alpha_n/a)}{ (1-2Mr_+ \alpha_n/a) } (\omega- \frac{a}{2Mr_+} m){\rm Im}(u'\bar u)\\
\nonumber
&=&
-\frac{\chi'_{{\rm Killing},n}(2Mr_+ \alpha_n/a)}{ (1-2Mr_+ \alpha_n/a) } (\omega- \frac{a}{2Mr_+} m){\rm Im}((u'-i\omega u)\bar u)\\
\nonumber
&&\qquad -\frac{\chi'_{{\rm Killing},n}(2Mr_+ \alpha_n/a)}{ (1-2Mr_+ \alpha_n/a) } (\omega- \frac{a}{2Mr_+} m){\rm Im}(i\omega u \bar u)\\
\nonumber
&=&
-\frac{\chi'_{{\rm Killing},n}(2Mr_+ \alpha_n/a)}{ (1-2Mr_+ \alpha_n/a) } (\omega- \frac{a}{2Mr_+} m){\rm Im}((u'-i\omega u)\bar u)\\
&&\qquad -\frac{\chi'_{{\rm Killing},n}(2Mr_+ \alpha_n/a)}{ (1-2Mr_+ \alpha_n/a) } (\omega- \frac{a}{2Mr_+} m)\omega | u |^2\\
\nonumber
&\ge &-\frac{\chi'_{{\rm Killing},n}(2Mr_+ \alpha_n/a)}{ (1-2Mr_+ \alpha_n/a) } (\omega- \frac{a}{2Mr_+} m){\rm Im}((u'-i\omega u)\bar u)
\end{eqnarray}
where for the last inequality we have used the fact that
$\alpha_n (\omega-\frac{a}{2Mr_+} m) \omega/a = \alpha_n a m^2  \frac{\omega}{am}   (\frac{\omega}{am} - \frac{1}{2Mr_+}) \ge 0$ in the non-superradiant
 regime to drop the second term.

 Let us now define $\chi_{{\rm Killing}, n}$ such that
 \begin{equation}
 \label{chikillingconstraint}
 \chi_{{\rm Killing}, n} =1 \text{\ for\ }r\le 1.1R_{\rm pot}, \qquad
 0\le -\chi'_{{\rm Killing}, n} \le \epsilon r^{-1} \text{\ for\ }r\ge 1.1R_{\rm pot},
 \qquad 
 \chi_{{\rm Killing},n} =0 \text{\ for\ }r\ge .9R_{\rm freq}. 
 \end{equation}
 Note that given arbitrary $\epsilon>0$, we can ensure~\eqref{chikillingconstraint}
 provided we take $R_{\rm freq}$ sufficiently large.

 It follows from the above and inequality~\eqref{freqindepestimate} 
 of Proposition~\ref{coercoftheelements} and $\eqref{makesame}$ (and the fact that
 by~\eqref{Vprimefar}, we have $V'<0$
 for $r\ge R_{\rm pot}$),
 that given $E>0$, we have
 in the region $1.2 R_{\rm pot}\le r\le .9R_{\rm freq}$
\begin{eqnarray}
\nonumber
\mathfrak{P}^{y_n}[u]+ \mathfrak{P}^{f_{\rm fixed}} [u] +\mathfrak{P}^{y_{\rm fixed}} [u]+\mathfrak{P}^{\hat{z}} [u] +
  E\, \mathfrak{P}^{T+\chi_{{\rm Killing},n}\Omega_1}[u] 
&\ge &
\nonumber
c r^{-2} ( |u'|^2 + \omega^2|u|^2 + \Lambda |u|^2 )
					+c |u'-i\omega u|^2 \\
				\nonumber
					&&\qquad
  -E\chi'_{{\rm Killing}, n} \alpha_n m{\rm Im} ((u' -i\omega u)\bar{u})\\
 &\ge &
 \nonumber
 c r^{-2} ( |u'|^2 + \omega^2|u|^2 + \Lambda |u|^2 )
					+c |u'-i\omega u|^2 \\
					\label{absorbthekilling}
					&&\quad
					- r^{-2} E\epsilon m^2|u|^2  - E\epsilon |u'-i\omega u|^2.
\end{eqnarray}
Recalling that $m^2\le \Lambda$, we see that given arbitrary $E>0$, we may chose
$\epsilon$ small enough (thus requiring $R_{\rm freq}$ large enough!) so that we may
absorb the terms in line~\eqref{absorbthekilling} into the previous terms.
It follows then from the support properties of the elements of the current
and from inequality~\eqref{farfaroutfreqindepestimate} of 
Proposition~\ref{coercoftheelements} that~\eqref{fixedfrequencycoerciveestimate}
 holds
for $r\ge 1.2R_{\rm pot}$.
 
 Thus it suffices to focus on the region $r_+\le  r\le 1.2R_{\rm pot}$.
 
 With all our current functions defined, let us note that 
 the zeroth order terms arising in the bulk, including those additional terms 
 arising from replacing   $V$, $V'$ with $V_0$, $V_0'$,
 may be estimated by the expression 
 \begin{equation}
 \label{lowerordererrorsestimatedby}
 C\Delta r^{-6}\tilde\chi_n |u|^2 
 \end{equation}
 appearing on the right hand side of~\eqref{fixedfrequencycoerciveestimate}
 where $\tilde\chi_n$ is the auxiliary degeneration function defined
 in Section~\ref{nowforthebulk} (we may in fact multiply the above by $\iota_{r \le 1.2R_{\rm pot}}$).
 Thus, in what follows we may indeed ignore zeroth order terms in the bulk identities
 and, moreover, replace $V$ by $V_0$, etc.

 We will first show that
 \eqref{fixedfrequencycoerciveestimate}  holds  in this region
(with $e_{\rm red}=0$).

We consider separately the subranges of $r_+\le r\le 1.2R_{\rm pot}$ 
for which
\begin{itemize}
\item[(i)] $V_0-\omega^2\ge b_{\rm elliptic}\, \omega^2$,
\item[(ii)] $ \omega^2 -V_0 \ge b_{\rm pot}\, \omega^2$,
\item[(iii)] $-b_{\rm pot}\,\omega^2  \le V_0-\omega^2\le  b_{\rm elliptic}\,\omega^2$, \qquad $|V_0'| \le \Delta b'_{\rm pot} \,\omega^2$,
\item[(iv)] $-b_{\rm pot}\, \omega^2  \le V_0-\omega^2\le  b_{\rm elliptic}\, \omega^2$, \qquad  $ V_0' > \Delta b'_{\rm pot} \, \omega^2$,
\item[(v)] $-b_{\rm pot}\, \omega^2 \le V_0-\omega^2\le  b_{\rm elliptic}\, \omega^2$, \qquad $ V_0' <- \Delta b'_{\rm pot} \, \omega^2$,
\end{itemize}
where $b_{\rm pot}>0$, $b'_{\rm pot}>0$ are the parameters from Proposition~\ref{determiningthecoveringprop}.

Let us note that for $0<b_{\rm elliptic}\le b_{\rm pot}$,
it follows from~\eqref{athirdsetherebundled} 
of Proposition~\ref{determiningthecoveringprop}  
that the range (iii) is completely
contained in the vanishing set of $\chi_n$. Thus inequality~\eqref{fixedfrequencycoerciveestimate}
trivially holds in this range.  We henceforth require $b_{\rm elliptic}$ to be sufficiently small so that (iii) is indeed contained
in the vanishing set of $\chi_n$. 
We will  further restrict the smallness of $b_{\rm elliptic}$ further down.

Concerning the range (i), let us recall that restricted to $r\in [r_+, 1.2R_{\rm pot}]$, 
the support of $\iota_{\rm elliptic}$
is precisely the set $V_0-\omega^2 \ge b_{\rm elliptic} \omega^2$, since by our 
choice~\eqref{choiceofgammazero} of $\gamma_{\rm elliptic}$  and the fact that $n\ne 1$, we have
(recall the definition~\eqref{rellipticdef}) that $r_{\rm elliptic}(\omega, m)=r_+$.
Moreover, we have~\eqref{officialLambdaboundfromVzero} for non-superradiant $n$.
Note that for $r\le r_{\rm pot}$, we have $-y_nV_0' \gtrsim \Delta \Lambda$ in the range (i).
Since moreover we have 
$V_0-\omega^2 \gtrsim V_0+\omega^2$  in this range, 
it follows that
\[
\mathfrak{P}^{y_n}[u] + \mathfrak{P}^{f_n} [u] +\mathfrak{P}^{f_{\rm fixed}} [u]+\mathfrak{P}^{y_{\rm fixed}}[u]\ge
c\Delta \chi_n (|u'|^2+ \omega^2 |u|^2 +  \Lambda|u|^2 )
-C\Delta \iota_{\rm elliptic} (V_0 + \omega^2)|u|^2 -C\Delta  \tilde\chi_n |u|^2 
\]
whence
inequality~\eqref{fixedfrequencycoerciveestimate} indeed holds
(with $e_{\rm red}=0$).

In the range (ii), we note that
by our requirement~\eqref{ypropertiestwo}, we have if 
$1.1 R_{\rm pot}\le r\le 1.2 R_{\rm pot}$ then
\[
 -f_{\rm fixed} V_0' |u|^2 + 
y_{\rm fixed} (\omega^2-V_0) |u|^2 ) - y_{\rm fixed}  V_0 ' |u|^2 +
 (y_n'(\omega^2-V_0) -y_nV_0' )|u|^2 \gtrsim (\omega^2 +\Lambda )|u|^2  
\]
while if 
$r\le 1.1R_{\rm pot}$ then
\[
(y_n'(\omega^2-V_0) -y_nV_0' )|u|^2 \ge (\Delta r^{-2} b_{\rm pot}\omega^2 (\chi_n + B_{\rm large} |y_n| )
-y_n V_0'   )|u|^2.
\]
Now let us note that from~\eqref{onesidedwhereneg}, we have the one sided
bound 
\[
- V_0'\lesssim \Delta \omega^2
\] 
in $r_+\le r\le r_{\rm pot}$.
Now let us note that we have $\omega^2 \gtrsim \Lambda \Delta $ in the range (ii) for
$r_+\le r\le 1.2R_{\rm pot}$. (This follows
because  we have $\omega^2 \gtrsim a^2m^2$ in view of the restriction to nonsuperradiant frequencies,
and thus $\omega^2 \gtrsim |V_\gamma|$ and hence $\omega^2 \gtrsim V_0 - V_\gamma$.)  On the other hand, 
in $r\le 1.2R_{\rm pot}$ we have 
\[
|V'_0| \lesssim \Delta \Lambda + \Delta \omega^2.
\]
It follows from the above and the sign of $y_n$ near $r_+$ that if $r\le 1.1R_{\rm pot}$,
we may fix $B_{\rm large}$ in~\eqref{ypropertiestwo} (depending on the choice of $b_{\rm pot}$ already
fixed!)~so that 
\[
(y_n'(\omega^2-V_0) -y_nV_0' ) \gtrsim (\chi_n + |y_n| )r^{-2} (\Delta \omega^2 +\Delta\Lambda).
\]
This again yields~\eqref{fixedfrequencycoerciveestimate} for this range 
(with $e_{\rm red}=0$). Note that by our comments after the definition of $f_n$, fixing
$B_{\rm large}$ determines the sufficiently largeness restriction  on $B_{\rm indep}$.

On the other hand, it follows from~\eqref{onesetherebundled} 
and~\eqref{anothersetherebundled} of Proposition~\ref{determiningthecoveringprop} and
the definition of $y_n$ 
that range (iv) is contained in the set $\{ f_n\le 0, y_n \le 0\}$ whereas the range
(v) is contained in the set $\{f_n\ge 0, y_n\ge 0\}$.   Note  then that, the $u'$ term aside,
the signs of terms in these two ranges are only competing
if $0<V_0 -\omega^2  \le b_{\rm elliptic}$, in which case, by restricting $b_{\rm elliptic}$ to be
even smaller (depending on  $B_{\rm large}$ which has already been fixed and depending on $B_{\rm indep}$), we may indeed
ensure that   
\[
-y_nV_0'  -  y_n' (V_0-\omega^2) -f_nV'_0 \gtrsim \chi_n \Delta r^{-2} ( \omega^2 + V_0 ).
\]
Here we are using that $0\le y_n' \lesssim \chi_n $ with an implicit constant which depends
on $B_{\rm indep}$ and  the already chosen $B_{\rm large}$, which
follows from~\eqref{alsoboundheremakeexplicit}, 
while $y_n+f_n \gtrsim \chi_n$, which follows by combining~\eqref{ypropertiestwo} 
and~\eqref{fndefnonsuper}. 
This again implies that~\eqref{fixedfrequencycoerciveestimate} holds in this range
(with $e_{\rm red}=0$).

We have thus shown~\eqref{fixedfrequencycoerciveestimate} in the entire region
$r_+ \le r\le  1.1R_{\rm pot}$ with $e_{\rm red}=0$.

By our comments at the beginning of the proof, the full statement~\eqref{fixedfrequencycoerciveestimate} for sufficiently small $e_{\rm red}>0$
as well as the additional statement regarding $r_{\rm close}(e_{\rm red})$ follow.
\end{proof}

\subsubsection{The generalised superradiant ranges}
\label{superradiantregimecurrentdef}
We recall again the $n$-independent functions
$f$, $y_{\rm fixed}$, $\hat{z}$, $z$ defined already in Section~\ref{firstchoiceoffetc}.

The crux of the proof is again a fixed frequency statement. We distinguish the
cases $n\ne0$ and $n=0$. 

We first give a proposition for the case $n\ne0$.

\begin{proposition}[Generalised superradiant bulk coercivity for $n\ne 0$]
\label{superrangepropfreq}
Let $\mathcal{F}_n$ be a generalised superradiant range for $n\ne 0$.
Let  $B_{\rm indep}>0$ be sufficiently
large. Let the paramater $b_{\rm elliptic}>0$ from~\eqref{elllipticdefinitionreg}  be sufficiently small (with smallness
depending on $B_{\rm indep}$).
Then there exist  functions 
 $y_n(r)$, $f_n(r)$, $\chi_{n,{\rm Killing}}(r)$
 with $y_n(r)=B_{\rm indep}$ for $r\ge 1.1R_{\rm pot}$ and $\chi_{n,{\rm Killing}}(r)=0$, $f_n(r)=0$ for $r\ge R_{\rm pot}$ such that, given 
 $E>0$ arbitrary,  the following holds for all sufficiently small $e_{\rm red}\ge 0$:
For all  admissible frequencies $(\omega, m, \Lambda)$ such that  $(\omega, m) \in \mathcal{F}_n$, we have in the region
$r_+ \le r\le R_{\rm freq}$ the coercivity bound
\begin{align}
\nonumber
\mathfrak{P}^{y_n}[u]+&
\mathfrak{P}^{f} [u] +\mathfrak{P}^{y_{\rm fixed}}[u] +\mathfrak{P}^{\hat z}[u]+e_{\rm red} \mathfrak{P}^z_{\rm red}[u]-E\, \mathfrak{P}^{T+\chi_{n,{\rm Killing}}\alpha_n\Omega_1}[u]   \\
&\nonumber
 \ge    c\Delta r^{-4} ( |u'|^2 + r^{-1}\omega^2|u|^2 + r^{-1} \Lambda |u|^2 )
   +ce_{\rm red}\chi_{\rm red} \Delta^{-1}  |u' + i(\omega-\frac{a}{2Mr_+} m) u|^2\\
   \label{fixedfrequencycoerciveestimatesuperduper}
&\qquad\qquad					-C\Delta r^{-6} |u|^2 
					- C\Delta \iota_{\rm elliptic} (V_0+\omega^2)|u|^2.
\end{align}

Moreover, for $\epsilon_{\rm red}>0$ sufficiently small, there exists
an $r_{\rm close}(e_{\rm red})>r_+$ such that 
in the region $r_+\le r \le r_{\rm close}$, inequality~\eqref{fixedfrequencycoerciveestimatesuperduper}
in fact holds without the final term on the right hand side.
\end{proposition}

\begin{remark}
As with Proposition~\ref{bulkfornonsuperduper}, 
we again reiterate that, according to our conventions, the constants $c$, $C$ on the right hand side
of~\eqref{fixedfrequencycoerciveestimatesuperduper} will depend on the yet to be chosen
parameters $B_{\rm indep}$, $E$ and $b_{\rm elliptic}$ (though not $e_{\rm red}$).
\end{remark}

\begin{proof}

As in the proof of Proposition~\ref{bulkfornonsuperduper}, 
it suffices to show~\eqref{fixedfrequencycoerciveestimatesuperduper} for $e_{\rm red}=0$.
The more general statement and the claim about $r_{\rm close}$ will then follow immediately
as before.

Recall the intervals $[r_{n,1}, r_{n,2}]$ from Proposition~\ref{determiningthecoveringprop},
where $r_{n,2}\le R_{\rm pot}$.
Requiring $0<b_{\rm elliptic} \le b_{\rm pot}$, we have that 
 for all admissible $(\omega, m, \Lambda)$ with  $(\omega, m)\in \mathcal{F}_n$,  
\begin{equation}
\label{entirelyhere}
[r_{n,1},r_{n,2}]\subset {\bf R}_{\rm elliptic} (\omega, m, \Lambda),
\end{equation}
i.e.~$\iota_{\rm elliptic} (r) =1$  on $[r_{n,1}, r_{n,2}]$.
Here we use the fact that by our choice of $\gamma_{\rm elliptic}$,
if $n\ne 1$ then $r_{\rm elliptic}(\omega,m) =r_+$, while if $n =1$, 
$r_+(\omega,m) <4M<r_{n,1}<r_{n,2}<6M<1.2R_{\rm pot}$.

We define $\chi_{{\rm Killing}, n}$ so that
\begin{equation}
\label{definekilling}
\chi_{{\rm Killing}, n} =1 \text{\ for all } r\le r_{n,1},  \qquad
\chi_{{\rm Killing}, n} =0 \text{\ for all } r\ge r_{n,2}, 
\end{equation}
and define
$y_n(r)$ to be smooth in $r\in [r_+,\infty)$ and to satisfy
\begin{equation}
\label{defineytosathere}
y_n \left(\frac{r_{n,1}+r_{n,2}}2\right)  =0,\quad  
2\Delta r^{-2}B_{\rm large} (1+|y_n|) \ge 
y_n'\ge \Delta r^{-2}  B_{\rm large} (1+|y_n|) \text{\ for all\ }r\in [r_+,1.1R_{\rm pot}],
\end{equation}
and also~\eqref{ypropertiesthree},~\eqref{makesame},~\eqref{alsoboundheremakeexplicit}.
Again, $B_{\rm large}\gg 1$ will be a large parameter to be determined in what
follows, and this will also determine the largeness requirement on $B_{\rm indep}$. We may 
again take $y_n$
independent of $B_{\rm indep}$ in the region $[r_+,1.1R_{\rm pot}]$. 

In the case $n\ne1$, we define $f_n(r)=0$ identically.
In the case $n=1$, 
we define $f_n$ such that
$f_n' \ge 0$ for all $r>r_+$,  $f_n=0$ for $r\ge 2.1M$, and
\begin{equation}
\label{route206}
f_n'  \ge  \Delta, \qquad  
f_n \le - \tilde{B}_{\rm large}  \text{\ for\ all\ } r\in [r_+, 2.06M]
\end{equation}
where $\tilde{B}_{\rm large}\gg 1$ is again a large parameter (depending on $B_{\rm indep}$) which must be determined.

In view of the support of $\chi_{{\rm Killing}, n}$ and Proposition~\ref{coercoftheelements},
it follows as in the proof of Proposition~\ref{bulkfornonsuperduper} that~\eqref{fixedfrequencycoerciveestimatesuperduper} indeed holds for $r\ge 1.2R_{\rm pot}$.

We will now show that~\eqref{fixedfrequencycoerciveestimatesuperduper} 
holds for $r_+\le r \le 1.2R_{\rm pot}$ (with  $e_{\rm red}=0$).
Recall as before that the $y' |u'|^2$ term in the bulk already yields control
of the $|u'|^2$ term, and in view of the presence of~\eqref{lowerordererrorsestimatedby} on
the right hand side of~\eqref{fixedfrequencycoerciveestimatesuperduper}, without
now the $\tilde\chi_n$ factor, we may replace $V$, $V'$ by $V_0$, $V_0'$ in all calculations modulo 
error terms controlled by this.

 Let us consider first the case $n\ne 1$.
As in the proof of Proposition~\ref{bulkfornonsuperduper}, 
we
consider separately the subranges of $r_+\le r\le 1.2R_{\rm pot}$ 
 for which
\begin{itemize}
\item[(i)]
 $V_0-\omega^2\ge b_{\rm elliptic}\, \omega^2$,
\item[(ii)] $ \omega^2 -V_0 \ge b_{\rm pot}\, \omega^2$,
\item[(iii)] $-b_{\rm pot}\,\omega^2  \le V_0-\omega^2\le  b_{\rm elliptic}\,\omega^2$, \qquad $|V_0'| \le \Delta  b'_{\rm pot} \,\omega^2$,
\item[(iv)] $-b_{\rm pot}\, \omega^2  \le V_0-\omega^2\le  b_{\rm elliptic}\, \omega^2$, \qquad  $ V_0' > \Delta b'_{\rm pot} \, \omega^2$,
\item[(v)] 
$-b_{\rm pot}\, \omega^2 \le V_0-\omega^2\le  b_{\rm elliptic}\, \omega^2$, \qquad $ V_0' < -\Delta b'_{\rm pot} \, \omega^2$,
\end{itemize}
where $b_{\rm pot}>0$, $b'_{\rm pot}>0$ are the parameters from Proposition~\ref{determiningthecoveringprop}.

If require $0<b_{\rm elliptic}\le b_{\rm pot}$, then 
the range (iii) is empty by~\eqref{itsemptycompletely} 
of
Proposition~\ref{determiningthecoveringprop} (superradiant
frequencies are not trapped). We  henceforth restrict to $b_{\rm elliptic}>0$ sufficiently small so that
the range (iii) is indeed empty.

Note that the support of $\chi'_{{\rm Killing}, n}$ is contained entirely in the range (i)
in view of~\eqref{entirelyhere}. We may bound the part of the bulk term arising
from the multiplier $T+\chi_{n,{\rm Killing}}\alpha_n\Omega_1$  as follows:
\begin{equation}
\label{providecontrol}
|E\chi'_{{\rm Killing}, n} \alpha_n m{\rm Im} ((u' -i\omega u)\bar{u}) |\lesssim
\epsilon |u'|^2 + (\epsilon^{-1} m^2  +\epsilon \omega^2 )|u|^2
\lesssim \epsilon |u'|^2 + \epsilon^{-1} (V_0 -\omega^2)|u|^2  
\end{equation}
where we use the fact that $n\ne 1$ to bound $m^2 \lesssim \omega^2$ (note the constant
in this $\lesssim$ is worse the finer our partition, but this is of no relevance here).
For small enough $\epsilon$ we may absorb the first term on the right hand side of~\eqref{providecontrol} in the control provided
by $y' |u'|^2$, while the second term on the right hand side of~\eqref{providecontrol} is controlled by the $\Delta \iota_{\rm elliptic} (V_0-\omega^2)|u|^2$ term on the right hand side of~\eqref{fixedfrequencycoerciveestimatesuperduper}.
Thus, the bulk term of the $T+\chi_{n,{\rm Killing}}\alpha_n\Omega_1$ multiplier is now completely accounted for.
Moreover, we have~\eqref{officialLambdaboundfromVzero} since $n\ne 1$, and
again, for $r\le r_{\rm pot}$, we have $-y_nV_0' \gtrsim \Delta \Lambda$ in the range (i).
Thus in the range (i), inequality~\eqref{fixedfrequencycoerciveestimatesuperduper} indeed holds
(with $e_{\rm red}=0$).

In the range (ii), we note that for $B_{\rm large}$ sufficently large (depending on the choice of $b_{\rm pot}$ which we recall has been fixed), inequality~\eqref{fixedfrequencycoerciveestimatesuperduper} again holds 
provided $r\ge r_{\rm pot}$ (with $e_{\rm red}=0$). On the other hand,
by~\eqref{rstatementonpotentialforallsuper},
 if $r\le r_{\rm pot}$, then 
$V_0' \gtrsim \Delta (\Lambda +\omega^2)$ while $y_n\le -b$. Thus~\eqref{fixedfrequencycoerciveestimatesuperduper} again holds
(with $e_{\rm red}=0$).

Similarly, for $b_{\rm elliptic}>0$ suitably small,
it follows now from~\eqref{superonesetherebundled} 
and~\eqref{superanothersetherebundled} 
of Proposition~\ref{determiningthecoveringprop} that 
the range (iv) is contained in the set $y_n\le -b$ whereas the range
(v) is contained in the set $y_n\ge b$.   Note then that the signs of terms in these two ranges is only competing
if $0<V_0 -\omega^2  \le b_{\rm elliptic}$, in which case, by restricting $b_{\rm elliptic}$ to be
even smaller (depending on $B_{\rm large}$ which is already fixed and $B_{\rm indep}$), we may indeed
ensure that  $-y_nV_0'  -  y_n' (V_0-\omega^2) \gtrsim \Delta r^{-2} (\Lambda+\omega^2)$,
where we have used~\eqref{defineytosathere} and~\eqref{alsoboundheremakeexplicit}.
Thus again, inequality~\eqref{fixedfrequencycoerciveestimatesuperduper} holds 
in these ranges (with $e_{\rm red}=0$).

In view of our comments at the beginning about the case $e_{\rm red}>0$ and $r_{\rm close}$, 
this completes the proof in the case $n\ne 1$.

For the case $n=1$, the above arguments establish~\eqref{fixedfrequencycoerciveestimatesuperduper} in the region $r\ge  1.2R_{\rm pot}$
since the definition of $y_n$ is the same as above, and in this region we have $f_n=0$.
Recall that for $n=1$, $4M<r_{n,1}< r_{n,2}< 6M$  and thus
the  bulk term of the $T+\chi_{n,{\rm Killing}}\alpha_n\Omega_1$ multiplier
in particular vanishes in $r\le 2.06M$. Let us fix $B_{\rm large}$ in~\eqref{defineytosathere} to be $1$.
 From 
statement~\eqref{strongerstatementonpotentialone} of Proposition~\ref{determiningthecoveringprop} 
and~\eqref{route206},
it follows that $-f_nV_0' \ge b\tilde{B}_{\rm large}\Delta r^{-2}( \Lambda+\omega^2)$ in $r\le 2.06M$,
where $\tilde{B}_{\rm large}$ is the parameter in~\eqref{route206}. 
Thus, for $\tilde{B}_{\rm large}$ sufficiently large (depending on $B_{\rm indep}$),  it follows,
again using~\eqref{defineytosathere} and~\eqref{alsoboundheremakeexplicit},
that 
\[
-f_nV_0' -y_nV_0' -y_n'(V_0-\omega^2)   \ge \frac12b\tilde{B}_{\rm large}\Delta r^{-2}( \Lambda+\omega^2),
\]
whence~\eqref{fixedfrequencycoerciveestimatesuperduper} holds in $r\le 2.06M$ 
 (with $e_{\rm red}=0$
and without the $\iota_{\rm elliptic}$ terms).

Considering now the region $2.06M\le r\le 1.2R_{\rm pot}$, we note that by~\eqref{ellipticregionfornequalzero} this
lies entirely in the region~(i), and thus we may argue exactly as with the $n\ne1$ case for region (i),
where
we now use~\eqref{Vzerocontrols} to control $m^2|u|^2$ and $\Lambda|u|^2$ terms. In
particular, note that
to argue~\eqref{providecontrol}, we use
that for $n=1$, by~\eqref{Vzerocontrols}, we have that  $\omega^2+ m^2 \lesssim V_0$ on the set $2.06M\le r\le 1.2R_{\rm pot}$. 
We obtain~\eqref{fixedfrequencycoerciveestimatesuperduper} (with
$e_{\rm red}=0$, but now with the $\iota_{\rm elliptic}$ term). 

(Note that we can indeed set $B_{\rm large}=1$ in~\eqref{defineytosathere} because the
 largeness was not necessary in region (i).)

It follows that~\eqref{fixedfrequencycoerciveestimatesuperduper} holds
in $r_+\le r\le 1.2R_{\rm pot}$
with $e_{\rm red}=0$.

The full statement in the $e_{\rm red}>0$ case again follows in view of the remarks at the beginning
of the proof.
\end{proof}

For the remaining low frequency range $n=0$, we have the following:

\begin{proposition}[Bulk coercivity for low frequency range $n=0$]
\label{bulkforzero}
Let $\mathcal{F}_0$ denote the low frequency range.
Let $B_{\rm indep}>0$ be sufficiently
large.
There exist 
 $y_0(r)$,  $f_0(r)$,  $\chi_{0,{\rm Killing}}(r)$
 with $y_0(r)=B_{\rm indep}$ for $r\ge 1.1R_{\rm pot}$ such that
 the following hold, for all
 $E>0$ arbitrary and for $e_{\rm red}\ge 0$ sufficiently small.
For all admissible frequencies $(\omega, m, \Lambda)$  such that  $(\omega, m) \in \mathcal{F}_0$,
we have in the region 
$r_+ \le r\le R_{\rm freq}$ the coercivity bound
\begin{align}
\nonumber
\mathfrak{P}^{y_0}[u]+&\mathfrak{P}^{f_0}[u]+
\mathfrak{P}^{f} [u] +\mathfrak{P}^{y_{\rm fixed}}[u] +\mathfrak{P}^{\hat z}[u]+e_{\rm red} \mathfrak{P}^z_{\rm red}[u]-E\, \mathfrak{P}^{T+\chi_{0,{\rm Killing}}\alpha_0\Omega_1}[u] \\
\nonumber
& \ge   c\Delta  ( |u'|^2 + \omega^2|u|^2 + \Lambda |u|^2 )+ ce_{\rm red}\chi_{\rm red} \Delta^{-1}  |u' + i(\omega-\frac{a}{2Mr_+} m) u|^2\\
\label{fixedfrequencycoerciveestimatelow}
&\qquad\qquad					-C\Delta r^{-5} |u|^2 
					- C \Delta \iota_{\rm elliptic} (V_0+\omega^2)|u|^2.
\end{align}
Moreover, for $\epsilon_{\rm red}>0$ sufficiently small, there exists
a $r_{\rm close}(e_{\rm red})>r_+$ such that 
in the region $r_+\le r \le r_{\rm close}$, inequality~\eqref{fixedfrequencycoerciveestimatelow}
in fact holds without the final term on the right hand side.
\end{proposition}

\begin{proof}
Recall that for $n=0$ we have $[r_{0,1},r_{0,2}]=[4M,5M]$.

We define  $\chi_{{\rm Killing},0}$ as in~\eqref{definekilling} 
with $n=0$. 
Note the bound 
\begin{equation}
\label{notetheboundhere}
E\, \mathfrak{P}^{T+\chi_{0,{\rm Killing}}\Omega_1}[u]=
E\alpha_0
 \chi'_{\rm Killing} \, {\text Q}^{\Omega_1} [u]   \ge  -E \epsilon^{-1} C|u|^2 -E\epsilon C |u'|^2,
\end{equation}
since $m^2\le B_{\rm low}$. 
Now consider  $y_0$ defined again to satisfy~\eqref{defineytosathere} with $n=0$
and also~\eqref{ypropertiesthree},~\eqref{makesame} and~\eqref{alsoboundheremakeexplicit}
and define $f_0$ by~\eqref{route206}.
In the region $r\ge 1.2 R_{\rm pot}$,~\eqref{fixedfrequencycoerciveestimatelow}
holds as before, whereas in the region $r_+\le r\le 1.2 R_{\rm pot}$, we note that
we may choose $\epsilon$ small enough so that we may absorb 
the second term on the right hand side  of~\eqref{notetheboundhere}
in the term $y' |u'|^2$. 
The proof then proceeds as in the $n=1$ case, 
noting that in the region $r_+\le r\le 2.06$, we have 
$-V_0' \ge c\Delta \Lambda -C\Delta$, while 
in the region $2.06\le r\le 1.2R_{\rm pot}$,  
we have $\Delta \Lambda \lesssim V_0+1$.
\end{proof}

\subsubsection{Fixing all remaining parameters}
\label{fixparamsheremust}

Let us fix $B_{\rm indep}$  so as to satisfy all the largeness assumptions of the propositions of this section. 
We may also now fix
the parameter $b_{\rm elliptic}$  of~\eqref{elllipticdefinitionreg}  which was
constrained (independently of $E$, $e_{\rm red}$)
by  the proof of Propositions~\ref{bulkfornonsuperduper} and~\ref{superrangepropfreq}).
Now we may
fix  $b_{\rm trap}$ to satisfy the constraints of Proposition~\ref{determiningthecoveringprop} 
(related to~\eqref{finalchinatpropertyone} as well as~\eqref{finalchinatpropertytwo})
and~\eqref{givethisalabeltoo}. (We recall that 
$\gamma_{\rm elliptic}$ was
fixed by~\eqref{choiceofgammazero}.)

Given arbitrary $E>0$, let us
 fix $R_{\rm freq}(E)$ \emph{as a function of $E$}
from Proposition~\ref{bulkfornonsuperduper} (cf.~Remark~\ref{justforcf}),

With the above choices, we have that 
all parameters appearing in the definition of the functions
of~\eqref{willtakethefollowingformnonsuper} 
and~\eqref{willtakethefollowingformsuper}
are completely determined from the stipulation of an arbitrary $E>0$ and a sufficiently 
small $e_{\rm red}>0$, with the dependence as indicated above.

\subsubsection{Finishing the proof of Theorem~\ref{globalgenbulkcoercprop}}
\label{completingtheproofsectionhere}

In view of the identifications~\eqref{translationone}--\eqref{finaltranslation} from Section~\ref{translationsection} and the Plancherel relation~\eqref{corollaryofPlancherel}
applied with 
\[
u:= \mathfrak{C}\circ\mathfrak{F} [P_n \mathring{\mathfrak{D}}^{\bf k}\chi^2_{\tau_0,\tau_1}\psi],
\]
the proof of Theorem~\ref{globalgenbulkcoercprop} follows
from Propositions~\ref{bulkfornonsuperduper},~\ref{superrangepropfreq} and~\ref{bulkforzero}, according to whether $n\ge N_s+1$,
$1\le n\le N_s$, or $n=0$, respectively, upon summing and 
integrating over Carter frequencies $(\omega, m, \ell)$,
as long as we may bound the contribution of the two negative-signed terms
on the right hand side of~\eqref{fixedfrequencycoerciveestimate},~\eqref{fixedfrequencycoerciveestimatesuperduper} and~\eqref{fixedfrequencycoerciveestimatelow}, respectively.

Concerning the contribution of the zeroth order
negative signed terms,
In view of the support of $\tilde{\chi}_n$, 
note that from~\eqref{finalchinatpropertyone}, \eqref{finalchinatpropertytwo} we have
\[
\tilde\chi_n |u|^2\lesssim \chi_{\natural} |u|^2 
\]
for all $n=0,\ldots,N$.
By definition of $\mathring{\mathfrak{D}}^{\bf k}$, then
 there either exists a  ${\bf \tilde{k}}$ with $|{\bf \tilde{k}}|=k-1$
such that
\[
u =  i\omega\, \mathfrak{C}\circ\mathfrak{F} [P_n \mathring{\mathfrak{D}}^{\bf \tilde{k}}\chi^2_{\tau_0,\tau_1}\psi]
\]
or there exists such a  ${\bf \tilde{k}}$ such that 
\[
u = im\, \mathfrak{C}\circ\mathfrak{F} 
[P_n\mathring{\mathfrak{D}}^{\bf \tilde{k}}\chi^2_{\tau_0,\tau_1}\psi].
\]
 Thus, denoting $\tilde{u}=  \mathfrak{C}\circ\mathfrak{F} [P_n\mathring{\mathfrak{D}}^{\bf \tilde{k}}\chi^2_{\tau_0,\tau_1}\psi]$,
we have
\begin{eqnarray*}
\int_{-\infty}^{R^*_{\rm freq}} \int_{-\infty}^\infty \sum_{m, \ell\ge |m|}   \tilde{\chi}_{n} \Delta r^{-5} |u|^2
&\lesssim & \int_{-\infty}^{R^*_{\rm freq}} \int_{-\infty}^\infty \sum_{m, \ell\ge |m|} 
\chi_{\natural}  \Delta r^{-5} |u|^2 \\
&\lesssim & \int_{-\infty}^{R^*_{\rm freq}} \int_{-\infty}^\infty \sum_{m, \ell\ge |m|} 
\chi_{\natural}  \Delta r^{-5} (m^2+\omega^2)  |\tilde{u}|^2 \\
 &\lesssim& {}^{\natural}\Xkminusone (\tau_0,\tau_1)[\psi].
\end{eqnarray*}

Concerning the contribution of the  first-order negative signed terms, we notice that
if $(\omega,m)$ are such that  $m=0$ or $\gamma=\omega/ma \in I_n$ for  $n\ge 2$, then
for all admissible $(\omega, m, \Lambda)$,
\begin{equation}
\label{tobereferredtoearlier}
\iota_{\rm elliptic}(r,\omega, m, \Lambda) (m^2+\omega^2) (V_0 +\omega^2)  \lesssim \iota_{\rm elliptic}(r,\omega, m, \Lambda)   (V_0 -\omega^2)^2
\end{equation}
(with the $\sim$ depending on a bound $|\omega| \gtrsim m$ provided by $\gamma$ in the $m\ne 0$ case).
In the case where $\omega/ma \in I_n$ for $n=1$, we similarly have~\eqref{tobereferredtoearlier},
where we use  that 
$\iota_{\rm elliptic}(r,\omega, m, \Lambda) $ is supported in $r\ge 2.01M$ since 
$2.01M<r_+(\omega, m)$,
and thus 
and we use~\eqref{Vzerocontrols} to bound $m^2\lesssim V_0$ in this region.
Finally, if $n=0$, then~\eqref{tobereferredtoearlier} clearly holds for all admissible $(\omega, m, \Lambda)$ with $(\omega, m)\in \mathcal{F}_0$.

If $\tilde{u}$, $|{\bf \tilde{k}}|=k-1$ are as above, 
then we have
\begin{eqnarray*}
\int_{-\infty}^{R^*_{\rm freq}} \int_{-\infty}^\infty \sum_{m, \ell\ge |m|}  \iota_{\rm elliptic}\Delta(V_0-  \omega^2) |u|^2    dr^* d\omega  &\lesssim& \int_{-\infty}^{R^*_{\rm freq}} \int_{-\infty}^\infty \sum_{m, \ell\ge |m|}   \iota_{\rm elliptic}\Delta (V_0-\omega^2) (m^2+\omega^2) |\tilde{u}|^2 dr^* d\omega \\
&\lesssim &\int_{-\infty}^{R^*_{\rm freq}} \int_{-\infty}^\infty \sum_{m, \ell\ge |m|}  
  \iota_{\rm elliptic}\Delta (V_0-\omega^2)^2 |\tilde{u}|^2 dr^* d\omega \\
&\lesssim &{}^{\rm elp}_{\scalebox{.6}{\mbox{\tiny{\boxed{+1}}}}}\,\Xkminusone(\tau_0,\tau_1)[\psi].
\end{eqnarray*}
This completes the proof.

\subsection{The global boundary coercivity at $\mathcal{S}$: Proof of Proposition~\ref{globalboundarySpositivity}}
\label{fundcoercivityboundapp}

We first consider the boundary term of  $J^{{\rm main},n}$  at $\mathcal{H}^+$.
One can of course analyse this directly using~\eqref{willtakethefollowingformnonsuperbefore}
and~\eqref{willtakethefollowingformsuperbefore}, but given our formalism, it is convenient
to infer it from a horizon computation at the fixed Carter frequency level:
\begin{proposition}[Non-superradiant horizon coercivity]
\label{nshorcoerprop}
Let $e_{\rm red}>0, E>0$ be arbitrary.  Let
$y_n, \alpha_n$ correspond to a non-superradiant range $\mathcal{F}_n$. Then 
for all admissible frequencies $(\omega, m ,\Lambda)$, not necessarily satisfying
$(\omega, m)\in \mathcal{F}_n$,
we have
\begin{eqnarray}
\nonumber
\left({\text Q}^{y_n}[u] + e_{\rm red}{\text Q}^{z}_{\rm red}[u]  - E_n\,  {\text Q}^{T} [u]
  - E_n\alpha_n  \, {\text Q}^{\Omega_1} [u] \right)(r_+) &=&E  B_1 (\omega -\alpha_n m )(\omega-\frac{a}{2Mr_+}m)|u|^2 \\
\nonumber
&&\quad   +B_2 \omega m|u|^2  + B_3 \omega^2|u|^2  +B_4 m^2|u|^2\\
\label{toprovethishorizoncoercivity}
&&\quad +  e_{\rm red}(B_5 \Lambda +B_6)|u|^2
\end{eqnarray}
where $E_n:=E(1-2Mr_+\alpha_n/a)^{-1}$ and
where $B_1, B_5, B_6>0$. In particular,
for admissible $(\omega, m,\Lambda)$ with $(\omega, m)\in \mathcal{F}_n$,
then 
\begin{equation}
\label{toprovethishorizoncoercivityactual}
\left( {\text Q}^{y_n}[u] + e_{\rm red}{\text Q}^{z}_{\rm red}[u]  - E_n{\text Q}^T[u]
  - E_n\alpha_n  \, {\text Q}^{\Omega_1} [u] \right)(r_+) \gtrsim
     E \left(\omega -\frac{a}{2Mr_+} m \right)^2 |u|^2 +e_{\rm red}\Lambda |u|^2
\end{equation}
if $E$ is chosen sufficiently large depending on $e_{\rm red}$.
\end{proposition}
\begin{proof}
The identity~\eqref{toprovethishorizoncoercivity} follows from the definition of the currents
together with the boundary behaviour of $u$ as $r^*\to -\infty$.
On the other hand, recall~\eqref{alphanbound}.
Note that for all nonsuperradiant ranges we have that either $m=0$ or
$\omega/am$ is
bounded away from $0$ and $\frac{a}{2Mr_+}$.
Thus we
have a uniform lower bound
\[
(\omega -\alpha_n m )(\omega-\frac{a}{2Mr_+}m) \gtrsim (\omega-\frac{a}{2Mr_+}m) ^2 + \omega^2 +m^2 
\]
on the distance of the non-superradiant intervals to the threshold.
We may now absorb the terms multiplying $B_2$, $B_3$ and $B_4$ for large enough $E$. 
This yields~\eqref{toprovethishorizoncoercivityactual}.
\end{proof}

Now note that we may re-express the boundary contribution of currents
on $r=\hat{r}$
with respect to the volume form $ r^2 dt^* d\phi \sin\theta d\theta $
and a rescaled normal $\hat{\rm n}_{\rm r=\hat{r}}$ which corresponds to our chosen
normal at the horizon. The coefficients of the quadratic expression
\[
J^{{\rm main},n}  [\psi ]\cdot   \hat{\rm n}_{\rm r=\hat{r}} 
\] 
are then smooth as spacetime functions.
Let us note that when expressed in the coordinate
basis $\partial_r$, $\partial_{t^*}$, $\sin^{-1}\theta \partial_{\phi^*}$, and $\partial_\theta$
of Kerr star coordinates, 
it follows  that 
\begin{eqnarray*}
J^{{\rm main},n}  [\psi ]  \cdot \hat{\rm n}_{r=r_0}   & \ge& 
 E_n{\rm Re}\left((\partial_{t^*}\psi +\alpha_n\partial_{\phi^*}\psi)\overline{(\partial_{t^*}\psi+\frac{a}{2M r_+} \partial_{\phi^*}\psi)}\right)  + B_2{\rm Re} \left( \partial_{t^*}\psi\overline{\partial_{\phi^*} \psi}\right) +B_3| \partial_{t^*} \psi|^2
+B_4|\partial_{t^*} \psi|^2 \\
&&\qquad- C(r_+-r_0) \left(  
|\partial_{t^*}\psi|^2+ \sin^{-2}\theta |\partial_{\phi*}\psi|^2 + |\partial_{\theta}\psi|^2 \right)-C (r_+-r_0)^2 | \partial_r \psi|^2- C |\psi|^2 \\
&&\qquad + ce_{\rm red} \left(  
|\partial_{t^*}\psi|^2+ \sin^{-2}\theta |\partial_{\phi*}\psi|^2 + |\partial_{\theta}\psi|^2 \right)
+ce_{\rm red}(r_+-r_0)|\partial_r \psi|^2.
\end{eqnarray*}

Now from Plancherel and the nonsuperradiant condition, we obtain that for $E$ sufficiently
large
\begin{eqnarray*}
\int_{\mathcal{S}} 
J^{\rm main} [\psi_{n,{\bf k}} ] \cdot \hat{\rm n}_{r=r_0} r^2 dt^* d\phi \sin\theta d\theta 
&\ge&\int_{\mathcal{S}}\Big( cE\left( | \partial_{t^*} \psi_{n,{\bf k}}|^2+ | \partial_{\phi^*} \psi_{n,{\bf k}}|^2 \right)
\\ 
&&\qquad + ce_{\rm red}   \left( |\partial_{t^*}\psi_{n,{\bf k}}|^2+ \sin^{-2}\theta |\partial_{\phi*}\psi_{n,{\bf k}}|^2 + |\partial_{\theta}\psi_{n,{\bf k}}|^2 \right) \\
&&\qquad 
+ce_{\rm red}(r_+ - r_0) |\partial_r\psi_{n,{\bf k}}|^2 -C|\psi_{n,{\bf k}}|^2\Big)  \, r^2 dt^* d\phi \sin\theta \,d\theta .
\end{eqnarray*}

This yields~\eqref{globalboundpos} for $r_0<r_+$ sufficiently close to $r_0$,
upon rewriting it so as to
replace $\hat{\rm n}_{r=r_0}$ with ${\rm n}_{r=r_0}$ 
 and the volume form with the induced
volume form. The second term on the right hand side of~\eqref{globalboundpos} arises from bounding the integral
of  $|\psi_{n,{\bf k}}|^2$ by the flux term on $\mathcal{S}$.

\bibliographystyle{DHRalpha}
\bibliography{notelargearefs}

\end{document}